\documentclass[12pt]{amsart}
\usepackage{inputenc}
\usepackage{geometry}
\usepackage{color}
\usepackage{amsmath}
\usepackage{amsbsy}
\usepackage{amstext}
\usepackage{amsthm}
\usepackage{amssymb}
\usepackage{graphicx}
\usepackage{setspace}
\usepackage{esint}
\usepackage{hyperref}
\hypersetup{unicode=true,
 bookmarks=false,
 breaklinks=false,pdfborder={0 0 1},backref=section,colorlinks=false}
 \usepackage{epsfig}
 \usepackage{float}

\makeatletter
\numberwithin{equation}{section}
\numberwithin{figure}{section}
\theoremstyle{plain}
\newtheorem*{conjecture*}{\protect\conjecturename}
\theoremstyle{plain}
\newtheorem{thm}{\protect\theoremname}[section]
\theoremstyle{remark}
\newtheorem{rem}[thm]{\protect\remarkname}
\theoremstyle{definition}
\newtheorem{defn}[thm]{\protect\definitionname}
\theoremstyle{plain}
\newtheorem{lem}[thm]{\protect\lemmaname}
\theoremstyle{plain}
\newtheorem{conjecture}[thm]{\protect\conjecturename}
\theoremstyle{plain}
\newtheorem{prop}[thm]{\protect\propositionname}
\theoremstyle{plain}
\newtheorem{cor}[thm]{\protect\corollaryname}

\@ifundefined{date}{}{\date{}}
\usepackage{amscd}
\usepackage{comment}
\usepackage{mathrsfs}
\usepackage{setspace}
\usepackage{color}
\usepackage{url}
\usepackage{hyperref}
\usepackage[T1]{fontenc}

\providecommand{\conjecturename}{Conjecture}
\providecommand{\corollaryname}{Corollary}
\providecommand{\definitionname}{Definition}
\providecommand{\lemmaname}{Lemma}
\providecommand{\propositionname}{Proposition}
\providecommand{\remarkname}{Remark}
\providecommand{\theoremname}{Theorem}

\makeatother

\providecommand{\conjecturename}{Conjecture}
\providecommand{\corollaryname}{Corollary}
\providecommand{\definitionname}{Definition}
\providecommand{\lemmaname}{Lemma}
\providecommand{\propositionname}{Proposition}
\providecommand{\remarkname}{Remark}
\providecommand{\theoremname}{Theorem}

\begin{document}
\global\long\def\av{\vec{\alpha}}%
\global\long\def\ztrip{\left(\boldsymbol{z}_{1},z_{e},\boldsymbol{z}_{2}\right)}%
\global\long\def\ztwo{\boldsymbol{z}_{2}}%
\global\long\def\zi{\boldsymbol{z}_{i}}%
\global\long\def\zone{\boldsymbol{z}_{1}}%
\global\long\def\ze{z_{e}}%
\global\long\def\zv{\vec{z}}%
\global\long\def\norm#1{\left\Vert #1\right\Vert }%
\global\long\def\I{\mathcal{I}}%

\global\long\def\A{\mathcal{A}}%
\global\long\def\tr{\mathrm{trace}}%

\global\long\def\G{\mathcal{G}}%
\global\long\def\adj{adj}%

\global\long\def\M{\mathcal{M}}%
\global\long\def\LocInd#1{\text{\ensuremath{\mathcal{M}_{\left(#1\right)}}}}%
\global\long\def\LocHess#1{\text{\ensuremath{\mathrm{Hess}}}_{\left(#1\right)}}%

\global\long\def\opH{\mathcal{H}}%
\global\long\def\L{\mathcal{L}}%

\global\long\def\D{\mathcal{D}}%

\global\long\def\U{\mathcal{U}}%

\global\long\def\I{\mathcal{I}}%

\global\long\def\P{\mathrm{P}}%

\global\long\def\Rev{\mathcal{R}_{v,e}}%

\global\long\def\d{\partial}%

\global\long\def\dg{\partial\Gamma}%

\global\long\def\do{\partial\Omega}%

\global\long\def\ndv{\Omega_{n}^{(v)}}%

\global\long\def\E{\mathcal{E}}%

\global\long\def\V{\mathcal{V}}%

\global\long\def\Vint{\mathcal{V}\setminus\partial\Gamma}%

\global\long\def\Vin{V_{\textrm{in}}}%

\global\long\def\la{\lambda}%

\global\long\def\hess{\mathrm{Hess}}%

\global\long\def\H{H}%

\global\long\def\Z{\mathbb{Z}}%

\global\long\def\R{\mathbb{R}}%

\global\long\def\C{\mathbb{C}}%

\global\long\def\N{\mathbb{N}}%

\global\long\def\Q{\mathbb{Q}}%

\global\long\def\lap{\Delta}%

\global\long\def\na{\nabla}%
\global\long\def\floor#1{\left\lfloor #1\right\rfloor }%

\global\long\def\opcl#1{\left(#1\right]}%

\global\long\def\clop#1{\left[#1\right)}%

\global\long\def\bs#1{\boldsymbol{#1}}%

\global\long\def\deg#1{\mathrm{deg}(#1)}%

\global\long\def\T{\mathbb{T}}%

\global\long\def\TE{\mathbb{T^{\left|\E\right|}}}%

\global\long\def\BGm{\mu_{\vec{l}}}%

\global\long\def\msing{\Sigma^{\mathrm{sing}}}%

\global\long\def\mreg{\Sigma^{\mathrm{reg}}}%

\global\long\def\mgen{\Sigma^{\mathrm{gen}}}%

\global\long\def\lv{\vec{l}}%

\global\long\def\ts{t_{S}}%

\global\long\def\Lv{\vec{L}}%

\global\long\def\tv{\vec{\theta}}%

\global\long\def\kv{\vec{\kappa}}%

\global\long\def\xv{\vec{x}}%

\global\long\def\Tv{\vec{\kappa}}%

\global\long\def\dL{d_{\vec{L}}}%

\global\long\def\sgn{\mathrm{sgn}}%

\global\long\def\bs#1{\boldsymbol{#1}}%

\global\long\def\undercom#1#2{\underset{_{#2}}{\underbrace{#1}}}%
\global\long\def\set#1#2{\left\{  #1\,\,:\,\,#2\right\}  }%
\pagestyle{empty}
\pagenumbering{gobble}
\begin{center}
\textbf{\Huge{}Quantum graphs - }\\
\textbf{\Huge{} Generic eigenfunctions and their nodal count and Neumann count statistics}{\Huge\par}
\par\end{center}

%\title{{\Huge{}Quantum graphs - }\\
%{\Huge{} Generic eigenfunctions and their nodal count and Neumann count statistics}}
%\maketitle
\begin{center}
{\LARGE{}\vspace{9.1cm}
}{\LARGE\par}
\par\end{center}

\begin{center}
{\LARGE{}Lior Alon}\cleardoublepage\textbf{\Huge{}Quantum graphs - }\\
\textbf{\Huge{} Generic eigenfunctions and their nodal count and Neumann count statistics}{\Huge\par}
\par\end{center}

\begin{center}
{\LARGE{}\vspace{2.5cm}
}{\LARGE\par}
\par\end{center}

\begin{center}
{\LARGE{}Research Thesis\vspace{1cm}
}{\LARGE\par}
\par\end{center}

\begin{center}
{\LARGE{}In Partial Fulfillment of The}{\LARGE\par}
\par\end{center}

\begin{center}
{\LARGE{}Requirements for the Degree of}{\LARGE\par}
\par\end{center}

\begin{center}
{\LARGE{}Doctor of Philosophy\vspace{2.5cm}
}{\LARGE\par}
\par\end{center}

\begin{center}
{\LARGE{}Lior Alon\vspace{2.5cm}
}{\LARGE\par}
\par\end{center}

\begin{center}
{\LARGE{}Submitted to the Senate of the }{\LARGE\par}
\par\end{center}

\begin{center}
{\LARGE{}Technion - Israel institute of Technology\vspace{1cm}
}{\LARGE\par}
\par\end{center}

\begin{center}
{\LARGE{}Av 5780, Haifa, August 2020}{\LARGE\par}
\par\end{center}

\cleardoublepage{}

\textcolor{white}{\Large{}
\[
\;
\]
}

\begin{center}
{\large{}\vspace{2cm}
To my mother }{\large\par}
\par\end{center}

\begin{center}
{\Huge{} Tami Alon \large{Z"L}}\\
\par\end{center}

\cleardoublepage{}

\textcolor{white}{\Large{}
\[
\;
\]
}

\begin{center}
{\large{}\vspace{2cm}
The Research Thesis Was Done Under The Supervision of }{\large\par}
\par\end{center}

\begin{center}
{\large{}Associate Prof. Ram Band in The Faculty of Mathematics.}\\
\par\end{center}
{\LARGE{}\vspace{5cm}
}{\LARGE\par}
{\normalsize{}Some results in this thesis have been published as articles
by the author together with collaborators:} 
\begin{enumerate}
\item L. Alon, R. Band, and G. Berkolaiko, \textit{Nodal statistics on quantum graphs}, Communications in Mathematical Physics, (2018).
\item L. Alon, R. Band, M. Bersudsky, and S. Egger, \textit{Neumann domains on graphs and manifolds}, Analysis and geometry on graphs and manifolds, 461 (2020), p.203. 
\end{enumerate}

\begin{center}
{\LARGE{}\vspace{3cm}
}{\LARGE\par}
\par\end{center}

\begin{center}
{\small{}The Generous Financial Help of the Technion, the Irwin and Joan Jacobs Fellowship, and the Ruth and Prof. Arigo Finzi Fellowship is Gratefully Acknowledged.}\\
{\large{}\vspace{3cm}
}{\large\par}
\par\end{center}
\newpage{}

\textcolor{white}{\Large{}
\[
\;
\]
}

\begin{center}
{\large{}
Acknowledgements }{\large\par}
\par\end{center}

\normalsize{}\vspace{2cm} This dissertation is the last milestone of my Ph.D. Journey. Throughout this Journey I have received a great deal of support and assistance for which I am thankful.\\ 

I would first like to thank my Ph.D. advisor, Rami Band, for his guidance and support, for believing in me when I had doubts, and for his ability to teach me so much and in the same time acknowledge and appreciate the ideas I bring with me. Thank you Rami for your insightful and uncompromising feedback along every step of the way, you've helped me grow and pushed my work to higher levels.\\

I would also like to thank Gregory Berkolaiko, my non-formal co-advisor, colleague and friend. Thank you for your advice and thank you for always being honest and straightforward. \\

I would like to thank Yehuda Pinchover and Uzy Smilansky for reading and commenting on my thesis. I would also like to thank Uzy for many insightful conversations, for teaching me physics, mathematics and history all together, and most importantly, for setting the ground on which my work stands.\\

I would also like to thank the mathematics faculty of the Technion for being my second home in the last eight years. I would like to thank the administrative staff, and Anat in particular, for all their caring and support. I would also like to thank the faculty members, for their willingness to help, give advice, teach and discuss mathematics beyond any formal course or office hours. In particular, Amos Nevo, Uri Shapira, Dani Neftin, Orr Shalit, Ron Rosenthal, Tali Pinsky and Nir Lazarovitch.\\  

Finally, I would like to thank my family, my father and sisters, who supported me during the hard times. I would like to thank my son, Itamar, whose arrival (two and half years ago) had given me new hopes and new purpose in life.\\
Last but not least, I want to thanks my wife, Adi, for her infinite support. Thank you for believing in me, and for not allowing me to stop believing in myself. Non of this would have been possible without you by my side.\\ 

This thesis is dedicated to the memory of my mother, Tami Alon, who left us eight years ago and did not get the opportunity to see me pursuing my dream.

\cleardoublepage{}

\begin{doublespace}
\tableofcontents{}
\end{doublespace}

\cleardoublepage{}
\pagestyle{empty}
\listoffigures

\cleardoublepage{}

\pagestyle{plain} 
\pagenumbering{arabic}
\setcounter{page}{1}

\section*{Abstract}
In this thesis, we study Laplacian eigenfunctions on metric graphs,
also known as quantum graphs. We restrict the discussion to standard
quantum graphs. These are finite connected metric graphs with functions that satisfy Neumann vertex conditions. 

The first goal of this thesis is the study of the nodal count problem.
That is the number of points on which the $n$th eigenfunction vanishes.
We provide a probabilistic setting using which we are able to define the nodal count\textquoteright s
statistics. We show that the nodal count\textquoteright s
statistics admits a topological symmetry by which the first Betti number of the graph can be obtained. This result generalizes a result by which the nodal count is 0,1,2,3... if and only if the graph is a tree. We revise a conjecture that predicts a
universal Gaussian behavior of the nodal count\textquoteright s statistics for large graphs,
and prove it for a certain family of graphs which we call \textquoteleft trees
of cycles\textquoteright.

The second goal is to formulate and study a new closely related counting problem which we call the Neumann count,
in which one counts the number of local extrema of the
$n$th eigenfunction. This counting problem is motivated by the Neumann partitions of planar domains, a novel concept in spectral geometry. We provide uniform bounds on the Neumann count and investigate the Neumann count\textquoteright s
statistics using our probabilistic setting. We show that
the Neumann count\textquoteright s statistics admits a symmetry by which the number of leafs of the graph can be obtained. In particular, we show that the Neumann count provides a complementary geometrical information to that obtained from the nodal count. We show that for a certain family of tree graphs the Neumann count\textquoteright s statistics can be calculated explicitly and it approaches a Gaussian distribution for large enough graphs, similarly to the nodal count conjecture.

The third goal is a genericity result, which justifies the generality
of the Neumann count discussion. To this day it was known that generically, eigenfunctions do not vanish on vertices. We generalize this result to derivatives at vertices as well. That is, generically, the derivatives of an eigenfunction on interior vertices do not vanish.

\newpage{}
\section*{List of Symbols}
\begin{tabular}{ll}
    $\Gamma$ & A discrete graph\\
    $\E,\V$ & The sets of edges and vertices of $\Gamma$ \\
    $E, V$ & The number of edges and the number of vertices of $\Gamma$ \\
    $\beta$ & The first Betti number of $\Gamma$ \\
    $\vec{\E}$ & The set of directed edges of $\Gamma$ \\
    $\partial\Gamma, \V_{in}$ & The boundary vertices (leafs) of $\Gamma$ and the set of interior vertices in $\Gamma$\\
    $\E_{v}$ & The set of edges connected to a vertex $v$\\
    $\Gamma_{\lv{}}$ & A standard quantum graph with edge lengths $\lv{}$\\
    $L$ & The total length of $\Gamma_{\lv{}}$\\
    $\partial_{e}f(v)$ & The outgoing derivative of $f$ at $v$ in the direction of $e\in\E_{v}$\\
    $f_{n}, k_{n}$ & The $n$th eigenfunction of $\Gamma_{\lv{}}$ and its (square root) eigenvalue\\
    $\G, \L$ & The index sets of generic eigenfunctions, and of loop-eigenfunctions\\
    $d(A)$ & The natural density of an index set $A$\\
    $\phi(n)$ & The nodal count of the $n$th eigenfunction\\
    $\sigma(n)$ & The nodal surplus of the $n$th eigenfunction, $\phi(n)-n$ \\    
    $\mu(n)$ & The Neumann count of the $n$th eigenfunction\\
    $\omega(n)$ & The Neumann surplus of the $n$th eigenfunction, $\mu(n)-n$ \\
    $\Omega^{v}$ & The Neumann domain containing $v$\\
    $N(\Omega)$ & The spectral position of the Neumann domain $\Omega$\\
    $\rho(\Omega)$ & The wavelength capacity of the Neumann domain $\Omega$\\
    $Eig(\Gamma_{\lv{}},k^2)$ & The $k^2$ eigenspace of $\Gamma_{\lv{}}$\\
    $\T^{\E}$ & The characteristic torus of $\Gamma$, $(\R/2\pi\Z)^{\E}$\\
    $\{*\}$ & The quotient map from $\R^{\E}$ to $\T^{\E}$, $\{\xv\}:=\xv\  mod\  2\pi$\\
    $\Gamma_{\kv{}}$ & The standard quantum graph associated to $\kv\in\T^{\E}$\\
    $U_{\kv{}}$ & The unitary evolution matrix associated to $\kv\in\T^{\E}$\\
    $F(\kv)$ & The secular function of $\Gamma$\\
    $\Sigma,\Sigma^{reg}$ & The secular manifold of $\Gamma$ and its regular part\\
    $f_{\kv{}}$ & The canonical eigenfunction associated to $\kv\in\Sigma^{reg}$\\
    $\mu_{\lv{}}$ & The Barra-Gaspard measure associated to $\lv$\\

\end{tabular}

\newpage{}

\section{Introduction}

The following thesis lies in the mathematical field of \emph{spectral
geometry, }but can be regarded also as a work in the field of \emph{quantum
chaos.} In the following section we provide the needed context for
our main results. We first review the field of quantum chaos, as the
motivation of our research, after which we briefly present the aspects
of spectral geometry relevant to the subjects of our work: \emph{nodal
count, Neumann count }and \emph{genericity.} We then present each
subject, first describing the known results on manifolds for comparison
and motivation, then known results for quantum graphs, following which
we present and discuss our new result. But first, let us introduce
quantum graphs.

\subsection*{Quantum graphs}

\emph{A Quantum Graph} is a model for a quantum particle on a network.
Mathematically, a \emph{quantum graph} is a metric graph, a 1-d simplicial
complex, equipped with a differential operator (usually a Schr\"{o}dinger
operator). This model was introduced in the 30's by Pauling \cite{Pau_jcp36}
to describe free electrons of organic molecules, and was further developed
in the 50's by Ruedenberg and Scherr \cite{RueSch_jcp53} that considered
quantum graphs as an idealization of a network of wires of very small
cross-section. For modern analysis of the zero cross-section limit
see \cite{Rub_incol06} and \cite{ExnerPost2008a}. The list of successful
applications of quantum graphs in the study of complex phenomena include
superconductivity in granular and artificial materials \cite{Ale_prb85},
Anderson localization \cite{And_pmb85}, electromagnetic waveguide
networks \cite{ExnPos_jmp07,MR2738109} and nanotechnology \cite{ExnSeb_incol89}
to name but a few. The name `quantum graph' was first coined in the
late 90's by Smilansky and Kottos \cite{KotSmi_prl97} in their work
on quantum chaos. Following their work and subsequent works, such
as \cite{BarGas_jsp00,BarGas_pre02,BerBogKea_jpa01,BerKea_jpa99},
quantum graphs gained popularity as models for quantum chaos. A thorough
introduction to quantum graphs and their applications can be found
in the following (partial list of) reviews on the subject \cite{BerKuc_graphs,GnuSmi_ap06,Kuc_incol08}. 

\subsection*{Quantum chaos}

Quantum chaos is a field of research in physics that studies the relation
between quantum mechanics and classical (Hamiltonian) mechanics, through
the scope of chaos. For a reader not familiar with these notions,
here is a brief description: 

\subsubsection*{Quantum versus classical (Hamiltonian) mechanics in a nutshell:}

Classical Hamiltonian mechanics, the modern description of Newtonian
mechanics, describes the dynamics of one or many particles on a domain $\Omega$,
for simplicity assume $\Omega$ is a domain in $\R^{n}$. The dynamics of a particle is described by its position $q\in\Omega$ and momentum $p\in\R^{n}$.
The \emph{phase space,} $\Omega\times\R^{n}$, is the space of all
pairs $\left(q,p\right)$, and the physical setting is encoded in
the\emph{ Hamiltonian, }$H\left(q,p\right)=\norm p^{2}+V\left(q\right)$
which describes the energy of a particle at $\left(q,p\right)$. Hamilton's
equations $\frac{d}{dt}\left(q,p\right)=\left(\frac{\partial H}{\partial_{p}},-\frac{\partial H}{\partial_{q}}\right)$
describe the particle's dynamics in phase space\footnote{This is a simplified description. In general $\Omega$ can be a manifold,
possibly with boundary, the phase space is the tangent bundle over
$\Omega$, and Hamilton's equations can be described using Poisson
brackets. The Hamiltonian function, in general, can be time dependent
and may include other terms.}. In the simple case of $V\equiv0$, the Hamiltonian is $H=\norm p^{2}$
and so $p$ is piece-wise constant by Hamilton's equations,
with discontinuities at the boundary. 

In quantum mechanics, the phase space $\Omega\times\R^{n}$ is replaced
by the Hilbert space $L^{2}\left(\Omega\right)$, and the pair $\left(q,p\right)$
is replaced by a wave-function $f\in L^{2}\left(\Omega\right)$. The
position and momentum of the particle are no longer deterministic
and are given on average by $\left\langle \boldsymbol{q}\right\rangle =\left\langle f,\boldsymbol{q}f\right\rangle $
and $\left\langle \boldsymbol{p}\right\rangle =\left\langle f,\boldsymbol{p}f\right\rangle $.
Where $\left\langle *,*\right\rangle $ denotes the $L^{2}$ inner
product, $\boldsymbol{q}$ is a multiplicative operator and $\boldsymbol{p}=i\hbar\nabla$
is the derivative operator scaled by a constant $\hbar$. The classical
Hamiltonian $H\left(p,q\right)$ is upgraded to a self-adjoint differential
operator $\boldsymbol{H}=H\left(\boldsymbol{q},\boldsymbol{p}\right)$.
In particular the term $\norm p^{2}$ is upgraded to the (positive)
rescaled Laplacian $\hbar^{2}\Delta$. In general, the operators $\boldsymbol{q}$,
$\boldsymbol{p}$ and $\boldsymbol{H}$ are unbounded and one should
specify the domain $D_{\boldsymbol{H}}\subset L^{2}\left(\Omega\right)$
on which $\boldsymbol{H}$ is (weakly) defined and is self-adjoint.
The dynamics of the wave-function, namely the time evolution of a
wave function $f_{t}$ at time $t\in\R$, is according to \emph{Schr\"{o}dinger's
equation} $i\hbar\frac{d}{dt}f_{t}=\boldsymbol{H}f_{t}$. A standard
separation of variables usually reduce the problem to the ``\emph{stationary
Schr\"{o}dinger's equation''} $\boldsymbol{H}f=\lambda f$.

\subsubsection*{Classical chaos:}

A standard classification of classical systems distinguishes between
\emph{chaotic }and \emph{integrable }systems. Systems where certain
symmetries and constants of motion can reduce the number of degrees
of freedom are called integrable. From a dynamical point of view,
a system is integrable if it has a maximal number of integrals of motion such that the phase space can be foliated by the level sets of $H$ and the integrals of motion. Chaotic systems are, in a sense, as far from integrable
as possible. These are systems with dense trajectories in phase space,
which are extremely sensitive to perturbations. That is, the distance
between two close trajectories grows exponentially with time. Chaotic
systems are very hard to investigate, thus in the study of chaos even
the simplest systems that can exhibit chaotic behavior are of interest.
A simple study case of chaotic behavior is that of billiard domains:
A free particle ( $H=\norm p^{2}$) on a compact planar domain $\Omega$,
that bounces (symmetrically) when hitting the boundary. The classification
of a billiard as integrable or chaotic is dictated by the shape of
$\Omega$. Rectangles and ellipses are integrable, but ``simple''
chaotic billiards can be constructed, for example Sinai's billiard,
a square with a round hole in the middle.

\subsubsection*{Quantum chaos:}

The starting point of \emph{quantum chaos} is the quantum (miss) behavior
of systems that are classically chaotic. Einstein's Theory of Special Relativity, in the limit of $c\rightarrow\infty$, provides the same predictions as classical Newtonian mechanics. One should expect the same from a Quantum Mechanics Theory. In the beginning of the formalizing process of Quantum Mechanics, Bohr introduced the correspondence principle, stating that at a certain scale, the Quantum Mechanics predictions should agree with the classical mechanics predictions. After Schr\"{o}dinger's probabilistic formalism of Quantum Mechanics was introduced, the correspondence principle was reformulated to state that in the $\hbar\rightarrow0$ limit the quantum expected values of position and momentum should behave according to the classical mechanics predictions. However, unlike Special Relativity, it appears that the quantum dynamics predictions may not converge to the classical predictions, in the $\hbar\rightarrow0$ limit, if the system is classically chaotic. For example, consider the phase-space trajectories $\left(q(t),p(t)\right)$ which governs the  chaotic behaviour. The quantum
expected values of position and momentum, $\left\langle \boldsymbol{q}\right\rangle $
and $\left\langle \boldsymbol{p}\right\rangle $, have widths of uncertainty $\delta\boldsymbol{q}$ and $\delta\boldsymbol{p}$ which
 obey the \emph{uncertainty
principle} $\delta\boldsymbol{q}\delta\boldsymbol{p}\ge\hbar$. Due to the uncertainty constrain, different trajectories of quantum expected values cannot be distinguished if the spacing between them is of order much smaller than $\hbar$ and so chaos in its classical sense loses its meaning. It is believed
that this fundamental problem of correspondence can shed light on
the very nature of Quantum Mechanics. The $\hbar\rightarrow0$ limit of a quantum system is called semi-classical limit. 

One definition for quantum chaos was presented in the Bakerian lecture
1987 by M.V. Berry \cite{berry1987bakerian}:

\textbf{Definition}: \emph{``Quantum chaology} is the study of semiclassical,
but nonclassical, behaviour characteristic of systems whose classical
motion exhibits chaos.'' 

These ``behavior characteristics'' that we wish to study are properties
of high eigenvalues and their eigenfunctions (to capture the semiclassical
regime where $\hbar\ll1$) that can distinguish chaotic from integrable.
It is believed that such properties should be \textbf{universal},
namely, insensitive to the details of the specific system. The most
famous example of such a property is \emph{spectral statistics}, the
statistical behavior of fluctuations of eigenvalues $\lambda_{n}$
around their asymptotic growth predicted by Weyl's law. For example,
in planar domains the asymptotic growth is linear $\lambda_{n}\sim\frac{4\pi}{\left|\Omega\right|}n$.
The spectral statistics behave differently for integrable systems
and chaotic systems. It was proven in 1977 by Berry and Tabor \cite{BerTab_prsl77}
that the spectrum of an integrable system has level spacing statistics
corresponding to a Poisson process. The chaotic case, by nature, is
much harder to analyze. The famous \emph{BGS conjecture} by Bohigas,
Giannoni and Schmidt \cite{BohGiaSch_prl84} (which followed a nuclear
physics folklore\footnote{The folklore was that the spectrum of complicated quantum systems,
like electrons of a very large atoms, can be well predicted by the
spectrum of a suitable random matrix.}), states that the spectral statistics of chaotic systems can be predicted
by the spectral statistics of a corresponding random matrix ensemble.\emph{
}This conjecture was a significant milestone in the field of \emph{Random
matrix theory} (RMT). The BGS conjecture was affirmed numerically
in many chaotic models, together with related works that provided
analytical supporting evidences. See \cite{BGSreview} for a 2016
overview. It is now widely accepted that RMT spectral statistics is
an indication for quantum chaos. 

\subsubsection*{Quantum chaos on quantum graphs }

Kottos and Smilansky \cite{KotSmi_prl97} provided numerical evidence
for RMT spectral statistics in quantum graphs with edge lengths linearly
independent over $\Q$ (we call them \emph{rationally independent}).
This was the first evidence for ``chaotic fingerprints'' in quantum
graphs. They also provided an exact trace formula for quantum graphs, similar in nature to Gutzwiller's trace formula. Gutzwiller's trace formula \cite{Gutzwiller1971}
(following the works of Weyl, Selberg, Krein and Schwinger) is the
main tool relating the spectrum of a quantum system to periodic orbits
in the classical phase space. It is an approximation of the spectral
density of a quantum system in the semmiclassical limit by means of
the Hamiltonian action on periodic orbits. Unlike Gutzwiller's trace
formula which has an error term that vanishes in the $\hbar\rightarrow0$
limit, the quantum graph's trace formula is exact without error terms.
It is now a common belief that a single (finite) quantum graph is
not enough to properly model quantum chaos, but in the limit of large
graphs, it is a good paradigm for quantum chaos. Barra and Gaspard
\cite{BarGas_jsp00} provided an implicit analytic formula for the
level spacing distribution of quantum graphs with rationally independent
edge lengths. Their work shows that the statistics of a single graph
has a small but not neglectable deviation from the RMT statistics.
However, they noticed, numerically, that this deviation was independent
of the choice of edge lengths (under the rationality condition) and
that the deviation decreases as the graph grows. Further works on
quantum chaos on quantum graphs are \cite{BerBogKea_jpa01,BerKea_jpa99,BerKeaSmi_cmp07,GnuAlt_pre05,GnuSmi_ap06,Kea_incol08}
for example. 

Another behavior characteristic of chaos in quantum graphs, suggested
by \cite{GnuSmiWeb_wrm04}, is the \emph{nodal statistics}. This is
the subject of our work \cite{AloBanBer_cmp18}. To elaborate, let
us first introduce the \emph{nodal count problem} from the scope of
\emph{spectral geometry.} 

\subsection*{Spectral geometry }

Spectral geometry aims to study the relations between spectral properties
of differential operators (usually the Laplacian) and the geometric
structure of the space on which they act (usually a Riemannian manifold).
The term ``spectral properties'' is not confined to properties of
the spectrum alone, but also to the ``landscape'' of eigenfunctions.
A popular theme of spectral geometry is inverse problems. That is,
what geometrical information can be recovered from spectral properties.
Such a question was famously popularized by Mark Kac, asking `Can
one hear the shape of a drum?' in \cite{Kac66}. While Milnor provided
a counter example for isospectral sixteen dimensional manifolds in
\cite{Milnor64}, the question for planar domains held open for three
decades until 1992 when a counter example of isospectral planar domains
was found by Gordon, Webb and Wolpert in \cite{GordonWebbWolpert92}
based on ideas from Sunada's theory \cite{Sun_am85}. Classifying
planar domains that have unique spectrum is still an active topic,
with recent works to this day such as \cite{HezariZelditch19}. A
2014 survey is found in \cite{zelditch14survey}. It was only natural
to look for other spectral properties, those related to the ``landscape''
of eigenfunctions, to resolve isospectrality. Smilansky, Gnutzmann
and Sondergaard conjectured in \cite{GnuSmiSon_jpa05} that the \emph{nodal
count }(which will be presented next) would resolve isospectrality
for planar domains. In a following work of Gnutzmann and Smilansky
with Karageorge \cite{GnuKarSmi_prl06}, a nodal count ``trace formula''
is provided for certain families of planar domains, using which one
can reconstruct these drums. It was affirmed that the nodal count
can solve isospectrality in certain settings, as seen in \cite{bruning2007comment,bruning2007nodal}
for example, but counter examples were given in \cite{BruFaj_cmp12},
and the general validity of this conjecture is still open.

\subsubsection*{Isospectrality on quantum graphs}

The isospectral problem is a good example for a quantum graphs problem
arising from manifolds and planar domains. The question `can one hear
the shape of a graph?' was asked by Gutkin and Smilansky in \cite{GutSmi_jpa01}.
They showed that any simple graph with rationally independent edge
lengths has a unique spectrum. The meaning of such result is that
generically ``one can hear the shape of a graph''. They also provided
an algorithm to reconstruct such graphs from their spectrum, and gave
a counter example of isospectral quantum graphs that do not have rationally
independent edge lengths. This work led to construct more isospectral
graphs \cite{BanParBen_jpa09,BanShaSmi_jpa06,ParBan_jga10} together
with a generalization of Sunada's method to graphs. The conjecture
raised in \cite{GnuSmiSon_jpa05} and the work in \cite{GnuKarSmi_prl06}
on the resolution of isospectrality by the nodal count led to similar
works on quantum graphs. It was affirmed under certain settings \cite{BanOreSmi_incoll08,BanShaSmi_jpa06,BanSmi_epj07,Ore_jpa07}
but counter examples were given in \cite{OreBan_jpa12,JuuJoy_jphys18},
and the general validity of this conjecture is still open. In this
thesis we prove that the first Betti number of a graph (a topological
characterization) can be obtained from its nodal count sequence (a
result from \cite{AloBanBer_cmp18}). In particular, this result implies
that graphs of different Betti number can be distinguished by their
nodal count. 

\subsection*{Nodal count}

Given an eigenfunction of a manifold or a planar domain $\Omega$,
the \emph{nodal partition} is a partition of $\Omega$ according to
the nodal lines (the zero set) of the eigenfunction.\emph{ }The \emph{nodal
domains }are the connected components of this partition, and are the
largest connected subdomains on which the eigenfunction has a fixed
sign\emph{. The nodal count, }is the number of nodal domains. Given
a sequence of eigenfunctions $\left\{ f_{n}\right\} _{n\in\N}$ that
span $L^{2}\left(\Omega\right)$, arranged according to their eigenvalues,
we obtain a nodal count sequence. We denote the nodal count of $f_{n}$
by $\phi\left(n\right)$. The works of Albert \cite{Alb_psp71,Albert_thesis72}
and Uhlenbeck \cite{Uhl_ajm76} assures that generically, nodal lines
(zero sets) of eigenfunctions are of co-dimension one and therefore
partition $\Omega$. Uhlenbeck also showed that eigenvalues are generically
simple. Therefore the nodal count sequence, generically, is well defined
and independent of the choice of basis. The motivation for nodal count
goes back to physical experiments from the 17th century, done by DaVinci
\cite{daVinci}, Galileo \cite{Galileo} and Hooke \cite{Hooke},
later to be further developed by Chladni \cite{Chladni} in the 18th
century (probably using his skills both as a physicist and a musician).
In what is now known as ``Chladni figures'' the vibration patterns
of sound waves are visualized by spreading sand on a brass plate which
is then brought to different resonances using a violin bow. The sand
accumulates into the non-vibrating parts of the plate, forming the
figure of the nodal lines. 

In dimension one, Sturm's oscillation theorem \cite{Stu_jmpa36} states
that $f_{n}$ will have $n-1$ nodal points (zeros). The first generalization
of nodal count to planar domains and manifolds was done by Courant
\cite{Cou_ngwg23} in 1923. The famous \emph{Courant bound }is $\phi\left(n\right)\le n$.
The problem of whether there are eigenfunctions for which $\phi\left(n\right)=n$
was addressed by Pleijel who showed that $\phi\left(n\right)=n$ can
occur only finitely many times, by proving that $\limsup_{n\rightarrow\infty}\frac{\phi\left(n\right)}{n}\le c<1$
in \cite{Ple_cpam56}. This asymptotic bound is known as \emph{Pleijel's
bound,} where $c\approx0.691...$ is given explicitly in terms of the first zero of the zeroth Bessel function. Bourgain and Steinerberger \cite{Bourgain_imrn15,Steinerberger_ahp14}
showed that $c$ is not optimal (improving the bound by order of $10^{-9}$).
More Pleijel-like bounds can be found in \cite{Polterovich09,lena2019pleijel,charron2018pleijel}.

Both Courant's bound and Pleijel's bound, together with many other
results on nodal count, are based on the following observation. If
$f$ is an eigenfunction on $\Omega$ with eigenvalue $\lambda$ and
its nodal domain are denoted by $\left\{ \Omega_{j}\right\} _{j=1}^{N}$,
then the restriction $f|_{\Omega_{j}}$ to a nodal domain $\Omega_{j}$
is the first eigenfunction of the Dirichlet problem on $\Omega_{j}$.
In particular if $\lambda_{1}\left(\Omega_{j}\right)$ denotes the
first Dirichlet eigenvalue of $\Omega_{j}$, then $\lambda_{1}\left(\Omega_{j}\right)=\lambda$
for all $j$. This is a special property of the nodal partition of
an eigenfunction. A variational characterization of nodal partitions
was given in \cite{AncHelfHof_doc04,BHH09,Hel_sem07}. They considered
all partitions of $\Omega$ into $N$ subdomains $\left\{ \Omega_{j}\right\} _{j=1}^{N}$
and considered $\lambda=\max_{j}\lambda_{1}\left(\Omega_{j}\right)$
as a functional over these partitions. It appeared that the nodal
partitions were critical points of $\lambda$, and minimum in the
case of $\phi\left(n\right)=n$. This result led to a characterization
of $\phi\left(n\right)-n$, called \emph{nodal deficiency,} as a Morse
index of the functional $\lambda$ under certain variations \cite{BerKucSmi_gafa12}.

The number of nodal domains, is not the only generalization of Sturm's
oscillations to higher dimensions. Another generalization of the ``number
of zeros'' for a $d$ dimensional manifold is $\mathcal{H}^{d-1}\left(f^{-1}\left(0\right)\right)$,
the $d-1$ dimensional Hausdorff measure of the nodal set of an eigenfunction
$f$. S.T. Yau famously conjectured that $c\sqrt{\lambda}\le\mathcal{H}^{d-1}\left(f^{-1}\left(0\right)\right)\le C\sqrt{\lambda}$
for any eigenfunction $f$ of eigenvalue $\lambda$ with some system
dependent constants $c,C$. Yau's conjecture was affirmed for real
analytic manifolds by \cite{DonFef_ancet90} and the upper bound was
later upgraded to smooth manifolds by Lagunov \cite{logunov18} (see
\cite{logunovMalinnikova19} for a recent review by Logunov and Malinnikova). 

\emph{Nodal statistics - }According to Courant's bound, the normalized
nodal count is bounded by $\frac{\phi\left(n\right)}{n}\le1$ and
is asymptotically bounded by Pleijel's bound. It was shown numerically
by Blum, Gnutzmann and Smilansky in \cite{BGS02}, that the statistics
of\emph{ $\frac{\phi\left(n\right)}{n}$ }can distinguish between
chaotic and integrable planar domains and obeys a universal behavior.
Moreover, for integrable planar domains they proved that the $\frac{\phi\left(n\right)}{n}$
statistics is well defined and calculated its universal characteristics.
However, for chaotic domains the numerics predict a concentration
of measure at a single value. It appears numerically to have a universal
Gaussian concentration with variance of order $\frac{1}{n}$. The
problem of well-posedness of the statistics in the chaotic setting,
not to mention proving the universal behavior, is still open. One
may call it the \emph{nodal BGS conjecture} as it also deals with
a spectral property of chaotic systems, like the BGS conjecture, and
it agrees with the initials of the authors \cite{BGS02}. There were
several related works on the nodal count of a random eigenfunction
that are believed to describe the nodal statistics. Bogomolny and
Schmidt developed a percolation model for the nodal count of random
eigenfunctions for which the conjectured chaotic nodal count behavior
is obtained \cite{BS02}. The credibility
of the percolation model as a prediction for nodal count is discussed
in \cite{BeliaevKereta13}. Another important work on the nodal count
of a random eigenfunction was done by Sodin and Nazarov in \cite{NS07}
for a random eigenfunction on a sphere based on Berry's random wave
model. In a recent work of Sodin and Nazarov \cite{nazarov2020fluctuations}
yet to be published, they improved their result using methods that
resemble Bogomolny and Schmidt's model. The work in \cite{NS07} opened
the door for many works in the area, such as the statistics of the
total length of the nodal lines, $\mathcal{H}^{d-1}\left(f^{-1}\left(0\right)\right)$
for $d=2$ \cite{Wigman_jmp09,Wigman_cmp12,SarWig_cm16,BouRud_inv11}.
In \cite{krishnapurKurlbergWigma13,CamMarWig_jga16} it was shown
that the nodal length statistics for random eigenfunctions on the
torus do not satisfy a universal behavior. 

\subsection*{Nodal count on quantum graphs }

Nodal count on quantum graphs first appeared in the work of Al-Obeid
\cite{AlO_viniti92}, treating only tree graphs. A decade later, independently,
Gnutzmann, Smilansky and Weber raised the question of ``nodal counting
on quantum graphs'' (for all graphs) in \cite{GnuSmiWeb_wrm04}.
The context of their work was the nodal BGS conjecture, after quantum
graphs were established as good paradigms for quantum chaos \cite{GutSmi_jpa01}.
The conjecture that the nodal count can resolve isospectrality \cite{GnuSmiSon_jpa05},
led to a sequence of works on nodal count for quantum graphs after
quantum graphs' isospectrality was introduced \cite{GutSmi_jpa01,BelLub_nhm09}.
This conjecture was affirmed for quantum graphs in certain settings
\cite{BanSmi_epj07,Ore_jpa07,BanOreSmi_incoll08,BanShaSmi_jpa06}
but counter examples were given in \cite{OreBan_jpa12}, and the general
validity of the conjecture is still open. A particular study case
in this research are tree graphs. It was shown in \cite{AlO_viniti92,PokPryObe_mz96,PokPry_rms04},
and independently in \cite{Schapotschnikow06}, that Strum's result
holds for trees, and hence all trees have the same nodal count. In
addition, Band proved in \cite{Ban_ptrsa14} that Sturm's result does
not hold for any other graph and therefore the nodal count distinguishes
between trees and the rest of the graphs. 

The nodal count for quantum graphs can be defined either as the number
of nodal domains (as in the manifolds settings) or as the number of
nodal points (as in Sturm's theorem). These two counts differ by a
constant for all but finitely many eigenfunctions and hence share
the same statistics. Our convention is the number of nodal points,
$\phi\left(n\right):=\left|f_{n}^{-1}\left(0\right)\right|$ (where
the numbering of eigenfunctions starts from $n=0$ for the constant
eigenfunction). The nodal count $\phi\left(n\right)$ is well defined
under the assumption that $f_{n}$ does not vanish on vertices and
its eigenvalue $\lambda_{n}$ is simple. This assumption was shown
in \cite{BerLiu_jmaa17} to hold generically. A Courant like upper
bound on $\phi\left(n\right)$ was proven in \cite{GnuSmiWeb_wrm04}
and a lower bound for trees was found in \cite{Schapotschnikow06,PokPryObe_mz96,PokPry_rms04,AlO_viniti92}
and later generalized to every graph by Berkolaiko in \cite{Ber_cmp08}.
Altogether, we get the following bounds: 
\begin{equation}
n\le\phi\left(n\right)\le n+\beta,\label{eq: nodal bounds}
\end{equation}
where $\beta$ is the first Betti number of the graph. The non-negative
deviation of $\phi\left(n\right)$ from its linear growth is called
the \emph{nodal surplus,} $\sigma\left(n\right):=\phi\left(n\right)-n$,
and it fully characterize the nodal count. A variational characterization
of $\phi\left(n\right)-n$ for quantum graphs was shown in \cite{BanBerRazSmi_cmp12}
in analog to the planar domains result \cite{BerKucSmi_gafa12}. A
similar work for discrete graphs \cite{BerRazSmi_jpa12} led to the
\emph{nodal-magnetic} work of Berkolaiko in \cite{Ber_apde13}. He
showed, for discrete graphs, that the analog of $\phi\left(n\right)-n$
is equal to the Morse index of the $n$th eigenvalue with respect
to magnetic perturbations. A complementary work was done by Colin
de Verdière in \cite{Col_apde13}. This relation, called \emph{the
nodal magnetic relation, }was later upgraded to quantum graphs by
Berkolaiko and Weyand \cite{BerWey_ptrsa14}. We will discuss it in
Section \ref{sec: Magnetic}. The nodal magnetic relation was a key
ingredient in \cite{Ban_ptrsa14} and plays an important role in our
works in \cite{AloBanBer_cmp18} and \cite{AloBanBer_conj}. It also
provides a different physical motivation for nodal statistics from
a ``solid state physics'' point of view, as seen in \cite{BanBer_prl13}.

The behavior of the nodal surplus already appears in \cite{GnuSmiWeb_wrm04},
one of the first nodal count works on quantum graphs, where the number
of nodal domains, $\nu_{n}$, was investigated. For large enough $n$,
$\nu_{n}$ and $\phi\left(n\right)$ are related by $\nu_{n}=\phi\left(n\right)-\beta+1$.
Gnutzmann, Smilansky and Weber showed in \cite{GnuSmiWeb_wrm04} that
$\nu_{n}-n$ is bounded in some fixed interval and considered the
distribution of $\nu_{n}-n$ in this interval, which is the nodal
surplus distribution (up to a constant). They raised the following
\emph{quantum graphs' nodal statistics conjecture:}
\begin{conjecture*}
\cite{GnuSmiWeb_wrm04}~``For well connected graphs, with incommensurate
bond lengths, the distribution of the number of nodal domains in the
interval mentioned above approaches a Gaussian distribution in the
limit when the number of vertices is large''.
\end{conjecture*}
Where by `incommensurate bond lengths' they mean edge lengths which
are linearly independent over the rationals. We will use the name
\emph{rationally independent }edge lengths\emph{.} The term `well
connected graphs' is not classified in their paper and so an important
observation should be made. The limit of large graphs should be taken
as $\beta\rightarrow\infty$ since (\ref{eq: nodal bounds}) implies
that $\sigma\left(n\right)\in\left\{ 0,1,...\beta\right\} $. In particular,
as showed in \cite{Schapotschnikow06}, tree graphs (defined by $\beta=0$)
have $\sigma\left(n\right)\equiv0$ which clearly does not obey a
Gaussian limit.

\subsection*{New results on the nodal statistics for quantum graphs}

The main nodal statistics results of this thesis, which appear in
\cite{AloBanBer_cmp18}, set the mathematical well-posedness of quantum
graphs' nodal statistics and prove the Gaussian limit conjecture for
a certain family of graphs. The well-posedness will be shown in Section
\ref{sec: proof-of-existence}, Theorem \ref{thm:First}, where we
prove that nodal statistics is well defined when the edge lengths
are rationally independent. That is, we prove that the following limit
exists for every $j\in\left\{ 0,1,...\beta\right\} $, 
\[
p_{j}=\lim_{N\rightarrow\infty}\frac{\set{n\le N}{\sigma\left(n\right)=j}}{N}.
\]
We also prove that there is a common symmetry $p_{j}=p_{\beta-j}$
for all graphs, and so the excepted value is $\mathbb{E\left(\sigma\right)=}\frac{\beta}{2}$.
This result can be considered as a generalization of Band's result
for trees, showing that the nodal count distinguishes between graphs
of different Betti number $\beta$. 

In Section \ref{sec: Binomial}, Theorem \ref{thm:Second}, we present
a family of graphs which we call \emph{trees of cycles} for which
the nodal statistics can be explicitly calculated and shown to have
binomial distribution $\sigma\sim Bin\left(\frac{1}{2},\beta\right)$.
The Gaussian limit at $\beta\rightarrow\infty$ follows, together
with the variance estimate $Var\left(\sigma\right)=\frac{\beta}{4}$.
Our modification to the conjecture of Gnutzmann, Smilansky and Weber
is thus: 
\begin{conjecture*}
\cite{AloBanBer_conj} The nodal surplus distribution for a quantum
graph, with rationally independent edge lengths and first Betti number
$\beta$, approaches a Gaussian distribution in the limit of $\beta\rightarrow\infty$
as follows:
\[
\frac{\sigma-\frac{\beta}{2}}{\sqrt{\mathrm{Var}\left(\sigma\right)}}\xrightarrow[\beta\rightarrow\infty]{\mathcal{D}}N\left(0,1\right).
\]
Where the convergence above is in distribution and the variance is
of order $\mathrm{Var}\left(\sigma\right)=\mathcal{O}\left(\beta\right)$.
\end{conjecture*}
In a work in progress \cite{AloBanBer_conj} we prove this conjecture
for several other families of graphs, different than the ones in \cite{AloBanBer_cmp18},
together with a vast numerical evidence.

\subsection*{Neumann count }

It was first noticed by Stern in her\footnote{Antonie Stern (1892-1967) was a Ph.D. student of Courant at G\"{o}ttingen. As a woman,
she could not get a position and was not able to proceed with mathematical
research. In 1939 she escaped Nazi Germany and made Aliyah \cite{SternsInfo}.} Ph.D. thesis in 1925, that there can be arbitrarily large eigenvalues
with nodal domains as small as $\phi\left(n\right)=2$ \cite{Stern_Bemerkungen25}.
This counter intuitive fact is unavoidable. As shown by Uhlenbeck
in \cite{Uhl_ajm76}, a crossing of nodal lines is unstable and can
be omitted by ``as small as we want'' perturbations. Nodal partitions
are determined by such crossings and are therefore usually unstable. 

A novel idea of a more stable partition, which reflects the topography
of an eigenfunction, was first suggested by Zelditch in a paragraph
in \cite{Zel_sdg13}, and (independently) was studied by McDonald
and Fulling in \cite{McDFul_ptrs13}. The partition, now called \emph{Neumann
partition,} is the Morse-Smale complex (see \cite{EdeHarZon_dcg03})
of a planar domain or a 2d manifold according to a given eigenfunction
$f$. A description of such partition, following the definitions and
notations of \cite{BanFaj_ahp16}, is as follows. Let $M$ be a 2d
manifold (for simplicity assume no boundary) and let $f$ be an eigenfunction
of $M$. Consider the gradient $\nabla f$ as a vector field on $M$
and consider gradient flow lines $\varphi:\R\rightarrow M$ such that
$\frac{d}{dt}\varphi\left(t\right)=\nabla f\left(\varphi\left(t\right)\right)$.
It is not hard to deduce that each gradient flow line start and ends
at critical points of $f$ and that these flow lines cover $M$. The
naive picture one should have in mind is that each point $x\in M$
which is not a critical point of $f$ lies on a unique gradient flow
line that starts from a local minimum $q$ and ends at a local maximum
$p$. This is not necessarily the case, in general, and so the discussion
is restricted to eigenfunctions which are Morse, which is a generic
property \cite{Uhl_ajm76}. An eigenfunction is said to be Morse if
its set of critical points is discrete and each critical point is
non-degenerate (i.e Hessian of full rank).

Given a Morse eigenfunction $f$, every pair of a local minimum $q$
and a local maximum $p$ define a \emph{Neumann domain} $\Omega_{p,q}$
as the union (possibly empty) of gradient flow lines of $f$ going
from $p$ to $q$. On the boundary of a Neumann domain are the \emph{Neumann
lines,} gradient flow lines that go through a saddle point. The \emph{Neumann
partition }is the partition of $M$ into Neumann domains. Such a partition
can be shown to be stable under small perturbations of $f$ or of
the metric on $M$. Heuristically, the Neumann partition is changed
only if critical points meet\textbackslash appear\textbackslash disappear,
which generically does not happen under small enough perturbations. 

A main feature of the Neumann partition and the origin of its name
is that the restriction $f|_{\Omega}$ to a Neumann domain $\Omega$
is an eigenfunction of $\Omega$ with Neumann boundary conditions
\cite{BanEggTay_arXiv17}. It is analogous to the known fact that
the restriction to a nodal domain is a Dirichlet eigenfunction of
that domain. However, unlike the restriction to a nodal domain, where
the eigenfunction is known to be the first Dirichlet eigenfunction
of that domain, for a Neumann domain this is not the case. The \emph{spectral
position} $N\left(\Omega\right)$ of a Neumann domain $\Omega$ of
$f$ with eigenvalue $\lambda$ is the position of $\lambda$ in the
spectrum of $\Omega$. Namely, the number of eigenvalues of $\Omega$
(with Neumann boundary conditions) smaller than $\lambda$. It was
previously believed that like in the case of nodal domains, $N\left(\Omega\right)$
should be one, or at least very low, but the works of \cite{BanFaj_ahp16,BanEggTay_arXiv17}
show , counter intuitively, that $N\left(\Omega\right)$ can be as
high as we wish, even for simple cases like the flat torus. In analogy
to the nodal count problem, the \emph{Neumann count, }$\mu\left(n\right)$
is defined as the number of Neumann domains of $f_{n}$, the $n^{th}$
eigenfunction. It was shown in \cite{BanFaj_ahp16} that $\mu\left(n\right)\ge\frac{1}{2}\phi\left(n\right)$
but it is still unknown whether the Neumann count holds more geometric
information on the manifold than the nodal count. In particular, Neumann
statistics properties and resolution of isospectrality are still unknown
in general. For more information see the review paper \cite{AloBanBerEgg_Neumann}.

\subsection*{New results on Neumann count and statistics for quantum graphs}

The novel study of Neumann partitions led naturally to the question
of a Neumann partition on a quantum graph. We raised this question
in \cite{AloBanBerEgg_Neumann} and compared between properties of
Neumann partitions on quantum graphs and manifolds. The wealth of
questions on Neumann partitions, Neumann domains and Neumann count
on quantum graphs was further studied in \cite{AloBan19}. In analogy
to nodal points, the \emph{Neumann points} of an eigenfunction $f$
on a quantum graph are the interior critical points (which are either
local minima or maxima). The \emph{Neumann count} $\mu\left(n\right)$
for a quantum graph is the number of Neumann points of $f_{n}$, the
$n^{th}$ eigenfunction, and it is convenient to discuss the deviation
$\omega\left(n\right)=\mu\left(n\right)-n$, called the \emph{Neumann
surplus} in analogy to the nodal surplus (although it can be negative).
The Neumann surplus was bounded uniformly in \cite{AloBanBerEgg_Neumann,AloBan19},
in analogy to the nodal surplus bounds (\ref{eq: nodal bounds}). 

The main results of this thesis in the context of Neumann count and
statistics were obtained in \cite{AloBan19}. In Section \ref{sec: proof-of-existence},
Theorem \ref{thm:First}, we prove, alongside the nodal statistics,
that Neumann statistics (that is the statistics of the Neumann surplus)
is well defined, by existence of the limits, 
\[
p_{j}=\lim_{N\rightarrow\infty}\frac{\set{n\le N}{\omega\left(n\right)=j}}{N},
\]
for all possible values of the Neumann surplus. We also prove a symmetry,
similar to that of the nodal statistics, which provides the expected
value $\mathbb{E}\left(\omega\right)=\frac{\beta-\left|\partial\Gamma\right|}{2}$,
where $\left|\partial\Gamma\right|$ is the number of vertices of
degree one in the graph. As a consequence, we can recover both $\beta$
and $\left|\partial\Gamma\right|$ from $\mathbb{E}\left(\omega\right)$
and $\mathbb{E}\left(\sigma\right)$. This is a major improvement
to the inverse problem of the nodal count, as the number of (discrete)
graphs with fixed $\beta$ and $\left|\partial\Gamma\right|$ is finite.
This also proves that the nodal count and the Neumann count cannot
be obtained one from the other. The question of how correlated are
the nodal statistics and the Neumann statistics is discussed in \cite{AloBan19}
and is still open. In Section \ref{sec: Binomial}, Theorem \ref{thm:Second},
we prove, alongside the binomial nodal statistics of a certain family
of graphs, that tree graphs whose interior vertices (those of degree
larger than one) are of degree 3 have a shifted binomial Neumann statistics:
\[
\omega+\left|\V_{in}\right|+1\sim Bin\left(\left|\V_{in}\right|,\frac{1}{2}\right).
\]
Where $\left|\V_{in}\right|$ is the number of internal vertices.
A Gaussian limit at $\left|\V_{in}\right|\rightarrow\infty$ appears
here, and the universality of this limit for other families of graphs
is currently investigated. 

\subsection*{Generic properties of eigenfunctions on quantum graphs}

As already stated, Uhlenbeck's seminal work ``generic properties
of eigenfunctions'' \cite{Uhl_ajm76} was needed in order to discuss
and define the nodal and Neumann counts on manifolds. In fact this
work is crucial for almost every spectral property, as in the words
of Uhlenbeck, it ``set up machinery to consider the eigenfunctions
of curves of operators...suggest an approach to the problem of characterizing
the $n$th eigenfunction of a family of operators''. As metric graphs
are not manifolds, the results of Uhlenbeck do not apply and new machinery
is needed. The first genericity result for quantum graphs was obtained
by Friedlander \cite{Fri_ijm05}, who showed that for any graph structure
not homeomorphic to a cycle, there is a residual set of edge lengths
for which every eigenvalue of the quantum graph is simple. A decade
later, Berkolaiko and Liu had found \cite{BerLiu_jmaa17} that for
graphs without loops (an edge connecting a vertex to itself) and a
generic choice of edge lengths (in the sense of \cite{Fri_ijm05})
none of the eigenfunctions vanish on a vertex. This property is needed
to define nodal count, and is also crucial for the nodal magnetic
connection as seen in \cite{BerWey_ptrsa14}. 

In the case where a graph has a loop, for any choice of edge lengths,
there will be infinitely many eigenfunctions supported on that loop.
Nevertheless, it is proven in \cite{BerLiu_jmaa17}, that for a generic
choice of edge lengths, every eigenfunction not supported on a loop,
does not vanish on any vertex. In \cite{AloBanBer_cmp18} we show
that the implicit generic choice of edge lengths can be replaced by
the explicit restriction to rationally independent edge lengths, at
the cost of a density zero sequence. Namely, for almost every eigenfunction
(a density one sequence), the eigenvalue is simple and either the
eigenfunction is supported on a loop or it does not vanish on any
vertex. 

\subsection*{New genericity results for quantum graphs}

The main genericity result in this thesis, a work from \cite{Alon}
yet to be published, is that generically the derivatives of an eigenfunction
do not vanish on any interior vertex (that is any vertex which is
not of degree one). Here, by generically we mean either in the sense
of every eigenfunction for a residual set of edge lengths, or in the
sense of a density one sequence of eigenfunction for any choice of
rationally independent edge lengths. We also prove that the two choices
of genericity are equivalent in this case. This additional property,
is needed in order to define the Neumann count and statistics.

\subsection{The structure of the thesis }

This thesis incorporates the nodal statistics works of \cite{AloBanBer_cmp18}
together with the Neumann count and statistics works of \cite{AloBan19,AloBanBerEgg_Neumann}.
The structure of the thesis and the partition into ``Neumann'' results
of \cite{AloBan19,AloBanBerEgg_Neumann} versus ``nodal'' results
of \cite{AloBanBer_cmp18} is illustrated in the diagram in Figure
\ref{fig: thesis structure}. Section \ref{sec: Neumann-count-bounds}
is a short section in which we present a uniform bound on the Neumann
surplus. Section \ref{sec: The-secular-manifold} is the core of this
thesis, in it we present the \emph{secular manifold} and provide the
needed machinery for statistical investigation of eigenfunctions.
A new result presented in this section is a generalization of a result
from \cite{BerLiu_jmaa17} on the number of connected components of
the secular manifold. Section \ref{sec: genericity} is devoted to
the generalization of the genericity result in \cite{BerLiu_jmaa17}.
The extended generic properties from this work are needed in order
to define the Neumann count. In Section \ref{sec: proof-of-existence}
we prove existence and symmetry of both the nodal surplus distribution
(a main result of \cite{AloBanBer_cmp18}) and the Neumann surplus
distribution (a main result of \cite{AloBan19}). In Section \ref{sec: Binomial},
we present families of graphs for which we can prove that the nodal\textbackslash Neumann
surplus distributions are binomials and converge to a Gaussian limit
as conjectured. These results were obtained in \cite{AloBanBer_cmp18}
for the ``nodal'' case and in \cite{AloBan19} for the Neumann ``case''.
Both results make use of statistical behavior of ``local'' properties,
which are described in Sections \ref{sec: Properties-of-a Neumann domain}
and \ref{sec: Magnetic}. In Section \ref{sec: Properties-of-a Neumann domain}
we present local properties of Neumann domains, providing both bounds
and statistical analysis as done in \cite{AloBan19,AloBanBerEgg_Neumann}.
In Section \ref{sec: Magnetic} we present a brief introduction to
the nodal magnetic connection that was proved in \cite{BerWey_ptrsa14}.
Using the nodal magnetic theorem we prove, as in \cite{AloBanBer_cmp18},
that the nodal surplus is given by a sum of local magnetic stability
indices and analyze their statistics.

\begin{figure}
\includegraphics[width=0.75\paperwidth]{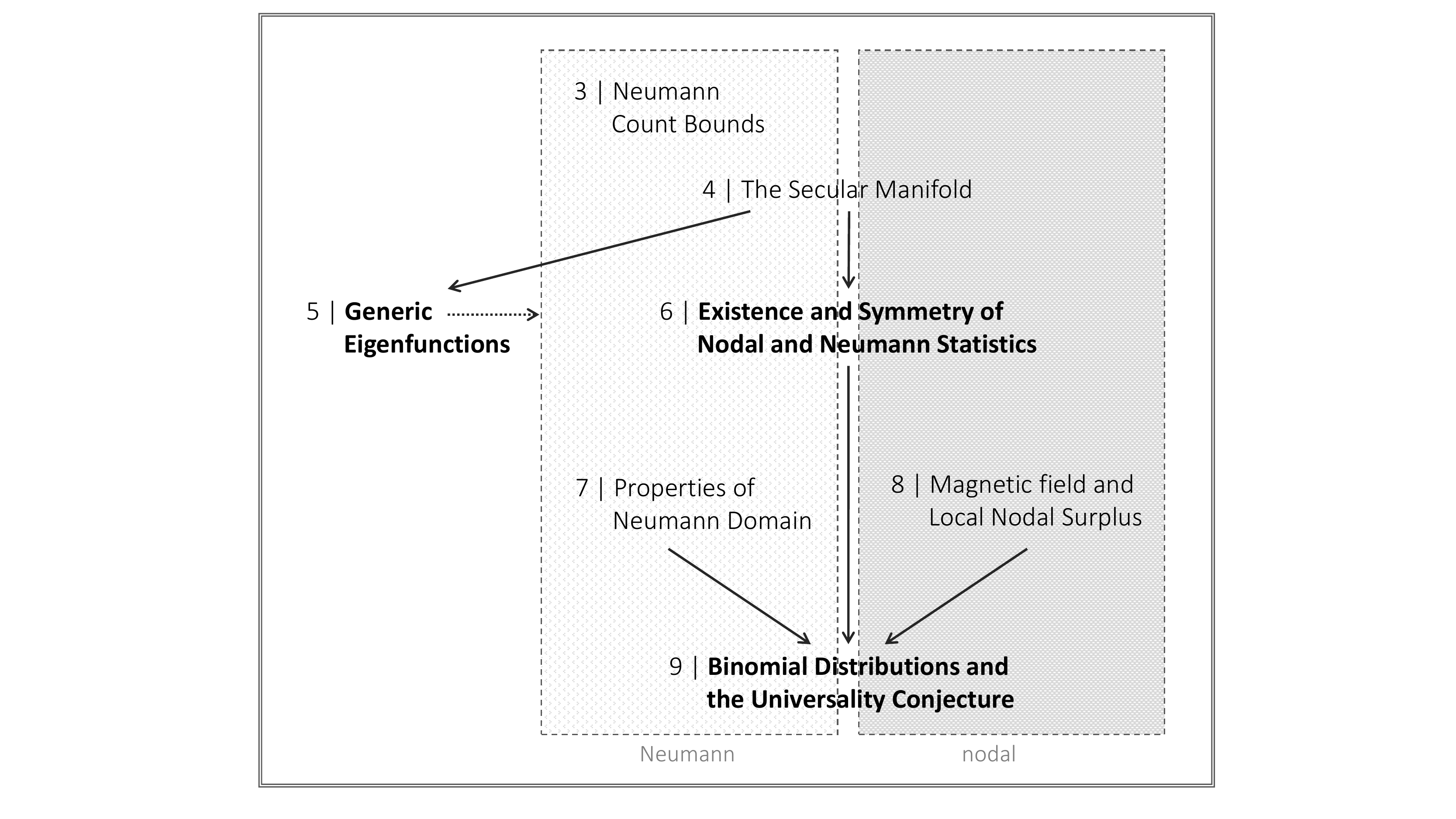}\caption[Thesis structure]{\label{fig: thesis structure}A diagram of the thesis structure. The
arrows indicate dependence. Sections 5,6 and 9 in bold as these hold
the main results (although new results appear in every chapter from
4 to 9). The 'Neumann' (or 'nodal') block indicates results that appear
in \cite{AloBan19} (or \cite{AloBanBer_cmp18}). }
\end{figure}

\newpage{}
\section{\label{sec: Preliminaries} Preliminaries}

\subsection{Basic graph definitions and notations\label{subsec:metric_graphs_introduction}}

Throughout this manuscript, the graphs we consider are finite and
connected. We denote by $\V$ the set of the graph vertices and by
$\E$ the set of its edges. We denote their cardinality by $V:=\left|\V\right|$
and $E:=\left|\E\right|$. Our discussion is not restricted to simple
graphs. Namely, two vertices may be connected by more than one edge
and it is also possible for an edge to connect a vertex to itself.
An edge connecting a vertex to itself is called a \emph{loop}. Given
a vertex $v\in\V$ we denote the multi-set of edges connected to $v$
by $\E_{v}$. We note that every loop connected to $v$ will appear
twice in $\E_{v}$. The \emph{degree} of a vertex is denoted by $\deg v:=\left|\E_{v}\right|$.
\begin{rem}
\textbf{Throughout this manuscript we assume no $\deg v=2$ vertices
and that $E>1$.} 

We will show in Remark \ref{rem: deg 2} why adding\textbackslash removing
vertices of $\deg v=2$ does not affect the quantum graphs we discuss.
The restriction to $E>1$, is to exclude the loop graph which is the
only exception in most of the following theorems (as it has a continuous
symmetry which gives a completely non-simple spectrum). By considering
$E>1$ we also exclude the interval, which is fully analyzed in the
famous works of Sturm and Liouville. 
\end{rem}

\begin{defn}
We call a vertex of degree one, a \emph{boundary }vertex, and define
the boundary of the graph as $\dg:=\left\{ v\in\V\thinspace|\,\deg v=1\right\} .$
The rest of the vertices are called \emph{interior vertices} and we
denote the set of interior vertices by $\V_{in}:=\Vint$. 
\end{defn}

\begin{defn}
We define a \emph{tail} as an edge connected to a boundary vertex,
and we define a \emph{bridge }as an edge whose removal disconnects
the graph. 
\end{defn}

\begin{rem}
In particular a tail is a bridge.
\end{rem}

\begin{defn}
The first Betti number of a finite connected graph is given by 
\begin{equation}
\beta:=E-V+1.\label{eq:Betti_number}
\end{equation}
A graph with $\beta=0$ is called a \emph{tree graph.} Throughout
this manuscript we will always use $\beta$ to denote the first Betti
number of a graph.
\end{defn}

\begin{rem}
The first Betti number should be thought of as the number of ``independent
cycles'' on the graph. Here is a brief explanation. The general definition
of the first Betti number for a topological space $X$ is the dimension
of its first Homology group $H_{1}\left(X,\Z\right)$. For a graph
$\Gamma$, with some choice of orientation for each edge, every closed
path $\gamma$ induce a formal sum $\gamma\mapsto\sum_{e\in\E}\gamma_{e}\cdot e$
where each $\gamma_{e}\in\Z$ is the number of times (with sign that
indicates direction) in which $\gamma$ passes through $e$. The space
of all such (formal sums of) closed paths is $H_{1}\left(\Gamma,\Z\right)$.
It can be shown to be a vector space. Therefore, $\beta:=\dim H_{1}\left(\Gamma,X\right)$
is the number of linearly independent elements in $H_{1}\left(\Gamma,\Z\right)$
that span $H_{1}\left(\Gamma,X\right)$. See chapter 4 in \cite{sunada2012topological}
for more details on homology groups on graphs.
\end{rem}

\begin{defn}
A \emph{metric graph} is a graph $\Gamma$ with edge lengths $\lv\in\R_{+}^{\E}$
such that every edge $e\in\E$ is given an edge length $l_{e}$. We
denote such a graph by $\Gamma_{\lv}$. We denote the total length
of $\Gamma_{\lv}$ by $L:=\sum_{j=1}^{E}l_{j}$.
\end{defn}

A common assumption in this paper is that the set of edge lengths
form a linear independent set over $\Q$.
\begin{defn}
\label{def: rationally indepdnet}A vector $\lv$ is called\emph{
rationally independent} if its entries are linearly independent over
$\Q$. That is, the only rational $\vec{q}\in\Q^{\E}$ that satisfies
$\sum_{e}q_{e}l_{e}=0$ is $\vec{q}=0$. 
\end{defn}

\begin{rem}
\label{rem: rationally indepdnet}We will later use the fact that the set
of rationally independent edge lengths is \emph{residual }in $\R_{+}^{\E}$\emph{.
}Where a residual set is a countable intersection of sets with dense
interior (equivalently, it is the complement of a countable union
of nowhere-dense sets). To show that the set of rationally independent
edge lengths is residual, notice that the set $\set{\lv\in\R_{+}^{\E}}{\lv\cdot\vec{q}\ne0}$
is open and dense in $\R_{+}^{\E}$ for any given $\vec{q}\in\Q^{\E}$.
The set we are after, $\cap_{\vec{q}\in\Q^{\E}}\set{\lv\in\R_{+}^{\E}}{\lv\cdot\vec{q}\ne0}$,
is therefore residual.
\end{rem}

\subsection{Standard quantum graphs \label{subsec:Spectral-theory-of-Quantum-Graphs}}

It is convenient to describe a function $f$ on a metric graph $\Gamma$
in terms of its restrictions to edges. If $v\in\V$ and $e\in\E_{v}$
is of length $l_{e}$, then we can define an arc-length coordinate
$x_{e}\in\left[0,l_{e}\right]$ such that $x_{e}=0$ at $v$. If $e$
is not a loop, then $x_{e}$ is the distance from $v$ along the edge.
The restriction of $f$ to $e$, given by $f|_{e}\left(x_{e}\right)$
is a function $f|_{e}:\left[0,l_{e}\right]\rightarrow\C$. The choice
of coordinates dictates a direction for each edge. We denote the edge
$e$ with opposite direction by $\hat{e}$ with arc-length coordinate
$x_{\hat{e}}=l_{e}-x_{e}$. We define the set of directed edges by
$\vec{\E}$ such that $\left|\vec{\E}\right|=2E$. The choice of orientation
does not effect functions, namely $f|_{e}\left(x_{e}\right)=f|_{\hat{e}}\left(l_{e}-x_{\hat{e}}\right)$.
But does effect (odd order) derivatives, $\frac{d}{dx_{e}}f|_{e}\left(x_{e}\right)=-\frac{d}{dx_{\hat{e}}}f|_{\hat{e}}\left(l_{e}-x_{\hat{e}}\right)$.
The $L^{2}$ space and $H^{2}$ (also known as $W^{2,2}$) Sobolev
space of $\Gamma_{\lv}$ are defined according to the restrictions
of the functions to edge: 
\begin{equation}
L^{2}\left(\Gamma_{\lv}\right):=\oplus_{e\in\E}L^{2}\left(\left[0,l_{e}\right]\right),\quad H^{2}\left(\Gamma_{\lv}\right):=\oplus_{e\in\E}H^{2}\left(\left[0,l_{e}\right]\right).
\end{equation}
The Laplace operator $\Delta:H^{2}\left(\Gamma_{\lv}\right)\rightarrow L^{2}\left(\Gamma_{\lv}\right)$
is defined edgewise by 
\[
\Delta\ :\ f|_{e}\mapsto-\frac{d^{2}}{dx_{e}^{2}}f|_{e}.
\]

If $f|_{e}$ is a solution to $\frac{d^{2}}{dx_{e}^{2}}f|_{e}=-k^{2}f|_{e}$
for some $k>0$, then it is determined by the initial values $f|_{e}\left(0\right)$
and $\frac{df|_{e}}{dx_{e}}\left(0\right)$. Such initial values are
assigned to every pair of vertex $v\in\V$ and edge $e\in\E_{v}$
by considering the arc-length parameterization with $x_{e}=0$ at
$v$. In \cite{BerLatSuk19}, \emph{$\tr\left(f\right)$} is defined
as the collection of these values. We will use this terminology:
\begin{defn}
Given a function $f\in H^{2}\left(\Gamma_{\lv}\right)$ and a pair
$\left(v,e\right)$ such that $v\in\V$ and $e\in\E_{v}$, we define
$\tr\left(f\right)$ at $\left(v,e\right)$ as a pair $f|_{e}\left(v\right),\partial_{e}f\left(v\right)$
of the value and outgoing derivative of $f|_{e}$ at $v$, which are
given by 
\begin{align*}
f|_{e}\left(v\right): & =f|_{e}\left(0\right)\\
\partial_{e}f\left(v\right): & =\frac{df|_{e}}{dx_{e}}\left(0\right).
\end{align*}
If $f$ is continuous at $v$, namely $f|_{e}\left(v\right)=f|_{e'}\left(v\right)\,\,\forall e,e'\in\E_{v}$,
we will write $f\left(v\right)$ instead of $f|_{e}\left(v\right)$.
\end{defn}

\begin{rem}
If $e\in\E_{v}$ is connecting $v$ to $u$, then $f|_{e}\left(l_{e}\right)=f|_{e}\left(u\right)$
and $-\partial_{e}f\left(u\right)=f|_{e}\left(0\right)$. If $e$
is a loop, we denote the two outgoing derivatives by $\partial_{e}f\left(v\right)$
and $\partial_{\hat{e}}f\left(v\right)$.
\end{rem}

The Laplacian is not self-adjoint on $H^{2}\left(\Gamma_{\lv}\right)$.
Using the $L^{2}$ inner product $\left\langle *,*\right\rangle $
and integration by parts, one can show that: 
\[
\left\langle \Delta f,g\right\rangle -\left\langle f,\Delta g\right\rangle =-\sum_{e\in\E}\left(\frac{df|_{e}}{dx_{e}}\overline{g}|_{e}-\frac{d\overline{g}|_{e}}{dx_{e}}f|_{e}\right)|_{0}^{l_{e}}=\sum_{v\in\V}\sum_{e\in\E_{v}}\partial_{e}f\left(v\right)\overline{g}|_{e}\left(v\right)-\partial_{e}\overline{g}\left(v\right)f|_{e}\left(v\right).
\]
Therefore the Laplacian is self-adjoint on domains of functions in
$H^{2}\left(\Gamma_{\lv}\right)$ for which the RHS of the above vanish.
A description of all vertex conditions for which the Laplacian is
self-adjoint can be found for example in \cite{BerKuc_graphs}, and
in \cite{BerLatSuk19} there is a description of the RHS above as
a simplectic form on $\tr\left(f\right)$, and the possible ``good''
domains as Lagrangian manifolds. Throughout this paper we only consider
the domain of functions that satisfy\emph{ Neumann vertex conditions}
for which the RHS above vanish and the Laplacian is self-adjoint. 
\begin{defn}
A function $f\in H^{2}\left(\Gamma_{\lv}\right)$ is said to satisfy
\emph{Neumann vertex conditions} if it satisfies the following condition
at every vertex $v\in\V$. The \emph{Neumann} (also known as Kirchhoff
or standard) condition of $f$ on $v$ is: 
\begin{enumerate}
\item $f$ is continuous at $v$, namely $f|_{e}\left(v\right)=f|_{e'}\left(v\right)\,\,\forall e,e'\in\E_{v}$.
\item The sum of outgoing derivatives vanish, namely $\sum_{e\in\E_{v}}\partial_{e}f\left(v\right)=0.$
\end{enumerate}
\end{defn}

\begin{rem}
First notice that indeed if $f$ and $g$ satisfy Neumann vertex
conditions, then $\sum_{v\in\V}\sum_{e\in\E_{v}}\partial_{e}f\left(v\right)\overline{g}|_{e}\left(v\right)-\partial_{e}\overline{g}\left(v\right)f|_{e}\left(v\right)=0$
and thus the Laplacian is self-adjoint. One may also observe that
if $\deg v=1$, namely it is a boundary vertex, then the Neumann condition
is simply $\partial_{e}f\left(v\right)=0$ which is the Neumann boundary
condition on a segment in one dimension. 
\end{rem}

\begin{rem}
\label{rem: deg 2} If $f\in H^{2}\left(\Gamma_{\lv}\right)$ and
$x\in\Gamma\setminus\V$ is an interior point, then both $f$ and
$f'$ are continuous at $x$. If we consider $x$ as a vertex of degree
two, then $f$ satisfies Neumann vertex condition at $x$. The inverse
argument is also true, that is if $v$ is of degree two and $f\in H^{2}\left(\Gamma_{\lv}\right)$
satisfies Neumann vertex condition at $v$, then we can consider $v$
as an interior point, and $f$ will remain in $H^{2}$. It follows
that the eigenfunctions and eigenvalues will not change by adding/removing
vertices of degree two with Neumann vertex conditions. 
\end{rem}

\begin{defn}
We define a \emph{standard quantum graph }as a finite connected metric
graph $\Gamma_{\lv}$ (assuming $E>1$ and $\deg v\ne2\,\,\forall v\in\V$)
equipped with the Laplace operator restricted to the domain of Neumann
vertex conditions. We abbreviate it to a \emph{standard graph }and
denote it by $\Gamma_{\lv}$ as well. The \emph{spectrum/eigenvalues/eigenfunctions
}of $\Gamma_{\lv}$ are referred to the spectrum/eigenvalues/eigenfunctions
of the Laplacian on the domain of Neumann vertex conditions.
\end{defn}

The spectrum of a standard graph $\Gamma_{\lv}$ is real, non-negative
and discrete. The eigenvalues are indexed according to their magnitude,
including multiplicity: 
\begin{align}
0 & =\lambda_{0}<\lambda_{1}\le\lambda_{2}\nearrow\infty,
\end{align}
and the corresponding eigenfunctions are indexed accordingly $\left\{ f_{n}\right\} _{n=0}^{\infty}$.
Where the lowest eigenvalue is always $\lambda_{0}=0$, it is \emph{simple}
(has multiplicity one) and corresponds to the constant eigenfunction
$f_{0}\equiv c$ \cite{BerKuc_graphs}. As the Laplacian has real
coefficients (as a differential operator) and the Neumann vertex conditions
are real, then every eigenfunction is real up to a global constant
\cite{BerKuc_graphs} and we may choose them to be real. The choice
of $\left\{ f_{n}\right\} _{n=0}^{\infty}$ is not unique if there
are non-simple eigenvalues, but unless stated otherwise every result
in this manuscript is independent of that choice.

A common convention that we will use is to denote the eigenvalues
by $\lambda_{n}=k_{n}^{2}$ for $k_{n}\ge0$ and its is common abuse
of notations to refer to $\left\{ k_{n}\right\} _{n=0}^{\infty}$
as the eigenvalues of $\Gamma_{\lv}$ as well.

As discussed above, it is convenient to describe a non-constant eigenfunction
$f$ of eigenvalue $k^{2}>0$ by its restriction $f|_{e}$. Every
restriction satisfies $f|_{e}''=-k^{2}f|_{e}$ and the space of functions
satisfying this ODE has two standard bases $\left\{ \cos\left(kx_{e}\right),\sin\left(kx_{e}\right)\right\} $
and $\left\{ e^{ikx_{e}},e^{-ikx_{e}}\right\} $ so $f|_{e}$ can
be described by a pair of parameters. For later use we introduce the
following such pairs.
\begin{defn}
\label{def: edge restriction notations}~Let $f$ be a real eigenfunction
of eigenvalue $k^{2}>0$. Let $v\in\V$, $e\in\E_{v}$ and consider
the arc-length parameterization $x_{e}\in\left[0,l_{e}\right]$ with
$x_{e}=0$ at $v$. 
\begin{enumerate}
\item \label{enu: complex amplitudes}We define the \emph{complex-amplitudes}
pair $a_{e},a_{\hat{e}}\in\C$ such that 
\[
f|_{e}\left(x_{e}\right)=a_{e}e^{-ikl_{e}}e^{ikx}+a_{\hat{e}}e^{-ikx}.
\]
The relation between $\tr\left(f\right)$ at $\left(v,e\right)$ and
the complex amplitudes can be expressed as 
\begin{align}
f\left(v\right)=a_{e}e^{-ikl_{e}}+a_{\hat{e}},\,\,\,\,\,\,\,\,\,\,\,\,\,\, & \frac{\partial_{e}f\left(v\right)}{k}=i\left(a_{e}e^{-ikl_{e}}-a_{\hat{e}}\right)\label{eq: v values using a}\\
a_{\hat{e}}=\frac{1}{2}\left(f\left(v\right)+i\frac{\partial_{e}f\left(v\right)}{k}\right),\,\,\,\,\,\,\,\,\,\,\,\,\,\, & a_{e}=\frac{1}{2}e^{ikl_{e}}\left(f\left(v\right)-i\frac{\partial_{e}f\left(v\right)}{k}\right).\label{eq: a using v values}
\end{align}
Notice that if $\deg v=1$ ($e$ is a tail), then $a_{\hat{e}}=a_{e}e^{-ikl_{e}}.$
\\
We define the \emph{amplitudes vector }of $f$, $\boldsymbol{a}\in\C^{\vec{\E}}$,
as the tuple of complex-amplitudes pairs for all edges. 
\item \label{enu: Acos(x+phi)}We define the \emph{amplitude-phase} pair
$A_{e}\in\R\,,\varphi_{e}\in\clop{0,\pi}$ such that 
\[
f|_{e}\left(x_{e}\right)=A_{e}\cos\left(kx_{e}-\varphi_{e}\right),
\]
with $f\left(v\right)=A_{e}\cos\left(\varphi_{e}\right)$ and $\frac{\partial_{e}f\left(v\right)}{k}=A_{e}\sin\left(\varphi_{e}\right)$.
If $\deg v=1$ ($e$ is a tail), then $\varphi_{e}=0$. 
\item \label{enu: Acos+Bsin}If $e$ is a loop, the \emph{mid-edge} pair
$\boldsymbol{A}_{e},\boldsymbol{B}_{e}\in\R$ is sometimes more convenient.
Consider a different arc-length parameterization $x_{e}\in\left[-\frac{l_{e}}{2},\frac{l_{e}}{2}\right]$
such that both $x_{e}=\pm\frac{l_{e}}{2}$ correspond to $v$. Then
the pair $\boldsymbol{A}_{e},\boldsymbol{B}_{e}$ is such that
\[
f|_{e}\left(x_{e}\right)=\boldsymbol{A}_{e}\cos\left(kx_{e}\right)+\boldsymbol{B}_{e}\sin\left(kx_{e}\right).
\]
\end{enumerate}
\end{defn}

\begin{rem}
Throughout this manuscript, \textbf{unless stated otherwise, we will
use the complex-amplitudes notation}. The other notations will be
useful for various proofs and so we bring them here.
\end{rem}

\begin{lem}
\label{lem: |ae|=00003D|ae'|} Let $f$ be a real eigenfunction, then
each of its complex-amplitudes pair $a_{e},a_{\hat{e}}$ satisfy $\left|a_{e}\right|=\left|a_{\hat{e}}\right|$.
The relation between the amplitude-phase pair $A_{e},\varphi_{e}$
to the complex-amplitudes pair $a_{e},a_{\hat{e}}$ is given by 
\begin{equation}
A_{e}e^{i\varphi}=2a_{\hat{e}},\,\,A_{e}e^{-i\varphi}=2a_{e}e^{-ikl_{e}}.\label{eq: amp-phase to cpx amp}
\end{equation}
In particular, 
\begin{align}
2\left(\left|a_{e}\right|^{2}+\left|a_{\hat{e}}\right|^{2}\right) & =f\left(v\right)^{2}+\frac{\partial_{e}f\left(v\right)^{2}}{k^{2}}=A_{e}^{2},\,and\label{eq: amp square}\\
if\,\,\left|a_{\hat{e}}\right|\ne0,\,\, & e^{2i\varphi}=\frac{a_{\hat{e}}}{a_{e}}e^{ikl_{e}}=\frac{f\left(v\right)+i\frac{\partial_{e}f\left(v\right)}{k}}{f\left(v\right)-i\frac{\partial_{e}f\left(v\right)}{k}}.\label{eq: amp to phase}
\end{align}
\end{lem}

\begin{proof}
The equality $\left|a_{e}\right|=\left|a_{\hat{e}}\right|$ follows
from (\ref{eq: amp-phase to cpx amp}) which follows from (\ref{eq: a using v values})
together with $f\left(v\right)=A_{e}\cos\left(\varphi_{e}\right)$
and $\frac{\partial_{e}f\left(v\right)}{k}=A_{e}\sin\left(\varphi_{e}\right)$.
It is now immediate that 
\[
A_{e}^{2}=f\left(v\right)^{2}+\frac{\partial_{e}f\left(v\right)^{2}}{k^{2}}=\left|f\left(v\right)+i\frac{\partial_{e}f\left(v\right)}{k}\right|^{2}=4\left|a_{\hat{e}}\right|^{2}=2\left(\left|a_{e}\right|^{2}+\left|a_{\hat{e}}\right|^{2}\right),
\]
 and that if the above is non zero, then
\[
e^{i2\varphi}=\frac{A_{e}e^{i\varphi}}{A_{e}e^{-i\varphi}}=\frac{a_{\hat{e}}}{a_{e}}e^{ikl_{e}}=\frac{f\left(v\right)+i\frac{\partial_{e}f\left(v\right)}{k}}{f\left(v\right)-i\frac{\partial_{e}f\left(v\right)}{k}}.
\]
\end{proof}
The notion of \emph{nodal count} that will be later defined is discussed
for eigenfunctions that do not vanish on vertices. Similarly the
notion of \emph{Neumann count} that will be defined later requires
that the derivatives of the eigenfunctions on interior vertices do
not vanish. We therefore define the following:\emph{ }
\begin{defn}
\label{def: Properties I,II,III-1} Let $\Gamma_{\lv}$ be a standard
graph and let $f$ be an eigenfunction. We say that $f$ satisfies,
\begin{enumerate}
\item \label{enu:generic non vanishing on vetrtex-1}\emph{Property I }-
if $\forall v\in\V,\,\,f\left(v\right)\ne0.$
\item \label{enu: generic non vanishing on derivatives-1}\emph{Property
II }-if\emph{ }$\forall v\in\V_{in},\,\forall e\in\E_{v},\,\,\,\partial_{e}f\left(v\right)\ne0.$
\end{enumerate}
\end{defn}

In Section \ref{sec: genericity} we will show that generically (under
certain restrictions) eigenfunctions satisfy both properties above.
We therefore define the notion of \emph{generic eigenfunctions:}
\begin{defn}
\label{def:generic_eigenfunction} We call an eigenfunction $f$ \emph{generic
}if it has a simple eigenvalue and it satisfies both properties $I$
and $II$. Given a standard graph $\Gamma_{\lv}$ with eigenfunctions
$\left\{ f_{n}\right\} _{n=0}^{\infty}$, we define the integer subset:
\[
\G:=\set{n\in\N}{f_{n}\,\,\text{is generic}}.
\]
\end{defn}

\begin{rem}
\label{rem: differ on degernate e.s}The choice of a basis of eigenfunctions
$\left\{ f_{n}\right\} _{n=0}^{\infty}$ may change, but only on eigenspaces
of non-simple eigenvalues. As generic eigenfunctions have simple eigenvalues
then $\G$ is independent of such choice of a basis of eigenfunctions.
\end{rem}

The following lemma is immediate from the latter definition and Lemma
\ref{lem: |ae|=00003D|ae'|}: 
\begin{lem}
\label{lem: phase for generic}Let $f$ be a generic eigenfunction,
let $v\in\V,\,e\in\E_{v}$ and let $a_{e},a_{\hat{e}},A_{e},\varphi_{e}$
be the complex-amplitudes and amplitude-phase pairs. Then $a_{e},a_{\hat{e}},A_{e}\ne0$,
and if $v$ is not a boundary vertex, then $e^{i2\varphi}=\frac{a_{\hat{e}}}{a_{e}}e^{ikl_{e}}\notin\R$. 
\end{lem}

\subsection{Nodal and Neumann count}

A partition of a metric graph $\Gamma_{\lv}$ at a set of interior
points $\left\{ x_{j}\right\} _{j=1}^{n}\in\Gamma_{\lv}\setminus\V$
is the procedure of cutting the graph at these points, replacing each
point $x_{j}$ with two distinct vertices of degree one $x_{j}^{+},x_{j}^{-}$.
The resulting partitioned graph may not be connected and we will refer
to its connected components as components of the partition\emph{.} 
\begin{defn}
\label{def: Neumann_points_Nodal_points_etc} Let $f\in H^{2}\left(\Gamma_{\lv}\right)$.
An interior point $x\in\Gamma_{\lv}\setminus\V$ is called a \emph{nodal
point} of $f$ if $f\left(x\right)=0$ and is called a \emph{Neumann
point }of $f$ if $f'\left(x\right)=0$. If an eigenfunction $f$
is generic, then it has a finite number of nodal and Neumann points. 

The partition of $\Gamma_{\lv}$ according to the nodal points is
the \emph{nodal partition }and its connected components are the \emph{nodal
domains}. We define the \emph{nodal count }of $f$ as the number of
nodal points, 
\begin{equation}
\phi(f):=\left|\left\{ x\in\Gamma_{\lv}\setminus\V\,|\,f\left(x\right)=0\right\} \right|.
\end{equation}
We abuse notation and define the nodal count sequence of a standard
graph $\Gamma_{\lv}$, $\phi:\G\rightarrow\N$, by $\phi\left(n\right):=\phi\left(f_{n}\right)$
for any generic eigenfunction $f_{n}$. 

Similarly, the \emph{Neumann} \emph{partition} is the partition of
$\Gamma_{\lv}$ according to the Neumann points and its connected
components are the \emph{Neumann domains} (see Figure \ref{fig:Neumann domain example}).
We define the \emph{Neumann count }of $f$ as the number of Neumann
points, 
\begin{equation}
\mu(f):=\left|\left\{ x\in\Gamma_{\lv}\setminus\V\,|\,f'\left(x\right)=0\right\} \right|,
\end{equation}
and define the Neumann count sequence of a standard graph $\Gamma_{\lv}$,
~$\mu:\G\rightarrow\N$, by $\mu\left(n\right):=\mu\left(f_{n}\right)$
for any generic eigenfunction $f_{n}$. 
\end{defn}

\begin{rem}
The nodal count is usually defined for eigenfunctions of simple eigenvalue
that satisfy property $I$ (see Definition \ref{def: Properties I,II,III-1}).
This is to avoid the ambiguity of how to count a nodal point on a
vertex, and to make sure that it is independent of choice of basis
of eigenfunctions. By the same reasoning the Neumann count should
be defined for eigenfunctions of simple eigenvalues that satisfy property
$II$ (see Definition \ref{def: Properties I,II,III-1}). We restrict
the discussion to generic eigenfunctions in order to have both nodal
and Neumann count sequences defined on the same set of eigenfunctions,
and as we prove in Section \ref{sec: genericity}, generically, the
set of eigenfunctions excluded by this restriction is neglectable. 
\end{rem}

\begin{figure}[h]
\centering{}\includegraphics[width=0.4\textwidth]{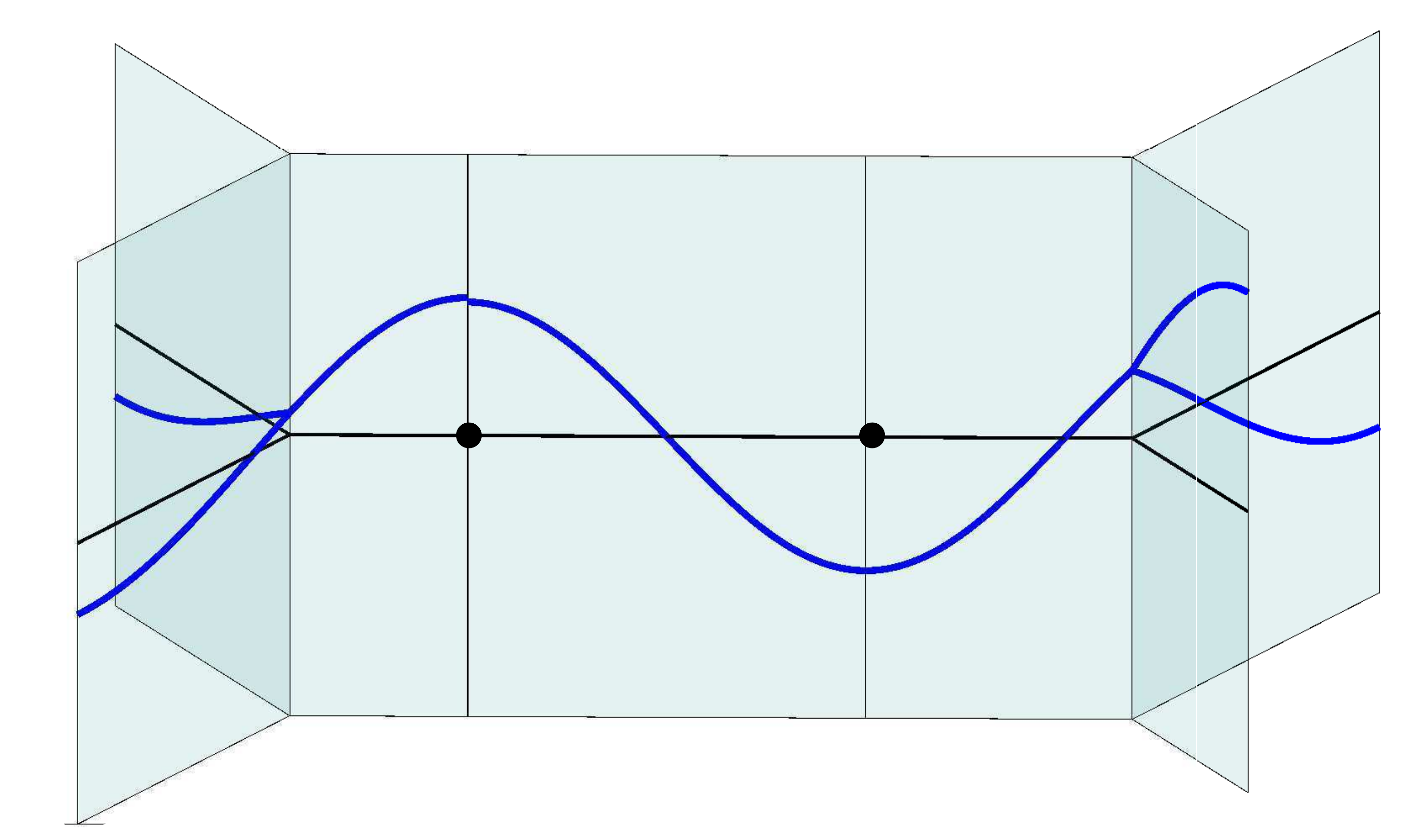}~~~\includegraphics[width=0.4\textwidth]{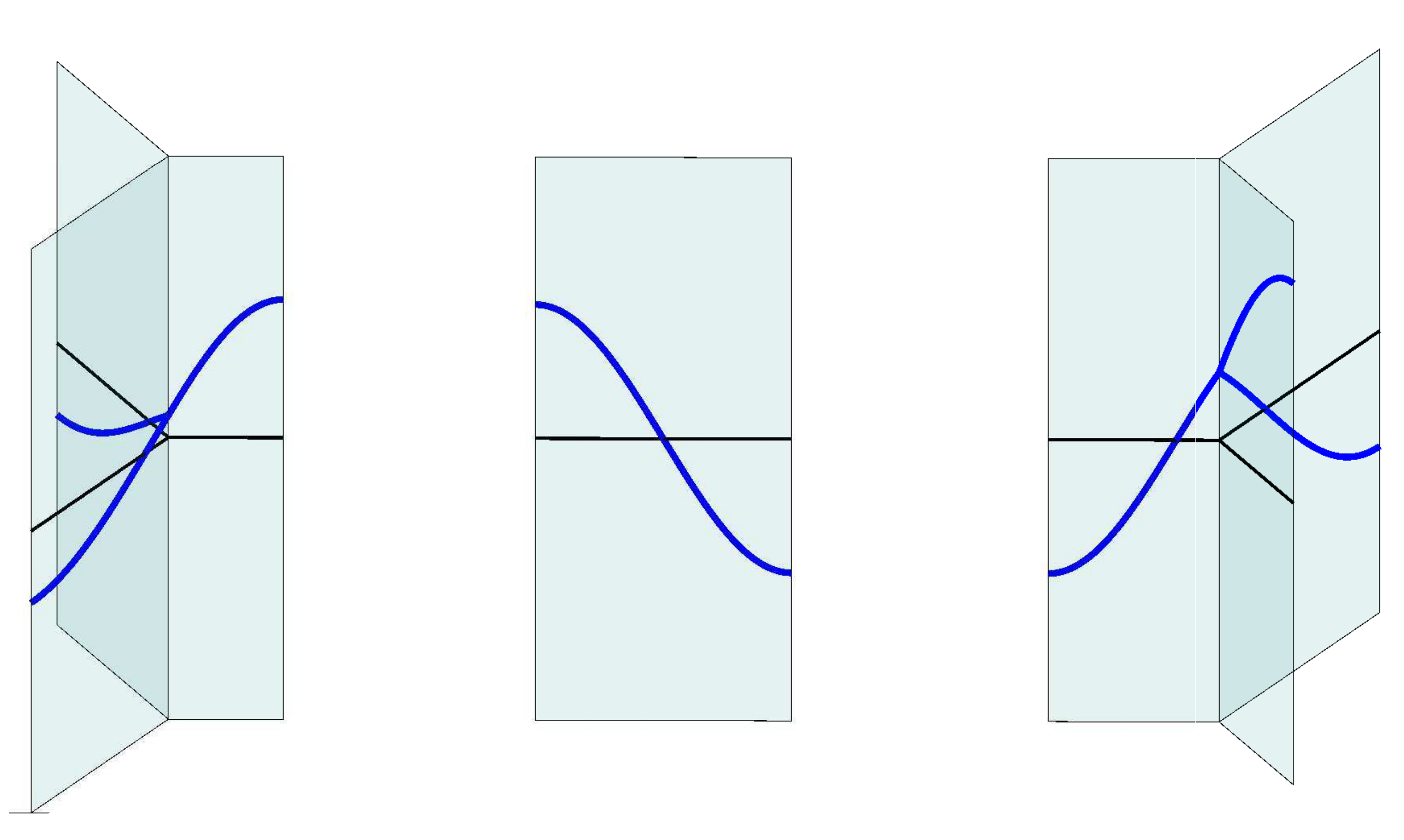}
\caption[Neumann partition]{On the left, a graph of an eigenfunction $f$ over a tree graph $\Gamma_{\protect\lv}$,
$\protect\set{\left(x,y,f\left(x,y\right)\right)}{\left(x,y\right)\in\Gamma_{\protect\lv}}$.
$f$ has two marked Neumann points. On the right, the Neumann partition
of $\Gamma_{\protect\lv}$ according to $f$.}
\label{fig:Neumann domain example}
\end{figure}

\begin{rem}
A nodal\textbackslash Neumann domain is a connected metric graph
on its own. We will consider Neumann domains as standard quantum graphs.
\end{rem}

\newpage{}
\section{\label{sec: Neumann-count-bounds} Neumann count bounds}

In this short section we prove uniform bounds on the Neumann count,
similarly to the bounds obtained for the nodal count. The nodal count
bounds are given by:
\begin{thm}
\cite{GnuSmiWeb_wrm04,Ber_cmp08,BanBerSmi_ahp12} \label{thm: nodal count bounds}Let
$\Gamma_{\lv}$ be a standard graph and assume that $f_{n}$ is generic,
then its nodal count is bounded by 
\begin{equation}
0\le\phi\left(f_{n}\right)-n\le\beta,\label{eq: nodal surplus bound}
\end{equation}
where $\beta$ is the first Betti number of $\Gamma$ (\ref{eq:Betti_number})\footnote{To avoid confusion when comparing to those works, recall that we start
indexing the eigenvalues from $n=0$.}.
\end{thm}

\begin{defn}
\label{def: Surpluses} Let $\Gamma_{\lv}$ be a standard graph and
assume that $f_{n}$ is generic. The \emph{nodal surplus }of $f_{n}$
is define by\emph{ 
\begin{equation}
\sigma\left(f_{n}\right):=\phi\left(f_{n}\right)-n,\label{eq: nodal surplus}
\end{equation}
}and the nodal surplus sequence\footnote{We describe the sequence as a function on its index set to emphasize it is not defined for every $n$.} $\sigma:\G\rightarrow\left\{ 0,1,...\beta\right\} $
is given by $\sigma\left(n\right):=\sigma\left(f_{n}\right)$. We
define the \emph{Neumann surplus} of $f_{n}$ in the same way, denoting\emph{
\begin{equation}
\omega\left(f_{n}\right):=\mu\left(f_{n}\right)-n,\label{eq: Neumann surplus}
\end{equation}
}with the Neumann surplus sequence $\omega:\G\rightarrow\Z$ given
by $\omega\left(n\right):=\omega\left(f_{n}\right)$. 
\end{defn}

\begin{rem}
We call it Neumann surplus as an analog of the nodal surplus, however
one should be aware that we do not suggest it is a non-negative
quantity (unlike the nodal surplus). We do allow a 'negative surplus'
. For example, the Neumann surplus is always negative for trees, as
follows from the next theorem.
\end{rem}

\begin{thm}
\cite{AloBan19}\label{thm:Neumann_surplus_main-1} Let $\Gamma_{\lv}$
be a standard graph whose first Betti number is $\beta$ and whose
boundary size is $\left|\d\Gamma\right|$. Then the Neumann surplus
sequence is uniformly bounded by 
\begin{equation}
\forall n\in\G\,\,\,\,\,\,\,\,1-\beta-\left|\d\Gamma\right|\leq\omega(n)\leq2\beta-1.\label{eq:Neumann_Surplus_bounds-1}
\end{equation}
\end{thm}

This theorem follows from the next lemma that will be useful for other
proofs as well.
\begin{lem}
\label{lem: nodal minus Neumann count} Let $f$ be a real generic
eigenfunction of $\Gamma_{\lv}$ with eigenvalue $k>0$, and let $\phi\left(f\right)$
and $\mu\left(f\right)$ be its nodal and Neumann counts. Then the
difference $\phi\left(f\right)-\mu\left(f\right)$ is given by: 

\begin{align}
\phi\left(f\right)-\mu\left(f\right) & =\frac{\left|\partial\Gamma\right|}{2}-\frac{1}{2}\sum_{v\in\V_{in}}\sum_{e\in\E_{v}}\mathrm{sign}\left(f\left(v\right)\partial_{e}f\left(v\right)\right).\label{eq:Diff-nodal-Neumann_by_vertices}
\end{align}
\end{lem}

\begin{proof}
Let $e$ be an edge of length $l_{e}$ and vertices $v,u$ (not necessarily
distinct). Let $\phi\left(f|_{e}\right)$ be the number of nodal point
on $e$ and similarly $\mu\left(f|_{e}\right)$ for Neumann points.
Consider the amplitude-phase pair $A_{e},\varphi_{e}$ (see Definition
\ref{def: edge restriction notations}) such that $f|_{e}\left(x_{e}\right)=A_{e}\cos\left(kx_{e}-\varphi\right)$
and $f|_{e}'\left(x_{e}\right)=-kA_{e}\sin\left(kx_{e}-\varphi\right)$.
As $f$ is generic, then $A_{e}\ne0$ and so the nodal points and Neumann
points on $e$ interlace and their union is the finite set $\set{x_{e}\in\left(0,l_{e}\right)}{kx_{e}-\varphi\in\frac{\pi}{2}\Z}$.
Let $N=\phi\left(f|_{e}\right)+\mu\left(f|_{e}\right)$ and number
the points 
\[
\left\{ x_{j}\right\} _{j=1}^{N}:=\set{x_{e}\in\left(0,l_{e}\right)}{kx_{e}-\varphi\in\frac{\pi}{2}\Z},
\]
in increasing order according to the distance from $v$. Let $\delta\left(x_{j}\right)=\begin{cases}
1 & x_{j}\,\,is\,\,nodal\\
-1 & x_{j}\,\,is\,\,Neumann
\end{cases}$ so that $\phi\left(f|_{e}\right)-\mu\left(f|_{e}\right)=\sum_{j=1}^{N}\delta\left(x_{j}\right)$.
If $N=0$, then $\phi\left(f|_{e}\right)-\mu\left(f|_{e}\right)=0$.
If $N=1$, then 
\[
\phi\left(f|_{e}\right)-\mu\left(f|_{e}\right)=\delta\left(x_{1}\right)=\frac{\delta\left(x_{1}\right)+\delta\left(x_{N}\right)}{2}.
\]
If $N>1$, then 
\[
\phi\left(f|_{e}\right)-\mu\left(f|_{e}\right)=\frac{\delta\left(x_{1}\right)+\delta\left(x_{N}\right)}{2}+\sum_{j=1}^{N-1}\frac{\delta\left(x_{j}\right)+\delta\left(x_{j+1}\right)}{2}=\frac{\delta\left(x_{1}\right)+\delta\left(x_{N}\right)}{2}.
\]
Where the last equality follows from the interlacing. We deduce that:
\begin{align}
\phi\left(f|_{e}\right)-\mu\left(f|_{e}\right)=\begin{cases}
\frac{\delta\left(x_{1}\right)+\delta\left(x_{N}\right)}{2} & N>0\\
0 & N=0
\end{cases} & .\label{eq: nodal minum nodal according to endpoints}
\end{align}
Assume that $N>0$. Recall that $f\left(v\right)\ne0$ as $f$ is
generic, and let g$\left(x_{e}\right):=\frac{f\left(v\right)}{k}f|_{e}\left(x_{e}\right)$
so that it is positive and monotone in the interval $x_{e}\in\left(0,x_{1}\right)$.
If $x_{1}$ is a nodal point, then $g$ must be decreasing in the interval
and if $x_{2}$ is Neumann, then $g$ must be increasing in the interval.
Notice that $g$ is monotone and $g''\left(0\right)=\frac{f\left(v\right)}{k}f|_{e}''\left(0\right)=-k\left(f\left(v\right)\right)^{2}<0$
so $g$ is decreasing unless $g'\left(0\right)>0$. Since $g'\left(0\right)=f\left(v\right)\partial_{e}f\left(v\right)$,
we get 
\begin{equation}
\delta\left(x_{1}\right)=\begin{cases}
1 & f\left(v\right)\partial_{e}f\left(v\right)\le0\\
-1 & f\left(v\right)\partial_{e}f\left(v\right)>0
\end{cases}.
\end{equation}
The same argument relates $\delta\left(x_{N}\right)$ to $f\left(u\right)\partial_{e}f\left(u\right)$
such that 
\begin{equation}
\frac{\delta\left(x_{1}\right)+\delta\left(x_{N}\right)}{2}=\begin{cases}
\frac{1}{2} & f\left(v\right)\partial_{e}f\left(v\right)\le0\\
-\frac{1}{2} & f\left(v\right)\partial_{e}f\left(v\right)>0
\end{cases}+\begin{cases}
\frac{1}{2} & f\left(u\right)\partial_{e}f\left(u\right)\le0\\
-\frac{1}{2} & f\left(u\right)\partial_{e}f\left(u\right)>0
\end{cases}.\label{eq: sum deltas}
\end{equation}
If $N=0$, then $g$ is monotone on $\left(0,l_{e}\right)$, $g'\left(0\right)=f\left(v\right)\partial_{e}f\left(v\right)$
and $-g'\left(l_{e}\right)=f\left(u\right)\partial_{e}f\left(u\right)$.
At least one vertex of $e$ must be interior vertex, with out loss
of generality assume that $\deg u\ne1$. Then $f$ being generic implies
that $g'\left(l_{e}\right)\ne0$. Since $g''\left(0\right)<0$ and
$g'$ does not vanish on $\opcl{0,l_{e}}$, then either $g'\left(0\right)=0$
and $g'\left(l_{e}\right)<0$ or $g'\left(0\right)$ is of the same
sign as $g'\left(l_{e}\right)$. Namely, the RHS of (\ref{eq: sum deltas})
vanish. The genericity assumption gives $g'\left(0\right)=0\iff\deg v=1$,
and so the latter argument together with (\ref{eq: nodal minum nodal according to endpoints})
and (\ref{eq: sum deltas}) would give, both for $N\ne0$ and $N=0$,
that 
\begin{equation}
\phi\left(f|_{e}\right)-\mu\left(f|_{e}\right)=\begin{cases}
-\frac{\mathrm{sign}\left(f\left(v\right)\partial_{e}f\left(v\right)\right)+\mathrm{sign}\left(f\left(u\right)\partial_{e}f\left(u\right)\right)}{2} & \deg v,\deg u\ne1\\
\frac{1-\mathrm{sign}\left(f\left(u\right)\partial_{e}f\left(u\right)\right)}{2} & \deg v=1
\end{cases}.\label{eq:Diff-nodal-Neumann-on-edge}
\end{equation}
Summing up over all edges, and rearranging the sum to vertices and
adjacent edges, gives
\begin{align*}
\phi\left(f\right)-\mu\left(f\right) & =\frac{1}{2}\sum_{v\in\V}\sum_{e\in\E_{v}}\begin{cases}
1 & \deg v=1\\
-\mathrm{sign}\left(f\left(v\right)\partial_{e}f\left(v\right)\right) & \deg v\ne1
\end{cases}\\
= & \frac{\left|\partial\Gamma\right|}{2}-\frac{1}{2}\sum_{v\in\V_{in}}\sum_{e\in\E_{v}}\mathrm{sign}\left(f\left(v\right)\partial_{e}f\left(v\right)\right).
\end{align*}
\end{proof}
The proof of Theorem \ref{thm:Neumann_surplus_main-1} is almost immediate:
\begin{proof}
Let $f_{n}$ be a real generic eigenfunction. By definition $\phi(n)-\mu(n)=\sigma(n)-\omega(n)$
and according to Lemma \ref{lem: nodal minus Neumann count}, 
\begin{align*}
\left|\sigma(n)-\omega(n)-\frac{\left|\partial\Gamma\right|}{2}\right| & \le\left|\frac{1}{2}\sum_{v\in\V_{in}}\sum_{e\in\E_{v}}\mathrm{sign}\left(f_{n}\left(v\right)\partial_{e}f_{n}\left(v\right)\right)\right|\\
\le & \frac{1}{2}\sum_{v\in\V_{in}}\left|\sum_{e\in\E_{v}}\mathrm{sign}\left(f_{n}\left(v\right)\partial_{e}f_{n}\left(v\right)\right)\right|.
\end{align*}
Given an interior vertex $v\in\V_{in}$, every $f_{n}\left(v\right)\partial_{e}f_{n}\left(v\right)$
is real and non zero since $f$ is real and generic. The Neumann condition
on $v$ implies that $\sum_{e\in\E_{v}}f_{n}\left(v\right)\partial_{e}f_{n}\left(v\right)=0$,
so at least one summand is positive and one is negative and so 
\[
\left|\sum_{e\in\E_{v}}\mathrm{sign}\left(f_{n}\left(v\right)\partial_{e}f_{n}\left(v\right)\right)\right|\le\deg v-2,
\]
which means that 
\begin{align*}
\left|\sigma(n)-\omega(n)-\frac{\left|\partial\Gamma\right|}{2}\right| & \le\frac{1}{2}\sum_{v\in\V_{in}}\left(\deg v-2\right)\\
= & \frac{1}{2}\sum_{v\in\V}\left(\deg v-2\right)-\frac{1}{2}\sum_{v\in\partial\Gamma}\left(-1\right)\\
= & E-V+\frac{1}{2}\left|\partial\Gamma\right|=\beta-1+\frac{1}{2}\left|\partial\Gamma\right|.
\end{align*}
It follows that 
\begin{equation}
1-\beta\leq\sigma(n)-\omega(n)\leq\beta-1+\left|\partial\Gamma\right|.\label{eq:bounds_on_nodal_Neumann_diff}
\end{equation}
Substituting the nodal surplus bounds, $0\le\sigma(n)\le\beta$, gives
\[
1-\beta-\left|\d\Gamma\right|\leq\omega(n)=\sigma(n)-\left(\sigma(n)-\omega(n)\right)\leq\beta+\beta-1=2\beta-1.
\]
\end{proof}
In \cite{AloBan19} we provide examples of graphs for which the sequence
$\sigma\left(n\right)-\omega\left(n\right)$ achieves both its upper
an lower bounds of (\ref{eq:bounds_on_nodal_Neumann_diff}). Unlike
the difference $\sigma\left(n\right)-\omega\left(n\right)$, we conjecture
in \cite{AloBan19} that the Neumann surplus bounds in the case of
$\beta\ge2$ can be improved:
\begin{conjecture}
\cite{AloBan19} \label{conj: Neumann surplus bounds}The Neumann
surplus sequence of a standard graph $\Gamma_{\lv}$ with first Betti
number $\beta\ge2$ is bounded by 
\[
-1-\left|\partial\Gamma\right|\le\omega\left(n\right)\le\beta+1.
\]
\end{conjecture}

\begin{rem}
\label{rem: bounda correlation} It is shown in \cite{AloBan19} that
both $\sigma\left(n\right)$ and $\sigma\left(n\right)-\omega\left(n\right)$
can achieve their upper and lower bounds. Therefore, if the conjecture
is true, then the nodal surplus sequence and the Neumann surplus sequence
are not independent, in terms of the statistics developed in Section
\ref{sec: proof-of-existence}. 
\end{rem}

\newpage{}
\section{\label{sec: The-secular-manifold} The secular manifold}

\subsection{Introduction to the secular manifold}

The secular manifold, that we will define and discuss in the following
section, was first introduced in the work of Barra and Gaspard in
\cite{BarGas_jsp00} in which they use ergodic arguments to calculate
the level spacing statistics, by means of averages on the secular
manifold. The secular manifold appeared to be useful even beyond spectral
statistics. For example, it was used by Colin de Verdière in \cite{CdV_ahp15}
to prove that there is no quantum unique ergodicity in quantum graph
and to describe the possible semiclassical limiting measures. It is
also used by Band in \cite{Ban_ptrsa14}, proving an inverse problem
of showing that the nodal count $\phi\left(n\right)=n$ implies the
graph is a tree. In a recent work of Kurasov and Sarnak \cite{Sarnakurasov},
they analyze the secular manifold from an algebraic point of view.
In this work, they classified the spectrum of certain quantum graphs
as crystalline measures that contain only finite arithmetic progression
with a uniform bound on the lengths of these progressions.

\subsection{Abstract definition of the secular manifold}

Given $k>0$ we denote the $k^{2}$ eigenspace of $\Gamma_{\lv}$
by $Eig\left(\Gamma_{\lv},k^{2}\right)$. That is $Eig\left(\Gamma_{\lv},k^{2}\right)=\ker\left(\Delta-k^{2}\right)$
restricted to the domain of Neumann vertex conditions. A scaling of
the graph by a factor of $t>0$, $\Gamma_{\lv}\mapsto\Gamma_{t\lv}$,
induces a bijection $f\left(x\right)\mapsto f\left(\frac{x}{t}\right)$
between $H^{2}\left(\Gamma_{\lv}\right)$ and $H^{2}\left(\Gamma_{t\lv}\right)$.
It is not hard to deduce that it preserves the values of $f$ on vertices
and scale the derivatives by a factor of $\frac{1}{t}$. In particular,
it preserve the Neumann vertex conditions. If $f\in Eig\left(\Gamma_{\lv},k^{2}\right)$
for $k>0$, then $f\left(\frac{x}{k}\right)$ on $\Gamma_{k\lv}$ satisfy
Neumann vertex conditions and $-\frac{d^{2}}{d^{2}x_{e}}f|_{e}\left(\frac{x_{e}}{k}\right)=\frac{1}{k^{2}}k^{2}f|_{e}\left(\frac{x_{e}}{k}\right)=f|_{e}\left(\frac{x_{e}}{k}\right)$
on every edge and so $f\left(\frac{x}{k}\right)\in Eig\left(\Gamma_{k\lv},1\right)$.
A simple check shows that all amplitudes in Definition \ref{def: edge restriction notations}
are preserved under such scaling, and so for any $k>0$ and $\lv\in\left(\R_{+}\right)^{\E}$,
$Eig\left(\Gamma_{\lv},k^{2}\right)$ and $Eig\left(\Gamma_{k\lv},1\right)$
are isomorphic by a map which preserves all amplitudes.

Consider $Eig\left(\Gamma_{\lv},1\right)$ and let $\lv$ range over
$\left(\R_{+}\right)^{\E}$. The restrictions of any $f\in Eig\left(\Gamma_{\lv},1\right)$
to edges are $2\pi$ periodic (see Definition \ref{def: edge restriction notations}),
and so $f$ can be extended uniquely to $\Gamma_{l+2\pi\boldsymbol{n}}$
for any $\boldsymbol{n}\in\N^{\E}$. This extension is a linear bijection
between $Eig\left(\Gamma_{\lv},1\right)$ and $Eig\left(\Gamma_{\lv+2\pi\boldsymbol{n}},1\right)$
which preserves all amplitudes. In fact this is true for all $\boldsymbol{n}\in\Z^{\E}$
such that $\lv+2\pi\boldsymbol{n}\in\R_{+}^{\E}$. We can conclude
that for any $\lv\in\left(\R_{+}\right)^{E},k>0$ and $\boldsymbol{n}\in\Z^{\E}$
such that $k\lv+2\pi\boldsymbol{n}\in\R_{+}^{\E}$: 
\begin{equation}
Eig\left(\Gamma_{\lv},k^{2}\right)\undercom{\cong}{Scaling\,by\,k}Eig\left(\Gamma_{k\lv},1\right)\undercom{\cong}{extending\,by\,\boldsymbol{n}}Eig\left(\Gamma_{k\lv+2\pi\boldsymbol{n}},1\right),\label{eq: equivalence of eigenspaces}
\end{equation}
with an isomorphism that preserves all amplitudes from Definition
\ref{def: edge restriction notations}. 
\begin{rem}
One may ask what happens for $k=0$, as we cannot scale the graph
by $0$, but we would expect an isomorphism between $Eig\left(\Gamma_{\lv},0\right)$
and $Eig\left(\Gamma_{2\pi\boldsymbol{n}},1\right)$ for any $\boldsymbol{n}\in\N^{\E}$.
It appears that this isomorphism holds only for tree graphs. In the
following subsection we will show that $\dim Eig\left(\Gamma_{2\pi\boldsymbol{n}},1\right)=\beta+1$,
where $\beta$ is the first Betti number of $\Gamma_{\lv}$. As we
already mentioned that $0$ is always a simple eigenvalue (since $\Gamma$
is connected), then $Eig\left(\Gamma_{\lv},0\right)\not\cong Eig\left(\Gamma_{2\pi\boldsymbol{n}},1\right)$
if $\beta>0$. In the $\beta=0$ case, namely a tree, the mapping
of $f\equiv C$ to $\tilde{f}\in Eig\left(\Gamma_{2\pi\boldsymbol{n}},1\right)$
given by $\tilde{f}|_{e}\left(x_{e}\right):=C\cos\left(x_{e}\right)$,
is a bijection between $Eig\left(\Gamma_{\lv},0\right)$ and $Eig\left(\Gamma_{2\pi\boldsymbol{n}},1\right)$
that preserve the trace of the functions. 
\end{rem}

It is clear from (\ref{eq: equivalence of eigenspaces}) that given
any $\lv$ and any $k^{2}>0$ eigenvalue of $\Gamma_{\lv }$, if we denote $\tilde{l}\in\opcl{0,2\pi}^{\E}$ such that $\tilde{l}=k\lv\,\,\,mod\,2\pi$
then, $Eig\left(\Gamma_{\lv},k^{2}\right)\cong Eig\left(\Gamma_{\tilde{l}},1\right)$
by (\ref{eq: equivalence of eigenspaces}). The \emph{secular manifold},
to be defined next, is the set of all such $\tilde{l}$ for every
possible pair of $\lv$ and $k^{2}>0$. 
\begin{defn}
\label{def: charecteristic torus}We denote the flat torus $\T^{n}:=\left(\R/2\pi\Z\right)^{n}$
for every $n\in\N$, and define the \emph{characteristic torus} of
a graph by $\T^{\E}$. We use the notation $\left\{ *\right\} :\R^{\E}\rightarrow\T^{\E}$
to denote $\left\{ \vec{x}\right\} =\vec{x}\,\,\,mod\,2\pi$. 
\end{defn}

\begin{defn}
\label{def: secular manifold}Given a discrete graph $\Gamma$, for
every point $\kv\in\T^{\E}$ we associate a standard graph denoted
by $\Gamma_{\kv}$, whose edge lengths are $\lv\in\opcl{0,2\pi}^{\E}$
such that $\left\{ \lv\right\} =\kv$. The \emph{secular manifold
}of $\Gamma$ is defined as follows: 
\begin{equation}
\Sigma:=\set{\kv\in\T^{E}}{\dim Eig\left(\Gamma_{\kv},1\right)\ge1}.
\end{equation}
We partition $\Sigma$ into ``singular'' and ``regular'' parts
which are define as follows:
\begin{align}
\Sigma^{reg} & :=\set{\kv\in\Sigma}{\dim Eig\left(\Gamma_{\kv},1\right)=1},\label{eq: sigman reg def}\\
\Sigma^{sing} & :=\set{\kv\in\Sigma}{\dim Eig\left(\Gamma_{\kv},1\right)>1}.\label{eq: Sigma sing def}
\end{align}
\end{defn}

It is now clear from (\ref{eq: equivalence of eigenspaces}) that: 
\begin{lem}
\label{lem: Sigma reg and simple eigenvalues}Let $\Gamma_{\lv}$
be a standard graph, let $k^{2}>0$ and let $\kv=\left\{ k\lv\right\} $.
Then $k^{2}$ is an eigenvalue of $\Gamma_{\lv}$ if and only if $\kv\in\Sigma$
and it is a simple eigenvalue if and only if $\kv\in\Sigma^{reg}$.
\end{lem}

The terms ``regular'' and ``singular'' in Definition \ref{def: secular manifold}
refer to the structure of the secular manifold as presented in the
following proposition.
\begin{defn}
\cite{RealAnalyticFunction}\label{def: real analytic variety} A
\emph{real analytic manifold} is a smooth manifold whose transition
maps are real analytic. Given a real analytic function $f:\R^{n}\rightarrow\R^{m}$,
its zero set $Z_{f}$ is called a \emph{real analytic variety. }A
point $x\in Z_{f}$ is called \emph{regular} if it has a neighborhood
of $Z_{f}$ which is a manifold, and is called \emph{singular} if
it is not regular. The regular part of $Z_{f}$ which we denote by
$Z_{f}^{reg}$ is the union of regular points. The dimension of $Z_{f}$
is the maximal dimension of neighborhoods of regular points. 
\end{defn}

\begin{prop}
\label{prop: Secular manifold structure} Given a graph $\Gamma$
with $E$ edges, its secular manifold $\Sigma$ is a real analytic
variety of dimension $E-1$. The set $\Sigma^{reg}$ that was defined
in (\ref{eq: sigman reg def}) is the set of regular points of $\Sigma$,
and it is a real analytic manifold of dimension $E-1$. The set $\Sigma^{sing}$
that was defined in (\ref{eq: Sigma sing def}) is the set of singular
points of $\Sigma$, and it is a real analytic variety of dimension
strictly smaller than $E-1$ (positive co-dimension in $\Sigma$).
\end{prop}

A proof for this proposition will be given in Subsection \ref{proof of structure},
but we will state here two usefull lemmas regarding real analytic
varieties that will be used in the proof of Proposition \ref{prop: Secular manifold structure}
as well as in other proofs.
\begin{lem}
\cite{RealAnalyticFunction} \label{lem: stratification} ~If $Z_{f}$
is a real analytic variety of dimension $n$, then it has the following
stratification $Z_{f}=\cup_{j=0}^{n}S_{j}$ where every $S_{j}$,
which is called a \emph{strata}, is a $j$-dimensional real analytic
manifold (possibly empty). 
\end{lem}

\begin{lem}
\label{lem: real analytic lemma}Let $M$ be a connected real analytic
manifold and let $h$ be a real analytic function on $M$ which is
not the zero function. Then the zero set
\begin{align}
Z_{h}:= & \set{\xv\in M}{h\left(\xv\right)=0},
\end{align}
is a real analytic variety of positive co-dimension in $M$. 
\end{lem}

\begin{proof}
Clearly, if $M$ is a zero set of a real analytic function $f$, then
$Z_{h}$ is the zero set of the real analytic function $\left(f,h\right)$
and is therefore an analytic variety. Let $\left\{ O_{n},\varphi_{n}\right\} _{n\in\N}$
be the real analytic atlas of $M$, with $\left\{ U_{n}\right\} _{n\in\N}$
open subsets of $\R^{d}$ (where $d=\dim\left(M\right)$), each $U_{n}$
diffeomorphic to $O_{n}$ by $U_{n}=\varphi_{n}\left(O_{n}\right)$.
Let $h_{n}:=h\circ\varphi_{n}^{-1}$ so that by definition $h_{n}:U_{n}\rightarrow\R$
is real analytic. By proposition 3 in \cite{Mit15} the zero set of
$h_{n}$ is of positive co-dimension in $U_{n}$ and therefore its
preimage by $\varphi_{n}$, $Z_{h}\cap O_{n}$ has positive co-dimension
in $M$. Therefore $Z_{h}$ is a countable union of sets of positive
co-dimension and as such has positive co-dimension.
\end{proof}

\subsection{Canonical eigenfunctions}

In the context of quantum mechanics and spectral theory it is standard
to consider eigenfunctions which are $L^{2}$ normalized. However,
in our context a different normalization would be more convenient. 
\begin{defn}
\label{def: normalization} Given an eigenfunction $f$ with amplitudes
vector $\boldsymbol{a}\in\C^{\vec{\E}}$ as defined in Definition
\ref{def: edge restriction notations}, we say that $f$ is normalized
if $\norm{\boldsymbol{a}}=1$ (Euclidean norm).
\end{defn}

For a regular point $\kv\in\Sigma^{reg}$, where $\dim Eig\left(\Gamma_{\kv},1\right)=1$,
we define a canonical choice (up to sign) of eigenfunction $f_{\kv}\in Eig\left(\Gamma_{\kv},1\right)$. 
\begin{defn}
\label{def: canonical e.f} Let $\kv\in\Sigma^{reg}$, we define its
\emph{canonical eigenfunction} as the unique (up to a sign) normalized
real eigenfunction $f_{\kv}\in Eig\left(\Gamma_{\kv},1\right)$.
\end{defn}

\begin{lem}
\label{lem: vertex values canonical ef and other ef }~Let $\Gamma_{\lv}$
be a standard graph, let $k^{2}>0$ be a simple eigenvalue of $\Gamma_{\lv}$
with a normalized real eigenfunction $f$ and let $\kv=\left\{ k\lv\right\} \in\Sigma^{reg}$
with canonical eigenfunction $f_{\kv}$. Then $f$ and $f_{\kv}$
share the same amplitudes vector and so their traces are related as
follows:
\begin{equation}
\forall v\in\V,\,\,\forall e\in\E_{v}\,\,\,\,\,\,f_{\kv}\left(v\right)=f\left(v\right)\,\,\,\,\,and\,\,\,\,\partial_{e}f_{\kv}\left(v\right)=\frac{\partial_{e}f\left(v\right)}{k}.\label{eq: canonical eigenfunction and other function}
\end{equation}
Where the equalities are up to a global sign. 
\end{lem}

\begin{proof}
The bijection between $Eig\left(\Gamma_{\lv},k^{2}\right)$ and $Eig\left(\Gamma_{\kv},1\right)$
is given as a composition of two maps, scaling by $k$ and extensions
of the edges by integer multiples of $2\pi$. Let $\tilde{f}\in Eig\left(\Gamma_{k\lv},1\right)$
given by the scaling bijection, namely $\tilde{f}|_{e}\left(x_{e}\right)=f|_{e}\left(\frac{x_{e}}{k}\right)$
$\forall e\in\E,\forall x_{e}\in\left[0,kl_{e}\right]$. It is not
hard to deduce that $f\left(v\right)=\tilde{f}\left(v\right)$ and
$\partial_{e}f\left(v\right)=k\partial_{e}\tilde{f}\left(v\right)$
for every $v\in\V,e\in\E_{v}$. Therefore, according to the relations
between the trace and amplitudes in Lemma \ref{def: edge restriction notations}
(\ref{enu: complex amplitudes}), $\tilde{f}$ and $f$ share the
same amplitudes and therefore $\tilde{f}$ is also real with normalized
amplitudes vector. It is not hard to deduce that the second bijection,
the extension by integer multiples of $2\pi$, does not change the
trace and the amplitudes vector, and therefore $\tilde{f}$ is mapped
to a real eigenfunction in $Eig\left(\Gamma_{\kv},1\right)$ with
normalized amplitudes vector, namely $\pm f_{\kv}$. This proves the
lemma. 
\end{proof}

\subsection{\label{subsec:Wave-scattering}Wave scattering and explicit construction
of the secular manifold.}

The scattering approach was first applied to quantum graphs by Von
Below \cite{Bel_laa85} and later on by Kottos and Smilansky in \cite{KotSmi_ap99,KotSmi_prl97}.
In this subsection we review this procedure and use it to construct
the secular manifold. For more information on the scattering approach
we refer the reader to \cite{BerKuc_graphs,GnuSmi_ap06}.

Let $f$ be a real eigenfunction of a standard graph $\Gamma_{\lv}$
with eigenvalue $k^{2}>0$, and denote its amplitudes vector by $\boldsymbol{a}\in\C^{\vec{\E}}$
(see Definition \ref{def: edge restriction notations} (\ref{enu: complex amplitudes})).
The vertex conditions on a vertex $v\in\V$ provides $\deg v$ linear
equations on $\tr\left(f\right)$:
\begin{align*}
f|_{e}\left(v\right) & =f|_{e'}\left(v\right)\,\,\,\forall e,e'\in\E_{v},\\
\sum_{e\in\E_{v}}\partial_{e}f\left(v\right) & =0,
\end{align*}
which can be written, using Definition \ref{def: edge restriction notations},
(\ref{enu: complex amplitudes}), as $\deg v$ linear equations on
$a$ with coefficients that are linear in $\left\{ e^{ikl_{e}}\right\} _{e\in\E}$.
Over all, it is a system of $\sum_{v\in\V}\deg v=2E=\left|\vec{\E}\right|$
linear equations on $a$ that can be rearranged as such (see \cite{GnuSmi_ap06,BerKuc_graphs}
for detailed explanation): 
\begin{equation}
\left(1-U_{k\lv}\right)a=0,\label{eq:sec_condition}
\end{equation}
where $U_{k\lv}$ is unitary matrix from $\C^{\vec{\E}}$ to itself,
called the \emph{unitary evolution matrix, }whose\emph{ }entries are
linear functions of $\left\{ e^{ikl_{e}}\right\} _{e\in\E}$. We define
$U_{\kv}$ accordingly as function of $\kv\in\T^{\E}$ such that $U_{k\lv}=U_{\kv}$
for $\kv=\left\{ k\lv\right\} $. Notice that the mapping $a\mapsto f$
given by the restrictions $f|_{e}\left(x_{e}\right)=a_{e}e^{-ikl_{e}}e^{ikx}+a_{\hat{e}}e^{-ikx}$
is a linear bijection (according to Definition \ref{def: edge restriction notations}
(\ref{enu: complex amplitudes}) ) between $\C^{\vec{\E}}$ and functions
that satisfy $f''=-k^{2}f$ edgewise. As the vertex conditions of
$f$ in terms of $\boldsymbol{a}$ are given in (\ref{eq:sec_condition}),
then 
\begin{align}
f\in Eig\left(\Gamma_{\kv},1\right) & \iff\boldsymbol{a}\in\ker\left(1-U_{k\lv}\right).
\end{align}
We may deduce the following lemma:
\begin{lem}
\label{lem: secular bijection}The mapping $\boldsymbol{a}\mapsto f$
given by Definition \ref{def: edge restriction notations} (\ref{enu: complex amplitudes})
is a linear bijection between $\ker\left(1-U_{k\lv}\right)$ and $Eig\left(\Gamma_{\lv},k^{2}\right)$
for $k>0$. In particular, $Eig\left(\Gamma_{\kv},1\right)\cong\ker\left(1-U_{\kv}\right)$
for any $\kv\in\T^{\E}$ and $\Sigma$ is the zero set of $\det\left(1-U_{\kv}\right)$.
This allows to express (\ref{eq: sigman reg def}) and (\ref{eq: Sigma sing def})
as 
\begin{align}
\Sigma^{reg} & :=\set{\kv\in\Sigma}{\dim\ker\left(1-U_{\kv}\right)=1},\label{eq: sigman reg def-1}\\
\Sigma^{sing} & :=\set{\kv\in\Sigma}{\dim\ker\left(1-U_{\kv}\right)>1}.\label{eq: Sigma sing def-1}
\end{align}
\end{lem}

The structure of $U_{\kv}$ and its $\kv$ dependence are given by
a product of two unitary matrices 
\begin{equation}
U_{\kv}=e^{i\kv}S.\label{eq: U=00003Dexp J S}
\end{equation}
This is a decomposition of $U$ to a $\kv$ dependent matrix $e^{i\kv}$
and a fixed\footnote{This is true for Neumann vertex conditions. For other vertex conditions, this matrix may be $k$ dependent, in which the secular manifolds approach will need some modifications.} real orthogonal matrix $S$ which is called the \emph{bond-scattering
matrix}. The matrix $e^{i\kv}$ is a unitary diagonal matrix with
diagonal entries 
\begin{equation}
\forall e\in\E\,\,\left(e^{i\kv}\right)_{e,e}=\left(e^{i\kv}\right)_{\hat{e},\hat{e}}=e^{i\kappa_{e}}.
\end{equation}
Given two directed edges $e$ and $e'$ that are connected by a vertex
$v$, we write $e\xrightarrow[v]{}e'$ if $e$ is directed into $v$
and $e'$ is directed out of $v$. The matrix $S$ satisfies $\left(S\right)_{e,e'}\ne0$
if and only if $e'\xrightarrow[v]{}e$. In such case, if $e\xrightarrow[v]{}e'$,
then \emph{ }
\begin{equation}
\left(S\right)_{e,e'}=\begin{cases}
\frac{2}{\deg v}-1 & if\,e'=\hat{e}\\
\frac{2}{\deg v} & otherwise
\end{cases},\label{eq: S explicietly}
\end{equation}
where $\hat{e}$ denotes the opposite direction of $e$. If we introduce the reflection $J$, a permutation matrix defined
by $J\left(e\right)=\hat{e}$ for every directed edge, then it is
not hard to deduce that $JS$ is block diagonal with blocks that we
denote by $\left(JS\right)_{v}$ corresponding to edges that are directed
into the vertex $v$. Every off-diagonal entry of the block $\left(JS\right)_{v}$
is equal to $\frac{2}{\deg v}$ and every diagonal entry is equal
to $\frac{2}{\deg v}-1$. Notices that $JS=2P-1$ where $P$ is the
$\deg v\times\deg v$ matrix all of whose entries are $\frac{1}{\deg v}$.
It is a simple check to see that $P$ is a rank one orthogonal projection
and therefore $\det\left(JS\right)_{v}=\det\left(2P-1\right)_{v}=\left(-1\right)^{\deg v-1}$.
It is easy to show that in the basis of $\left(e_{1},\hat{e}_{1},e_{2},\hat{e}_{2}...\right)$,
$J$ is block diagonal with blocks $\begin{pmatrix}0 & 1\\
1 & 0
\end{pmatrix}$ and so $\det\left(J\right)=\left(-1\right)^{E}$. Therefore, $\det\left(S\right)=\left(-1\right)^{\sum_{v}\deg v-V+E}=\left(-1\right)^{\beta-1}$,
using $\sum_{v}\deg v=2E$ and $\beta=E-V+1$. This gives 
\begin{equation}
\det\left(U_{\kv}\right)=\left(-1\right)^{\beta-1}\det\left(e^{i\kv}\right)=\left(-1\right)^{\beta-1}e^{i\sum_{e\in\E}2\kappa_{e}}.\label{eq: det U}
\end{equation}

\begin{rem}
This construction can be done to arbitrary vertex conditions (by changing
$S$ accordingly) and not only Neumann vertex conditions. But if the
vertex conditions are not preserved under the scaling $f\left(x\right)\mapsto f\left(\frac{x}{k}\right)$
then the matrix $S$ depends on the eigenvalue $k^{2}$. In such case
$\ker\left(1-U_{k\lv}\right)\cong Eig\left(\Gamma_{k\lv},1\right)$
but these are not isomorphic to $Eig\left(\Gamma_{\lv},k^{2}\right)$
in general, and the analysis on the secular manifold can not be proceeded.
A possible solution for general vertex conditions is to consider an
asymptotic $k$ dependent scattering matrix \cite{BerKuc_graphs,GnuSmi_ap06,rueckriemen2012trace}. 
\end{rem}

As stated in Lemma \ref{lem: secular bijection}, $\Sigma$ is the
zero set of $\det\left(1-U_{\kv}\right)$. As $\det\left(1-U_{\kv}\right)$
is a complex function we define a real function that vanish whenever
$\det\left(1-U_{\kv}\right)$ does.
\begin{defn}
\label{def: secular function} We define the \emph{secular function
on $\T^{\E}$ }by
\begin{align}
F\left(\kv\right):= & \det\left(U_{\kv}\right)^{-\frac{1}{2}}\det\left(1-U_{\kv}\right),
\end{align}
where we consider the square-root branch $\det\left(U_{\kv}\right)^{-\frac{1}{2}}=i^{\left(\beta-1\right)}e^{-i\sum_{e\in\E}\kappa_{e}}$.
\end{defn}

Our ``implicit'' definition of the secular manifold and its partition
to regular and singular in Definition \ref{def: secular manifold}
is different than its ``explicit'' definition as the zero set of
the secular function, as was defined in prior works such as \cite{BarGas_jsp00,BerWin_tams10,CdV_ahp15}
for example. The following lemma justify the equivalence of these
definitions:
\begin{lem}
\label{lem:The-secular-functions properties}The secular function
$F$ is a real (multi-variable) trigonometric polynomial. The secular
manifold $\Sigma$ is the zero set of $F$, and its singular part
is the set on which $\nabla F=0$. That is, 
\begin{align}
\Sigma & =\set{\kv\in\T^{\E}}{F\left(\kv\right)=0}\label{eq: secular manifold in F}\\
\Sigma^{reg} & =\set{\kv\in\T^{\E}}{F\left(\kv\right)=0\,\,\text{and}\,\,\nabla F\left(\kv\right)\ne0}\label{eq: regular part in F}\\
\Sigma^{sing} & =\set{\kv\in\T^{\E}}{F\left(\kv\right)=0\,\,\text{and}\,\,\nabla F\left(\kv\right)=0}.\label{eq: singular part in F}
\end{align}
\end{lem}

\begin{rem}
In \cite{CdV_ahp15} the secular manifold is called the `determinant
manifold'. Both the names follows from the description of $\Sigma$
as the zero set of $F$ and $\det\left(1-U_{\kv}\right)$.

We will prove Lemma \ref{lem:The-secular-functions properties}, together
with the next lemma that relates the gradient of $F$ at a regular
point $\kv$ to the canonical eigenfunction $f_{\kv}$ at that point.
\end{rem}

\begin{defn}
\label{def: mk weights}Let $\kv\in\Sigma^{reg}$ and let $f_{\kv}$
be the canonical eigenfunction with amplitudes vector $\boldsymbol{a}$.
We define its weights vector $m_{\kv}\in\left[0,1\right]^{\E}$ by
\begin{equation}
\left(m_{\kv}\right)_{e}=\left|a_{e}\right|^{2}+\left|a_{\hat{e}}\right|^{2}.\label{eq: define m_e}
\end{equation}
Equivalently, if $v\in\V$ is connected to $e$, then $\left(m_{\kv}\right)_{e}=\frac{f_{\kv}\left(v\right)^{2}+\partial_{e}f_{\kv}\left(v\right)^{2}}{2}$.
\end{defn}

\begin{rem}
The above weights play a special role in the work of Colin de Verdière
\cite{CdV_ahp15} on quantum ergodicity fro quantum graphs. 
\end{rem}

\begin{defn}
\label{def: p}We define an auxiliary function
\begin{align}
p\left(\kv\right):= & -i\det\left(U_{\kv}\right)^{-\frac{1}{2}}\mathrm{trace}\left(\mathrm{adj}\left(1-U_{\kv}\right)\right),
\end{align}
where $\det\left(U_{\kv}\right)^{-\frac{1}{2}}=i^{\left(\beta-1\right)}e^{-i\sum_{e\in\E}\kappa_{e}}$
and $\mathrm{adj}\left(1-U_{\kv}\right)$ is the adjugate matrix of
$\left(1-U_{\kv}\right)$. 
\end{defn}

\begin{lem}
\label{lem: F p and mk}The auxiliary function $p$ is a trigonometric
polynomial, and it is real when restricted to $\Sigma$. The gradient
of the secular function $\nabla F$ is proportional to the weights
vector $m_{\kv}$ with a factor $p$:
\begin{equation}
\forall\kv\in\Sigma^{reg}\,\,\nabla F\left(\kv\right)=p\left(\kv\right)m_{\kv}.\label{eq: grad F}
\end{equation}
In particular, all non-vanishing entries of $\nabla F$ share the
same sign. Moreover, the regular and singular parts of $\Sigma$ are
characterized by $p$,
\begin{align}
\Sigma^{reg} & =\set{\kv\in\T^{\E}}{F\left(\kv\right)=0\,\,and\,\,p\left(\kv\right)\ne0},\label{eq: Sigma reg with p}\\
\Sigma^{sing} & =\set{\kv\in\T^{\E}}{F\left(\kv\right)=0\,\,and\,\,p\left(\kv\right)=0},\label{eq: singular part using p}
\end{align}
and the normal to $\Sigma^{reg}$ at $\kv$ is given by 
\begin{equation}
\hat{n}\left(\kv\right)=\frac{m_{\kv}}{\norm{m_{\kv}}}.\label{eq: normal}
\end{equation}
\end{lem}

In order to prove Lemmas \ref{lem:The-secular-functions properties}
and \ref{lem: F p and mk} we need to discuss some properties of the
adjugate matrix:
\begin{lem}
\label{lem: adjugate}Let $U$ be a unitary $n$-dimensional matrix
with eigenvalues $\left\{ e^{i\theta_{j}}\right\} _{j=1}^{n}$ and
(orthonormal) eigenvectors $\left\{ \boldsymbol{a}_{j}\right\} _{j=1}^{n}$,
then
\begin{enumerate}
\item The adjugate matrix $\mathrm{adj}\left(1-U\right)$ satisfies
\begin{equation}
\mathrm{adj}\left(1-U\right)=\sum_{j=1}^{n}\Pi_{i\ne j}\left(1-e^{i\theta_{i}}\right)\boldsymbol{a}_{j}\boldsymbol{a}_{j}^{*}.\label{eq: adjugate in general -inverted}
\end{equation}
In particular $\mathrm{adj}\left(1-U\right)=0$ if and only if $\dim\ker\left(1-U\right)\ge2$.
\\
\item Let $\dim\ker\left(1-U\right)=1$ and let $\boldsymbol{a}\in\ker\left(1-U\right)$
be a normalized vector. Let us number the eigenvalues such that $e^{i\theta_{1}}=1$
and $e^{i\theta_{j}}\ne1$ for $j\ge2$. Then 
\begin{align}
\mathrm{adj}\left(1-U\right) & =\Pi_{j=2}^{n}\left(1-e^{i\theta_{j}}\right)\boldsymbol{a}\boldsymbol{a}^{*},\label{eq: adjugate in general}\\
\mathrm{trace}\left(\mathrm{adj}\left(1-U\right)\right) & =\Pi_{j=2}^{n}\left(1-e^{i\theta_{j}}\right)\ne0.\label{eq: trace of adjugate}
\end{align}
\item Let $U_{t}$ be a $t\in\R$ dependent family of unitary matrices with
a self-adjoint matrix $A$ such that $\frac{d}{dt}U_{t}=iAU_{t}$.
Let $t_{0}$ such that $\dim\ker\left(1-U_{t_{0}}\right)=1$ and let
$\boldsymbol{a}\in\ker\left(1-U_{t_{0}}\right)$ be a normalized vector
. Then 
\begin{align}
\frac{d}{dt}\det\left(1-U_{t}\right)= & -i\cdot\mathrm{trace}\left(\mathrm{adj}\left(1-U_{t_{0}}\right)\right)\left\langle \boldsymbol{a},A\boldsymbol{a}\right\rangle .\label{eq: d(det(1-U)) in general}
\end{align}
\end{enumerate}
\end{lem}

\begin{proof}
If $1-U$ is invertible, then its adjugate matrix satisfies,
\begin{align}
\mathrm{adj}\left(1-U\right)=\det\left(1-U\right)\left(1-U\right)^{-1}.
\end{align}
Since $\det\left(1-U\right)=\Pi_{j=1}^{n}\left(1-e^{i\theta_{j}}\right)$
and 
\[
\left(1-U\right)^{-1}=\sum_{j=1}^{n}\frac{1}{1-e^{i\theta_{j}}}\boldsymbol{a}_{j}\boldsymbol{a}_{j}^{*},
\]
then
\begin{equation}
\mathrm{adj}\left(1-U\right)=\sum_{j=1}^{n}\frac{1}{1-e^{i\theta_{j}}}\Pi_{i=1}^{n}\left(1-e^{i\theta_{i}}\right)\boldsymbol{a}_{j}\boldsymbol{a}_{j}^{*}=\sum_{j=1}^{n}\Pi_{i\ne j}\left(1-e^{i\theta_{i}}\right)\boldsymbol{a}_{j}\boldsymbol{a}_{j}^{*}.\label{eq: adjugate general}
\end{equation}
By definition every entry of $\mathrm{adj}\left(1-U\right)$ is a
minor of $1-U$ up to a sign, so it is continuous in the entries of
$1-U$. As the eigenvalues and eigenvectors are also continuous and
invertible matrices are dense within all matrices, then (\ref{eq: adjugate general})
may be extended to every matrix with such spectral decomposition,
thus proving (\ref{eq: adjugate in general -inverted}). Observe that
since $\left\{ \boldsymbol{a}_{j}\right\} _{j=1}^{n}$ are orthogonal
then a matrix $\sum c_{j}\boldsymbol{a}_{j}\boldsymbol{a}_{j}^{*}$
is equal to zero if and only if every $c_{j}=0$. For the above adjugate
matrix, all $\Pi_{i\ne j}\left(1-e^{i\theta_{i}}\right)$ coefficients
vanish if and only if there at at least two $e^{i\theta_{j}}$'s that
are equal to 1. Namely, 
\[
\mathrm{adj}\left(1-U\right)=0\iff\dim\ker\left(1-U\right)\ge2.
\]
Clearly, if $e^{i\theta_{1}}=1$ and $e^{i\theta_{j}}\ne1$ for $j\ne1$,
then the only non-vanishing coefficient is $\Pi_{j=2}^{n}\left(1-e^{i\theta_{j}}\right)$
and so (\ref{eq: adjugate in general}) follows. As $\boldsymbol{a}$
in (\ref{eq: adjugate in general}) is normalized, then $\mathrm{trace}\left(\boldsymbol{a}\boldsymbol{a}^{*}\right)=1$
and therefore $\mathrm{trace}\left(\mathrm{adj}\left(1-U\right)\right)=\Pi_{j=2}^{n}\left(1-e^{i\theta_{j}}\right)\ne0$
proving (\ref{eq: trace of adjugate}). To prove (\ref{eq: d(det(1-U)) in general})
we use Jacobi's identity for the derivative of a matrix and the given
relation $\frac{d}{dt}U_{t}=iAU_{t}$:
\begin{align*}
\frac{d}{dt}\left(\det\left(1-U_{t}\right)\right) & =\mathrm{trace}\left(\mathrm{adj}\left(1-U_{t}\right)\frac{d}{dt}\left(1-U_{t}\right)\right)\\
= & -i\cdot\mathrm{trace}\left(\mathrm{adj}\left(1-U_{t}\right)AU_{t}\right).
\end{align*}
At $t=t_{0}$, using (\ref{eq: adjugate in general}) and (\ref{eq: trace of adjugate}),
we get 
\begin{align*}
\frac{d}{dt}\left(\det\left(1-U_{t_{0}}\right)\right) & =-i\cdot\mathrm{trace}\left(\mathrm{adj}\left(1-U_{t_{0}}\right)\right)\cdot\mathrm{trace}\left(\boldsymbol{a}\boldsymbol{a}^{*}AU_{t_{0}}\right)\\
= & -i\cdot\mathrm{trace}\left(\mathrm{adj}\left(1-U_{t_{0}}\right)\right)\cdot\mathrm{trac}e\left(\boldsymbol{a}\boldsymbol{a}^{*}A\right)\\
= & -i\cdot\mathrm{trace}\left(\mathrm{adj}\left(1-U_{t_{0}}\right)\right)\left\langle \boldsymbol{a},A\boldsymbol{a}\right\rangle .
\end{align*}
Where in the second step we used $\mathrm{trace}\left(\boldsymbol{a}\boldsymbol{a}^{*}AU_{t_{0}}\right)=\mathrm{trace}\left(U_{t_{0}}\boldsymbol{a}\boldsymbol{a}^{*}A\right)$
and $U_{t_{0}}\boldsymbol{a}=\boldsymbol{a}$ (by the definition of
$\boldsymbol{a}$).
\end{proof}
We may now prove lemmas \ref{lem:The-secular-functions properties}
and \ref{lem: F p and mk}. Although presented separately, it will
be convenient to prove both lemmas together:
\begin{proof}
First notice that both $\det\left(1-U_{\kv}\right)$ and $\mathrm{trace}\left(\mathrm{adj}\left(1-U_{\kv}\right)\right)$
are trigonometric polynomials as they are polynomial in the entries
of $U_{\kv}$ which are linear in $\left\{ e^{i\kappa_{e}}\right\} _{e\in\E}$.
Recall that we consider $\det\left(U_{\kv}\right)^{-\frac{1}{2}}=i^{\left(\beta-1\right)}e^{-i\sum_{e\in\E}\kappa_{e}}$
so it is also a trigonometric polynomial, and so both $F$ and $p$
are trigonometric polynomials. Let $\left\{ e^{i\theta_{j}}\right\} _{j=1}^{2E}$
be the $\kv$ dependent eigenvalues of $U_{\kv}$ with $\theta_{j}\in\R/2\pi\Z$
for every $j$. Then 
\[
F\left(\kv\right)=\det\left(U_{\kv}\right)^{-\frac{1}{2}}\det\left(1-U_{\kv}\right)=\pm\Pi_{n=1}^{2E}e^{-i\frac{\theta_{j}}{2}}\left(1-e^{i\theta_{j}}\right)=\pm\left(-2i\right)^{2E}\Pi_{n=1}^{2E}\sin\left(\frac{\theta_{j}}{2}\right),
\]
where the $\pm$ ambiguity is due to possible square-root branch choices.
The above RHS is real, which proves that $F$ is real and hence a real trigonometric polynomial (and in particular
real analytic). Since $\left|p\left(\kv\right)\right|=\left|\mathrm{trace}\left(\mathrm{adj}\left(1-U_{\kv}\right)\right)\right|$
we use Lemma \ref{lem: adjugate} to conclude that $p\left(\kv\right)\ne0$
if $\dim\left(\ker\left(1-U_{\kv}\right)\right)=1$ and that $p\left(\kv\right)=0$
if $\dim\left(\ker\left(1-U_{\kv}\right)\right)\ge2$ as in such case
$\mathrm{adj}\left(1-U_{\kv}\right)=0$. By Lemma \ref{lem: secular bijection}
we may conclude that $p$ is a trigonometric polynomial that vanish
on $\Sigma^{sing}$ and is non-zero on $\Sigma^{reg}$. To show that
$p$ is real on $\Sigma$, it enough to show it for $\Sigma^{reg}$.
If $\kv\in\Sigma^{reg}$, and with out loss of generality $e^{i\theta_{1}}=1$
then according to Lemma \ref{lem: adjugate}, $\mathrm{trace}\left(\mathrm{adj}\left(1-U_{\kv}\right)\right)=\Pi_{n=2}^{2E}\left(1-e^{i\theta_{j}}\right)\ne0$,
and therefore 
\begin{align*}
p\left(\kv\right) & =-i\det\left(U_{\kv}\right)^{-\frac{1}{2}}\mathrm{trace}\left(\mathrm{adj}\left(1-U_{\kv}\right)\right)\\
= & \pm i\Pi_{n=2}^{2E}e^{-i\frac{\theta_{j}}{2}}\left(1-e^{i\theta_{j}}\right)\\
= & \pm i\left(-2i\right)^{2E-1}\Pi_{n=2}^{2E}\sin\left(\frac{\theta_{j}}{2}\right)\in\R\setminus\left\{ 0\right\} .
\end{align*}
It follows that $p$ is real on $\Sigma$. Notice that we cannot deduce
from the above that $p$ is a real trigonometric polynomial, but it
is clear that $p$ is a real analytic function on $\Sigma$. Both
(\ref{eq: Sigma reg with p}) and (\ref{eq: singular part using p})
follows from  

Let $\kv\in\Sigma$ so that $\det\left(U_{\kv}\right)^{\frac{1}{2}}\ne1$
and $\det\left(1-U_{\kv}\right)=0$, then
\begin{equation}
\nabla F\left(\kv\right)=\det\left(U_{\kv}\right)^{\frac{1}{2}}\nabla\det\left(1-U_{\kv}\right).
\end{equation}
If $\kv\in\Sigma^{sing}$, namely $\dim\ker\left(1-U_{\kv}\right)\ge2$
then $adj\left(1-U_{\kv}\right)=0$ and by Jacobi identity, 
\begin{equation}
\nabla\det\left(1-U_{\kv}\right)=\mathrm{trace}\left(\mathrm{adj}\left(1-U_{\kv}\right)\nabla U_{\kv}\right)=0,
\end{equation}
so $\nabla F\left(\kv\right)=0$. If $\kv\in\Sigma^{reg}$, let $\boldsymbol{a}$
be the amplitudes vector of $f_{\kv}$ so that it is normalized and
in $a\in\ker\left(1-U_{\kv}\right)$. Notice that 
\[
\forall e\in\E\,\,\frac{\partial}{\partial\kappa_{e}}U_{\kv}=\left(\frac{\partial}{\partial\kappa_{e}}e^{i\kv}\right)S=iA_{e}U_{\kv}
\]
where $A_{e}$ is diagonal with $\left(A\right)_{e',e'}=\begin{cases}
1 & e'=e\,or\,e'=\hat{e}\\
0 & \mathrm{otherwise}
\end{cases}$. In particular, $\left\langle \boldsymbol{a},A_{e}\boldsymbol{a}\right\rangle =\left|a_{e}\right|^{2}+\left|a_{\hat{e}}\right|^{2}=\left(m_{\kv}\right)_{e}$
(see (\ref{eq: define m_e})). We can apply Lemma \ref{lem: adjugate}
to get that 
\begin{align*}
\forall e\in\E\,\,\,\frac{\partial}{\partial\kappa_{e}}\det\left(1-U_{\kv}\right) & =-i\cdot\mathrm{trace}\left(\mathrm{adj}\left(1-U_{\kv}\right)\right)\left\langle \boldsymbol{a},A_{e}\boldsymbol{a}\right\rangle \\
= & -i\cdot\mathrm{trace}\left(\mathrm{adj}\left(1-U_{\kv}\right)\right)\left(m_{\kv}\right)_{e}.
\end{align*}
We can deduce that 
\[
\nabla F\left(\kv\right)=-i\cdot\det\left(U_{\kv}\right)^{-\frac{1}{2}}\mathrm{trace}\left(\mathrm{adj}\left(1-U_{\kv}\right)\right)m_{\kv}=p\left(\kv\right)m_{\kv}.
\]
This proves (\ref{eq: grad F}) and as both $p\left(\kv\right)$ and
$m_{\kv}$ does not vanish on $\Sigma^{reg}$, then so does $\nabla F$.
We thus showed that $\nabla F$, like $p$, vanish in $\Sigma$ only
on $\Sigma^{sing}$ which finish the proof of Lemma \ref{lem:The-secular-functions properties}.
To finish the proof of Lemma \ref{lem: F p and mk} it is only left
to notice that the normal to $\Sigma$ which is the zero set of $F$
is proportional to $\nabla F$ and thus to $m_{\kv}$, which proves
(\ref{eq: normal}).
\end{proof}
We may now prove Proposition \ref{prop: Secular manifold structure}:
\begin{proof}
\label{proof of structure} Lemma \ref{lem:The-secular-functions properties}
characterize $\Sigma$ as the zero set of the real analytic function
$F$ and $\Sigma^{sing}$ as the zero set of the real analytic function
$\norm{\nabla F}^{2}+F^{2}$, so both are real analytic varieties.
By our assumption on the graph it is not homeomorphic to a single
loop, and so according to \cite{Fri_ijm05} there is a choice of $\lv$
with a simple eigenvalue $k>0$, and therefore $\kv=\left\{ k\lv\right\} \in\Sigma^{reg}$,
and so $\Sigma^{reg}\ne\emptyset$. For such $\kv\in\Sigma^{reg}$,
according to Lemma \ref{lem:The-secular-functions properties}, $\nabla F\left(\kv\right)\ne0$
and so $\nabla F$ does not vanish on a small neighborhood of $\kv$,
$O_{\kappa}\subset\Sigma$, which is therefore a manifold of dimension
$E-1$. Therefore, $\Sigma^{reg}$ which is the union of such points
is an $E-1$ dimensional manifold, and so $\Sigma$ is $E-1$ dimensional.
Since $\Sigma^{reg}$ is also an open subset of (the real analytic
variety) $\Sigma$, then it is a real analytic manifold. The singular
part $\Sigma^{sing}\subset\Sigma$ is therefore of dimension smaller or equal to $E-1$. Assume by contradiction that $\dim\Sigma^{sing}=E-1$,
and let $\kv_{0}\in\Sigma^{sing}$ such that it has an $E-1$ dimensional
(real analytic) neighborhood $O\subset\Sigma^{sing}$ around it. Then
the normal $\hat{n}\left(\kv\right)$ is well defined and smooth on
$O$. Let $\hat{n}\left(\kv_{0}\right)$ be the normal at $\kv_{0}$
and denote $\hat{n}^{\perp}:=\set{\lv\in\left(\R_{+}\right)^{\E}}{\lv\cdot\hat{n}\left(\kv_{0}\right)\ne0}$.
According to Friedlander's genericity result in \cite{Fri_ijm05},
there exists a residual set $G\subset\left(\R_{+}\right)^{\E}$ such
that for every $\lv\in G$, every eigenvalue of $\Gamma_{\lv}$ is
simple. The set of rationally independent $\lv$'s in $\left(\R_{+}\right)^{\E}$
is residual according to Remark \ref{rem: rationally indepdnet} and
$\hat{n}^{\perp}$ is residual as it is open with complement of positive
co-dimension. Therefore the set of rationally independent $\lv$'s
in $\hat{n}^{\perp}\cap G$ is residual. Let $\lv$ be such, then
$\lv\cdot\hat{n}\left(\kv_{0}\right)\ne0$ and so there exists a neighborhood
$\tilde{O}\subset O$ on which $\hat{n}\cdot\lv\ne0$. The linear
flow $k\mapsto\left\{ k\lv\right\} $ is dense in $\T^{\E}$ because
$\lv$ is rationally independent, and is transversal to $\tilde{O}$
since $\hat{n}\cdot\lv\ne0$ on every point in $\tilde{O}$. Since
$\tilde{O}$ is $E-1$ dimensional, and the flow is dense and transversal
to $\tilde{O}$, then there are infinitely many intersections $\left\{ k_{m}\lv\right\} \in\tilde{O}\subset\Sigma^{sing}$.
However, by the definition of $\Sigma^{sing}$, if $\left\{ k_{m}\lv\right\} \in\Sigma^{sing}$
then $k_{m}$ is a non-simple eigenvalue of $\Gamma_{\lv}$, in contradiction
to the choice of $\lv\in G$. Therefore, $\Sigma^{sing}$ is of dimension
strictly smaller than $E-1$.
\end{proof}
\begin{rem}
Colin de Verdière proves Friedlander's result in section 7 of \cite{CdV_ahp15}
by proving that $\Sigma^{reg}$ is always non-empty. His proof requires
an argument (which does not appear in \cite{CdV_ahp15}) that states
that either the set where both $F$ and $\nabla F$ vanish is $\Sigma$
or it is of positive co-dimension in $\Sigma$. We haven't found such
an argument, but if the conjecture in \cite{CdV_ahp15} regarding
the irreducibility of $F$ holds, then the needed argument follows. 
\end{rem}

\begin{defn}
\label{def: Inversion}We define the inversion $\mathcal{I}$ on $\T^{\E}$
(also for any $\T^{n}$) by $\I\left(\kv\right):=\left\{ -\kv\right\} $. 
\end{defn}

\begin{lem}
\label{lem: inversion on F} The inversion $\mathcal{\I}$ is an isometry
of the secular manifold $\Sigma$ that preserves both $\Sigma$ and
its partition to $\Sigma^{reg},\Sigma^{sing}$. The secular function
$F$, together with $p$ and $m_{\kv}$, transform under the inversion
as follows:
\begin{align}
F\circ\I & =\left(-1\right)^{\beta-1}F,\label{eq: Inversion on F}
\end{align}
and for any $\kv\in\Sigma^{reg}$, 
\begin{align}
p\left(\I\left(\kv\right)\right) & =\left(-1\right)^{\beta}p\left(\kv\right),\label{eq: inversion on p}\\
m_{\I\left(\kv\right)} & =m_{\kv}.\label{eq: inversion on m}
\end{align}
\end{lem}

\begin{proof}
To prove Lemma \ref{lem: inversion on F}, notice that $e^{i\mathcal{I}\left(\kv\right)}=\overline{e^{i\kv}}$
and $S$ is real so $U_{\mathcal{I}\left(\kv\right)}=\overline{U_{\kv}}$
and in particular 
\begin{align*}
F\left(\mathcal{I}\left(\kv\right)\right) & =i^{\left(\beta-1\right)}e^{i\sum_{e\in\E}\kappa_{e}}\det\left(1-\overline{U_{\kv}}\right)\\
= & \left(-1\right)^{\left(\beta-1\right)}\overline{i^{\left(\beta-1\right)}e^{-i\sum_{e\in\E}\kappa_{e}}\det\left(1-U_{\kv}\right)}=\left(-1\right)^{\left(\beta-1\right)}\overline{F\left(\kv\right)}.
\end{align*}
As $F\left(\kv\right)$ is real it proves (\ref{eq: Inversion on F}),
and we can deduce that $\Sigma$ is invariant under $\I$. Since $\I$
is an isometry of $\T^{\E}$ in which $\Sigma$ is embedded, then it
is also an isometry of $\Sigma$. In the same manner, 
\begin{align*}
p\left(\mathcal{I}\left(\kv\right)\right) & =-i^{\beta}e^{i\sum_{e\in\E}\kappa_{e}}\mathrm{trace}\left(\mathrm{adj}\left(1-\overline{U_{\kv}}\right)\right)\\
= & \left(-1\right)^{\beta}\left(\overline{-i^{\beta}e^{-i\sum_{e\in\E}\kappa_{e}}\mathrm{trace}\left(\mathrm{adj}\left(1-U_{\kv}\right)\right)}\right)=\left(-1\right)^{\beta}\overline{p\left(\kv\right)},
\end{align*}
and since $p|_{\Sigma}$ is real it proves (\ref{eq: inversion on p}).
We can therefore deduce that both $\Sigma^{reg}$ and $\Sigma^{sing}$
are invariant to $\I$. Let $\kv\in\Sigma^{reg}$ and let $a\in\ker\left(1-U_{\kv}\right)$,
then $\overline{a}\in\ker\overline{\left(1-U_{\kv}\right)}=\ker\left(1-U_{\mathcal{I}\left(\kv\right)}\right)$.
It follows that $m_{\mathcal{I}\left(\kv\right)}=m_{\kv}$ as needed. 
\end{proof}
\begin{lem}
\label{lem: vertex values of canonical function}Let $S$ and $J$
as defined in Subsection \ref{subsec:Wave-scattering}. Then for any
$v,u\in\V$ and $e\in\E_{v},\,e'\in\E_{u}$, all $f_{\kv}\left(v\right)f_{\kv}\left(u\right)$,
$\partial_{e}f_{\kv}\left(v\right)f_{\kv}\left(u\right)$ and $\partial_{e}f_{\kv}\left(v\right)\partial_{e'}f_{\kv}\left(u\right)$
are real analytic functions of $\kv\in\Sigma^{reg}$ that are given
explicitly as the following matrix elements: 
\begin{align}
f_{\kv}\left(v\right)f_{\kv}\left(u\right) & =\left(\frac{1}{\mathrm{trace}\left(\mathrm{adj}\left(1-U_{\kv}\right)\right)}\left(S+J\right)adj\left(1-U_{\kv}\right)\left(S+J\right)^{T}\right)_{e,e'}.\label{eq: vertex values matrix}\\
\partial_{e}f_{\kv}\left(v\right)f_{\kv}\left(u\right) & =\left(\frac{-i}{\mathrm{trace}\left(\mathrm{adj}\left(1-U_{\kv}\right)\right)}\left(S-J\right)adj\left(1-U_{\kv}\right)\left(S+J\right)^{T}\right)_{e,e'}.\label{eq: vertex and derivatives matrix}\\
\partial_{e}f_{\kv}\left(v\right)\partial_{e'}f_{\kv}\left(u\right) & =\left(\frac{-1}{\mathrm{trace}\left(\mathrm{adj}\left(1-U_{\kv}\right)\right)}\left(S-J\right)adj\left(1-U_{\kv}\right)\left(S-J\right)^{T}\right)_{e,e'}.\label{eq: derivatives of f by matrix}
\end{align}
Moreover, if $\boldsymbol{a}$ is the amplitudes vector of $f_{\kv}$
and $f_{\I\left(\kv\right)}$ is the canonical eigenfunction at the
point $\I\left(\kv\right)$, then the amplitudes vector $f_{\I\left(\kv\right)}$
is $\pm\overline{\boldsymbol{a}}$, and their traces satisfy
\begin{align}
f_{\I\left(\kv\right)}\left(v\right)f_{\mathcal{\I}\left(\kv\right)}\left(u\right) & =f_{\kv}\left(v\right)f_{\kv}\left(u\right),\label{eq: inversion vertex values}\\
\partial_{e}f_{\I\left(\kv\right)}\left(v\right)f_{\I\left(\kv\right)}\left(u\right) & =-\partial_{e}f_{\kv}\left(v\right)f_{\kv}\left(u\right).\label{eq: inversion derivatives}\\
\partial_{e}f_{\I\left(\kv\right)}\left(v\right)\partial_{e'}f_{\I\left(\kv\right)}\left(u\right) & =\partial_{e}f_{\kv}\left(v\right)\partial_{e'}f_{\kv}\left(u\right).\label{eq: inversion der der}
\end{align}
\end{lem}

\begin{proof}
Let $\kv\in\Sigma^{reg}$ and let $\boldsymbol{a}$ be the amplitudes
vector of $f_{\kv}$. According to Definition \ref{def: edge restriction notations},
$f\left(v\right)=a_{e}e^{-i\kappa_{e}}+a_{\hat{e}}$ and $-i\partial_{e}f\left(v\right)=a_{e}e^{-i\kappa_{e}}-a_{\hat{e}}$.
Notice that $a_{\hat{e}}=\left(J\boldsymbol{a}\right)_{e}$. Using
$\boldsymbol{a}=U_{\kv}\boldsymbol{a}=e^{i\hat{\kappa}}S\boldsymbol{a}$
we get $\left(Sa\right)_{e}=\left(e^{-i\hat{\kappa}}a\right)_{e}=e^{-i\kappa_{e}}a_{e}$.
This can be written as 
\begin{align*}
\left(\left(S+J\right)\boldsymbol{a}\right)_{e} & =f\left(v\right),\,\,\,and\\
\left(\left(S-J\right)\boldsymbol{a}\right)_{e} & =-i\partial_{e}f\left(v\right),
\end{align*}
which means that 
\begin{align}
f_{\kv}\left(v\right)f_{\kv}\left(u\right) & =\left(\left(S+J\right)\boldsymbol{a}\boldsymbol{a}^{*}\left(S+J\right)^{T}\right)_{e,e'},\\
\partial_{e}f_{\kv}\left(v\right)f_{\kv}\left(u\right) & =-i\left(\left(S-J\right)\boldsymbol{a}\boldsymbol{a}^{*}\left(S+J\right)^{T}\right)_{e,e'},\,\,\,and\,\\
\partial_{e}f_{\kv}\left(v\right)\partial_{e'}f_{\kv}\left(u\right) & =-\left(\left(S-J\right)\boldsymbol{a}\boldsymbol{a}^{*}\left(S-J\right)^{T}\right)_{e,e'}.
\end{align}
where $e$ and $e'$ emits out of $v$ and $u$. According to Lemma
\ref{lem: adjugate}, $\mathrm{trace}\left(\mathrm{adj}\left(1-U_{\kv}\right)\right)\ne0$
and $\boldsymbol{a}\boldsymbol{a}^{*}=\frac{1}{\mathrm{trace}\left(\mathrm{adj}\left(1-U_{\kv}\right)\right)}\mathrm{adj}\left(1-U_{\kv}\right)$
which concludes the proof of (\ref{eq: vertex values matrix}) and
(\ref{eq: vertex and derivatives matrix}). Since both $\mathrm{trace}\left(\mathrm{adj}\left(1-U_{\kv}\right)\right)$
and the entries of $\mathrm{adj}\left(1-U_{\kv}\right)$ are polynomials
in $\left\{ e^{i\kappa_{e}}\right\} _{e\in\E}$, then the entries of
$\frac{1}{\mathrm{trace}\left(\mathrm{adj}\left(1-U_{\kv}\right)\right)}\mathrm{adj}\left(1-U_{\kv}\right)$
are rational functions in $\left\{ e^{i\kappa_{e}}\right\} _{e\in\E}$
(with no poles on $\Sigma^{reg}$ as $\mathrm{trace}\left(\mathrm{adj}\left(1-U_{\kv}\right)\right)\ne0$)
and hence so does the matrix elements in (\ref{eq: vertex values matrix}),(\ref{eq: vertex and derivatives matrix})
and (\ref{eq: derivatives of f by matrix}). Since $f_{\kv}$ is real
then both $f_{\kv}\left(v\right)f_{\kv}\left(u\right),\partial_{e}f_{\kv}\left(v\right)f_{\kv}\left(u\right)$
and $\partial_{e}f_{\kv}\left(v\right)\partial_{e'}f_{\kv}\left(u\right)$
are real on $\Sigma^{reg}$ and so we may conclude that they are real
analytic functions on $\Sigma^{reg}$.

To finish the proof, consider the point $\I\left(\kv\right)\in\Sigma^{reg}$.
Since $\boldsymbol{a}\in\ker\left(1-U_{\kv}\right)$, then $\boldsymbol{\overline{a}}\in\ker\left(1-\overline{U_{\kv}}\right)=\ker\left(1-U_{\I\left(\kv\right)}\right)$
and therefore $\boldsymbol{\overline{a}}$ is the amplitudes vector
of an eigenfunction $\tilde{f}\in Eig\left(\Gamma_{\I\left(\kv\right)},1\right)$.
The traces of $f_{\kv}$ and $\tilde{f}$ are related as follows.
For every vertex $v$ and edge $e\in\E_{v}$:
\begin{align*}
\tilde{f}\left(v\right)=\overline{a_{e}}e^{-i\I\left(\kappa_{e}\right)}+\overline{a_{\hat{e}}}=\overline{a_{e}e^{-i\kappa_{e}}+a_{\hat{e}}} & =\overline{f_{\kv}\left(v\right)}=f_{\kv}\left(v\right),\,\,\,\text{and}\\
\partial_{e}\tilde{f}\left(v\right)=i\left(\overline{a_{e}}e^{-i\I\left(\kappa_{e}\right)}-\overline{a_{\hat{e}}}\right)=i\left(\overline{a_{e}e^{-i\kappa_{e}}-a_{\hat{e}}}\right) & =\overline{-\partial_{e}f_{\kv}\left(v\right)}=-\partial_{e}f_{\kv}\left(v\right).
\end{align*}
Therefore the trace of $\tilde{f}$ are real so $\tilde{f}$ is real
with normalized amplitudes and therefore $\tilde{f}=\pm f_{\I\left(\kv\right)}$.
Hence the traces of $f_{\I\left(\kv\right)}$ and $f_{\kv}$ are related
as above (up to a global sign). Both (\ref{eq: inversion vertex values}),
(\ref{eq: inversion derivatives}) and \ref{eq: inversion der der}
follows.
\end{proof}

\subsection{The equidistribution of $\left\{ k_{n}\protect\lv\right\} $ on $\Sigma$
and the Barra-Gaspard measure.}

As already discussed, the first work on the secular manifold, by Barra
and Gaspard in \cite{BarGas_jsp00}, used ergodic arguments to calculate
the level spacing statistics. The ergodic argument used was later
formalized both in \cite{BerWin_tams10,CdV_ahp15}. In the following
subsection we present this mechanism, which we will use in the following
sections. This mechanism is based on the equidistribution of the points
$\left\{ k_{n}\lv\right\} $ on $\Sigma$ according to the Barra-Gaspard
measure, and a measure preserving inversion of $\Sigma$. 

To do so we will need to define the notion of \emph{equidistribution}
(see \cite{einsiedler2013ergodic} section 4.4.2). 
\begin{defn}
Let $X$ be a compact metric space and let $m$ be a Borel measure
on $X$. A sequence $\left\{ x_{n}\right\} _{n\in\N}$ of points in
$X$ is \emph{equidistributed} according to $m$ if for any continuous
function $f$, 
\begin{equation}
\lim_{N\rightarrow\infty}\frac{\sum_{n=1}^{N}f\left(x_{n}\right)}{N}=\int_{X}fdm.
\end{equation}
Equivalently, $\left\{ x_{n}\right\} _{n\in\N}$ is equidistributed
if the atomic measures $\frac{\sum_{n=1}^{N}\delta_{x_{n}}}{N}$ converges
to $m$ as $N\rightarrow\infty$ in the $weak^{*}-topology$.
\end{defn}

Given a subset $A\subset X$, the atomic measures $\frac{\sum_{n=1}^{N}\delta_{x_{n}}}{N}$
evaluated on $A$ gives:
\begin{align}
\left(\frac{\sum_{n=1}^{N}\delta_{x_{n}}}{N}\right)A=\frac{\left|\set{n\le N}{x_{n}\in A}\right|}{N}.
\end{align}
Next, we define the notion of \emph{natural density. }
\begin{defn}
\label{def: natural density} Given a subset $\A\subset\N$ we denote
$\A\left(N\right):=\A\cap\left\{ 1,2,...N\right\} $ for any $N\in\N$.
We say that $\A$ has \emph{density} and denote it by $d\left(\A\right)$,
if the following limit exists:
\[
d\left(\A\right)=\lim_{N\rightarrow\infty}\frac{\left|\A\left(N\right)\right|}{N}.
\]
\end{defn}

Our motivation for the above definitions is that many statistical
properties that we are after regards limits of the form $\lim_{N\rightarrow\infty}\frac{\left|\set{n\le N}{\left\{ k_{n}\lv\right\} \in A}\right|}{N}$
for some given $A$. A statement, \emph{the limit $\lim_{N\rightarrow\infty}\frac{\left|\set{n\le N}{\left\{ k_{n}\lv\right\} \in A}\right|}{N}$
exists}, is equivalent to the statement, \emph{the index set $\set{n\in\N}{\left\{ k_{n}\lv\right\} \in A}$
has density}. We will mainly use the density's terminology. 

The equidistribution will become useful for limits as above, by the
next lemma that will provide both a sufficient condition for an integers
set to have density, and its density in terms an equidistributed sequence.
\begin{defn}
\label{def: Jordan}If $X$ is a topological space with a Borel measure
$m$, we call a Borel subset $A\subset X$ \emph{Jordan }if its topological
boundary (closure minus interior) has measure zero, $m\left(\partial A\right)=0$.
\end{defn}

Using a standard approximation argument one can prove the following
lemma:
\begin{lem}
\label{lem: equidistribution Jordan}Let $X$ be a compact metric
space, let $m$ be a Borel regular measure and let $\left\{ x_{n}\right\} _{n\in N}\subset X$
be equidistributed with respect to $m$. If a Borel set $A\subset X$
is Jordan with respect to $m$, then, the index set $\set{n\in\N}{x_{n}\in A}$
has density:
\[
d\left(\set{n\in\N}{x_{n}\in A}\right)=m\left(A\right).
\]
\end{lem}

For completeness we will present a proof for this lemma in Appendix
\ref{sec: Equidistribution}.

The compact metric space that we will consider is the secular manifold
$\Sigma$. The measures we consider on $\Sigma$, are called the \emph{Barra-Gaspard
measures.} 
\begin{defn}
\label{def: BG measure}Given a standard graph $\Gamma_{\lv}$, the
\emph{Barra-Gaspard measure} $\mu_{\lv}$ (BG-measure) is an $\lv$
dependent Radon probability measure on $\Sigma$. It is defined on
$\Sigma^{reg}$ in terms of the euclidean surface element $ds$ and
the normal vector $\hat{n}$ as follows:
\begin{equation}
d\mu_{\lv}=\frac{\pi}{L}\cdot\frac{1}{\left(2\pi\right)^{E}}\left|\hat{n}\cdot\lv\right|ds.\label{eq: BG measure}
\end{equation}
As the singular part $\Sigma^{sing}$ is a closed subset of positive
co-dimension in $\Sigma$ (Proposition \ref{prop: Secular manifold structure})
we extend $\mu_{\lv}$ to $\Sigma$ by setting $\mu_{\lv}\left(\Sigma^{sing}\right)=0$.
\end{defn}

\begin{rem}
\label{rem: BG measuer positivity}For any $\lv\in\R_{+}^{\E}$, the
BG-measure $\mu_{\lv}$ is a Radon measure with strictly positive
density on $\Sigma^{reg}$. That is, for any open set $O\subset\Sigma$,
$\mu_{\lv}\left(O\right)>0$. This follows from $\left|\hat{n}\cdot\lv\right|$
being strictly positive which can be seen using (\ref{eq: normal}),
\[
\forall\kv\in\Sigma^{reg}\,\,\left|\hat{n}\left(\kv\right)\cdot\lv\right|=\frac{1}{\norm{m_{\kv}}}\sum_{e}\left(m_{\kv}\right)_{e}l_{e}>0.
\]
In particular, all BG measures $\mu_{\lv}$ for all $\lv$ agree on
measure zero sets, and therefore if a set $A\subset\Sigma$ is Jordan
with respect to $\mu_{\lv}$ for some $\lv$, then it is Jordan with respect
to $\mu_{\lv}$ for any $\lv$. 
\end{rem}

The following theorem was proven by Barra and Gaspard in \cite{BarGas_jsp00}
and a more detailed proof appeared both in \cite{BerWin_tams10} and
\cite{CdV_ahp15}. We present this theorem using the equidistribution
terminology rather than the original statement. 
\begin{thm}
\label{thm: BG equidistribution}Let $\Gamma_{\lv}$ be a standard
graph with $\lv$ rationally independent and (square-root) eigenvalues
$\left\{ k_{n}\right\} _{n=0}^{\infty}$. Then the sequence $\left\{ \left\{ k_{n}\lv\right\} \right\} _{n=0}^{\infty}$
is dense in $\Sigma$ and is equidistributed according to $\mu_{\lv}$.
In particular, if $A\subset\Sigma$ is Jordan with respect to $\mu_{\lv}$
then the index set $\set{n\in\N}{\left\{ k_{n}\lv\right\} \in A}$
has density:
\begin{equation}
d\left(\set{n\in\N}{\left\{ k_{n}\lv\right\} \in A}\right)=\mu_{\lv}\left(A\right).
\end{equation}
\end{thm}

\begin{rem}
\label{rem:-Sigma reg is-Jordan}$\Sigma^{reg}$ is Jordan in $\Sigma$
with measure $\mu_{\lv}$ for any choice of $\lv$. This is because
$\Sigma^{reg}$ is open and dense in $\Sigma$, by Proposition \ref{prop: Secular manifold structure},
and so $\partial\Sigma^{reg}=\Sigma^{sing}$ which has $\mu_{\lv}\left(\Sigma^{sing}\right)=0$.
It follows that each of the connected components of $\Sigma^{reg}$
is Jordan as well.
\end{rem}

\begin{thm}
\label{thm: inversion BG}The inversion $\mathcal{I}$ (see Definition
\ref{def: Inversion}) is $\mu_{\lv}$ measure preserving. In particular,
if $\lv$ is rationally independent, then for any Jordan set $A\subset\Sigma$,
both\\ $\set{n\in\N}{\left\{ k_{n}\lv\right\} \in A}$ and $\set{n\in\N}{\left\{ k_{n}\lv\right\} \in\mathcal{I}\left(A\right)}$
have equal densities given by:
\[
d\left(\set{n\in\N}{\left\{ k_{n}\lv\right\} \in A}\right)=d\left(\set{n\in\N}{\left\{ k_{n}\lv\right\} \in\mathcal{I}\left(A\right)}\right)=\mu_{\lv}\left(A\right).
\]
\end{thm}

\begin{proof}
As seen in Lemma \ref{lem: inversion on F}, the inversion $\I$ is
an isometry of $\Sigma$ (and $\Sigma^{reg}$), and hence it preserves
the Euclidean surface element $ds$. Using both (\ref{eq: inversion on m})
and (\ref{eq: normal}) it is clear that the normal vector to $\Sigma^{reg}$
is also preserved under $\I$. It follows from (\ref{eq: BG measure})
that $d\mu_{\lv}$ is preserved and therefore $\mathcal{I}$ is $\mu_{\lv}$
preserving. Since $\I$ is a measure preserving homeomorphism, then
a set $A$ is Jordan if and only if $\I\left(A\right)$ is Jordan.
In such case, if $A$ and hence $\I\left(A\right)$ are Jordan, then
by Theorem \ref{thm: BG equidistribution} we get that 
\begin{align*}
d\left(\set{n\in\N}{\left\{ k_{n}\lv\right\} \in A}\right) & =\mu_{\lv}\left(A\right),\\
d\left(\set{n\in\N}{\left\{ k_{n}\lv\right\} \in\I\left(A\right)}\right) & =\mu_{\lv}\left(\I\left(A\right)\right).
\end{align*}
As $\I$ is measure preserving, then $\mu_{\lv}\left(\I\left(A\right)\right)=\mu_{\lv}\left(A\right)$
and we are done. 
\end{proof}

\subsection{Bridges and the secular manifold.}

In this section we discuss the structure of the secular manifold in
the case a graph has a bridge. The results of this subsection will
be used later on in Sections \ref{sec: genericity}, \ref{sec: Magnetic}
and \ref{sec: Binomial}. Recall that an edge is called a bridge if
its removal disconnects the graph. In particular a tail is a bridge
under which removal its boundary vertex is disconnected from the rest
of the graph. 
\begin{defn}
Given a graph $\Gamma$, and a bridge $e$, we define the \emph{bridge
decomposition} of $\Gamma$ according to $e$ as a decomposition of
$\Gamma$ into $\left\{ e\right\} $ and the two connected components
of $\Gamma\setminus$$\left\{ e\right\} $ which we denote by $\Gamma_{1}$
and $\Gamma_{2}$. The corresponding (disjoint) sets of edges and
vertices are $\E_{1},\E_{2}$ and $\V_{1},\V_{2}$ such that $\V=\V_{1}\sqcup\V_{2}$
and $\E=\E_{1}\sqcup\left\{ e\right\} \sqcup\E_{2}$. This induces
a decomposition of the characteristic torus to $\T^{\E}=\T^{\E_{1}}\times\T\times\T^{\E_{2}}$
and we its coordinates respectively as $\kv=\left(\kv_{1},\kappa_{e},\kv_{2}\right)$.
\end{defn}

A detailed technical description of the decomposition of the secular
function according to such a decomposition of the graph (and more
complicated decompositions) can be found in Section 4 of \cite{AloBanBer_cmp18}.
We will only consider bridge decompositions and devote Appendix \ref{sec: appendix bridge}
to the technical details. The results are as follows, 
\begin{prop}
\label{prop: simple bridge det decomposition}Let $\Gamma$ be a graph
with a bridge $e$ and a bridge decomposition\\ $\Gamma\setminus\left\{ e\right\} =\Gamma_{1}\sqcup\Gamma_{2}$.
Then the secular function $F$ can be decomposed as follows:
\begin{equation}
F\left(\kv\right)=F\left(\kv_{1},\kappa_{e},\kv_{2}\right)=g_{1}\left(\kv_{1}\right)g_{2}\left(\kv_{2}\right)e^{-i\kappa_{e}}\left(1-e^{i2\kappa_{e}}e^{i\Theta_{1}\left(\kv_{1}\right)}e^{i\Theta_{2}\left(\kv_{2}\right)}\right).\label{eq: Bridge decomposition F}
\end{equation}
Where $g_{i}:\T^{\E_{i}}\rightarrow\C$ is a trigonometric polynomial
and $\Theta_{i}\left(\kv_{i}\right):\T^{\E_{i}}\rightarrow\R/2\pi\Z$
is real analytic on the set where $g_{i}\left(\kv_{i}\right)\ne0$
for both $i\in\left\{ 1,2\right\} $. 
\end{prop}

\begin{proof}
Using the terminology of Appendix \ref{sec: appendix bridge}, $U_{\kv}=U_{\ztrip}$
for\\ $\ztrip=\left(e^{i\hat{\kappa}_{1}},e^{i\kappa_{e}},e^{i\hat{\kappa}_{2}}\right)$
being the blocks of $e^{i\hat{\kappa}}$. We define $\tilde{g}_{i}\left(\kv_{i}\right):=\det D_{i}\left(\zi\right)$
and $e^{i\Theta_{i}\left(\kv_{i}\right)}=\mathcal{S}_{i}\left(\zi\right)$
as in Definition \ref{def: D,S}. By Lemma \ref{lem: Technical  bridge decomp},\[
\det\left(1-U_{\kv}\right)=\tilde{g}_{1}\left(\kv_{1}\right)\tilde{g}_{2}\left(\kv_{2}\right)\left(1-e^{i2\kappa_{e}}e^{i\Theta_{1}\left(\kv_{1}\right)}e^{i\Theta_{2}\left(\kv_{2}\right)}\right),\]
and therefore: 
\[
F\left(\kv\right)=i^{\left(\beta-1\right)}e^{-i\sum_{e\in\E}\kappa_{e}}\tilde{g}_{1}\left(\kv_{1}\right)\tilde{g}_{2}\left(\kv_{2}\right)\left(1-e^{i2\kappa_{e}}e^{i\Theta_{1}\left(\kv_{1}\right)}e^{i\Theta_{2}\left(\kv_{2}\right)}\right).
\]
This proves (\ref{eq: Bridge decomposition F}) by setting
\begin{align}
g_{1}\left(\kv_{1}\right):= & i^{\left(\beta-1\right)}e^{-i\sum_{e\in\E_{1}}\kappa_{e}}\tilde{g}_{1}\left(\kv_{1}\right),\ and\\
g_{2}\left(\kv_{2}\right):= & e^{-i\sum_{e\in\E_{2}}\kappa_{e}}\tilde{g}_{2}\left(\kv_{2}\right).
\end{align}
Clearly by definition of $\tilde{g}_{i}\left(\kv_{i}\right)$ is polynomial
in the entries of $e^{i\hat{\kappa}_{i}}$ and is therefore a trigonometric
polynomial, and so does $g_{i}\left(\kv_{i}\right)$. According to
Lemma \ref{lem: scattering phase}, $e^{i\Theta_{i}\left(\kv_{i}\right)}=\mathcal{S}_{i}\left(e^{i\hat{\kappa}_{i}}\right)$
is uni-modular for any $\kv_{i}$, and is analytic in the entries
of $e^{i\hat{\kappa}_{i}}$ in the region where $g_{i}\left(\kv_{i}\right)\ne0$.
This proves that $\Theta_{i}\left(\kv_{i}\right):\T^{\E_{i}}\rightarrow\R/2\pi\Z$
is well defined everywhere and is real analytic on the set where $g_{i}\left(\kv_{i}\right)\ne0$.
\end{proof}
The meaning of $e^{i\Theta_{1}\left(\kv_{1}\right)},e^{i\Theta_{2}\left(\kv_{2}\right)}$
in terms of eigenfunctions is as follows:
\begin{lem}
\label{lem: Rev phase} Let $f\in Eig\left(\Gamma_{\kv},1\right)$
and assume that $f|_{e}\not\equiv0$. Consider the direction of $e$
from $\Gamma_{1}$ to $\Gamma_{2}$. If $\boldsymbol{a}$ is the amplitudes
vector of some $f\in Eig\left(\Gamma_{\kv},1\right)$, and $f|_{e}\not\equiv0$,
then:
\begin{align}
e^{i\left(\kappa_{e}+\Theta_{1}\left(\kv_{1}\right)\right)}a_{\hat{e}} & =a_{e},\,\,\,and\label{eq:  phase1}\\
e^{i\left(\kappa_{e}+\Theta_{2}\left(\kv_{2}\right)\right)}a_{e} & =a_{\hat{e}}.\label{eq: phaze2}
\end{align}
Let $\varphi_{e}$ be the phase of the amplitude-phase pair in Definition
\ref{def: edge restriction notations}.\\ Then $e^{2i\varphi_{e}}=e^{-i\Theta_{1}\left(\kv_{1}\right)}$, and in the same way, $e^{2i\varphi_{\hat{e}}}=e^{-i\Theta_{2}\left(\kv_{2}\right)}$,
for the phase $\varphi_{\hat{e}}$ of the other direction. 
\end{lem}

\begin{proof}
Using Lemma \ref{lem: Technical  bridge decomp}, and since the amplitudes
vector of $f$ is in $\ker\left(1-U_{\kv}\right)$, then it is also
in the kernel of the $M$ matrix from the lemma and therefore\\ $\begin{pmatrix}a_{e}\\
a_{\hat{e}}
\end{pmatrix}\in\ker\begin{pmatrix}-1 & z_{e}\mathcal{S}_{1}\left(\zone\right)\\
z_{e}\mathcal{S}_{2}\left(\ztwo\right) & -1
\end{pmatrix}$. Using $z_{e}\mathcal{S}_{1}\left(\zone\right)=e^{i\kappa_{e}}e^{i\Theta_{i}\left(\kv_{i}\right)}$
we get both (\ref{eq:  phase1}) and (\ref{eq: phaze2}). Since $f|_{e}\not\equiv0$
then $a_{e}\ne0$ and therefore (\ref{eq:  phase1} implies $e^{-i\Theta_{1}\left(\kv_{1}\right)}=\frac{a_{\hat{e}}}{a_{e}}e^{i\kappa_{e}}$.
Using Lemma \ref{lem: phase for generic}, $e^{2i\varphi_{e}}=\frac{a_{\hat{e}}}{a_{e}}e^{i\kappa_{e}}$
which gives the needed result. In the same way, using (\ref{eq: phaze2}),
$e^{2i\varphi_{\hat{e}}}=e^{-i\Theta_{2}\left(\kv_{2}\right)}$.
\end{proof}
Let us denote the zero set of $g_{1}\left(\kv_{1}\right)g_{2}\left(\kv_{2}\right)$
by $Z_{g}$: 
\begin{equation}
Z_{g}:=\set{\kv=\left(\kv_{1},\kappa_{e},\kv_{2}\right)\in\T^{\E}}{g_{1}\left(\kv_{1}\right)g_{2}\left(\kv_{2}\right)=0}.\label{eq: Zg}
\end{equation}

\begin{lem}
\label{lem: Sigma minus Zg} The complement of $Z_{g}$ lies in $\Sigma^{reg}$
and can be described in two equivalent ways, using either the secular
function $F$ or canonical eigenfunctions $f_{\kv}$: 
\begin{align}
\Sigma\setminus Z_{g} & =\set{\kv\in\Sigma}{\frac{\partial F}{\partial\kappa_{e}}\left(\kv\right)\ne0},\,\,\,and\label{eq: Sigma minus Zg using df}\\
\Sigma\setminus Z_{g} & =\set{\kv\in\Sigma^{reg}}{f_{\kv}|_{e}\not\equiv0}.\label{eq: Sigma minus Zg as fke not zero}
\end{align}
\end{lem}

\begin{proof}
Taking derivative\footnote{Although $\Theta_{i}\left(\kv_{i}\right)$ may not be differentiable
on $Z_{g}$, we have no problem with derivatives with respect to $\kappa_{e}$.} of (\ref{eq: Bridge decomposition F}) we get 
\begin{equation}
\frac{\partial F}{\partial\kappa_{e}}\left(\kv_{1},\kappa_{e},\kv_{2}\right)=-ig_{1}\left(\kv_{1}\right)g_{2}\left(\kv_{2}\right)e^{-i\kappa_{e}}\left(1+e^{i2\kappa_{e}}e^{i\Theta_{1}\left(\kv_{1}\right)}e^{i\Theta_{2}\left(\kv_{2}\right)}\right).\label{eq: simple bridge dfdke}
\end{equation}
If $\kv\in Z_{g}$, then clearly $\frac{\partial F}{\partial\kappa_{e}}\left(\kv\right)=0$.
If $\kv\in\Sigma\setminus Z_{g}$ , namely $F\left(\kv\right)=0$
and $g_{1}\left(\kv_{1}\right)g_{2}\left(\kv_{2}\right)\ne0$, so
$e^{i2\kappa_{e}}e^{i\Theta_{1}\left(\kv_{1}\right)}e^{i\Theta_{2}\left(\kv_{2}\right)}=1$.
In such case 
\[
\frac{\partial F}{\partial\kappa_{e}}\left(\kv_{1},\kappa_{e},\kv_{2}\right)=-ig_{1}\left(\kv_{1}\right)g_{2}\left(\kv_{2}\right)e^{-i\kappa_{e}}2\ne0.
\]
This proves (\ref{eq: Sigma minus Zg using df}) and in particular
that $\nabla F\left(\kv\right)\ne0$ for $\kv\in\Sigma\setminus Z_{g}$
and therefore $\Sigma\setminus Z_{g}\subset\Sigma^{reg}$. Let $\kv\in\Sigma^{reg}$
and consider $f_{\kv}$ and its amplitudes vector $\boldsymbol{a}\in\ker\left(1-U_{\kv}\right)$.
Clearly $f_{\kv}|_{e}\equiv0\iff\left|a_{e}\right|^{2}+\left|a_{\hat{e}}\right|^{2}=0$
and by Lemma \ref{lem: F p and mk}, $\left|a_{e}\right|^{2}+\left|a_{\hat{e}}\right|^{2}=0\iff\frac{\partial F}{\partial\kappa_{e}}\left(\kv\right)=0$. 
\end{proof}
\begin{lem}
\label{lem: simple bridge inversion gi thetai}The maps $g_{i}$ and
$\Theta_{i}$ transform under the torus inversion $\mathcal{I}$ by
\begin{align}
\left|g_{i}\circ\mathcal{I}\right| & =\left|g_{i}\right|,\,\,and\label{eq: inversion gi}\\
e^{i\Theta_{i}}\circ\mathcal{I} & =e^{-i\Theta_{i}}.\label{eq: inversion Theta_i}
\end{align}
\end{lem}

\begin{proof}
Substituting $\ztrip=\left(e^{i\hat{\kappa}_{1}},e^{i\kappa_{e}},e^{i\hat{\kappa}_{2}}\right)$,
the inversion $e^{i\hat{\kappa}}\mapsto\I\left(e^{i\hat{\kappa}}\right)=\overline{e^{i\hat{\kappa}}}$
gives $\ztrip\mapsto\overline{\ztrip}$. As $D_{i}\left(\zi\right)$
is linear in $\zi$, then $\det\left(D_{i}\left(\overline{\zi}\right)\right)=\overline{\det\left(D_{i}\left(\zi\right)\right)}$
and therefore $\left|g_{i}\circ\mathcal{I}\right|=\left|g_{i}\right|$.
In the same way, and using Lemma \ref{lem: scattering phase}, $e^{i\Theta_{i}\left(\I\left(\kv_{i}\right)\right)}=\mathcal{S}_{i}\left(\overline{\zi}\right)=\overline{\mathcal{S}_{i}\left(\zi\right)}=e^{-i\Theta_{i}\left(\kv_{i}\right)}$.
\end{proof}
As we already showed that the bridge decomposition of the graph provides
a factorization of the secular function, one should expect that this
factorization will induce symmetries of the secular manifold. This
is indeed the case, and we will exploit these symmetries both in Section
\ref{sec: genericity} and in Section \ref{sec: Magnetic}. The two
bridge related symmetries of the secular manifold are as followed: 
\begin{defn}
\label{def: Tau_=00007Be=00007D}Let $\Gamma$ be a graph with a bridge
$e$, and bridge decomposition\\ $\Gamma\setminus e=\Gamma_{1}\sqcup\Gamma_{2}$
with notations as discussed above. We define the \emph{bridge extension,}
$\tau_{e}:\T^{\E}\rightarrow\T^{\E}$, by 
\begin{equation}
\tau_{e}\left(\kv_{1},\kappa_{e},\kv_{2}\right)=\left\{ \left(\kv_{1},\kappa_{e}+\pi,\kv_{2}\right)\right\} .\label{eq: Tau e}
\end{equation}
Let $v$ be the vertex connecting $\Gamma_{1}$ to the bridge $e$.
We define the \emph{torus cut-flip, }$\Rev:\T^{\E}\rightarrow\T^{\E}$,
by 
\begin{equation}
\Rev\left(\kv_{1},\kappa_{e},\kv_{2}\right)=\left\{ \left(\kv_{1},\kappa_{e}+\Theta_{2}\left(\kv_{2}\right),-\kv_{2}\right)\right\} ,\label{eq: Rev}
\end{equation}
where $\Theta_{2}$ is defined in Proposition \ref{prop: simple bridge det decomposition}.
\end{defn}

\begin{rem}
We call $\tau_{e}$ `bridge extension' as the two graphs $\Gamma_{\kv}$
and $\Gamma_{\tau_{e}\left(\kv\right)}$ are related by an extension
of the bridge by $\pi$. We call $\Rev$ `torus cut-flip' as the bridge
is a cut of the graph into $\Gamma_{1},\Gamma_{2}$, and $\Rev$ acts
as $\I$-inversion on $\Gamma_{2}$ without changing $\Gamma_{1}$.
That is, $\Gamma_{\Rev\left(\kv\right)}$ agree with $\Gamma_{\kv}$
on the restriction $\Gamma_{1}$ and agree with $\Gamma_{\I\left(\kv\right)}$
on the restriction to $\Gamma_{2}$. We specify that it is a torus
function not to be confused with the cut-flip function on discrete
graphs that we define in Section \ref{sec: Binomial}. 
\end{rem}

\begin{lem}
\label{lem: Re and tau e are measure preseving} All $\Sigma,\Sigma^{reg}$
and $Z_{g}$ are invariant to both $\tau_{e}$ and $\Rev$. Furthermore,
$\tau_{e}$ and $\Rev$ are $\mu_{\lv}$ preserving for any $\lv$
and satisfy $\tau_{e}^{2}=\Rev^{2}=identity$. Given $\kv\in\Sigma^{reg}$,
the traces of $f_{\kv},f_{\tau_{e}\left(\kv\right)}$ and $f_{\Rev\left(\kv\right)}$
are related as follows. For every $u\in\V$ and $e\in\E_{u}$,
\begin{align}
f_{\tau_{e}\left(\kv\right)}\left(u\right)=\begin{cases}
f_{\kv}\left(u\right) & u\in\V_{1}\\
-f_{\kv}\left(u\right) & u\in\V_{2}
\end{cases},\,\,\, & \partial_{e'}f_{\tau_{e}\left(\kv\right)}\left(u\right)=\begin{cases}
\partial_{e'}f_{\kv}\left(u\right) & u\in\V_{1}\\
-\partial_{e'}f_{\kv}\left(u\right) & u\in\V_{2}
\end{cases},\label{eq:f for tau_e}
\end{align}
and
\begin{align}
f_{\Rev\left(\kv\right)}\left(u\right)=\begin{cases}
f_{\kv}\left(u\right) & u\in\V_{1}\\
f_{\kv}\left(u\right) & u\in\V_{2}
\end{cases},\,\,\, & \partial_{e'}f_{\Rev\left(\kv\right)}\left(u\right)=\begin{cases}
\partial_{e'}f_{\kv}\left(u\right) & u\in\V_{1}\\
-\partial_{e'}f_{\kv}\left(u\right) & u\in\V_{2}
\end{cases}.\label{eq: f for R_e}
\end{align}
\end{lem}

\begin{rem}
Both (\ref{eq:f for tau_e}) and (\ref{eq: f for R_e}) hold up to
a global sign due to the sign ambiguity of canonical eigenfunctions.
\end{rem}

\begin{proof}
We begin with the bridge extension $\tau_{e}$. If $\kv\in\Sigma$
and $f\in Eig\left(\Gamma_{\kv},1\right)$, then we define the function
$\tau_{e}.f$ on $\Gamma_{\tau_{e}\left(\kv\right)}$ as follows.
On the edges of $\Gamma_{1},\Gamma_{2}$ where $\Gamma_{\kv}$ and
$\Gamma_{\tau_{e}\left(\kv\right)}$ share the same edge lengths,
we define: 
\begin{equation}
\tau_{e}.f|_{e'}=\begin{cases}
f|_{e'} & e'\in\E_{1}\\
-f|_{e'} & e'\in\E_{2}
\end{cases}.
\end{equation}
On the bridge $e$ of $\Gamma_{\tau_{e}\left(\kv\right)}$ which is
an extension\textbackslash reduction of the bridge of $\Gamma_{\kv}$
by $\pi$, we define $\tau_{e}.f|_{e}$ as an extension\textbackslash reduction
of $f|_{e}$, which is $2\pi$ periodic (see Definition \ref{def: edge restriction notations})
by half a period in the $\Gamma_{2}$ direction. In other words, if
$\boldsymbol{a}=\left(a_{1},a_{e},a_{\hat{e}},a_{2}\right)$ is the
amplitudes vector of $f$, then we define $\tau_{e}.f$ by the amplitudes
vector $\tau_{e}.a:=\left(a_{1},-a_{e},a_{\hat{e}},-a_{2}\right)$.
We may conclude from either of the latter descriptions of $\tau_{e}.f$
(using Definition \ref{def: edge restriction notations}) that 
\begin{equation}
\tau_{e}.f\left(u\right)=\begin{cases}
f\left(u\right) & u\in\V_{1}\\
-f\left(u\right) & u\in\V_{2}
\end{cases},\,\,\,\partial_{e'}\tau_{e}.f\left(u\right)=\begin{cases}
\partial_{e'}f\left(u\right) & u\in\V_{1}\\
-\partial_{e'}f\left(u\right) & u\in\V_{2}
\end{cases}.\label{eq: tau e vartex values}
\end{equation}
Therefore, $\tau_{e}.f$ satisfies Neumann vertex conditions on all
vertices and so\\ $\tau_{e}.f\in Eig\left(\Gamma_{\tau_{e}\left(\kv\right)},1\right)$.
Clearly this is an invertible linear map between $Eig\left(\Gamma_{\kv},1\right)$
and $Eig\left(\Gamma_{\tau_{e}\left(\kv\right)},1\right)$ so $\tau_{e}$
preserve both $\Sigma$ and $\Sigma^{reg}$. It is also clear that
$\tau_{e}.f|_{e}\equiv0\iff f|_{e}\equiv0$ so $\tau_{e}$ preserve
$Z_{g}$ according to Lemma \ref{lem: Sigma minus Zg}. Consider $f_{\kv}$
for some $\kv\in\Sigma^{reg}$ has amplitudes vector $\boldsymbol{a}$
then $\tau_{e}.f_{\kv}\in Eig\left(\Gamma_{\tau_{e}\left(\kv\right)},1\right)$
with amplitudes vector $\tau_{e}.\boldsymbol{a}$. It is obvious that
$\norm{\tau_{e}.a}=\norm a=1$ and that $\tau_{e}.f_{\kv}$ has real
trace, so $\tau_{e}.f_{\kv}=\pm f_{\tau_{e}\left(\kv\right)}$ and
so (\ref{eq: tau e vartex values}) implies (\ref{eq:f for tau_e}).
It is left to show that $\tau_{e}$ is BG-measure preserving. Clearly
by definition $\tau_{e}^{2}=\mathrm{identity}$ and $\tau_{e}.\boldsymbol{a}$
preserve the weights $\left|a_{e}\right|^{2}+\left|a_{\hat{e}}\right|^{2}=\left(m_{\kv}\right)_{e}$
and therefore according to (\ref{eq: normal}), $\tau_{e}$ preserve
the normal to $\Sigma^{reg}$. As $\tau_{e}$ is a translation, then
it is an isometry of $\T^{\E}$ and hence preserve the surface element
$ds$. It follows from Definition \ref{def: BG measure} and (\ref{eq: normal})
that $\tau_{e}$ preserve $\mu_{\lv}$ for any $\lv$.

We may now prove the same for $\Rev$. Given $f\in Eig\left(\Gamma_{\kv},1\right)$
with amplitudes vector $\boldsymbol{a}\in\left(1-U_{\kv}\right)$
for some $\kv\in\Sigma$, then as already discussed $\boldsymbol{\overline{a}}\in\left(1-U_{\I\left(\kv\right)}\right)$
and let us define $\I.f\in Eig\left(\Gamma_{\I\left(\kv\right)},1\right)$
by the amplitudes vector $\boldsymbol{\overline{a}}$. By Definition
\ref{def: edge restriction notations}, 
\[
\forall u\in\V,\,\forall e'\in\E_{u}\,\,\I.f\left(u\right)=f\left(u\right),\,\,\,\partial_{e'}\I.f\left(u\right)=-\partial_{e'}f\left(u\right).
\]
We define the function $\Rev.f$ on $\Gamma_{\Rev\left(\kv\right)}$
as follows. On the restriction to $\Gamma_{1}$ where $\Gamma_{\Rev\left(\kv\right)}$
and $\Gamma_{\kv}$ agree on the edge lengths, we define 
\begin{equation}
\Rev.f|_{e'}=f|_{e'}\,\,\forall e'\in\E_{1}.
\end{equation}
On the restriction to $\Gamma_{2}$ where $\Gamma_{\Rev\left(\kv\right)}$
and $\Gamma_{\I\left(\kv\right)}$ agree on the edge lengths, we define
\begin{equation}
\Rev.f|_{e'}=\I.f|_{e'}\,\,\forall e'\in\E_{2}.
\end{equation}
The bridge $e$ of $\Gamma_{\tau_{e}\left(\kv\right)}$ is an extension\textbackslash reduction
of the bridge of $\Gamma_{\kv}$ from $l_{e}$ (for which $\left\{ l_{e}\right\} =\kappa_{e}$)
to $\tilde{l_{e}}$ such that $\left\{ \tilde{l_{e}}\right\} =\left\{ \kappa_{e}+\Theta_{2}\left(\kv_{2}\right)\right\} $.
Let $u$ be the edge connecting $e$ to $\Gamma_{2}$, and $\hat{e}$
in the direction emitting from $u$. Let $A_{\hat{e}},\varphi_{\hat{e}}$
be the amplitude-phase pair of $f$ for $\hat{e}$ (see Definition
\ref{def: edge restriction notations}) so that $f|_{e}\left(x_{\hat{e}}\right)=A_{\hat{e}}\cos\left(x_{\hat{e}}-\varphi_{\hat{e}}\right)$
and define $\Rev.f|_{e}\left(x_{\hat{e}}\right)=A_{\hat{e}}\cos\left(x_{\hat{e}}+\varphi_{\hat{e}}\right)$.
Using Lemma \ref{lem: Rev phase}, $e^{i\left(\tilde{l_{e}}+\varphi_{\hat{e}}\right)}=e^{i\left(\kappa_{e}+\Theta_{2}\left(\kv_{2}\right)+\varphi_{\hat{e}}\right)}=e^{i\left(\kappa_{e}-\varphi_{\hat{e}}\right)}=e^{i\left(l_{e}-\varphi_{\hat{e}}\right)}$.
Therefore the traces of $f|_{e}$ and $\Rev.f|_{e}$ are given by
\begin{align}
f|_{e}\left(u\right)=A_{\hat{e}}\cos\left(-\varphi_{\hat{e}}\right) & =A_{\hat{e}}\cos\left(\varphi_{\hat{e}}\right)=\Rev.f|_{e}\left(u\right)\\
\partial_{e}f|_{e}\left(u\right)=A_{\hat{e}}\sin\left(\varphi_{\hat{e}}\right) & =-A_{\hat{e}}\sin\left(\varphi_{\hat{e}}\right)=-\partial_{e}\Rev.f|_{e}\left(u\right)\\
f|_{e}\left(v\right)=A_{\hat{e}}\cos\left(l_{e}-\varphi_{\hat{e}}\right) & =A_{\hat{e}}\cos\left(\tilde{l_{e}}+\varphi_{\hat{e}}\right)=\Rev.f|_{e}\left(v\right)\\
\partial_{e}f|_{e}\left(u\right)=A_{\hat{e}}\sin\left(l_{e}-\varphi_{\hat{e}}\right) & =A_{\hat{e}}\sin\left(\tilde{l_{e}}+\varphi_{\hat{e}}\right)=\partial_{e}\Rev.f|_{e}\left(v\right)
\end{align}
We may conclude that the trace of $\Rev.f$ agree with $f$ on $\V_{1}$
and agree with $\I.f$ on $\V_{2}$, so 
\[
\Rev.f\left(u\right)=\begin{cases}
f\left(u\right) & u\in\V_{1}\\
f\left(u\right) & u\in\V_{2}
\end{cases},\,\,\,\partial_{e'}\Rev.f\left(u\right)=\begin{cases}
\partial_{e'}f\left(u\right) & u\in\V_{1}\\
-\partial_{e'}f\left(u\right) & u\in\V_{2}
\end{cases}.
\]
Therefore,$\Rev.f$ satisfies Neumann vertex conditions on all vertices
and so\[\Rev.f\in Eig\left(\Gamma_{\Rev\left(\kv\right)},1\right).\]
Clearly this is a linear map between $Eig\left(\Gamma_{\kv},1\right)$
and $Eig\left(\Gamma_{\Rev\left(\kv\right)},1\right)$, to see that
it is invertible we use Lemma \ref{lem: simple bridge inversion gi thetai}
to conclude that \[ \Rev^{2}\left(\kv\right)=\left\{ \kv_{1},\kappa_{e}+\Theta_{2}\left(\kv_{2}\right)+\Theta_{2}\left(-\kv_{2}\right),\kv_{2}\right\} =\kv,\]
so $\Rev^{2}=\mathrm{identity}$ and we can deduce that $\Rev.\left(\Rev.f\right)=f$.
Hence $Eig\left(\Gamma_{\kv},1\right)\cong Eig\left(\Gamma_{\Rev\left(\kv\right)},1\right)$
and so $\Rev$ preserve both $\Sigma$ and $\Sigma^{reg}$. It is
also clear that $\Rev.f|_{e}\equiv0\iff f|_{e}\equiv0$ so $\Rev$
preserve $Z_{g}$ according to Lemma \ref{lem: Sigma minus Zg}. Exactly
as we did for $\tau_{e}$, $\Rev.f_{\kv}=\pm f_{\Rev\left(\kv\right)}$
for any $\kv\in\Sigma^{reg}$, which proves (\ref{eq: f for R_e}),
and $\Rev$ preserve the weights $\left(m_{\kv}\right)_{e}$ and the
normal to $\Sigma^{reg}$. It is left to prove that $\Rev$ preserve
the surface element $ds$ to finish the proof. Since $\Rev\left(\kv\right)=\left\{ \kv_{1},\kappa_{e}+\Theta_{2}\left(\kv_{2}\right),-\kv_{2}\right\} $
then its derivative (the Jacobian matrix) is triangular with $\pm1$
on the diagonal, it follows that $ds$ is preserved, and therefore
$\mu_{\lv}$ is preserved for any $\lv$.
\end{proof}

\subsection{Loops and the secular manifold}

Recall that a loop is an edge connecting a vertex to itself. As discussed
in the introduction, if $\Gamma_{\lv}$ is a standard graph with a
loop $e$, then there are infinitely many eigenfunction supported
on $e$. The nodal count and Neumann count cannot be defined on such
eigenfunctions. Moreover, as we will prove in Section \ref{sec: genericity},
our definition of generic eigenfunction (Definition \ref{def:generic_eigenfunction})
is only generic among the eigenfunctions that are not supported on
loops. The purpose of this subsection is to provide the machinery
needed in order to exclude loop eigenfunctions from discussions. We
partition $\Sigma^{reg}$ into two parts: 
\begin{equation}
\Sigma_{\L}:=\set{\kv\in\Sigma^{reg}}{f_{\kv}\,\,\text{is supported on a loop}},\label{eq: definition Sigma L}
\end{equation}
and its complement 
\begin{equation}
\Sigma_{\L}^{c}:=\Sigma^{reg}\setminus\Sigma_{\L}.\label{eq: definition Sigma L comp}
\end{equation}
We will show in this subsection that the secular manifold's machinery
can work on each part separately. To do so we first construct $\Sigma_{\L}$
explicitly.
\begin{defn}
\label{def: Ze}Given a graph $\Gamma$, we denote its set of loops
by $\E_{loops}$. For every loop $e\in\E_{loops}$, we define the
sub-torus
\[
Z_{e}:=\set{\kv\in\T^{\E}}{e^{i\kappa_{e}}=1},
\]
and we define the \emph{loop factor}, $\tilde{Z}_{e}$, as its intersection
with $\Sigma^{reg}$:
\[
\tilde{Z_{e}}:=Z_{e}\cap\Sigma^{reg}.
\]
\end{defn}

\begin{lem}
\label{lem: Ze} Let $\Gamma_{\lv}$ be a standard graph and let $e$
be a loop of length $l_{e}$. Then, 
\begin{enumerate}
\item \label{enu: Ze 1}There exists an eigenfunction $f\in Eig\left(\Gamma_{\lv},k^{2}\right)$
if and only if $k\in\frac{2\pi}{l_{e}}\N$. 
\item \label{enu: Ze2}There exists an eigenfunction $f\in Eig\left(\Gamma_{\lv},k^{2}\right)$
if and only if $\left\{ k\lv\right\} \in Z_{e}$.
\item \label{enu: Ze 3 - explicit fk}If $\kv\in\tilde{Z}_{e}$, then $f_{\kv}$
is supported on $e$, it is given (up to a sign) by 
\[
f_{\kv}|_{e}\left(x_{e}\right)=\frac{1}{\sqrt{2}}\sin\left(x_{e}\right),
\]
and its amplitudes vector satisfies $a_{e}=-a_{\hat{e}}=\frac{1}{\sqrt{2}}$
and zeros on all other entries.
\end{enumerate}
\end{lem}

\begin{proof}
Clearly $k\in\frac{2\pi}{l_{e}}\N$ is equivalent to $e^{ikl_{e}}=1$
and so (\ref{enu: Ze 1}) and (\ref{enu: Ze2}) are equivalent statements.
To prove (\ref{enu: Ze2}) first assume that there is a real normalized
eigenfunction $f\in Eig\left(\Gamma_{\lv},k^{2}\right)$, supported
on $e$, with amplitudes vector \textbf{$\boldsymbol{a}$}. Since
the constant eigenfunction is not supported on $e$, then $k>0$. Since
$f|_{e'}\equiv0$ for every edge $e'\ne e$, then \textbf{$\boldsymbol{a}$}
is supported on $e,\hat{e}$. Since $\boldsymbol{a}$ is normalized
then $\left|a_{e}\right|^{2}+\left|a_{\hat{e}}\right|^{2}=1$ and
by Lemma \ref{lem: |ae|=00003D|ae'|} $\left|a_{e}\right|=\left|a_{\hat{e}}\right|=\frac{1}{\sqrt{2}}$.
The Neumann condition on the vertex of the loop implies that $f|_{e}\left(0\right)=f|_{e}\left(l_{e}\right)=0$
and $f|_{e}'\left(0\right)=f|_{e}'\left(l_{e}\right)$. Using the
relation between $\boldsymbol{a}$ and the trace of $f$ as seen in
Definition \ref{def: edge restriction notations}, we get three conditions:
\begin{align*}
e^{-ikl_{e}}a_{e}+a_{\hat{e}} & =0,\\
a_{e}+e^{-ikl_{e}}a_{\hat{e}} & =0,\\
ik\left(e^{-ikl_{e}}a_{e}-a_{\hat{e}}\right) & =ik\left(a_{e}-e^{-ikl_{e}}a_{\hat{e}}\right).
\end{align*}
The first two and $\left|a_{e}\right|=\left|a_{\hat{e}}\right|=\frac{1}{\sqrt{2}}$
implies that $e^{ikl_{e}}=-\frac{a_{e}}{a_{\hat{e}}}=-\frac{a_{\hat{e}}}{a_{e}}$
and so $e^{ikl_{e}}=\pm1$. Dividing the third condition by $ika_{e}$,
as $k>0$, insure that $e^{ikl_{e}}=1$. Hence, $\left\{ k\lv\right\} \in Z_{e}$
and $a_{e}=-a_{\hat{e}}=\frac{1}{\sqrt{2}}$ which means that $f|_{e}\left(x_{e}\right)=\frac{1}{\sqrt{2}}\sin\left(kx_{e}\right)$.
On the other hand, if we assume that $k>0$ is such that $e^{ikl_{e}}=1$
then clearly the function $f$ constructed above is an eigenfunction
of $\Gamma_{\lv}$, with eigenvalue $k^{2}$, that is supported on
$e$. This proves (\ref{enu: Ze 1}) and (\ref{enu: Ze2}), and if
we consider a point $\kv\in\tilde{Z}_{e}$, namely $e^{i\kappa_{e}}=1$
and $\kv\in\Sigma^{reg}$, then the above construction provides a
real normalized eigenfunction $f\in Eig\left(\Gamma_{\kv},1\right)$
that is supported on $e$ with $a_{e}=-a_{\hat{e}}=\frac{1}{\sqrt{2}}$
and $f|_{e}\left(x_{e}\right)=\frac{1}{\sqrt{2}}\sin\left(x_{e}\right)$.
Since $\kv\in\Sigma^{reg}$, this function is (up to a sign) $f_{\kv}$. 
\end{proof}
As an immediate corollary:
\begin{cor}
\label{cor: Sigma L} If $\Gamma$ is a graph with set of loops $\E_{loops}$,
then $\Sigma_{\L}$, as defined in (\ref{eq: definition Sigma L}),
is the disjoint union of $\tilde{Z}_{e}$'s: 
\[
\Sigma_{\L}=\sqcup_{e\in\E_{loops}}\tilde{Z}_{e}.
\]
\end{cor}

See Figures \ref{fig: secman with loops 3flower}, \ref{fig: secman with loops 2-1 stower},
\ref{fig: secman with loops 12 stower} and \ref{fig: secman with loops dumbbell}
in Appendix \ref{sec: Appendix examples} for several examples of
$\Sigma$ and $\Sigma_{\L}$. In all of these figures $\Sigma_{\L}$
is colored in blue. 

Let us now construct an algebraic characterization of $\Sigma_{\L}$
and its complement $\Sigma_{\L}^{c}$. As seen in Lemma \ref{lem: Ze}
and its proof, a real normalized eigenfunction is supported on a loop
$e$ if and only if its amplitudes vector $\boldsymbol{a}\in\C^{\vec{\E}}$
is $a_{e}=-a_{\hat{e}}=\frac{1}{\sqrt{2}}$ and zero elsewhere. Let
us denote this normalized anti-symmetric vector by $\hat{e}_{-}$
and let $\hat{e}_{+}$ be the orthogonal symmetric vector. That is,
\begin{equation}
\left(\hat{e}_{\pm}\right)_{e'}=\begin{cases}
\frac{1}{\sqrt{2}} & e'=e\\
\pm\frac{1}{\sqrt{2}} & e'=\hat{e}\\
0 & else
\end{cases}.
\end{equation}
It follows from the construction of the real scattering matrix $S$
in (\ref{eq: S explicietly}) and the matrix $e^{i\hat{\kappa}}$
that they are invariant to the swapping $e\leftrightarrow\hat{e}$
if $e$ is a loop. Therefore so does $U_{\kv}=e^{i\hat{\kappa}}$.
Consider the basis of $\C^{\vec{\E}}$ in the following order: The
antisymmetric vectors $\hat{e}_{-}$ for any $e\in\E_{loops}$, then
the symmetric vectors $\hat{e}_{+}$ for any $e\in\E_{loops}$ and
then the rest of the directed edges. We denote the anti-symmetric
part of the basis by $\E_{as}$ and the rest (including the symmetric
part) by $\E_{0}$. With this order and edge grouping the bond-scattering
matrix $S$ and the unitary evolution matrix $U_{\kv}$ have the following
block structure:
\begin{align}
S & =\begin{pmatrix}\boldsymbol{1} & 0\\
0 & S_{0}
\end{pmatrix},\,\,and\\
U_{\kv} & =\begin{pmatrix}e^{i\hat{\kappa}_{loops}} & 0\\
0 & e^{i\hat{\kappa}_{0}}S_{0}
\end{pmatrix}.\label{eq: U loop decomposition}
\end{align}
Where, denoting $E_{loops}=\left|\E_{loops}\right|$ and $E_{0}=2E-E_{loops}$,
$\boldsymbol{1}$ is the identity matrix of dimension $E_{loops}$,
$e^{i\hat{\kappa}_{loops}}$ is diagonal of dimension $E_{loops}$
with entries $\left\{ e^{i\kappa_{e}}\right\} _{e\in\E_{loops}}$
(each appears once), $S_{0}$ is a real orthogonal matrix of dimension
$E_{0}$ and $e^{i\hat{\kappa}_{0}}$ is diagonal of dimension $E_{0}$.
The diagonal entries of $e^{i\hat{\kappa}_{0}}$ are $\left\{ e^{i\kappa_{e}}\right\} _{e\in\E_{loops}}$
that appear once and $\left\{ e^{i\kappa_{e}}\right\} _{e\in\E\setminus\E_{loops}}$
that appear twice. 

The purpose of this subsection, the exclusion of loops-eigenfunctions
from the discussion, is done by restricting $U_{\kv}$ to $\C^{\E_{0}}$: 
\begin{defn}
\label{def: Zo} Denote the unitary matrix $U_{0}\left(\kv\right):=e^{i\hat{\kappa}_{0}}S_{0}$.
We define the \emph{main factor of $\Sigma$ }as 

\begin{align}
Z_{0} & :=\set{\kv\in\T^{\E}}{\det\left(1-U_{0}\left(\kv\right)\right)=0},\label{eq: Z0-1}
\end{align}
and similarly to $\Sigma^{reg}$, we define the \emph{regular part}
of $Z_{0}$ by 
\begin{equation}
Z_{0}^{reg}:=\set{\kv\in\T^{\E}}{\dim\ker\left(1-U_{0}\left(\kv\right)\right)=1}.\label{eq: Z0reg by dimker}
\end{equation}
\end{defn}

\begin{rem}
\label{rem: Z0 reg as determinant} The same argument in the proof
of Lemma \ref{lem:The-secular-functions properties} will give:
\begin{equation}
Z_{0}^{reg}=\set{\kv\in Z_{0}}{\nabla\det\left(1-U_{0}\left(\kv\right)\right)\ne0}.\label{eq: Zoreg as det}
\end{equation}
\end{rem}

Examples of $Z_{0}$ and can be shown in Figures \ref{fig: secman with loops 3flower},
\ref{fig: secman with loops 2-1 stower}, \ref{fig: secman with loops 12 stower}
and \ref{fig: secman with loops dumbbell} in Appendix \ref{sec: Appendix examples}
where the right picture in each figure is $Z_{0}$. In Figures \ref{fig: secman with loops 3flower}
and \ref{fig: secman with loops dumbbell} $Z_{0}=Z_{0}^{reg}$. In
Figures \ref{fig: secman with loops 2-1 stower} and \ref{fig: secman with loops 12 stower}
one can spot singular points of $Z_{0}$ (at height $\pm\pi$) where
the layers of $Z_{0}$ meet. Therefore, in these cases, $Z_{0}\ne Z^{reg}$.
\begin{lem}
\label{lem: loop decomposition} Let $\Gamma$ be a graph with set
of loops $\E_{loops}$. Then, 
\begin{enumerate}
\item \label{enu: loops decomposition of F}The secular function, $F\left(\kv\right)=\det\left(U_{\kv}\right)^{-\frac{1}{2}}\det\left(1-U_{\kv}\right)$,
is factorized as follows: 
\begin{align}
F\left(\kv\right) & =\det\left(U_{\kv}\right)^{-\frac{1}{2}}\Pi_{e\in\E_{loops}}\left(1-e^{i\kappa_{e}}\right)\det\left(1-U_{0}\left(\kv\right)\right),\label{eq: det of lops}
\end{align}
and the secular manifold is decomposed to $\Sigma=Z_{0}\cup_{e\in\E_{loops}}Z_{e}$. 
\item \label{enu: loop secular manifold - union of factors}The complement
of $\Sigma_{\L}$ in $\Sigma^{reg}$ is 
\begin{equation}
\Sigma_{\L}^{c}=Z_{0}^{reg}\cap\Sigma^{reg}.\label{eq: Sigma L comp}
\end{equation}
Each of the $\tilde{Z}_{e}$'s and $\Sigma_{\L}^{c}$ are unions of
connected components of $\Sigma^{reg}$, and $\Sigma^{reg}$ is their
disjoint union:
\begin{equation}
\Sigma^{reg}=\Sigma_{\L}^{c}\sqcup_{e\in\E_{loops}}\tilde{Z_{e}}.\label{eq: decomp of Sigma reg}
\end{equation}
\item \label{enu: loops factor charecterization}Each loop factor $\tilde{Z_{e}}$
is characterized by
\begin{align}
\tilde{Z}_{e} & =\set{\kv\in\T^{\E}}{e^{i\kappa_{e}}=1\,\,and\,\,\frac{F\left(\kv\right)}{1-e^{i\kappa_{e}}}\ne0},\label{eq: Ze reg}
\end{align}
and $\Sigma_{\L}$ can be characterized as 
\begin{equation}
\Sigma_{\L}=\set{\kv\in\T^{\E}}{\Pi_{e\in\E_{loops}}\left(1-e^{i\kappa_{e}}\right)=0\,\,\,and\,\,\,\det\left(1-U_{0}\left(\kv\right)\right)\ne0}.\label{eq: Sigma L through F}
\end{equation}
\item \label{enu: loops factors measure} Given any choice of $\lv\in\left(\R_{+}\right)^{\E}$,
the Barra-Gaspard measure of the above sets is given by:
\begin{align*}
\forall e\in\E_{loops}\,\,\,\,\,\mu_{\lv}\left(\tilde{Z_{e}}\right) & =\frac{l_{e}}{2L},\,\,and\\
\mu_{\lv}\left(\Sigma_{\L}^{c}\right) & =1-\frac{\sum_{e\in\E_{loops}}l_{e}}{2L}\ge\frac{1}{2}.
\end{align*}
\end{enumerate}
\end{lem}

\begin{proof}
The decomposition $\det\left(1-U_{\kv}\right)=\Pi_{e\in\E_{loops}}\left(1-e^{i\kappa_{e}}\right)\det\left(1-U_{0}\left(\kv\right)\right)$
is immediate from (\ref{eq: U loop decomposition}). The factorization
of $F$ and the decomposition of $\Sigma$ follows. This proves (\ref{enu: loops decomposition of F}).
By definition $\tilde{Z}_{e}=\set{\kv\in\Sigma^{reg}}{\left(1-e^{i\kappa_{e}}\right)=0}$
and so the $\Sigma$ decomposition induce a $\Sigma^{reg}$ decomposition:
\begin{equation}
\Sigma^{reg}=\left(Z_{0}\cap\Sigma^{reg}\right)\cup_{e\in\E_{loops}}\tilde{Z_{e}}.\label{eq: predecomp of Sigma reg}
\end{equation}
To show that this is a disjoint union, we use the decomposition of
$U_{\kv}$ in (\ref{eq: U loop decomposition}) to get 
\begin{align}
\dim\ker\left(1-U_{\kv}\right) & =\dim\ker\left(1-e^{i\hat{\kappa}_{loops}}\right)+\dim\ker\left(1-U_{0}\left(\kv\right)\right)\label{eq: loops kernel decomposition}\\
= & \left|\set{e\in\E_{loops}}{\kv\in Z_{e}}\right|+\dim\ker\left(1-U_{0}\left(\kv\right)\right).\nonumber 
\end{align}
It follows that if $\kv\in\Sigma^{reg}$, namely $\dim\ker\left(1-U_{\kv}\right)=1$,
then either $\kv\notin Z_{0}$ and $\left|\set{e\in\E_{loops}}{\kv\in Z_{e}}\right|=1$
or $\kv\in Z_{0}^{reg}$ and $\left|\set{e\in\E_{loops}}{\kv\in Z_{e}}\right|=0$.
Therefore, $\left(Z_{0}\cap\Sigma^{reg}\right)=Z_{0}^{reg}\cap\Sigma^{reg}$
and (\ref{eq: predecomp of Sigma reg}) can be upgraded to a disjoint
union. As Corollary \ref{cor: Sigma L} states that $\Sigma_{\L}=\sqcup_{e\in\E_{loops}}\tilde{Z}_{e}$
then $\Sigma_{\L}^{c}=Z_{0}^{reg}\cap\Sigma^{reg}$. It is left to
show that each $\tilde{Z}_{e}$ is a union of connected components
of $\Sigma^{reg}$ in order to conclude that so does $\Sigma_{\L}^{c}$
and thus prove (\ref{eq: decomp of Sigma reg}). To do so we need
to prove that $\tilde{Z}_{e}$ is both open and closed in $\Sigma^{reg}$.
It is closed as it is the zero set of $1-e^{i\kappa_{e}}$ on $\Sigma^{reg}$.
To shoe that it is open, consider the sub-torus $Z_{e}$ which is
closed in $\T^{\E}$ and has $\dim\left(Z_{e}\right)=E-1$. Its intersection
with the closed variety $\Sigma^{sing}$ is closed in $Z_{e}$, and
so $\tilde{Z_{e}}=Z_{e}\setminus\Sigma^{sing}$ is open in $Z_{e}$.
So each point in $\tilde{Z_{e}}$ has a small enough $E-1$ dimensional
neighborhood in $Z_{e}$ (and hence in $\Sigma$) which does not intersect
$\Sigma^{sing}$. Therefore, $\tilde{Z_{e}}=Z_{e}\cap\Sigma^{reg}$
is open in $\Sigma^{reg}$. So $\tilde{Z}_{e}$ is both closed and
open in $\Sigma^{reg}$, which concludes the proof of (\ref{enu: loop secular manifold - union of factors}).
Now consider $\tilde{Z_{e}}$ and a factorization $F\left(\kv\right)=\left(1-e^{i\kappa_{e}}\right)\frac{F}{\left(1-e^{i\kappa_{e}}\right)}$
where $\frac{F}{\left(1-e^{i\kappa_{e}}\right)}$ is a trigonometric
polynomial according to (\ref{eq: det of lops}). The gradient of
$F$ is given by: 
\[
\nabla F=\left(1-e^{i\kappa_{e}}\right)\nabla\frac{F}{\left(1-e^{i\kappa_{e}}\right)}+\frac{F}{\left(1-e^{i\kappa_{e}}\right)}\nabla\left(1-e^{i\kappa_{e}}\right).
\]
If $\kv\in Z_{e}$, then $\nabla F\left(\kv\right)=\frac{F\left(\kv\right)}{\left(1-e^{i\kappa_{e}}\right)}\nabla\left(1-e^{i\kappa_{e}}\right)$
and so for $\kv\in Z_{e}$, $\nabla F\left(\kv\right)=0$ if only
if $\frac{F\left(\kv\right)}{\left(1-e^{i\kappa_{e}}\right)}=0$.
This proves (\ref{eq: Ze reg}). It now follows that $\kv\in\Sigma_{\L}=\sqcup_{e\in\E_{loops}}$
if and only if $\Pi_{e\in\E_{loops}}\left(1-e^{i\kappa_{e}}\right)=0$
and $\frac{F\left(\kv\right)}{\Pi_{e\in\E_{loops}}\left(1-e^{i\kappa_{e}}\right)}\ne0$.
To prove (\ref{eq: Sigma L through F}) we simply recall that 
\[
\frac{F\left(\kv\right)}{\Pi_{e\in\E_{loops}}\left(1-e^{i\kappa_{e}}\right)}=\det\left(U_{\kv}\right)^{-\frac{1}{2}}\det\left(1-U_{0}\left(\kv\right)\right),
\]
 by (\ref{eq: det of lops}), and that $\det\left(U_{\kv}\right)^{-\frac{1}{2}}\ne0$.
We are left with proving (\ref{enu: loops factors measure}). Since
$Z_{e}\setminus\tilde{Z}_{e}\subset\Sigma^{sing}$, then it is of positive
co-dimension and therefore $\int_{\tilde{Z}_{e}}ds=\int_{Z_{e}}ds=\left(2\pi\right)^{E-1}$.
As the normal at a point in $\tilde{Z_{e}}$ is in the direction of
$e$, then 
\[
\mu_{\lv}\left(\tilde{Z_{e}}\right)=\frac{\pi}{L}\frac{1}{\left(2\pi\right)^{E}}\int_{\tilde{Z_{e}}}\left|\hat{n}\cdot\lv\right|ds=\frac{\pi}{L}\frac{l_{e}}{\left(2\pi\right)^{E}}\int_{\tilde{Z_{e}}}ds=\frac{l_{e}}{2L}.
\]
We may deduce from the disjoint decomposition of $\Sigma^{reg}$ that
\[
\mu_{\lv}\left(Z_{0}^{reg}\cap\Sigma^{reg}\right)=\mu_{\lv}\left(\Sigma^{reg}\right)-\sum_{e\in\E_{loops}}\mu_{\lv}\left(\tilde{Z}_{e}\right)=1-\frac{\sum_{e\in\E_{loops}}l_{e}}{2L}\ge\frac{1}{2},
\]
thus proving \ref{enu: loops factors measure}. 
\end{proof}
We may now deduce that $Z_{0}$ and $Z_{0}^{reg}$ share the same
real analytic structure as $\Sigma$ and $\Sigma^{reg}$.
\begin{cor}
Like the secular manifold, the main factor $Z_{0}$ is a real analytic
variety of dimension $E-1$, its singular part is of positive co-dimension,
and its regular part, $Z_{0}^{reg}$, is a real analytic manifold of
the same dimension. 
\end{cor}

\begin{proof}
If we can show that $Z_{0}^{reg}\ne\emptyset$, then the proof is exactly
as the proof of Proposition \ref{prop: Secular manifold structure},
using Definition \ref{def: Zo} and Remark \ref{rem: Z0 reg as determinant}.
To show that $Z_{0}^{reg}\ne\emptyset$ we simply notice that $Z_{0}^{reg}\cap\Sigma^{reg}$
has a positive measure according to (\ref{enu: loops factors measure})
of the lemma above.
\end{proof}
It was shown in Lemma \ref{lem: vertex values of canonical function}
that quadratic combinations of the trace of canonical eigenfunctions
are real analytic on $\Sigma^{reg}$. This is a useful property of
$\Sigma^{reg}$ and one should expect $Z_{0}^{reg}$ to inherit this
property by continuation from $\Sigma_{\L}^{c}=Z_{0}^{reg}\cap\Sigma^{reg}$
to $Z_{0}^{reg}$. 
\begin{lem}
\label{lem: vertex values for loops} Let $\Gamma$ be a graph with
loops. Then for any $v,u\in\V$ and $e\in\E_{v},\,e'\in\E_{u}$, all
$f_{\kv}\left(v\right)f_{\kv}\left(u\right)$, $\partial_{e}f_{\kv}\left(v\right)f_{\kv}\left(u\right)$
and $\partial_{e}f_{\kv}\left(v\right)\partial_{e'}f_{\kv}\left(u\right)$
can be extended from $\Sigma_{\L}^{c}$ to real analytic functions
on $Z_{0}^{reg}$.
\end{lem}

\begin{proof}
Let $\kv\in Z_{0}^{reg}$ and let $\boldsymbol{a}_{0}\in\ker\left(1-U_{0}\left(\kv\right)\right)$
be a normalized vector. Using Lemma \ref{lem: adjugate} and $\dim\ker\left(1-U_{\kv}\right)=1$
we get that $\mathrm{trace}\left(\mathrm{adj}\left(1-U_{0}\left(\kv\right)\right)\right)\ne0$
and that,

\[
\boldsymbol{a}_{0}\boldsymbol{a}_{0}^{*}=\frac{\mathrm{adj}\left(1-U_{0}\left(\kv\right)\right)}{\mathrm{trace}\left(\mathrm{adj}\left(1-U_{0}\left(\kv\right)\right)\right)}.
\]
Let $O$ be the real orthogonal matrix transforming the standard basis
of $\C^{\vec{\E}}$ into the basis of $\C^{\E_{as}}\times\C^{\E_{0}}$
on which $U_{\kv}$ has the block structure $e^{i\hat{\kappa}_{loops}}\oplus U_{0}\left(\kv\right)$
(see (\ref{eq: U loop decomposition})). Let $\kv\in\Sigma^{reg}\cap Z_{0}^{reg}$
with $\boldsymbol{a}$ being the amplitudes vector of $f_{\kv}$ in
the standard basis and $O\boldsymbol{a}=\begin{pmatrix}\boldsymbol{a}_{loops}\\
\boldsymbol{a}_{0}
\end{pmatrix}$ in the basis of $\C^{\E_{as}}\times\C^{\E_{0}}$. Since $\kv\in\Sigma^{reg}\cap Z_{0}^{reg}$
then 
\[
\dim\ker\left(1-U_{\kv}\right)=\dim\ker\left(1-U_{0}\left(\kv\right)\right)=1,
\]
which means that $\dim\ker\left(1-e^{i\hat{\kappa}_{loops}}\right)=0$
according to (\ref{eq: loops kernel decomposition}). Therefore, $O\boldsymbol{a}=\begin{pmatrix}0\\
\boldsymbol{a}_{0}
\end{pmatrix}$ with $\boldsymbol{a}_{0}\in\ker\left(1-U_{0}\left(\kv\right)\right)$.
As $\boldsymbol{a}$ is normalized and $O$ is orthogonal, then $\boldsymbol{a}_{0}$
is normalized and therefore, as shown in Lemma \ref{lem: vertex values of canonical function},
\begin{align}
f_{\kv}\left(v\right)\cdot f_{\kv}\left(u\right)= & \left(\left(S+J\right)\boldsymbol{a}\boldsymbol{a}^{*}\left(S+J\right)\right)_{e,e'}\\
= & \left(\left(S+J\right)O^{T}\begin{pmatrix}0 & 0\\
0 & \boldsymbol{a}_{0}\boldsymbol{a}_{0}^{*}
\end{pmatrix}O\left(S+J\right)\right)_{e,e'}\\
= & \left(\left(S+J\right)O^{T}\begin{pmatrix}0 & 0\\
0 & \frac{\mathrm{adj}\left(1-U_{0}\left(\kv\right)\right)}{\mathrm{trace}\left(\mathrm{adj}\left(1-U_{0}\left(\kv\right)\right)\right)}
\end{pmatrix}O\left(S+J\right)\right)_{e,e'}.\label{eq: matrix element on Z0reg}
\end{align}
Where $e$ and $e'$ are directed edges going out of $v$ and $u$
correspondingly. Exactly as in the proof of Lemma \ref{lem: vertex values of canonical function},
the matrix $\frac{\mathrm{adj}\left(1-U_{0}\left(\kv\right)\right)}{\mathrm{trace}\left(\mathrm{adj}\left(1-U_{0}\left(\kv\right)\right)\right)}$
is a rational function in $\left\{ e^{i\kappa_{e}}\right\} _{e\in\E}$
with no poles on $Z_{0}^{reg}$ so the matrix element in (\ref{eq: matrix element on Z0reg})
can be extended to a rational function in $\left\{ e^{i\kappa_{e}}\right\} _{e\in\E}$
on $Z_{0}^{reg}$ (with no poles). Since $f_{\kv}\left(v\right)\cdot f_{\kv}\left(u\right)$
is defined and real on $Z_{0}^{reg}\cap\Sigma^{reg}$ which is open
and dense in $Z_{0}^{reg}$, then by continuity the matrix element
in (\ref{eq: matrix element on Z0reg}) is real and is therefore a
real analytic function on $Z_{0}^{reg}$.
\end{proof}

\subsection{Connectednes of the secular manifold }

It was conjectured by Colin de Verdière in \cite{CdV_ahp15} that
the regular part $\Sigma^{reg}$ is reducible\footnote{By reducible, we mean that the multivariate polynomial $p:\C^{\E}\rightarrow\C$
such that $\det\left(1-U_{\kv}\right)=p\left(e^{i\kappa_{1}},e^{i\kappa_{2}}...\right)$,
is completely reducible. } if and only if there is an isometry of $\Gamma_{\lv}$ that is $\lv$
independent. It is not hard to show that this happens only if a graph
has loops (in such case the symmetry is a reflection of the loop)
or if the graph is \emph{mandarin} (also known as \emph{pumpkin}),
a graph with two vertices such that every edge connects the two (in
such case the symmetry is a reflection of all edges). A proof, yet
to be published, for Colin de Verdière's conjecture was given by Kurasov
and Sarnak in \cite{Sarnakurasov}, where they also prove that if
the graph has loops, then $\Sigma$ is reducible but $Z_{0}$ is irreducible.
A related question, following this conjecture, was asked by Berkolaiko
and Liu in \cite{BerLiu_jmaa17} regarding the number of connected
components of $\Sigma^{reg}$. They prove the following theorem:
\begin{thm}
\cite{BerLiu_jmaa17} If $\Gamma$ is a graph with a tail (namely
$\left|\partial\Gamma\right|>0$) and no loops, then $\Sigma^{reg}$
has two connected components. 
\end{thm}

Berkoliako and Liu also suggested that the same proof should hold
for a graph with a bridge and no loops. In the following section we
will prove this result for bridges (including tails) and no loops,
providing also an isometry between the two components of $\Sigma^{reg}$. 

Similarly to the irreducibility result of \cite{Sarnakurasov} which
shows that for a graph with loops $Z_{0}$ is irreducible and not
$\Sigma$, we will show that in the case where the graph has loops,
$Z_{0}^{reg}$ is connected while $\Sigma^{reg}$ has at least $\left|\E_{loops}\right|+1$
connected components. As an example of $Z_{0}^{reg}$ see Figure \ref{fig: Zo reg}.
More examples are found in Appendix \ref{sec: Appendix examples}.

\begin{figure}
\includegraphics[width=0.35\paperwidth]{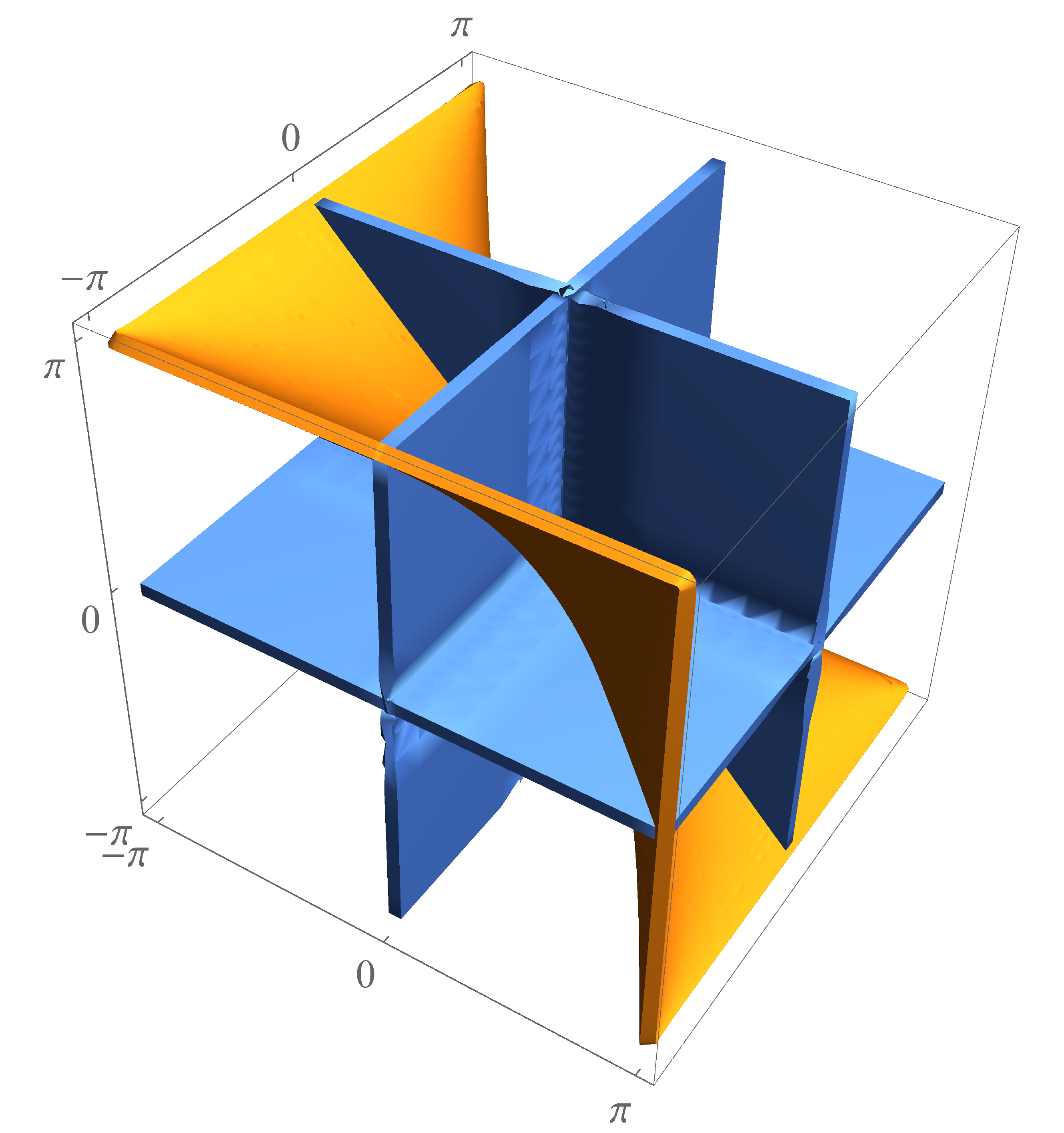}~~~~~\includegraphics[width=0.35\paperwidth]{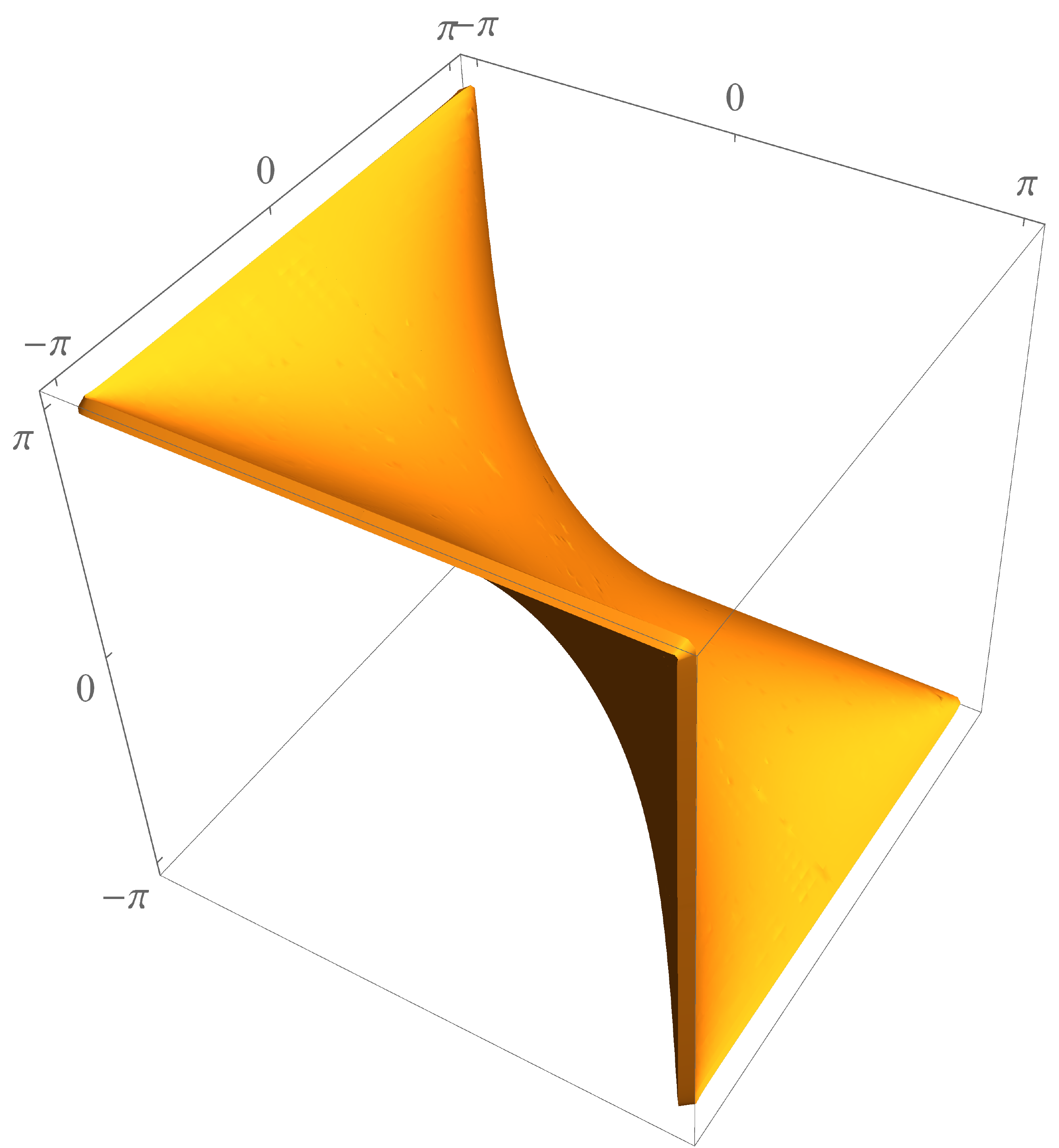}

\caption[ Secular manifold of 3-flower]{\label{fig: Zo reg} On the left, $\Sigma$ for $\Gamma$ a 3-flower
(one vertex and three loops), with loop factors $Z_{e_{1}},Z_{e_{2}}$
and $Z_{e_{3}}$ in blue and main factor $Z_{0}$ in orange. On the
right, only $Z_{0}^{reg}$. In both plots, $\protect\kv$ range in
$\left(-\pi,\pi\right)^{3}$ so that the loop factors are visible.}
\end{figure}

\begin{thm}
\label{thm: Connected components no loops}If $\Gamma$ is a graph
with a bridge $e$ and no loops, then $\Sigma^{reg}$ has two connected
components and the bridge extension $\tau_{e}$ (Definition \ref{def: Tau_=00007Be=00007D})
is an isometry between the two components. 
\end{thm}

\begin{proof}
Recall that if a graph has a bridge $e$ we defined its bridge decomposition\\
$\Gamma\setminus\left\{ e\right\} =\Gamma_{1}\sqcup\Gamma_{2}$ with
sets of edges $\E_{1}$ and $\E_{2}$ correspondingly and torus coordinates
$\kv=\left(\kv_{1},\kappa_{e},\kv_{2}\right)$ such that $\kv_{j}\in\T^{\E_{j}}$.
We have shown in Proposition \ref{prop: simple bridge det decomposition}
that the secular function can be factorized into $F\left(\kv\right)=g_{1}\left(\kv_{1}\right)g_{2}\left(\kv_{2}\right)\left(1-e^{i\left(2\kappa_{e}+\Theta_{1}\left(\kv_{1}\right)+\Theta_{2}\left(\kv_{2}\right)\right)}\right)$.
We have defined $Z_{g}:=\set{\left(\kv_{1},\kappa_{e},\kv_{2}\right)\in\T^{\E}}{g_{1}\left(\kv_{1}\right)g_{2}\left(\kv_{2}\right)=0}$,
and we will now prove that $Z_{g}$ is of positive co-dimension in
$\Sigma$. In fact, since $\Sigma^{sing}$ is of positive co-dimension
in $\Sigma$, then we only need to prove that $Z_{g}\cap\Sigma^{reg}$
is of positive co-dimension in $\Sigma^{reg}$. As each $\left|g_{i}\right|^{2}$
is a real trigonometric polynomial, by Proposition \ref{prop: simple bridge det decomposition},
and $Z_{g}\cap\Sigma^{reg}$ is the zero set of the real trigonometric
polynomial $\left|g_{1}\right|^{2}\left|g_{2}\right|^{2}$, then by
Lemma \ref{lem: real analytic lemma} either $Z_{g}\cap\Sigma^{reg}$
contains a connected component of $\Sigma^{reg}$ or it is of positive
co-dimension in $\Sigma^{reg}$. Assume by contradiction that $Z_{g}\cap\Sigma^{reg}$
contains a connected component of $\Sigma^{reg}$, say $M$. According
to \cite{BerLiu_jmaa17} there is a residual set $G\subset\left(\R_{+}\right)^{\E}$,
such that for every $\lv\in G$, every eigenvalue of $\Gamma_{\lv}$
is simple with eigenfunction that does not vanish on vertices (here
we use the fact that $\Gamma$ has no loops). Since $G$ is and the
set of rationally independent edge lengths is residual, then so does
their intersection. Let $\lv\in G$ be rationally independent, and
let $\left\{ f_{n}\right\} _{n=0}^{\infty}$ and $\left\{ k_{n}\right\} _{n=0}^{\infty}$
be the eigenfunctions and (square-root) eigenvalues of $\Gamma_{\lv}$.
So the sequence $\left\{ k_{n}\lv\right\} $ is dense in $\Sigma$
by Theorem \ref{thm: BG equidistribution}, and in particular there exists $n$ such that $\left\{ k_{n}\lv\right\} \in M\subset Z_{g}\cap\Sigma^{reg}$.
The characterization of $Z_{g}\cap\Sigma^{reg}$ in Lemma \ref{lem: Sigma minus Zg}
implies that the canonical eigenfunction $f_{\left\{ k_{n}\lv\right\} }$
vanish on the bridge $e$ and in particular on its vertices, and by
Lemma \ref{lem: vertex values canonical ef and other ef }, so does
$f_{n}$. But this is a contradiction to \cite{BerLiu_jmaa17} as
$\lv\in G$. We conclude that $\dim\left(Z_{g}\right)\le E-2$. 

Let us now define partition the sub-torus $\T^{\E_{1}}\times\T^{\E_{2}}$
into two parts:
\begin{align}
X & :=\set{\left(\kv_{1},\kv_{2}\right)\in\T^{\E_{1}}\times\T^{\E_{2}}}{g_{1}\left(\kv_{1}\right)g_{2}\left(\kv_{2}\right)\ne0},\,\,and\label{eq: definition of X}\\
X^{c} & :=\set{\left(\kv_{1},\kv_{2}\right)\in\T^{\E_{1}}\times\T^{\E_{2}}}{g_{1}\left(\kv_{1}\right)g_{2}\left(\kv_{2}\right)=0}.
\end{align}
So that $Z_{g}=\set{\left(\kv_{1},\kappa_{e},\kv_{2}\right)\in\T^{\E}}{\left(\kv_{1},\kv_{2}\right)\in X^{c}}$
and therefore $\dim\left(X^{c}\right)=\dim\left(Z_{g}\right)-1\le E-3$.
It follows that $X^{c}$ has co-dimension of at least two in $\T^{\E_{1}}\times\T^{\E_{2}}$
and therefore cannot bisect this space, so $X$ is connected. According
to Proposition \ref{prop: simple bridge det decomposition}, $e^{-i\left(\Theta_{1}\left(\kv_{1}\right)+\Theta_{2}\left(\kv_{2}\right)\right)}$
is smooth on $X$ and one can deduce from the decomposition of $F$
that, 
\[
\Sigma\setminus Z_{g}=\set{\left(\kv_{1},\kappa_{e},\kv_{2}\right)\in\T^{\E}}{\left(\kv_{1},\kv_{2}\right)\in X\,\,\text{and}\,\,e^{i2\kappa_{e}}=e^{-i\left(\Theta_{1}\left(\kv_{1}\right)+\Theta_{2}\left(\kv_{2}\right)\right)}}.
\]
Define the continuous function $y:X\rightarrow\R/2\pi\Z$ by $e^{iy}=e^{-i\left(\Theta_{1}\left(\kv_{1}\right)+\Theta_{2}\left(\kv_{2}\right)\right)}$.
Since $X$ is connected, then so does the graph of the continuous function
$y$ (embedded into $\T^{\E}$), 
\[
\mathrm{Graph}\left(y\right):=\set{\left(\kv_{1},y,\kv_{2}\right)\in\T^{\E}}{\left(\kv_{1},\kv_{2}\right)\in X\,\,\text{and}\,\,\,e^{iy}=e^{-i\left(\Theta_{1}\left(\kv_{1}\right)+\Theta_{2}\left(\kv_{2}\right)\right)}}.
\]
Notice that $y\left(\kappa_{e}\right)$ given by the equation $e^{i2\kappa_{e}}=e^{iy}$
is two to one, therefore $\Sigma\setminus Z_{g}=\set{\left(\kv_{1},\kappa_{e},\kv_{2}\right)}{e^{i2\kappa_{e}}=e^{-i\left(\Theta_{1}\left(\kv_{1}\right)+\Theta_{2}\left(\kv_{2}\right)\right)}}$
has at most two connected components. As $Z_{g}\cap\Sigma^{reg}$
is closed and of positive co-dimension, then $\Sigma\setminus Z_{g}$
is an open dense subset set of $\Sigma^{reg}$, and therefore $\Sigma^{reg}$
has at most two connected components. 

Recall that the bridge extension $\tau_{e}$ is an isometry of $\Sigma^{reg}$
as seen in the proof of Lemma \ref{lem: Re and tau e are measure preseving}.
We will now show that for any point $\kv\in\Sigma^{reg}$, the points
$\kv$ and $\tau_{e}\left(\kv\right)$ are not in the same connected
component. This will prove that $\Sigma^{reg}$ has at least and therefore
exactly two connected components, and that $\tau_{e}$ is an isometry
between the two. 

It is immediate from the decomposition of $F$ in Proposition \ref{prop: simple bridge det decomposition}
that $F\circ\tau_{e}=-F$. Since $\tau_{e}$ is a translation, then
it commutes with derivatives and therefore $-\nabla F=\nabla\left(F\circ\tau_{e}\right)=\nabla F\circ\tau_{e}$,
namely $\nabla F\left(\tau_{e}\left(\kv\right)\right)=-\nabla F\left(\kv\right)$
for all $\kv\in\T^{\E}$. Let $\kv\in\Sigma^{reg}$, according to
Lemma \ref{lem: F p and mk}, the non-vanishing components of $\nabla F\left(\kv\right)$
have the same sign as the sign of $p\left(\kv\right)$, and same for
$\tau_{e}\left(\kv\right)\in\Sigma^{reg}$. Therefore, $p\left(\tau_{e}\left(\kv\right)\right)=-p\left(\kv\right)$,
and since $p$ is real, non-vanishing and continuous on $\Sigma^{reg}$
then $\kv$ and $\tau_{e}\left(\kv\right)$ belongs to different connected
components. 
\end{proof}
\begin{thm}
\label{thm: connected components with loops}Let $\Gamma$ be a graph
with $\left|\E_{loops}\right|>0$ loops, then $\Sigma^{reg}$ has
at least $\left|\E_{loops}\right|+1$ connected components. However,
$Z_{0}^{reg}$ is connected.
\end{thm}

The lower bound on the number of connected components is straightforward from Lemma \ref{lem: loop decomposition}. The proof of $Z_{0}^{reg}$
being connected is very similar to that of Theorem \ref{thm: Connected components no loops}:
\begin{proof}
As discussed in Lemma \ref{lem: loop decomposition}, $\Sigma^{reg}=\Sigma_{\L}^{c}\sqcup_{e\in\E_{loops}}\tilde{Z}_{e}$
where each of these sets is a union of connected components of $\Sigma^{reg}$
and has positive measure (so it is not empty). It follows that there
must be at least $\left|\E_{loops}\right|+1$ connected components
of $\Sigma^{reg}$. 

Let $e$ be a distinguished loop and consider the notation $\kv=\left(\kappa_{e},\kv_{\overline{e}}\right)$
where $\kv_{\overline{e}}$ denotes the rest of the entries of $\kv$
which are not $\kappa_{e}$. Consider the decomposition of $U_{\kv}$
according to (\ref{eq: U loop decomposition}), so that $1-U_{\kv}=\left(1-e^{i\hat{\kappa}_{loops}}\right)\oplus\left(1-U_{0}\left(\kv\right)\right)$
where $U_{0}\left(\kv\right)=e^{i\hat{\kappa}_{0}}S_{0}$. As both
$e^{i\hat{\kappa}_{loops}}$ and $e^{i\hat{\kappa}_{0}}$ are unitary
diagonal with $e^{i\kappa_{e}}$ appearing ones in each of them, then
both $\det\left(1-U_{0}\left(\kv\right)\right)$ and $\det\left(1-e^{i\hat{\kappa}_{loops}}\right)$
are polynomials of degree one in $e^{i\kappa_{e}}$. We may define
$h_{0}\left(\kv_{\overline{e}}\right)=\Pi_{e'\in\E_{loops};\,e'\ne e}\left(1-e^{i\kappa_{e'}}\right)$
such that $\det\left(1-e^{i\hat{\kappa}_{loops}}\right)=\left(1-e^{i\kappa_{e}}\right)h_{0}\left(\kv_{\overline{e}}\right)$.
We may also define $h_{1}$ and $h_{2}$, by 
\begin{equation}
\det\left(1-U_{0}\left(\kv\right)\right)=h_{2}\left(\kv_{\overline{e}}\right)e^{i\kappa_{e}}+h_{1}\left(\kv_{\overline{e}}\right),\label{eq: h1 and h2}
\end{equation}
such that both $h_{1}$ and $h_{2}$ are polynomials in $\left\{ e^{i\kappa_{e'}}\right\} _{e'\ne e}$
with real coefficients (since $S_{0}$ is real). Therefore, 
\begin{align}
F\left(\kappa_{e},\kv_{\overline{e}}\right)= & \det\left(U_{\kv}\right)^{\frac{1}{2}}\det\left(1-U_{\kv}\right)\label{eq: h0}\\
= & \left(i^{\left(\beta-1\right)}e^{-i\sum_{e'\ne e}\kappa_{e'}}\right)e^{-i\kappa_{e}}\det\left(1-U_{\kv}\right)\\
= & \left(i^{\left(\beta-1\right)}e^{-i\sum_{e'\ne e}\kappa_{e'}}\right)e^{-i\kappa_{e}}\det\left(1-e^{i\hat{\kappa}_{loops}}\right)\det\left(1-U_{0}\left(\kv\right)\right)\\
= & \left(i^{\left(\beta-1\right)}e^{-i\sum_{e'\ne e}\kappa_{e'}}\right)e^{-i\kappa_{e}}\left(1-e^{i\kappa_{e}}\right)h_{0}\left(\kv_{\overline{e}}\right)\left(h_{2}\left(\kv_{\overline{e}}\right)e^{i\kappa_{e}}+h_{1}\left(\kv_{\overline{e}}\right)\right)\\
= & \left(i^{\left(\beta-1\right)}e^{-i\sum_{e'\ne e}\kappa_{e'}}\right)\left(1-e^{i\kappa_{e}}\right)h_{0}\left(\kv_{\overline{e}}\right)\left(h_{2}\left(\kv_{\overline{e}}\right)+e^{-i\kappa_{e}}h_{1}\left(\kv_{\overline{e}}\right)\right).
\end{align}
 Denoting the uni-modular prefactor $c\left(\kv_{\overline{e}}\right)=\left(i^{\left(\beta-1\right)}e^{-i\sum_{e'\ne e}\kappa_{e'}}\right)$
and suppressing the $\kv_{\overline{e}}$ dependence in $c,h_{0},h_{1}$
and $h_{3}$ for simplicity, this gives:
\begin{align}
F\left(\kappa_{e},\kv_{\overline{e}}\right) & =ch_{0}\left(1-e^{i\kappa_{e}}\right)\left(h_{2}+e^{-i\kappa_{e}}h_{1}\right),\,and\label{eq: F in loop comp}\\
\frac{\partial}{\partial\kappa_{e}}F\left(\kappa_{e},\kv_{\overline{e}}\right) & =ch_{0}\left(h_{2}\left(1-ie^{i\kappa_{e}}\right)-h_{1}\left(1+ie^{-i\kappa_{e}}\right)\right).\label{eq: dfdke in loop components}
\end{align}
One may check that the function $\left(1-e^{ix}\right)\left(A+Be^{-ix}\right)$
is real for all $x\in\R$ if and only if $A=-\overline{B}$. As $F$
is real, $\left|c\right|=1$ and $h_{0}\left(\kv_{\overline{e}}\right)\ne0$
for any $\kv_{\overline{e}}\in\left(0,2\pi\right)^{\E\setminus e}$
then $\left|h_{2}\right|=\left|h_{1}\right|$ for any $\kv_{\overline{e}}\in\left(0,2\pi\right)^{\E\setminus e}$.
This equality extends to $\T^{\E\setminus e}$ by continuity. Let
\begin{align}
X^{c} & :=\set{\kv_{\overline{e}}\in\T^{\E\setminus e}}{h_{2}\left(\kv_{\overline{e}}\right)=h_{1}\left(\kv_{\overline{e}}\right)=0}\\
X & :=\T^{\E\setminus e}\setminus X=\set{\kv_{\overline{e}}\in\T^{\E\setminus e}}{\left|h_{2}\left(\kv_{\overline{e}}\right)\right|=\left|h_{1}\left(\kv_{\overline{e}}\right)\right|\ne0},\,\,\text{and}\\
Z_{h} & :=\set{\left(\kappa_{e},\kv_{\overline{e}}\right)\in\T^{\E}}{\kv_{\overline{e}}\in X^{c}}.
\end{align}
Observe that $\left|\frac{\partial}{\partial\kappa_{e}}\det\left(1-U_{0}\left(\kv\right)\right)\right|=\left|h_{2}\left(\kv_{\overline{e}}\right)\right|$
and therefore, by the characterization in (\ref{eq: Zoreg as det}),
\begin{equation}
Z_{0}\setminus Z_{h}\subset Z_{0}^{reg}.
\end{equation}
If $\kv\in Z_{h}\cap\Sigma^{reg}$, namely $h_{1}=h_{2}=0$, then $\frac{\partial}{\partial\kappa_{e}}F=0$
according to (\ref{eq: dfdke in loop components}). Then by Lemma
\ref{lem: F p and mk}, the amplitudes vector of $f_{\kv}$ satisfies
$\left|a_{e}\right|^{2}+\left|a_{\hat{e}}\right|^{2}=0$, and hence
$f_{\kv}|_{e}\equiv0$. Therefore, if $\kv\in\Sigma^{reg}$ and $f_{\kv}$
does not vanish on vertices, then $\kv\notin Z_{h}$ and is also not
in any loop factor $Z_{e'}$. We may conclude, using the decomposition
in Lemma \ref{lem: loop decomposition}, that having $f_{\kv}$ does
not vanish on vertices implies $\kv\in\Sigma^{reg}\cap Z_{0}\setminus Z_{h}$. 

According to \cite{BerLiu_jmaa17}, if a graph has loops, then there
is a residual sets of edge lengths $G\subset\left(\R_{+}\right)^{\E}$
such that for every $\lv\in G$, every eigenvalue of $\Gamma_{\lv}$
is simple and every eigenfunction is either supported on a single
loop or does not vanish on vertices. As in the proof of Theorem \ref{thm: Connected components no loops},
we may choose $\lv\in G$ which is rationally independent, so that
the sequence $\left\{ k_{n}\lv\right\} _{n=0}^{\infty}$ is dense
in $\Sigma^{reg}$. If $f_{n}$ is supported on a loop $e'$, then
$\left\{ k_{n}\lv\right\} \in Z_{e'}\cap\Sigma^{reg}=Z_{e'}^{reg}$,
otherwise, $f_{n}$ does not vanish on any vertex, and so does $f_{\left\{ k_{n}\lv\right\} }$
(by Lemma \ref{lem: vertex values canonical ef and other ef }), so
$\left\{ k_{n}\lv\right\} \in\Sigma^{reg}\cap Z_{0}\setminus Z_{h}$.
Since $Z_{h}$ is the zero set of a real trigonometric polynomial
$\left|h_{1}\right|^{2}+\left|h_{2}\right|^{2}=0$ and $Z_{0}^{reg}$
is a real analytic manifold of dimension $E-1$, then the same argument
as in the proof of Theorem \ref{thm: Connected components no loops}
would show that $\dim\left(Z_{h}\cap\Sigma^{reg}\right)\le E-2$ and
therefore $\dim\left(Z_{h}\right)\le E-2$. It follows that $\dim\left(X^{c}\right)=\dim\left(Z_{h}\right)-1\le E-3$
which co-dimension of at least two in $\T^{\E\setminus e}$ and therefore
$X$ is connected. Define the function $y:X\rightarrow\R/2\pi\Z$
given by $e^{iy}=-\frac{h_{1}\left(\kv_{\overline{e}}\right)}{h_{2}\left(\kv_{\overline{e}}\right)}$
and notice that it is is continuous, so it has a connected graph (embedded
into $\T^{\E}$):
\begin{align*}
Graph\left(y\right) & :=\set{\left(y,\kv_{\overline{e}}\right)\in\T^{\E}}{\kv_{\overline{e}}\in X\,\,\,\text{and}\,\,\,e^{iy}=-\frac{h_{1}\left(\kv_{\overline{e}}\right)}{h_{2}\left(\kv_{\overline{e}}\right)}}\\
= & \set{\left(y,\kv_{\overline{e}}\right)\in\T^{\E}}{\kv_{\overline{e}}\in X\,\,\,\text{and}\,\,\,h_{2}\left(\kv_{\overline{e}}\right)e^{iy}+h_{1}\left(\kv_{\overline{e}}\right)=0}\\
= & \set{\left(\kappa_{e},\kv_{\overline{e}}\right)\in\T^{\E}}{\kv_{\overline{e}}\in X\,\,\,\text{and}\,\,\,\det\left(1-U_{0}\left(\kv\right)\right)=0},\\
= & Z_{0}\setminus Z_{h}.
\end{align*}
Therefore $Z_{0}\setminus Z_{h}$ is connected. But it is dense in
$Z_{0}^{reg}$ as $Z_{h}$ is of positive co-dimension so therefor
$Z_{0}^{reg}$ is connected.
\end{proof}

\newpage{}
\section{\label{sec: genericity}generic eigenfunctions}

Recall that we define (see Definition \ref{def:generic_eigenfunction})
a generic eigenfunction as an eigenfunction of a simple eigenvalue
that satisfies both properties $I$ and $II$. Where an eigenfunction
is said to satisfy property $I$ if it does not vanish on vertices,
and property $II$ if non of its outgoing derivatives vanish on any
interior vertex (see Definition \ref{def: Properties I,II,III-1}). 

The main goal of this chapter, as discussed in the introduction, is
Theorem \ref{thm: density-of-generic-and-loop-eigenfunctions} which
proves, for two related notions of genericity, that the eigenfunctions
we call generic are indeed generic (among the eigenfunctions which
are not supported on loops). 

The first genericity result for quantum graphs is duo to Friedlander,
who proved in \cite{Fri_ijm05} that given a graph $\Gamma$ there
is a residual set of edge length $G\subset\R_{+}^{\E}$ such that
for any $\lv\in G$, the spectrum of $\Gamma_{\lv}$ is simple (i.e
every eigenvalue is simple). This result was generalized by Berkolaiko
and Liu as follows:
\begin{thm}
\label{thm: Genericity of Grisha}\cite[Theorem 3.6]{BerLiu_jmaa17}
Given a graph $\Gamma$ there is a residual set of edge lengths $G\subset\R_{+}^{\E}$
such that for any $\lv\in G$, the spectrum of $\Gamma_{\lv}$ is
simple. Moreover, every eigenfunction is either supported on a single
loop (if such exists) or satisfies Property $I$.
\end{thm}

In \cite[Proposition A.1]{AloBanBer_cmp18} it was shown that the
non-explicit residual set $G$ can be replaced by the explicit residual
set of rationally independent edge lengths, if we restrict the discussion
to `almost every eigenvalue and eigenfunction' (in the sense of a
sub-sequence of density one). In order to state it, let us define the
following index sets:
\begin{defn}
\label{def: index sets} Given a standard graph $\Gamma_{\lv}$ with
(square-root) eigenvalues and eigenfunctions $\left\{ k_{n}\right\} _{n=0}^{\infty}$
and $\left\{ f_{n}\right\} _{n=0}^{\infty}$, we define the following
index sets:
\begin{equation}
\mathcal{S}:=\left\{ n\in\N~:~k_{n}\:\:\text{is simple}\right\} ,\label{eq: singular}
\end{equation}
\begin{align}
\mathcal{P}_{I} & :=\set{n\in\mathcal{S}}{f_{n}\,\,\text{satisfies property I}},\label{eq: Property I,II,III index sets}\\
\mathcal{P}_{II} & :=\set{n\in\mathcal{S}}{f_{n}\,\,\text{satisfies property II}},\,\,\text{and}\\
\L & :=\set{n\in\mathcal{S}}{f_{n}\,\,\text{is supported on a loop}}.\label{eq: loop eigenfunctions}
\end{align}
Where $\mathcal{L}=\emptyset$ if $\Gamma_{\lv}$ has no loops. The
index set of generic eigenfunction $\G$ (see Definition \ref{def:generic_eigenfunction})
can be written as 
\begin{equation}
\G=\mathcal{P}_{I}\cap\mathcal{P}_{II}.\label{eq: def G as PI cap PII}
\end{equation}
\end{defn}

\begin{rem}
As discussed in Remark \ref{rem: differ on degernate e.s}, these
sets are $\lv$ dependent, but are independent of the choice of eigenfunctions
$\left\{ f_{n}\right\} _{n=0}^{\infty}$. This is because different
choices of $L^{2}$ basis of eigenfunctions may differ only on eigenspaces
of non-simple eigenvalues, and so subsets of $\mathcal{S}$ are independent
of this choice. It is not hard to deduce that $\mathcal{L}$ and $\mathcal{P}_{I}\cup\mathcal{P}_{II}$
are disjoint and in particularly, $\L$ and $\G$ are disjoint.
\end{rem}

A Venn diagram of these index sets is presented Figure \ref{fig: Venn diagram integers}
for visualization.
\begin{figure}
\includegraphics[width=0.4\paperwidth]{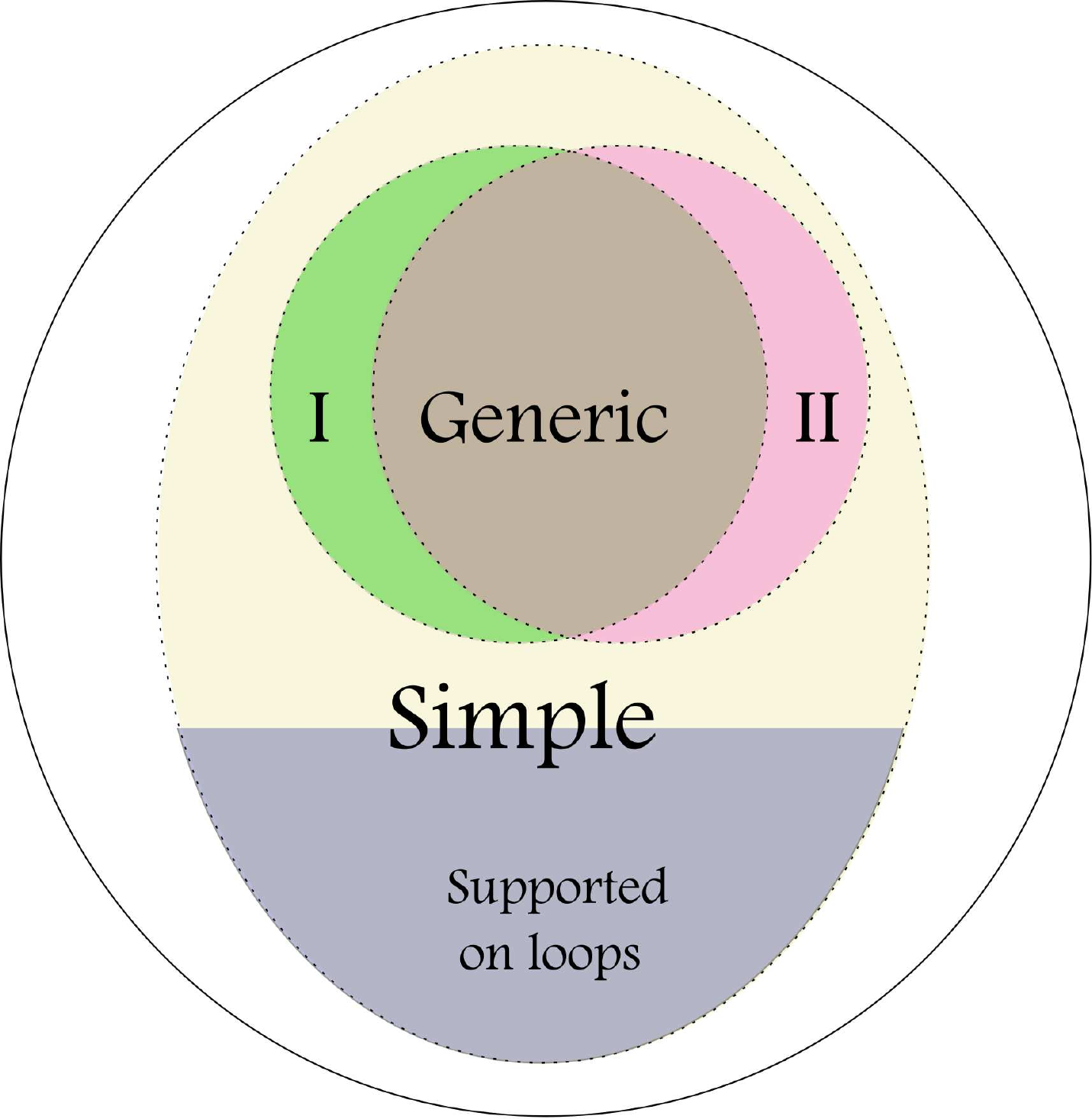}\caption[Diagram of index sets]{\label{fig: Venn diagram integers}A Venn diagram of the index sets.
Where $I$ and $II$ stands for $\mathcal{P}_{I}$ and $\mathcal{P}_{II}$.
`Simple' and `Generic' stands for $\mathcal{S}$ and $\protect\G$,
and `Supported on loops' stands for $\protect\L$. The outer white
region stands for the indices of the non-simple eigenvalues, namely
$\protect\N\setminus\mathcal{S}$. }
\end{figure}

Theorem \ref{thm: Genericity of Grisha} can be written, using this
terminology, as the existence of a residual set $G$ such that for
any $\lv\in G$, $\N=\mathcal{S}=\mathcal{P}_{I}\sqcup\mathcal{L}$.
Using the notations as above, the modification of \cite[Proposition A.1]{AloBanBer_cmp18}
can be written as follows. 
\begin{thm}
\cite[Proposition A.1]{AloBanBer_cmp18}\label{thm: weak genericity}
If $\lv$ is rationally independent with total length $L$, then $\mathcal{S}$,
$\mathcal{L}$ and $\mathcal{P}_{I}$ have densities given by 
\begin{align*}
d\left(\mathcal{S}\right) & =1,\\
d\left(\L\right) & =\frac{1}{2}\frac{\sum_{e\in\E_{loops}}l_{e}}{L},\,\,\text{and}\\
d\left(\mathcal{P}_{I}\right) & =1-d\left(\mathcal{L}\right)=1-\frac{1}{2}\frac{\sum_{e\in\E_{loops}}l_{e}}{\sum_{e\in\E}l_{e}}\ge\frac{1}{2}.
\end{align*}
\end{thm}

We will now prove the following theorem that generalizes both \cite[Proposition A.1]{AloBanBer_cmp18}
and \cite[Theorem 3.6]{BerLiu_jmaa17} from $\mathcal{P}_{I}$ to
$\G$: 
\begin{thm}
\cite{Alon}\label{thm: density-of-generic-and-loop-eigenfunctions}
Given a graph $\Gamma$,
\begin{enumerate}
\item \label{enu: gen always}There is a residual set of edge lengths $G\subset\R_{+}^{\E}$
such that for any $\lv\in G$, 
\begin{equation}
\N=\mathcal{S}=\mathcal{G}\sqcup\mathcal{L}.\label{eq: residual N=00003DG scup L}
\end{equation}
That is, the spectrum of $\Gamma_{\lv}$ is simple, and every eigenfunction
is either supported on a loop or generic. 
\item \label{enu: gen a.s}If $\Gamma_{\lv}$ is a standard graph with $\lv$
rationally independent, then $\G$ has density, and it is given by
\begin{equation}
d\left(\G\right)=1-d\left(\mathcal{L}\right)=1-\frac{1}{2}\frac{\sum_{e\in\E_{loops}}l_{e}}{\sum_{e\in\E}l_{e}}.\label{eq: d(g)}
\end{equation}
Namely almost every eigenfunction is either supported on a loop or
generic.
\end{enumerate}
\end{thm}

\begin{rem}
\label{rem: Morse functions}In the context of Neumann domains on
manifolds, one needs to restrict the discussion to eigenfunctions
which are Morse functions as discussed in \cite{BanFaj_ahp16}. These
are eigenfunctions with full rank Hessian at every critical point.
On a standard graph, an eigenfunction $f$ is not Morse if there is
an edge $e$ with an interior point $x_{0}\in e$ such that $f'\left(x_{0}\right)=f''\left(x_{0}\right)=0$.
If $f$ is not the constant eigenfunction, then by Definition \ref{def: edge restriction notations},
$f|_{e}\equiv0$. Therefore, on a standard graph, Morse eigenfunctions
are non-constant eigenfunctions such that $f|_{e}\not\equiv0$ for
any edge $e$. We conclude that every generic eigenfunction is Morse
and every Morse eigenfunction is not supported on a loop. Therefore,
by Theorem \ref{thm: density-of-generic-and-loop-eigenfunctions},
almost every eigenfunction is either supported on a loop or both Morse
and generic.
\end{rem}

\subsection{Outline of the proof.}

The proof of Theorem \ref{thm: density-of-generic-and-loop-eigenfunctions}
has the following sequence of deductions:
\begin{enumerate}
\item We show that Theorem \ref{thm: density-of-generic-and-loop-eigenfunctions}
follows from: 
\begin{equation}
\forall\lv\in\left(\R_{+}\right)^{\E}\,\,\mu_{\lv}\left(\Sigma_{\G}\sqcup\Sigma_{\mathcal{L}}\right)=1.\label{eq: measure statement}
\end{equation}
\item We show that if (\ref{eq: measure statement}) fails for some graph
$\Gamma$, then there is an open set $O\subset\Sigma^{reg}$ on which
\begin{equation}
\exists v\in\V_{in},e\in\E_{v}\,\,\,s.t\,\,\forall\kv\in O,\,\,\partial_{e}f_{\kv}\left(v\right)=0\,\,\text{and }\,\,f_{\kv}\left(v\right)\ne0.\label{eq: O result}
\end{equation}
\item We show that if there exists such an open set $O$ and a pair $v,e$
as above, then 
\begin{enumerate}
\item In the case that $\Gamma$ has loops and\textbackslash or bridges,
every canonical eigenfunction would satisfy $\partial_{e}f_{\kv}\left(v\right)=0$.
\item In the case that $\Gamma$ has no loops and no bridges, we can construct
a different graph $\tilde{\Gamma}$ with interior vertex $u_{1}$
and boundary vertex $u_{2}$ such that every canonical eigenfunction
of $\tilde{\Gamma}$ would satisfy $\left|f_{\kv}\left(u_{1}\right)\right|=\left|f_{\kv}\left(u_{2}\right)\right|$.
\end{enumerate}
\item We contradict (a) by proving that for every graph $\Gamma$ and any
choice of $v\in\V_{in}$ and $e\in\E_{v}$ there exists $\kv\in\Sigma^{reg}$
for which $\partial_{e}f_{\kv}\left(v\right)\ne0$. We contradict
(b) by proving that if $\Gamma$ is a graph with boundary, then for
any boundary vertex $u_{1}$ and interior vertex $u_{2}$ there exists
$\kv\in\Sigma^{reg}$ for which $\left|f_{\kv}\left(u_{1}\right)\right|\ne\left|f_{\kv}\left(u_{1}\right)\right|$.
This proves that no such $O$ exists, which proves that (\ref{eq: measure statement})
is not false, which proves Theorem \ref{thm: density-of-generic-and-loop-eigenfunctions}.
\end{enumerate}
In order to prove Theorem \ref{thm: density-of-generic-and-loop-eigenfunctions},
we first need to relate all properties discussed above to subsets
of the secular manifolds, and discuss the properties of these subsets:

\subsection{Special subsets of the secular manifold }
\begin{defn}
\label{def: Secular subsets}Let $\Gamma$ be a graph and consider
the regular part of the secular manifold $\Sigma^{reg}$. We define
the following subsets of $\Sigma^{reg}$:
\begin{equation}
\Sigma_{I}:=\set{\kv\in\Sigma^{reg}}{f_{\kv}\,\,\text{satisfy property I}},
\end{equation}
and similarly $\Sigma_{II}$ with property II. We define the generic
part by:
\begin{equation}
\Sigma_{\G}:=\set{\kv\in\Sigma^{reg}}{f_{\kv}\,\,\text{is generic}},\label{eq: Sigma g}
\end{equation}
so that $\Sigma_{\G}=\Sigma_{I}\cap\Sigma_{II}$. 
\end{defn}

\begin{rem}
\label{rem: Sigma L is a union of cc} Recall that by Lemma \ref{lem: loop decomposition},
if a graph has loops, then
\[
\Sigma^{reg}=\Sigma_{\L}\sqcup\Sigma_{\L}^{c},
\]
where 
\begin{align*}
\Sigma_{\L} & :=\set{\kv\in\Sigma^{reg}}{f_{\kv}\,\,\text{is supported on a loop}},\,\,\text{and}\\
\Sigma_{\L}^{c} & =Z_{0}^{reg}\cap\Sigma^{reg}.
\end{align*}
See Definition \ref{def: Zo} for $Z_{0}^{reg}$. By their definitions,
both $\Sigma_{I}$ and $\Sigma_{II}$ are subsets of $\Sigma_{\L}^{c}$
and so does their intersection $\Sigma_{\G}$.

According to Lemma \ref{lem: loop decomposition} both $\Sigma_{\L}$
and $\Sigma_{\L}^{c}$ are unions of connected components of $\Sigma^{reg}$.
As such, they are both closed and open in $\Sigma^{reg}$ (since its
a manifold) and are Jordan according to Remark \ref{rem:-Sigma reg is-Jordan}.
\end{rem}

\begin{rem}
If a graph has no loops we say that $\Sigma_{\L}=\emptyset$ and $\Sigma_{\L}^{c}=\Sigma^{reg}$.
In fact, the graph has no loops if and only if $\Sigma_{\L}=\emptyset$,
as Lemma \ref{lem: loop decomposition} says that for a graph with
at least one loop $\Sigma_{\L}$ has positive measure. 
\end{rem}

The index sets Venn diagram can be used for these sets as well, as
presented in Figure \ref{fig: Venn diagram integers-1}. 
\begin{figure}
\includegraphics[width=0.4\paperwidth]{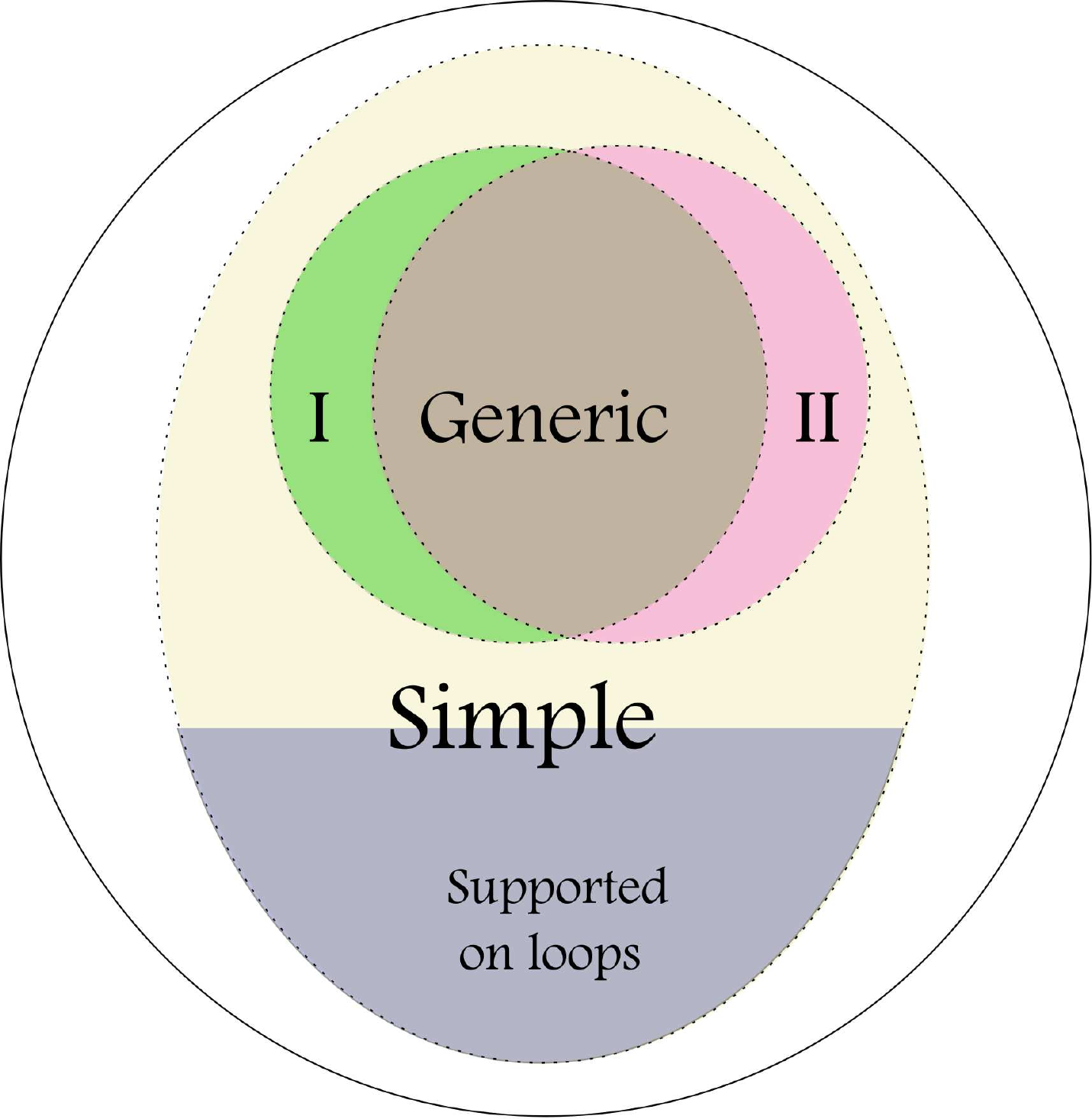}\caption[Diagram of secular manifold subsets]{\label{fig: Venn diagram integers-1}A Venn diagram of the partition
of $\Sigma$ to subsets. Where $I$ and $II$ stands for $\Sigma_{I}$
and $\Sigma_{II}$. `Simple' stands for $\Sigma^{reg}$ and `Generic'
stands for $\Sigma_{\protect\G}$. `Supported on loops' stands for
$\Sigma_{\protect\L}$. The outer white region stands for $\Sigma^{sing}$.}
\end{figure}

The relation to the index sets in Definition \ref{def: index sets}
is given by:
\begin{lem}
\label{lem: index sets and secular sets}Given a standard graph $\Gamma_{\lv}$
with (square-root) eigenvalues and eigenfunctions $\left\{ k_{n}\right\} _{n=0}^{\infty}$
and $\left\{ f_{n}\right\} _{n=0}^{\infty}$, then 
\begin{align}
\mathcal{S} & =\set{n\in\N}{\left\{ k_{n}\lv\right\} \in\Sigma^{reg}}.\\
\mathcal{P}_{I} & :=\set{n\in\N}{\left\{ k_{n}\lv\right\} \in\Sigma_{I}}.\\
\mathcal{P}_{II} & :=\set{n\in\N}{\left\{ k_{n}\lv\right\} \in\Sigma_{II}}.\\
\mathcal{G} & :=\set{n\in\N}{\left\{ k_{n}\lv\right\} \in\Sigma_{\mathcal{G}}}.\\
\L & :=\set{n\in\N}{\left\{ k_{n}\lv\right\} \in\Sigma_{\mathcal{L}}}.
\end{align}
\end{lem}

\begin{proof}
The characterizations of $\mathcal{S}$ in terms of $\Sigma^{reg}$
and $\L$ in terms of $\Sigma_{\L}$ are straightforward from their
definitions. According to Lemma \ref{lem: vertex values canonical ef and other ef },
an eigenfunction $f_{n}$ of a simple eigenvalue $k_{n}$ satisfies
property I or II if and $f_{\kv}$ satisfies property I or II, for
$\kv=\left\{ k_{n}\lv\right\} $. The relations between $\mathcal{P}_{I},\,\mathcal{P}_{II}$
and $\G$ to $\Sigma_{I},\,\Sigma_{II}$ and $\Sigma_{\G}$ follows.
\end{proof}
\begin{defn}
\label{def: vertex values subsets}We define the following zero sets
on $\Sigma_{\L}^{c}$: 
\begin{align}
\forall v\in\V\,\,\,\,\,\,Z_{v}:= & \set{\kv\in\Sigma_{\L}^{c}}{f_{\kv}\left(v\right)=0},\\
\forall v\in\V_{in},\,\forall e\in\E_{v}\,\,\,\,\,Z_{v,e}:= & \set{\kv\in\Sigma_{\L}^{c}}{\partial_{e}f_{\kv}\left(v\right)=0},\,\,\,\text{and}\\
\forall v,u\in\V\,\,\,\,\,Z_{v,u}:= & \set{\kv\in\Sigma_{\L}^{c}}{f_{\kv}\left(v\right)^{2}-f_{\kv}\left(u\right)^{2}=0}.
\end{align}
We will refer to these subsets as the \emph{$Z_{*}$}'s.
\end{defn}

\begin{rem}
\label{rem: Sigma g complement}Observe that $\Sigma_{\G}$ is the
complement (in $\Sigma_{\L}^{reg}$) of the union: 
\begin{equation}
\Sigma_{\L}\cup_{v\in\V_{in},e\in\E_{v}}Z_{v,e}\cup_{u\in\V}Z_{u}.
\end{equation}
\end{rem}

Since $\Sigma_{\L}^{c}$ is a real analytic manifold (as a union of
connected components of $\Sigma^{reg}$) and each of the \emph{$Z_{*}$}'s
subsets is the zero set of either $f_{\kv}\left(v\right)^{2},\,\,\partial_{e}f_{\kv}\left(v\right)^{2}$
or $f_{\kv}\left(v\right)^{2}-f_{\kv}\left(u\right)^{2}$ which are
real analytic according to (\ref{lem: vertex values canonical ef and other ef }),
then we can apply Lemma \ref{lem: real analytic lemma} for zero sets
of real analytic functions and conclude that:
\begin{lem}
\label{lem: verte values subsets}Let $A\subset\Sigma_{\L}^{c}$ be
a finite union of $Z_{*}$'s as in Definition \ref{def: vertex values subsets},
and denote its complement by $A^{c}=\Sigma_{\L}^{c}\setminus A$.
Then $A^{c}$ is an open Jordan subset of $\Sigma$. Moreover, if
we assume that $A$ has zero co-dimension in $\Sigma_{\L}^{c}$, then 
\begin{enumerate}
\item If $\Gamma$ has at least one loop, then $A=\Sigma_{\L}^{c}$.
\item If $\Gamma$ has no loops but has at least one bridge (including tail)
then $A=\Sigma^{reg}=\Sigma_{\L}^{c}$.
\item If $\Gamma$ has no loops and no bridges, then $A$ contains every
connected component of $\Sigma^{reg}$ on which it has zero co-dimension. 
\end{enumerate}
\end{lem}

\begin{proof}
Let $A=\cup_{j=1}^{N}Z_{h_{j}}$ where each $h_{j}$ is either $f_{\kv}\left(v\right)^{2},\,\,\partial_{e}f_{\kv}\left(v\right)^{2}$
or $f_{\kv}\left(v\right)^{2}-f_{\kv}\left(u\right)^{2}$ for some
$v,u\in\V$ and $e\in\E_{v}$ and $Z_{h_{j}}$ is the zero set of
$h_{j}$ in $\Sigma_{\L}^{c}$. Let $h\left(\kv\right)=\Pi_{j=1}^{N}h_{j}\left(\kv\right)$
so that its zero set $Z_{h}$ is equal to $A$. According to Lemma
(\ref{lem: vertex values canonical ef and other ef }), each $h_{j}$
is real analytic and therefore $h$ is real analytic. If $\Gamma$
has loops, then according to Lemma \ref{lem: vertex values for loops}
each $h_{j}$ can extended to a real analytic function $\tilde{h}_{j}$
on $Z_{0}^{reg}$. Therefore $\tilde{h}$, the product of theses $\tilde{h}_{j}$'s,
is a real analytic extension of $h$ from $\Sigma_{\L}^{c}$ to $Z_{0}^{reg}$.
Since $Z_{0}^{reg}$ is connected by Theorem \ref{thm: connected components with loops},
then we may use Lemma \ref{lem: real analytic lemma} and the assumption
that $Z_{h}$ has zero co-dimension in $\Sigma_{\L}^{c}$ to conclude
that $\tilde{h}$ vanish on all of $Z_{0}^{reg}$ and so $Z_{h}=\Sigma_{\L}^{c}$. 

Now assume that $\Gamma$ has no loops, namely $\Sigma^{reg}=\Sigma_{\L}^{c}$,
and let $M$ be a connected component of $\Sigma^{reg}$ such that
$Z_{h}\cap M$ has zero co-dimension in $M$. Such $M$ exists by the assumption
that $Z_{h}$ has zero co-dimension in $\Sigma_{\L}^{c}$. Since $M$
is a connected real analytic manifold and $h$ is a real analytic
function on $M$ whose zero set is $Z_{h}\cap M$, then by Lemma \ref{lem: real analytic lemma}
$M\subset Z_{h}$. 

If $\Gamma$ has a bridge $e'$ and $\tau_{e'}$ is the bridge extension
(see Definition \ref{def: Tau_=00007Be=00007D}), then according to
Theorem \ref{thm: Connected components no loops}, $\Sigma^{reg}=M\sqcup\tau_{e'}\left(M\right)$.
By Lemma \ref{lem: Re and tau e are measure preseving}, the functions
$f_{\kv}\left(v\right)^{2},\,\,\partial_{e}f_{\kv}\left(v\right)^{2}$
and $f_{\kv}\left(v\right)^{2}-f_{\kv}\left(u\right)^{2}$ are $\tau_{e'}$
invariant, and therefore so does every $h_{j}$. It follows that $h$
is $\tau_{e'}$ invariant and so $\tau_{e'}\left(Z_{h}\right)=Z_{h}$.
As we showed that $M\subset Z_{h}$, then so does $\tau_{e'}\left(M\right)\subset Z_{h}$
which means that $\Sigma^{reg}=Z_{h}$.

It is left to prove that $A^{c}$ is an open Jordan subset of $\Sigma$
for both cases where $A$ for both cases where $A$ has zero or positive
co-dimension in $\Sigma_{\L}^{c}$. Since $A$ is the zero set of
$h$, then it is closed and therefore $A^{c}=\Sigma_{\L}^{c}\setminus A$
is open and its boundary is given by $\partial A^{c}\subset\partial A\cup\partial\Sigma_{\L}^{c}$.
Since $\Sigma_{\L}^{c}$ is a union of connected components of $\Sigma^{reg}$
then $\partial\Sigma_{\L}^{c}\subset\Sigma^{sing}$ and is therefore
of measure zero for any $\mu_{\lv}$. It is thus left to prove that
$\partial A$ is of measure zero. If $A$ has positive co-dimension
in $\Sigma_{\L}^{c}$, then $\partial A$ also has positive co-dimension
(since $A$ was closed in $\Sigma^{reg}$) and has measure zero for
any $\mu_{\lv}$. Assume that $A$ has zero co-dimension, and let
$M$ be a connected component of $\Sigma_{\L}^{c}$. If $A\cap M$
has positive co-dimension in $M$, then $\partial A\cap M$ has measure
zero by the same argument as above. In the case that $A\cap M$ has
zero co-dimension in $M$, then we have showed that $M\subset A$ in
which case $\partial A\cap M=\emptyset$. Since $\Sigma_{\L}^{c}$
has at most a countable number of connected components\footnote{This is a property of real analytic manifolds in general, but in our
case one can simply consider a rationally independent $\lv$ so that
$\left\{ \left\{ k_{n}\lv\right\} \right\} _{n\in\N}$ is dense in
$\Sigma_{\L}^{c}$ by Theorem \ref{thm: BG equidistribution}, and
so each connected component contains at least one point $\left\{ k_{n}\lv\right\} $. }, each intersecting $\partial A$ in a measure zero set, then $\partial A\cap\Sigma_{\L}^{c}$
is of measure zero. As $\partial A\subset\Sigma_{\L}^{c}\cup\Sigma^{sing}$
then $\partial A$ is of measure zero and we are done.
\end{proof}
A similar result holds for the complements of the $\Sigma^{reg}$ subsets
 in Definition \ref{def: Secular subsets}:
\begin{cor}
\label{cor: key lemma in equivalence of genericity} The subsets $\Sigma_{I},\,\Sigma_{II}$
and $\,\Sigma_{\G}$, are open and Jordan in $\Sigma$.
\end{cor}

\begin{proof}
Denote the complements of the above sets in $\Sigma_{\L}^{c}$ by
$\Sigma_{I}^{c},\,\Sigma_{II}^{c}$ and $\Sigma_{\G}^{c}$. By their
definition,
\begin{align*}
\Sigma_{I}^{c} & =\cup_{v\in\V}Z_{v},\\
\Sigma_{II}^{c} & =\cup_{v\in\V_{in}}\cup_{e\in\E_{v}}Z_{v,e},\,\,\text{and}\\
\Sigma_{\G}^{c} & =\Sigma_{II}^{c}\cup\Sigma_{I}^{c}.
\end{align*}
The result now follows from Lemma \ref{lem: verte values subsets}.
\end{proof}
The following is a corollary of the above together with Lemma \ref{lem: index sets and secular sets}
and Theorem \ref{thm: BG equidistribution}:
\begin{cor}
\label{cor: irrational gives densities} If $\Gamma_{\lv}$ is a standard
graph with $\lv$ rationally independent, then the index sets $\mathcal{S},\,\L,\,\mathcal{P}_{I},\,\mathcal{P}_{II}$
and $\G$ have densities, and they are given by 
\begin{align}
d\left(\mathcal{S}\right) & =\mu_{\lv}\left(\Sigma^{reg}\right)=1,\\
d\left(\L\right) & =\mu_{\lv}\left(\Sigma_{\L}\right),\\
d\left(\mathcal{P}_{I}\right) & =\mu_{\lv}\left(\Sigma_{I}\right),\\
d\left(\mathcal{P}_{II}\right) & =\mu_{\lv}\left(\Sigma_{II}\right),\,\,\text{and}\\
d\left(\G\right) & =\mu_{\lv}\left(\Sigma_{\G}\right).
\end{align}
\end{cor}

\begin{proof}
In every equality above, the RHS is the BG-measure of a set $\Sigma_{*}$
that was shown to be Jordan in Corollary \ref{cor: key lemma in equivalence of genericity}
or Remark \ref{rem: Sigma L is a union of cc}, and the LHS is the
density of an index set $\N_{*}=\set{n\in\N}{\left\{ k_{n}\lv\right\} \in\Sigma_{*}}$
according to lemma \ref{lem: index sets and secular sets}. As $\lv$
is rationally independent, then according to Theorem \ref{thm: BG equidistribution},
$d\left(\N_{*}\right)=\mu_{\lv}\left(\Sigma_{*}\right)$, which concludes
the proof.
\end{proof}

\subsection{\label{subsec: proofs part}Proof of Theorem \ref{thm: density-of-generic-and-loop-eigenfunctions}.}

As discussed in the outline, the first deduction is the following
lemma:
\begin{lem}
\label{lem: proves Theorem}The statement 
\begin{equation}
\forall\lv\in\left(\R_{+}\right)^{\E}\,\,\,\mu_{\lv}\left(\Sigma_{\G}\sqcup\Sigma_{\mathcal{L}}\right)=1,\label{eq: measures equalitiy}
\end{equation}
implies Theorem \ref{thm: density-of-generic-and-loop-eigenfunctions}. 
\end{lem}

\begin{proof}
By Lemma \ref{lem: loop decomposition}, $\mu_{\lv}\left(\tilde{Z}_{e}\right)=\frac{l_{e}}{2L}$
for every loop $e$ and every $\lv\in\left(\R_{+}\right)^{\E}$. Therefore,
\[
\mu_{\lv}\left(\Sigma_{\mathcal{L}}\right)=\sum_{e\in\E_{loops}}\mu_{\lv}\left(Z_{e}^{reg}\right)=\frac{1}{2}\frac{\sum_{e\in\E_{loops}}l_{e}}{L}.
\]
If $\lv$ is rationally independent, then by Corollary \ref{cor: irrational gives densities},
\begin{equation}
d\left(\G\right)=\mu_{\lv}\left(\Sigma_{\G}\right),\,\,\,d\left(\mathcal{L}\right)=\mu_{\lv}\left(\Sigma_{\mathcal{L}}\right),\,\,\text{and }\,\,d\left(\mathcal{S}\right)=\mu_{\lv}\left(\Sigma^{reg}\right).
\end{equation}
So Theorem \ref{thm: density-of-generic-and-loop-eigenfunctions}
(\ref{enu: gen a.s}) now follows from (\ref{eq: measures equalitiy}). 

In order to prove Theorem \ref{thm: density-of-generic-and-loop-eigenfunctions}
(\ref{enu: gen always}), assume that (\ref{eq: measures equalitiy})
holds. As seen in the proof of Corollary \ref{cor: key lemma in equivalence of genericity},
$\Sigma_{\G}^{c}$ the complement of $\Sigma_{\G}$ in $\Sigma_{\L}^{c}$
is a finite union of $Z_{*}$'s and so by Lemma \ref{lem: verte values subsets}
is is either of positive co-dimension in $\Sigma^{reg}$ or it contains
a connected component of $\Sigma^{reg}$. According to (\ref{eq: measures equalitiy}),
\[
\mu_{\lv}\left(\Sigma_{\G}^{c}\right)=\mu_{\lv}\left(\Sigma^{reg}\right)-\mu_{\lv}\left(\Sigma_{\G}\sqcup\Sigma_{\mathcal{L}}\right)=0,
\]
and as $\mu_{\lv}$ is strictly positive on open sets (see Remark
\ref{rem: BG measuer positivity}), then $\Sigma_{\G}^{c}$ does not
contain any connected component of $\Sigma^{reg}$ and is therefore
of positive co-dimension in $\Sigma^{reg}$. Namely, $\dim\left(\Sigma^{reg}\setminus\left(\Sigma_{\G}\sqcup\Sigma_{\L}\right)\right)\le E-2$.
By adding $\Sigma^{sing}$ which is also of positive co-dimension
in $\Sigma$, we get $\dim\left(\Sigma\setminus\left(\Sigma_{\G}\sqcup\Sigma_{\L}\right)\right)\le E-2$.
As $\Sigma$ is closed in $\T^{\E}$ and $\Sigma_{\G}\sqcup\Sigma_{\L}$
in open in $\Sigma$, then $\Sigma\setminus\left(\Sigma_{\G}\sqcup\Sigma_{\L}\right)$
is closed in $\T^{\E}$. Denote the set $\tilde{B}:=\set{\lv\in\R^{\E}}{\left\{ \lv\right\} \in\Sigma\setminus\left(\Sigma_{\G}\sqcup\Sigma_{\L}\right)}$,
which is the lift of $\Sigma\setminus\left(\Sigma_{\G}\sqcup\Sigma_{\L}\right)$
from $\R^{\E}/2\pi\Z^{\E}$ to $\R^{\E}$. Then $\tilde{B}$ is closed
with $\dim\left(\tilde{B}\right)\le E-2$. Consider the ``bad''
subset of edge lengths $B\subset\left(\R_{+}\right)^{\E}$ given by
\[
B:=\set{\lv\in\left(\R_{+}\right)^{\E}}{\exists t>0\,\,\text{s.t}\,\,\,t\lv\in\tilde{B}}.
\]
It is the cone of the restriction $\tilde{B}\cap\left(\R_{+}\right)^{\E}$,
which is closed in $\left(\R_{+}\right)^{\E}$, and is therefore closed
and of dimension $\dim\left(B\right)\le\dim\left(\tilde{B}\cap\left(\R_{+}\right)^{\E}\right)+1\le E-1$.
It is therefore closed and nowhere dense. Hence, its complement, the
``good'' set $G=\left(\R_{+}\right)^{\E}\setminus B$ is residual.
By definition, if $\lv\in G$, then for every $k>0$ such that $\left\{ k\lv\right\} \in\Sigma$
we get that $\left\{ k\lv\right\} \in\Sigma_{\G}\sqcup\Sigma_{\L}$
which according to Lemma \ref{lem: index sets and secular sets} proves
\ref{thm: density-of-generic-and-loop-eigenfunctions} (\ref{enu: gen always}).
\end{proof}
\begin{rem}
The argument above is a machinery showing that every property of a
standard graph, that is described by a subset $\Sigma_{*}\subset\Sigma$
whose complement has positive co-dimension, is generic in both senses
described above.
\end{rem}

The second deduction is the following lemma:
\begin{lem}
\label{lem: fat Zve} If $\Gamma$ is such that $\mu_{\lv}\left(\Sigma_{\G}\sqcup\Sigma_{\mathcal{L}}\right)<1$
for some $\lv\in\left(\R_{+}\right)^{\E}$, then there exists an interior
vertex $v\in\V_{in}$ and edge $e\in\E_{v}$ such that $Z_{v,e}\setminus Z_{v}$
contains an open set.
\end{lem}

\begin{proof}
Let $G$ be the residual set in Theorem \ref{thm: Genericity of Grisha},
and let $G'$ be its (residual) intersection with the set of rationally
independent edge lengths. Let $\lv\in G'$, then as $\lv\in G$, $d\left(\mathcal{P}_{I}\right)=1-d\left(\L\right)$
according to Theorem \ref{thm: Genericity of Grisha}. Since $\lv$
is also rationally independent, Corollary \ref{cor: irrational gives densities}
gives, 
\begin{equation}
\mu_{\lv}\left(\Sigma_{I}\right)=d\left(\mathcal{P}_{I}\right)=1-d\left(\L\right)=1-\mu_{\lv}\left(\Sigma_{\L}\right).
\end{equation}
Since $G'$ is dense in $\R_{+}^{\E}$ and $\mu_{\lv}$ is continuous
in $\lv$ (immediate from Definition \ref{def: BG measure}), then
$\mu_{\lv}\left(\Sigma_{I}\right)=1-\mu_{\lv}\left(\Sigma_{\L}\right)$
for any $\lv$. 

As we assumed that $\mu_{\lv}\left(\Sigma_{\G}\right)+\mu_{\lv}\left(\Sigma_{\L}\right)<1=\mu_{\lv}\left(\Sigma_{I}\right)+\mu_{\lv}\left(\Sigma_{\L}\right)$
then $\mu_{\lv}\left(\Sigma_{\G}\right)<\mu_{\lv}\left(\Sigma_{I}\right)$
which means that $\mu_{\lv}\left(\Sigma_{I}\setminus\Sigma_{\G}\right)=\mu_{\lv}\left(\Sigma_{I}\setminus\Sigma_{II}\right)>0$.
The set $\Sigma_{I}\setminus\Sigma_{II}$ is a union of sets of the
form $Z_{v,e}\setminus Z_{v}$ and so if it has positive measure, then
there must be at least one $Z_{v,e}\setminus Z_{v}$ set which is
not of positive co-dimension. According to Lemma \ref{lem: verte values subsets},
we conclude that there is a connected component $M$ of $\Sigma^{reg}$
contained in $Z_{v,e}$ on which $Z_{v}\cap M$ is a closed set of
positive co dimension. Hence $M\cap Z_{v,e}\setminus Z_{v}$ is open
(and dense in $M$).
\end{proof}
The third deduction requires the notion of splitting a vertex. If
$v$ is an interior vertex of $\Gamma$ of $\deg v>3$ and $e\in\E_{v}$
is not a bridge, then we define $\tilde{\Gamma}$ as the graph obtained
by splitting $v$ into two vertices $v_{1}$ and $v_{2}$ such that
$v_{1}$ is only connected to $e$ and thus $v_{2}\in\partial\tilde{\Gamma}$,
and $v_{2}$ is connected to the remaining edges in $\E_{v}\setminus\left\{ e\right\} $.
See Figure \ref{fig: splitting} for example. In Appendix \ref{sec: App Gluing and contracting}
we discuss this process and the relations between of the secular manifolds
of $\Gamma$ and $\tilde{\Gamma}$, which we will use in the proof
of the next lemma. 

\begin{figure}
\includegraphics[width=0.4\paperwidth]{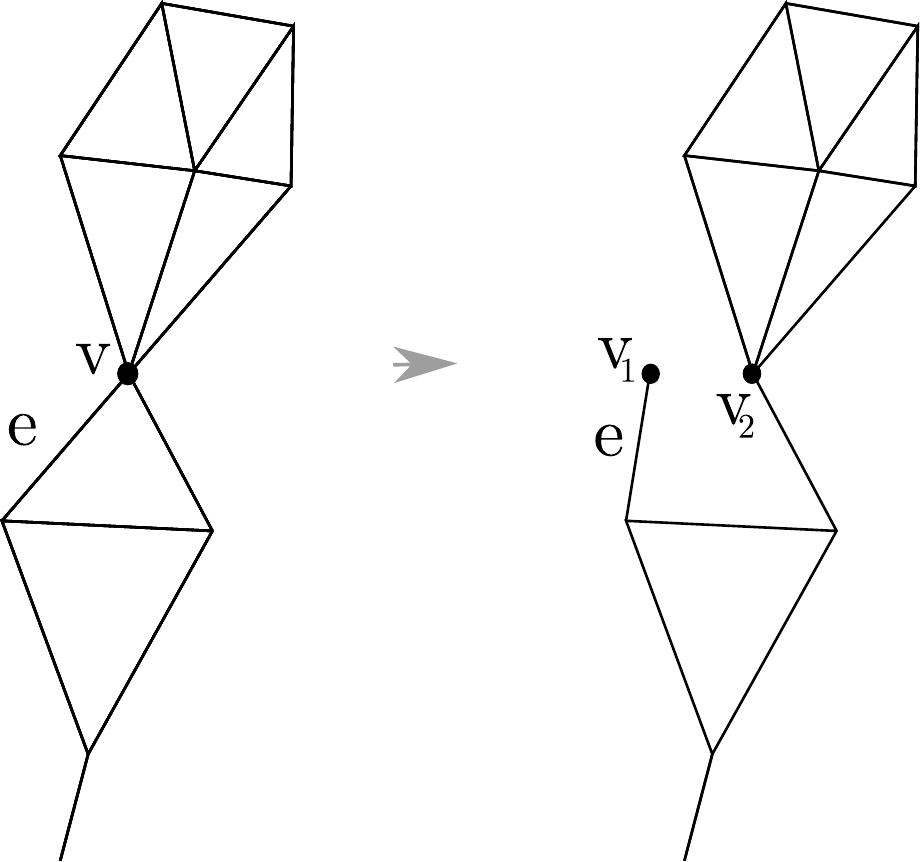}\caption[Vertex splitting example]{\label{fig: splitting} On the left, a graph $\Gamma$ with interior
vertex $v$ of $\protect\deg v=5$ and edge $e\in\protect\E_{v}$
which is not a bridge. On the right, $\tilde{\Gamma}$ obtained by
splitting $v$ into $v_{1}$ and $v_{2}$ such that $v_{1}$ is connected
to $e$ and has $\protect\deg{v_{1}}=1$ and $v_{2}$ is connected
to the remaining $\protect\E_{v}\setminus e$ and has $\protect\deg{v_{2}}=4$.}
\end{figure}

\begin{figure}
\includegraphics[width=0.4\paperwidth]{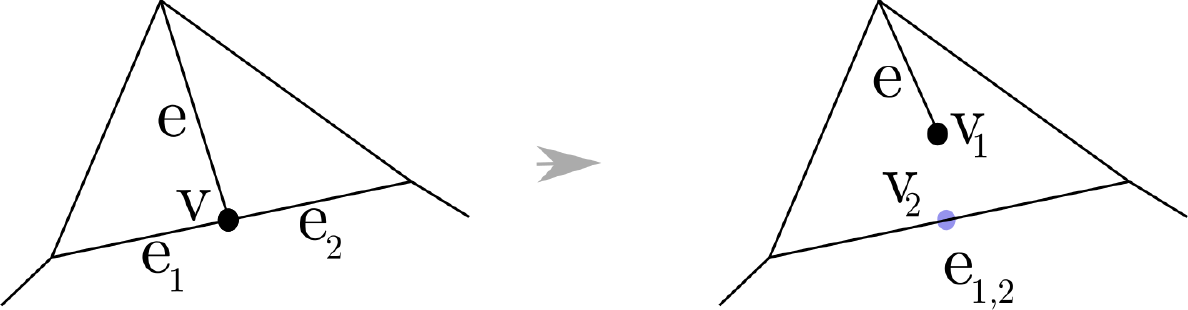}\caption[Vertex splitting example - deg(v)=3]{\label{fig: splitting-D3} On the left, a graph with interior vertex
$v$ of $\protect\deg v=3$ and $\protect\E_{v}=\left\{ e,e_{1},e_{2}\right\} $
such that $e$ is not a bridge. On the right, the graph obtained by
splitting $v$ into $v_{1}$ and $v_{2}$ such that $v_{1}$ is connected
to $e$ and has $\protect\deg{v_{1}}=1$ and $v_{2}$ is considered
as an interior vertex of the new edge $e_{1,2}$ which is the concatenation
of $e_{1}$ and $e_{2}$.}
\end{figure}
 
\begin{lem}
\label{lem: Zve constant}Let $\Gamma$ be a graph and assume there
exists an interior vertex $v\in\V_{in}$ with an edge $e\in\E_{v}$
such that $Z_{v,e}\setminus Z_{v}$ contains an open set in $\Sigma^{reg}$.
Then,
\begin{enumerate}
\item \label{enu: Zve const - loops and bridges}If $\Gamma$ has loops
or bridges, then for any $\kv\in\Sigma^{reg}$ either $f_{\kv}$ is
supported on a loop or $\partial_{e}f_{\kv}\left(v\right)=0$. That
is, $Z_{v,e}=\Sigma_{\L}^{c}$, and therefore $\Sigma_{II}=\emptyset$.
\item \label{enu: not deg 3}If $\Gamma$ has no loops and no bridges, then
$\deg v\ne3$.
\item \label{enu: Zvu const - no loops no bridges}If $\Gamma$ has no loops
and no bridges, and $\tilde{\Gamma}$ is the graph obtained by splitting
$v$ into $v_{1}$ and $v_{2}$ as discussed above, such that $v_{1}$
is a boundary vertex connected to $e$, then $\left|f_{\kv}\left(v_{1}\right)\right|=\left|f_{\kv}\left(v_{2}\right)\right|$
for any $\kv\in\tilde{\Sigma}^{reg}$. Where $\tilde{\Sigma}^{reg}$
is the regular part of the secular manifold of $\tilde{\Gamma}$.
\end{enumerate}
\end{lem}

\begin{proof}
We will prove each of the 3 cases separately. 

The proof of (\ref{enu: Zve const - loops and bridges}) is straightforward from Lemma \ref{lem: verte values subsets}. To see that,
observe that if $Z_{v,e}\setminus Z_{v}$ contains an open set in
$\Sigma^{reg}$, then $Z_{v,e}$ has zero co-dimension in $\Sigma^{reg}$.
It now follows from Lemma \ref{lem: verte values subsets}, and the
assumption that $\Gamma$ has loops or bridges, that $Z_{v,e}=\Sigma_{\L}^{c}$
as needed. 

Let us now prove (\ref{enu: not deg 3}):

Let $\Gamma$ be a graph with no loops or bridges, let $v$ be a vertex
of $\deg v=3$ and assume by contradiction that for some $e\in\E_{v}$,
$Z_{v,e}\setminus Z_{v}$ contains an open set $O$ in $\Sigma^{reg}$.
Denote $\E_{v}=\left\{ e,e_{1},e_{2}\right\} $. Let $\kv=\left(\kappa_{e},\kappa_{1},\kappa_{2},...\right)\in O$,
and notice that the Neumann condition at $v$ together with $\kv$
being in $Z_{v,e}\setminus Z_{v}$ implies that 
\[
f_{\kv}\left(v\right)^{2}+\partial_{e_{1}}f_{\kv}\left(v\right)^{2}=f_{\kv}\left(v\right)^{2}+\partial_{e_{2}}f_{\kv}\left(v\right)^{2},
\]
and so according to Lemma \ref{lem: |ae|=00003D|ae'|} and Lemma \ref{lem: F p and mk},
the normal to $\Sigma^{reg}$ at $\kv$, $\hat{n}\left(\kv\right)$,
has equal $e_{1}$ and $e_{2}$ components. In particular it follows
that for small enough $\epsilon$, the path \[\kv\left(t\right)=\left\{ \left(\kappa_{e},\kappa_{1}+t,\kappa_{2}-t,...\right)\right\}\  
\text{with}\  t\in\left(-\epsilon,\epsilon\right),\] is orthogonal to the normal and therefore contained in the open
set $O$. 

In order to prove (\ref{enu: not deg 3}), we need to consider a splitting
of $v$ as shown in Figure \ref{fig: splitting-D3} into $v_{1}$
and $v_{2}$ such that $v_{1}$ is connected to $e$ and has $\deg{v_{1}}=1$.
In this case, $v_{2}$ is of degree two and according to Remark \ref{rem: deg 2}
we may consider $v_{2}$ as an interior point of a new edge $e_{1,2}$
which is the concatenation of $e_{1}$ and $e_{2}$. Denote the new
graph by $\tilde{\Gamma}$ and consider the map $T\left(\kappa_{e},\kappa_{1},\kappa_{2},...\right):=\left\{ \left(\kappa_{e},\kappa_{1}+\kappa_{2},...\right)\right\} $
as discussed in Lemma \ref{lem: deg3 construction} in Appendix \ref{sec: App Gluing and contracting}.
Let $\kv=\left(\kappa_{e},\kappa_{1},\kappa_{2},...\right)\in O$
be the starting point of the path we constructed, and consider the
standard graphs $\Gamma_{\kv}$ and $\tilde{\Gamma}_{T\left(\kv\right)}$.
If $e_{1}$ and $e_{2}$ of $\Gamma_{\kv}$ have edge lengths $\kappa_{1}$
and $\kappa_{2}$, then $e_{1,2}$ of $\tilde{\Gamma}_{T\left(\kv\right)}$
is of length $\kappa_{1}+\kappa_{2}$ and if $x_{1,2}\in\left[0,\kappa_{1}+\kappa_{2}\right]$
is an arc-length coordinate on $e_{1,2}$ in the direction which goes
from $e_{1}$ to $e_{2}$, then the interior point $v_{2}$ is at $x_{12}=\kappa_{1}$.
Denote the secular manifold of $\tilde{\Gamma}$ by $\tilde{\Sigma}$.
Observe that according to Lemma \ref{lem: deg3 construction}, if
$\kv'\in\Sigma^{reg}$ with $\partial_{e}f_{\kv'}\left(v\right)=0$
then $T\left(\kv'\right)\in\tilde{\Sigma}$. Therefore, $T\left(Z_{v,e}\right)\subset\tilde{\Sigma}$.
It is not hard to deduce that $T\left(O\right)\subset\tilde{\Sigma}$
is of zero co-dimension in $\tilde{\Sigma}$, so $T\left(O\right)\cap\tilde{\Sigma}^{reg}$
is not empty and we may assume that $\kv$, the starting point of
the path, satisfies $T\left(\kv\right)\in\tilde{\Sigma}^{reg}$. Observe
that $T$ is constant on the path $\kv\left(t\right)$, namely $T\left(\kv\left(t\right)\right)\equiv T\left(\kv\right)\in\tilde{\Sigma}^{reg}$.
According to Lemma \ref{lem: deg3 construction}, it follows that
the $\Gamma_{\kv\left(t\right)}$ and $\tilde{\Gamma}_{T\left(\kv\right)}$
canonical eigenfunctions $f_{\kv\left(t\right)}$ and $\tilde{f}_{T\left(\kv\right)}$
satisfy:
\begin{equation}
\tilde{f}_{T\left(\kv\right)}|_{e_{1,2}}\left(\kappa_{1}+t\right)=\tilde{f}_{T\left(\kv\left(t\right)\right)}\left(v_{2}\right)=f_{\kv\left(t\right)}\left(v\right).
\end{equation}
In particular, $\tilde{f}_{T\left(\kv\right)}|_{e_{1,2}}$ is constant
on $x_{1,2}\in\left(\kappa_{1}-\epsilon,\kappa_{1}+\epsilon\right)$
and is non zero since $f_{\kv\left(t\right)}\left(v\right)\ne0$ by
$\kv\left(t\right)\in O\subset Z_{\tilde{v},e}\setminus Z_{\tilde{v}}$.
This is a contradiction for $\tilde{f}_{T\left(\kv\right)}$ being
an eigenfunction of eigenvalue $k^{2}=1$. As needed.

It is left to prove (\ref{enu: Zvu const - no loops no bridges}):

Assume that $\tilde{\Gamma}$ has no loops and no bridges. Then by
(\ref{enu: not deg 3}), and since $\tilde{v}$ is interior vertex,
$\deg{\tilde{v}}\ge4$. In such case, since non of the edges connected
to $\tilde{v}$ is a bridge or a loop, then $\tilde{\Gamma}$ can
be obtained by a graph $\Gamma$ with boundary vertex $v$ glued to
an interior vertex $u$ such that the new vertex is $\tilde{v}$.
Let $\kv\in Z_{\tilde{v},e}\setminus Z_{\tilde{v}}$, then $\partial_{e}f_{\kv}\left(\tilde{v}\right)=0$,
and by Lemma \ref{lem: gluing vertices} we get that $\kv\in\Sigma$
(the secular manifold of $\Gamma$). It follows that $Z_{\tilde{v},e}\setminus Z_{\tilde{v}}\subset\Sigma$,
and since they have the same dimension and $Z_{\tilde{v},e}\setminus Z_{\tilde{v}}$
contains an open set, then there is an open set $O\subset\Sigma^{reg}\cap Z_{\tilde{v},e}\setminus Z_{\tilde{v}}$.
In such case, for every $\kv\in O$, the $\Gamma$ canonical eigenfunction
$f_{\kv}$ satisfies $f_{\kv}\left(v\right)=f_{\kv}\left(u\right)$
(by Lemma \ref{lem: gluing vertices}) and so $Z_{v,u}$ (of $\Gamma$)
contains $O$ and so is not of positive co-dimension. But $v$ is
a boundary vertex of $\Gamma$, and so $\Gamma$ has a tail (which
is a bridge) and no loops (because $\tilde{\Gamma}$ has no loops).
Therefore, according to Lemma \ref{lem: verte values subsets}, $Z_{v,u}=\Sigma^{reg}$.
As needed. 
\end{proof}
The last step of the proof of Theorem \ref{thm: density-of-generic-and-loop-eigenfunctions},
as discussed in the sketch of the proof, is to provide the next `counter
example lemma' that would lead to a contradiction to the assumption
that there exists a graph for which $\mu_{\lv}\left(\Sigma_{\G}\sqcup\Sigma_{\L}\right)<1$.
We will state the lemma here and prove in the next subsection.
\begin{lem}
\label{lem: example for def(v) not zero} Let $\Gamma$ be a graph
with an interior vertex $v$ and edge $e\in\E_{v}$, then there exists
$\kv\in\Sigma_{\L}^{c}$ such that $\partial_{e}f_{\kv}\left(v\right)\ne0$.
Furthermore, if $\Gamma$ has a boundary vertex $u$, then we can
choose $\kv\in\Sigma_{\L}^{c}$ such that $\left|f_{\kv}\left(v\right)\right|\ne\left|f_{\kv}\left(u\right)\right|$.
\end{lem}

We may now collect the four lemmas to prove Theorem \ref{thm: density-of-generic-and-loop-eigenfunctions}:
\begin{proof}
Lemma \ref{lem: example for def(v) not zero} provides counter examples
for both Lemma \ref{lem: Zve constant} (\ref{enu: Zve const - loops and bridges})
and Lemma \ref{lem: Zve constant} (\ref{enu: Zvu const - no loops no bridges}).
Therefore, the assumption of Lemma \ref{lem: Zve constant} is false.
Namely, there cannot be a graph $\Gamma$ with interior vertex $v$
and edge $e\in\E_{v}$ such that $Z_{v,e}\setminus Z_{v}$ contains
an open set. It now follows from Lemma \ref{lem: fat Zve} that for
any graph $\Gamma$ and any $\lv\in\left(\R_{+}\right)^{\E}$, $\mu_{\lv}\left(\Sigma_{\G}\sqcup\Sigma_{\mathcal{L}}\right)=1$.
As was shown in Lemma \ref{lem: proves Theorem}, this condition is
sufficient to prove Theorem \ref{thm: density-of-generic-and-loop-eigenfunctions}.
\end{proof}

\subsection{Stowers and a proof for the `counter example lemma' }

In order to prove this lemma we will use a method of contracting edges
to get a reduction of the problem to a small family of graphs for
which we can construct $f_{\kv}$ explicitly. 
\begin{defn}
\label{def: stower}We call a graph $\Gamma$ a \emph{stower}\footnote{The name `stower', as a wedge of \emph{star} and \emph{flower} graphs,
was coined in \cite{BanLev17}.} if it has only one interior vertex $v_{0}$, which we call the \emph{central
vertex}. In such case every edge is either a loop or a tail, so we
characterize stowers by the number of tails and loops. Given a stower
of $n$ tails and $m$ loops, we denote its vertices by $\left\{ v_{j}\right\} _{j=0}^{n}$
with corresponding tails $\left\{ e_{i}\right\} _{i=1}^{n}$ and the
loops are denoted by $\left\{ e_{i}\right\} _{i=n+1}^{n+m}$. We number
the torus coordinates such that $\kappa_{j}$ correspond to $e_{j}$.
\end{defn}

\begin{rem}
We allow the cases of either $n=0$ or $m=0$ as long as $\deg{v_{0}}=n+2m\ge3$. 
\end{rem}

\begin{lem}
\label{lem: stower}Let $\Gamma$ be a stower with $n$ tails and
$m$ loops. Then, $\kv\in\Sigma_{\G}$ if the following holds: 
\begin{align}
\forall j\le n\,\,\,\kappa_{j}\notin & \frac{\pi}{2}\N\\
\forall n<j\le m\,\,\,\kappa_{j}\notin & \pi\N\\
\sum_{j=1}^{n}\tan\left(\kappa_{j}\right)+2\sum_{j=n+1}^{m}\tan\left(\frac{\kappa_{j}}{2}\right) & =0.\label{eq: secular stower}
\end{align}
Moreover, given $\kv\in\Sigma_{\G}$, 
\begin{equation}
\forall j\le n\,\,\,\,f_{\kv}\left(v_{j}\right)\cos\left(\kappa_{j}\right)=f_{\kv}\left(v_{0}\right).\label{eq: stower boundary vertices}
\end{equation}
\end{lem}

\begin{proof}
Let $f\in Eig\left(\Gamma_{\kv},1\right)$, and assume that $\kv$
satisfies $\kappa_{j}\notin\frac{\pi}{2}\N$ for every $j\le n$ and
$\kappa_{j}\notin\pi\N$ for every $n<j\le m$. Using Definition \ref{def: edge restriction notations},
for every tail $e_{j}$ directed from $v_{j}$ to $v_{0}$ with coordinate
$x_{j}\in\left[0,\kappa_{j}\right]$ we have $f|_{e_{j}}\left(x_{j}\right)=A_{j}\cos\left(x_{j}\right)$.
Clearly $A_{j}=f\left(v_{j}\right)$ and therefore, 
\begin{align}
f\left(v_{0}\right) & =f\left(v_{j}\right)\cos\left(\kappa_{j}\right),\,\,\text{and}\\
\partial_{j}f\left(v_{0}\right)=-f|_{e_{j}}'\left(\kappa_{j}\right) & =f\left(v_{j}\right)\sin\left(\kappa_{j}\right).
\end{align}
This proves (\ref{eq: stower boundary vertices}), and since $\cos\left(\kappa_{j}\right)\ne0$
by the assumption, then either $f$ is supported on the loops or $f\left(v_{0}\right)\ne0$.
For every loop $e_{j}$ we can use the real-amplitudes pair in Definition
\ref{def: edge restriction notations}, so that $f|_{e_{j}}\left(x_{j}\right)=A_{j}\cos\left(x_{j}\right)+B_{j}\sin\left(x_{j}\right)$
for $x_{j}\in\left[-\frac{\kappa_{j}}{2},\frac{\kappa_{j}}{2}\right]$.
By continuity, 
\begin{align*}
f\left(v_{0}\right) & =A_{j}\cos\left(\frac{\kappa_{j}}{2}\right)-B_{j}\sin\left(\frac{\kappa_{j}}{2}\right)=A_{j}\cos\left(\frac{\kappa_{j}}{2}\right)+B_{j}\sin\left(\frac{\kappa_{j}}{2}\right),\,\,\text{and\,so}\\
B_{j}\sin\left(\frac{\kappa_{j}}{2}\right) & =0,\,\,\text{and}\\
f\left(v_{0}\right) & =A_{j}\cos\left(\frac{\kappa_{j}}{2}\right).
\end{align*}
Since $\sin\left(\frac{\kappa_{j}}{2}\right)\ne0$ by the assumption,
then $B_{j}=0$. We conclude that $f|_{e_{j}}\left(x_{j}\right)=A_{j}\cos\left(x_{j}\right)$
with 
\begin{align*}
f\left(v_{0}\right) & =A_{j}\cos\left(\frac{\kappa_{j}}{2}\right),\,\,\text{and}\\
\partial_{j+}\left(v_{0}\right) & =\partial_{j-}\left(v_{0}\right)=A_{j}\sin\left(\frac{\kappa_{j}}{2}\right).
\end{align*}
As we also assume $\cos\left(\frac{\kappa_{j}}{2}\right)\ne0$, then
$f\left(v_{0}\right)=0$ if and only if $A_{j}=0$ for all edges and
hence $f$ is the zero function. We may assume that $f\left(v_{0}\right)\ne0$,
and so with out loss of generality $f\left(v_{0}\right)=1$. In such
case, the continuity implies that 
\[
A_{j}=\begin{cases}
\frac{1}{\cos\left(\kappa_{j}\right)} & j\le n\\
\frac{1}{\cos\left(\frac{\kappa_{j}}{2}\right)} & n<j\le m
\end{cases},
\]
and the Neumann condition on $v_{0}$ gives 
\[
\sum_{j=1}^{n}\tan\left(\kappa_{j}\right)+2\sum_{j=n+1}^{m}\tan\left(\frac{\kappa_{j}}{2}\right)=0.
\]
As needed. By the construction, it follows that any other $\tilde{f}\in Eig\left(\Gamma_{\kv},1\right)$
would be proportional to $f$, and so $\kv\in\Sigma^{reg}$. As the
boundary vertex values are the $A_{j}'s$ and $f\left(v_{0}\right)=1$
then $\kv\in\Sigma_{I}$. The derivatives on $v_{0}$ (the only interior
vertex) are given by 
\[
\begin{cases}
\partial_{e_{j}}f\left(v_{0}\right)=\tan\left(\kappa_{j}\right) & j\le n\\
\partial_{e_{j}}f\left(v_{0}\right)=\tan\left(\frac{\kappa_{j}}{2}\right) & n<j\le m
\end{cases},
\]
and so by the assumption are all non zero and $\kv\in\Sigma_{II}$
and therefore in $\Sigma_{\G}$.
\end{proof}
We can now use the above to prove Lemma \ref{lem: example for def(v) not zero}:
\begin{proof}
Let $\Gamma$ be a graph with an interior vertex $v$ and edge $e\in\E_{v}$. 

First assume that $e$ is a bridge, with the bridge decomposition
$\Gamma\setminus\left\{ e\right\} =\Gamma_{1}\sqcup\Gamma_{2}$, and
 edge sets $\E_{j}$ corresponding to $\Gamma_{j}$. Recall that we use
the coordinates $\kv=\left(\kv_{1},\kappa_{e},\kv_{2}\right)$ with
$\kv_{j}\in\T^{\E_{j}}$. According to Proposition \ref{prop: simple bridge det decomposition}
the secular function is factorized to 
\begin{equation}
F\left(\kv\right)=g_{1}\left(\kv_{1}\right)g_{2}\left(\kv_{2}\right)\left(1-e^{i2\kappa_{e}}e^{i\Theta_{1}\left(\kv_{1}\right)}e^{i\Theta_{2}\left(\kv_{2}\right)}\right),\label{eq: F dec 1}
\end{equation}
and $\set{\kv\in\Sigma^{reg}}{f_{\kv}|_{e}\not\equiv0}$ is open in
$\Sigma^{reg}$ and is given by 
\begin{equation}
\set{\left(\kv_{1},\kappa_{e},\kv_{2}\right)\in\T^{\E}}{g_{1}\left(\kv_{1}\right)g_{2}\left(\kv_{2}\right)\ne0\,\,\text{and }\,\,e^{i2\kappa_{e}}e^{i\Theta_{1}\left(\kv_{1}\right)}e^{i\Theta_{2}\left(\kv_{2}\right)}=1}.\label{eq: bridge fk|e}
\end{equation}
Let $\kv=\left(\kv_{1},\kappa_{e},\kv_{2}\right)\in\Sigma^{reg}$
such that $f_{\kv}|_{e}\not\equiv0$, then the Neumann vertex condition
at $v$ implies that $f_{\kv}|_{e'}\not\equiv0$ for some $e'\in\E_{v}\setminus e=\E_{v}\cap\E_{1}$.
According to Lemma \ref{lem: F p and mk}, $f_{\kv}|_{e'}\not\equiv0$
implies that $\frac{\partial F}{\partial\kappa_{e'}}\left(\kv\right)\ne0$.
As (\ref{eq: bridge fk|e}) implies that $g_{1}\left(\kv_{1}\right)g_{2}\left(\kv_{2}\right)\ne0$,
taking derivative of (\ref{eq: F dec 1}) gives: 
\[
\frac{\partial F}{\partial\kappa_{e'}}\left(\kv\right)=-ig_{1}\left(\kv_{1}\right)g_{2}\left(\kv_{2}\right)e^{i2\kappa_{e}}e^{i\Theta_{1}\left(\kv_{1}\right)}e^{i\Theta_{2}\left(\kv_{2}\right)}\frac{\partial\Theta_{1}}{\partial\kappa_{e'}}\left(\kv_{1}\right)\ne0,
\]
which means that $\frac{\partial\Theta_{1}}{\partial\kappa_{e'}}\left(\kv_{1}\right)\ne0$
and so $e^{i\Theta_{1}}$ is not constant around $\kv_{1}$. We can
thus assume that $\kv_{1}$ is such that $e^{i\Theta_{1}\left(\kv_{1}\right)}\ne1$
otherwise we vary $\kv_{1}$ while keeping $g_{1}\left(\kv_{1}\right)\ne0$.
Let $\boldsymbol{a}$ be the amplitudes vector of $f_{\kv}$, so by
Definition \ref{def: edge restriction notations} 
\[
\partial_{e}f_{\kv}\left(v\right)=0\iff\frac{a_{e}}{a_{\hat{e}}}e^{-i\kappa_{e}}=1,
\]
and according to Lemma \ref{lem: Rev phase}, $\frac{a_{e}}{a_{\hat{e}}}e^{-i\kappa_{e}}=e^{i\Theta_{1}\left(\kv_{1}\right)}$.
As we assume that $e^{i\Theta_{1}\left(\kv_{1}\right)}\ne1$, then
$\partial_{e}f_{\kv}\left(v\right)\ne0$. We thus found a point $\kv\in\Sigma^{reg}$
for which both $f_{\kv}|_{e}\not\equiv0$ (and so $\kv\in\Sigma_{\L}^{c}$)
and $\partial_{e}f_{\kv}\left(v\right)\ne0$ as needed.

We may now assume that $e$ is not a bridge, and denote the (possibly
empty) set of tails by $\E_{\partial\Gamma}$. Let $T_{\Gamma}$ be
a choice of a spanning tree\footnote{Given a connected graph $\Gamma$, a \emph{spanning tree }is a connected
subgraph $T\subset\Gamma$ which contains all vertices of $\Gamma$
and is a tree. A spanning tree always exists but may not be unique. } of $\Gamma\setminus\left\{ e\right\} $, and let $\tilde{\Gamma}$
be the graph obtained by contracting the edges of $T_{\Gamma}\setminus\E_{\partial\Gamma}$.
See Figure \ref{fig: Spaning tree} for example. Notice that $\tilde{\Gamma}$
is a stower with $\E_{\partial\Gamma}$ as its tails and $\beta$
loops, where $\beta$ is the first Betti number of $\Gamma$. Notice
that $e$ is a loop of $\tilde{\Gamma}$ and $v$ is identified with
the central vertex of the stower $v_{0}$. If $u$ was a boundary
vertex of $\Gamma$, then it is also a boundary vertex of the stower. 

\begin{figure}
\includegraphics[width=0.7\paperwidth]{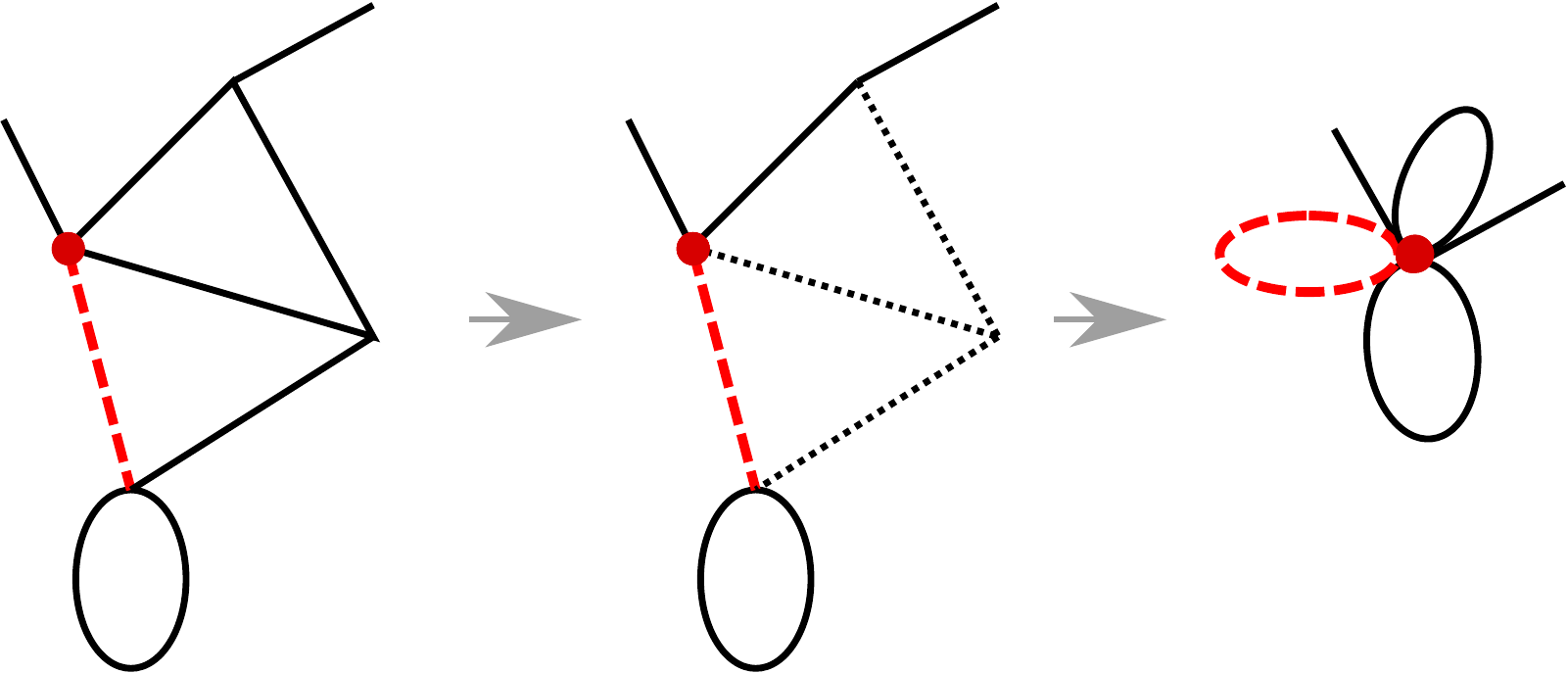}

\caption[Spanning tree contraction]{\label{fig: Spaning tree}On the left, a graph $\Gamma$ with a choice
of interior vertex $v$ and edge $e\in\protect\E_{v}$. $v$ is marked
in red and $e$ is dashed in red. On the middle, a choice of a spanning
tree $T$ of $\Gamma\setminus\left\{ e\right\} $. The edges of $T\setminus\protect\E_{\partial\Gamma}$
are dotted. On the right, $\tilde{\Gamma}$, the graph obtained by
contracting the dotted edges $T\setminus\protect\E_{\partial\Gamma}$.}
\end{figure}

Consider a real solution to $\sum_{j=1}^{\left|\partial\Gamma\right|}t_{j}+2\sum_{j=\left|\partial\Gamma\right|+1}^{\beta+\left|\partial\Gamma\right|}t_{j}=0$
such that all $t_{j}$'s are non zero. Let $\tilde{\kv}$ be such
that 
\begin{align*}
\kappa_{j} & =\begin{cases}
\tan^{-1}\left(t_{j}\right) & \forall j\le\left|\partial\Gamma\right|\\
2\tan^{-1}\left(t_{j}\right) & \forall\left|\partial\Gamma\right|<j\le\beta+\left|\partial\Gamma\right|
\end{cases},\,\,\text{with}\\
\tan^{-1}: & \R\setminus\left\{ 0\right\} \rightarrow\left(0,\frac{\pi}{2}\right)\cup\left(\frac{\pi}{2},\pi\right),
\end{align*}
so by Lemma \ref{lem: stower}, $\tilde{\kv}\in\tilde{\Sigma}_{\G}$.
Moreover, for any $u\in\partial\Gamma$ corresponding to the boundary
vertex $v_{j}$, 
\[
\left|f_{\tilde{\kv}}\left(v_{j}\right)\right|=\left|\cos\left(\kappa_{j}\right)f_{\tilde{\kv}}\left(v_{0}\right)\right|\ne\left|f_{\tilde{\kv}}\left(v_{0}\right)\right|.
\]
Let $\vec{\kappa}\in\T^{\E}$ given by the point $\tilde{\kv}$ above,
as follows: 
\[
\kv{}_{e'}=\begin{cases}
\tilde{\kappa}_{e'} & e'\in\tilde{\Gamma}\\
2\pi & e'\in T_{\Gamma}\setminus\E_{\partial\Gamma}
\end{cases},
\]
Lemma \ref{lem: contracting lemma-1} states that if $e''$ is not
a loop, $\Gamma_{\kv}$ is such that $\kv_{e''}=2\pi$, and $\Gamma_{\kv''}''$
is obtained by contracting $e''$ so that $\kv''$ and $\kv$ agree
on all other edges, then 
\[
\dim Eig\left(\Gamma_{\kv},1\right)=\dim Eig\left(\Gamma''_{\kv''},1\right),
\]
and if $\kv\in\Sigma^{reg}$, then the canonical eigenfunctions $f_{\kv}$
and $f_{\kv''}$ agree on all edges different than $e''$. Applying
Lemma \ref{lem: contracting lemma-1} finitely many times, for each
edge of $T_{\Gamma}\setminus\E_{\partial\Gamma}$, gives 
\[
\dim Eig\left(\Gamma_{\kv},1\right)=\dim Eig\left(\tilde{\Gamma}_{\tilde{\kv}},1\right).
\]
Since $\tilde{\kv}\in\tilde{\Sigma}_{\G}$, then the above implies
that $\kv\in\Sigma^{reg}$ and that for any $e'\in\tilde{\Gamma}$
(including $e$) the restriction of the canonical eigenfunction $f_{\kv}|_{e'}$
is equal to that of the $\tilde{\Gamma}$ canonical eigenfunction
$f_{\tilde{\kv}}|_{e'}$. In particular, since $\tilde{\kv}\in\tilde{\Sigma}_{\G}$,
then 
\[
\left|\partial_{e}f_{\kv}\left(v\right)\right|=\left|\partial_{e}f_{\tilde{\kv}}\left(v_{0}\right)\right|\ne0,
\]
and for every $u\in\partial\Gamma$, corresponding to some $v_{j}$
vertex of the stower, 
\[
\left|f_{\kv}\left(u\right)\right|=\left|f_{\tilde{\kv}}\left(v_{j}\right)\right|\ne\left|f_{\tilde{\kv}}\left(v_{0}\right)\right|=\left|f_{\tilde{\kv}}\left(v\right)\right|.
\]
\end{proof}

\newpage{}
\section{\label{sec: proof-of-existence} existence and symmetry of the nodal
and Neumann statistics}

In this section we construct the probabilistic setting needed in order
to discuss the statistics of the nodal surplus and Neumann surplus.
The results of this section appear in \cite{AloBanBer_cmp18,AloBan19}. 

Let $\Gamma_{\lv}$ be a standard graph and let $\G$ be the index
set of generic eigenfunctions, as in Definition \ref{def: index sets}.
Let $\sigma,\omega:\G\rightarrow\Z$ be the nodal surplus and Neumann
surplus sequences (see Definition \ref{def: Surpluses}). The discussion
on nodal and Neumann counts is restricted, by definition, to the subset
$\G$ and not to $\N$. The natural density (see Definition \ref{def: natural density})
is restricted accordingly:
\begin{defn}
\label{def: relative density} Let $\Gamma_{\lv}$ be a standard graph
and assume that $\G$ has positive density $d\left(\G\right)>0$.
If $A\subset\G$ has density, we define its relative density as, 
\[
d_{\G}\left(A\right):=\lim_{N\rightarrow\infty}\frac{\left|A\left(N\right)\right|}{\left|\G\left(N\right)\right|}=\frac{d\left(A\right)}{d\left(\G\right)}.
\]
Where $A\left(N\right)$ denotes $A\cap\left\{ 1,2,...N\right\} $
and same for $\G\left(N\right)$.
\end{defn}

\begin{rem}
\label{rem: dg not probabilty} The relative density $d_{\G}$ is
not a measure on $\G$ in general (that is with the power set as $\sigma$-algebra),
as it is not $\sigma$-additive. For example, $\forall j\in\G\,\,\,d_{\G}\left(\left\{ j\right\} \right)=0$
but $d_{\G}\left(\sqcup_{j\in\G}\left\{ j\right\} \right)=1$. Clearly,
the image of $d_{\G}$ is in $\left[0,1\right]$ with $d_{\G}\left(\emptyset\right)=0$
and $d_{\G}\left(\G\right)=1$. Therefore, given a $\sigma$-algebra
$\mathcal{F}$ on $\G$, $d_{\G}$ is a probability measure on $\left(\G,\mathcal{F}\right)$
if and only if every set in $\mathcal{F}$ has density and $d_{\G}$
is $\sigma$-additive on $\mathcal{F}$. 

We will now use the generic part of the secular manifold, $\Sigma_{\G}$
(see (\ref{eq: Sigma g})), to define a $\sigma$-algebra on $\G$
on which $d_{\G}$ will be a probability measure. 
\end{rem}

\begin{defn}
\label{def: sigma-algebra}Let $\Gamma_{\lv}$ be a standard graph
with (square-root) eigenvalues $\left\{ k_{n}\right\} _{n=0}^{\infty}$.
Let $\Sigma_{\G}$ as defined in (\ref{eq: Sigma g}). For any connected
component of $\Sigma_{\G}$, denoted by $\Sigma_{\G,i}$, we define
the index set $\G_{i}:=\set{n\in\G}{\left\{ k_{n}\lv\right\} \in\Sigma_{\G,i}}$.
We define $\mathcal{F}_{\G}$ to be the $\sigma$-algebra of $\G$
generated by all $\G_{i}$'s. 
\end{defn}

\begin{rem}
\label{rem: elements in the sigma algebra} As the $\G_{i}$'s are
disjoint, then they are the atoms of $\mathcal{F}_{\G}$, and any set
$A\in\mathcal{F}_{\G}$ is a (countable) disjoint union of atoms. 
\end{rem}

The main theorem of this section, combines Theorem 2.1 of \cite{AloBanBer_cmp18}
and Theorem 3.5 of \cite{AloBan19}: 
\begin{thm}
\label{thm:First}Let $\Gamma_{\lv}$ be a standard graph with $\lv$
rationally independent. Then,
\begin{enumerate}
\item \label{enu: triplet}The relative density, $d_{\G}$, is a probability
measure on $\left(\G,\mathcal{F}_{\G}\right)$ and we denote the probability
space by the triplet $\left(\G,\mathcal{F}_{\G},d_{\G}\right)$. Moreover,
every atom $\G_{i}$ has positive probability given by, 
\[
d_{\G}\left(\G_{i}\right)=\frac{\mu_{\lv}\left(\Sigma_{\G,i}\right)}{\mu_{\lv}\left(\Sigma_{\G}\right)}>0.
\]
\item \label{enu: sigma and omega random variables}Both $\sigma$ and $\omega$
are finite random variables on $\left(\G,\mathcal{F}_{\G},d_{\G}\right)$.
\\
In particular, for any possible value $j$, the following limits exist
and define the probabilities of the events $\sigma^{-1}\left(j\right):=\left\{ \sigma=j\right\} $
or $\omega^{-1}\left(j\right):=\left\{ \omega=j\right\} $:
\begin{align}
P\left(\sigma=j\right) & :=d_{\G}\left(\sigma^{-1}\left(j\right)\right)=\lim_{N\rightarrow\infty}\frac{\left|\set{n\in\G\left(N\right)}{\sigma_{n}=j}\right|}{\left|\G\left(N\right)\right|}\label{eq:Thm1}\\
P\left(\omega=j\right) & :=d_{\G}\left(\omega^{-1}\left(j\right)\right)=\lim_{N\rightarrow\infty}\frac{\left|\set{n\in\G\left(N\right)}{\omega_{n}=j}\right|}{\left|\G\left(N\right)\right|}.
\end{align}
\item \label{enu: symmetric simultaniuosly}The random variables $\sigma$
and $\omega$ are symmetric around $\frac{\beta}{2}$ and $\frac{\beta-\left|\partial\Gamma\right|}{2}$
simultaneously. That is, the joint probability of the event\[\left\{ \sigma=i\,\wedge\,\omega=j\right\} :=\sigma^{-1}\left(i\right)\cap\omega^{-1}\left(j\right),\]
which is given by 
\begin{equation}
P\left(\sigma=j\,\,\wedge\,\,\omega=i\right):=d_{\G}\left(\sigma^{-1}\left(j\right)\cap\omega^{-1}\left(i\right)\right),\label{eq: joint probability}
\end{equation}
satisfies the symmetry: 
\begin{align}
P\left(\sigma=j\,\,\wedge\,\,\omega=i\right) & =P\left(\sigma=\beta-j\,\,\wedge\,\,\omega=\beta-\left|\partial\Gamma\right|-i\right).\label{eq: joint symmetry}
\end{align}
\item \label{enu: levels sets are either ampty or of positive measure}Every
value of the pair $\left(\sigma,\omega\right)$ which is attained
once, appears infinitely often with positive density. Namely, if $\sigma\left(n\right)=i$
and $\omega\left(n\right)=j$ for some $n\in\G$, then $P\left(\sigma=j\,\,\wedge\,\,\omega=i\right)>0$.
\end{enumerate}
\end{thm}

Before proving this theorem we will first state a corollary of this
theorem. An inverse result, showing how $\beta$ and $\partial\Gamma$
can be obtained from the nodal and Neumann averages: 
\begin{cor}
\label{cor:Expected_values_of_nodal_and_Neumann_supluses} Let $\Gamma_{\lv}$
be a standard graph with $\lv$ rationally independent and first Betti
number $\beta$. Then, 
\begin{align*}
\mathbb{E}\left(\sigma\right)=\lim_{N\rightarrow\infty}\frac{1}{\left|\G(N)\right|}\sum_{n\in\G(N)}\sigma(n) & =\frac{\beta}{2},\,\,\text{and}\\
\mathbb{E}\left(\omega\right)=\lim_{N\rightarrow\infty}\frac{1}{\left|\G(N)\right|}\sum_{n\in\G(N)}\omega(n) & =\frac{\beta-\left|\partial\Gamma\right|}{2}.
\end{align*}
\end{cor}

\begin{rem}
If Conjecture
\ref{conj: Neumann surplus bounds} is true, as discussed in Remark
\ref{rem: bounda correlation}, then $\sigma$ and $\omega$ are not
independent. Such lack of independence can be spotted in Figure \ref{fig: random graph},
where the $\left(\sigma,\omega\right)$ statistics of a random regular
graph is provided. The support of the joint distribution of $\left(\sigma,\omega\right)$
is not rectangular which implies that they are not independent.

\begin{figure}
\includegraphics[width=0.7\paperwidth]{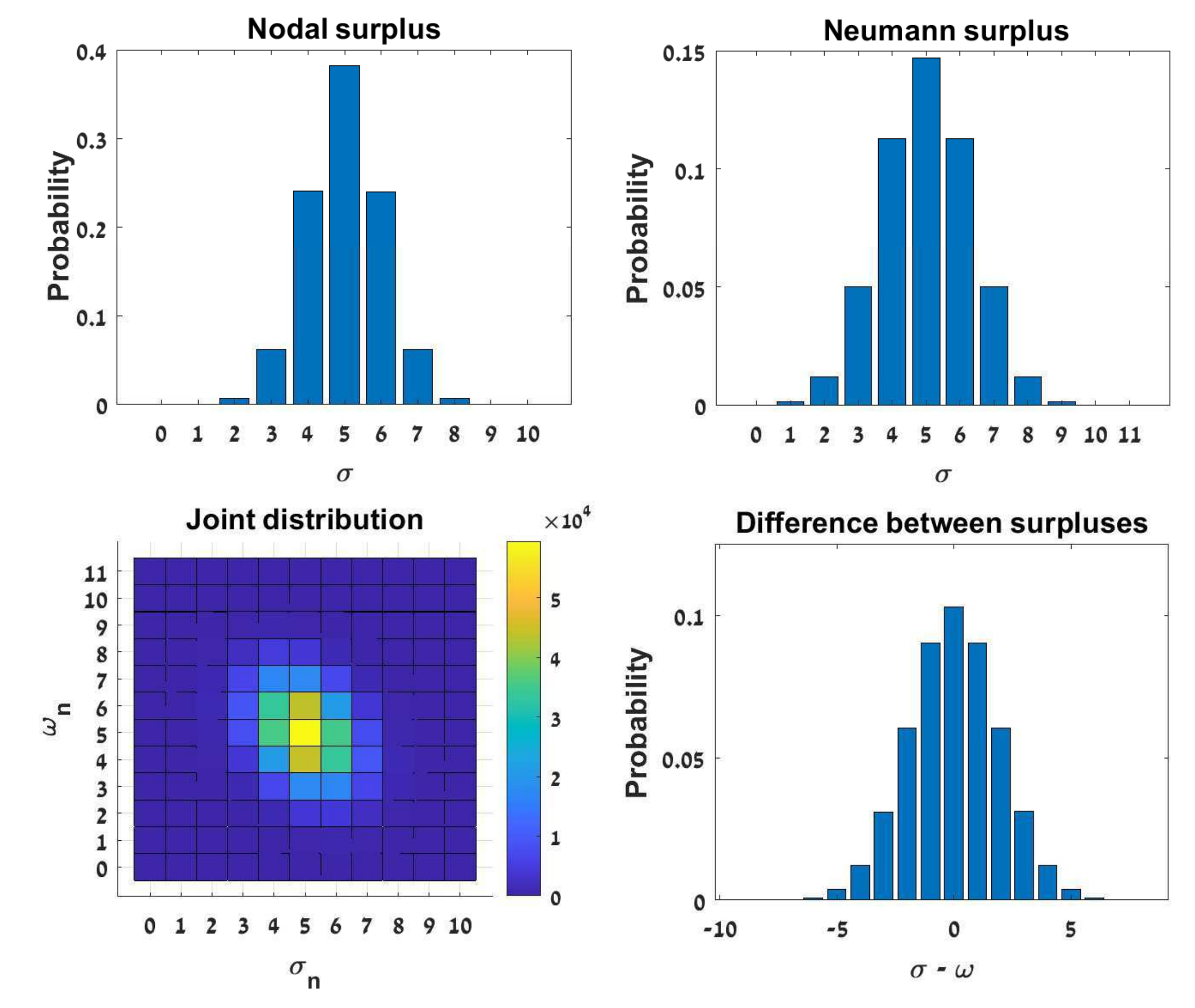}

\caption[Statistics for 6-complete graph]{\label{fig: random graph} Nodal and Neumann statistics for the complete
graph of 6 vertices using $10^{6}$ eigenfunctions. On the upper left,
the nodal surplus distribution. On the upper right, the Neumann surplus
distribution. On the bottom right, the distribution of the difference
$\sigma-\omega$. On the bottom left, a joint distribution of $\left(\sigma,\omega\right)$. }
\end{figure}
\end{rem}

The proof of Theorem 2.1 in \cite{AloBanBer_cmp18} relied on a result
called the \emph{nodal magnetic theorem }\cite{BerWey_ptrsa14} that
was discussed in the introduction and will be presented in Section
\ref{sec: Magnetic}. The result of Theorem 3.5 in \cite{AloBan19}
relied on Theorem 2.1 in \cite{AloBanBer_cmp18}, and therefore on
the nodal magnetic theorem. We will now present a proof of both, which
does not rely on the nodal magnetic theorem. 
\begin{defn}
\label{def:The-spectral-position} We define the \emph{spectral counting
function }of the standard graph $\Gamma_{\lv}$ as follows:
\[
N\left(\Gamma_{\lv},k\right):=\left|\set{0<\lambda\le k^{2}}{\lambda\,\,\text{is and eigenvalue of}\,\,\Gamma_{\lv}}\right|,
\]
where the counting includes multiplicity. If $k_{n}$ is a simple
eigenvalue of $\Gamma_{\lv}$, then its \emph{spectral position} is
$n$ and is equal to $N\left(\Gamma_{\lv},k_{n}\right)$. Given $\kv\in\Sigma^{reg}$,
we define the \emph{spectral position}\footnote{Not to be confused with $\hat{n}$, the normal to $\Sigma^{reg}$
at the point $\kv$}\emph{ }$n\left(\kv\right)$, as 

\begin{equation}
n\left(\kv\right):=N\left(\Gamma_{\kv},1\right).\label{eq: n definition-1}
\end{equation}
\end{defn}

See Figure \ref{fig: spectral position} for example of a secular
manifold colored according to the spectral position. One can see in
Figure \ref{fig: spectral position} that a line in $(0,2\pi)^\E$ connecting the origin to
a point on the $n=2$ layer will hit the $n=1$ layer once. In
the same way, such a line from the origin to a point with $n=3$ will first
cross the $n=1$ and $n=2$ layers.

\begin{figure}
\includegraphics[width=0.7\paperwidth]{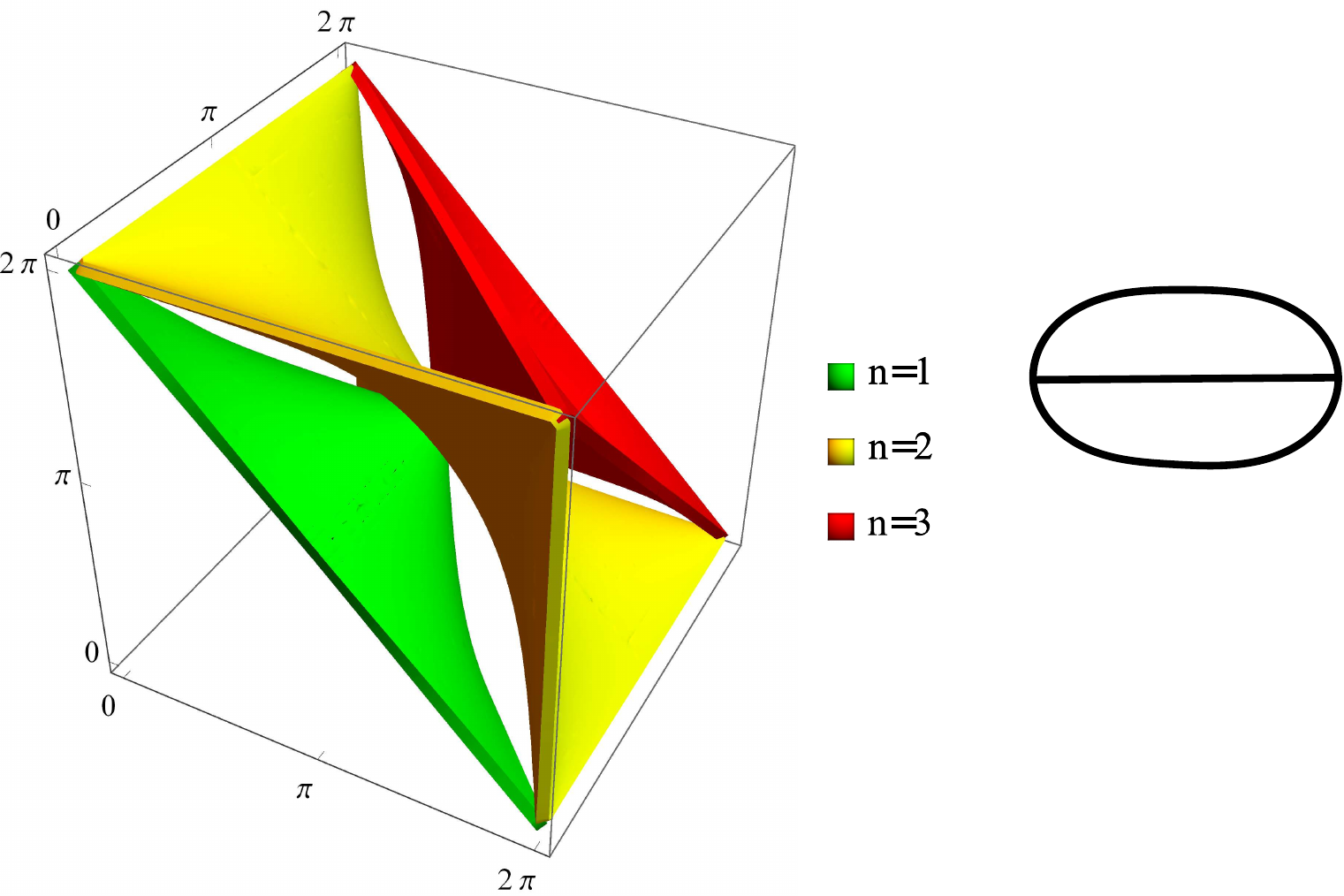}

\caption[Spectral position on the secular manifold]{\label{fig: spectral position} On the right, a `3-mandarin' graph
$\Gamma$. On the left, the secular manifold $\Sigma$, colored according
to the spectral position $n$. The first layer, $n=1$, in green.
The second layer, $n=2$, in yellow. The third layer, $n=3$, in red. }
\end{figure}

The torus $\T^{n}$ may be embedded into $\opcl{0,2\pi}^{n}$ or into
$\clop{0,2\pi}^{n}$ and the $l_{1}$ norm of a point in $\T^{n}$
will depend on the choice of embedding. To avoid obscurity, as we
will use both of these embeddings, we define the following. 
\begin{defn}
Define the two embeddings $r_{0}:\T\rightarrow\clop{0,2\pi}$ and
$r_{2\pi}:\T\rightarrow\opcl{0,2\pi}$ and extend them to $\T^{\E}$
by 
\[
r_{0}\left(\kv\right):=\sum_{e\in\E}r_{0}\left(\kappa_{e}\right),\,\,r_{2\pi}\left(\kv\right):=\sum_{e\in\E}r_{2\pi}\left(\kappa_{e}\right).
\]
\end{defn}

\begin{rem}
In particular, the total length of the standard graph $\Gamma_{\kv}$
is $r_{2\pi}\left(\kv\right)$.
\end{rem}

\subsection{\label{subsec: oscilatory part of the trace formula}The difference
$N\left(\Gamma_{\protect\lv},k\right)-\frac{L}{\pi}k$.}

Kottos and Smilansky derived the quantum graphs trace formula in \cite{KotSmi_prl97,KotSmi_ap99}
as the derivative of the trace formula of the spectral counting function
$k\mapsto N\left(\Gamma_{\lv},k\right)$. They showed that for a standard
graph $\Gamma_{\lv}$ of total length $L=\sum_{e\in\E}l_{e}$, $N\left(\Gamma_{\lv},k\right)$
is given in terms of the \emph{Weyl term $\frac{L}{\pi}k$ }and a
formal sum on the traces of the unitary evolution matrix $U_{\left\{ k\lv\right\} }$:\emph{
}
\begin{equation}
N\left(\Gamma_{\lv},k\right)=\frac{1}{2}+\frac{L}{\pi}k+\frac{1}{\pi}\Im\left[\sum_{n=1}^{\infty}\frac{1}{n}\mathrm{trace}\left(U_{\left\{ k\lv\right\} }^{n}\right)\right].
\end{equation}
The formal sum $\frac{1}{\pi}\Im\left[\sum_{n=1}^{\infty}\frac{1}{n}\mathrm{trace}\left(U_{\left\{ k\lv\right\} }^{n}\right)\right]$
should be evaluated at $k+i\epsilon$ with $\epsilon>0$, for which the sum converges, and then take the limit $\epsilon\rightarrow +0$. This sum is usually called the \emph{oscillatory part.} For convenience, we
will include the constant $\frac{1}{2}$ in our definition of the
oscillatory part: 
\begin{equation}
N_{osc}\left(\Gamma_{\lv},k\right):=N\left(\Gamma_{\lv},k\right)-\frac{L}{\pi}k.\label{eq: oscilatory part}
\end{equation}
Obviously, the summands of the formal sum depend only on $\kv=\left\{ k\lv\right\} $
and so $N_{osc}\left(\Gamma_{\lv},k_{n}\right)$ for simple $k_{n}$
should be a function on $\Sigma^{reg}$. We will present a proof that
$N_{osc}\left(\Gamma_{\lv},k_{n}\right)$ is a function of $\left\{ k_{n}\lv\right\} $
on $\Sigma^{reg}$, avoiding the formal sum, by following the method
of \cite{KotSmi_ap99} which is also presented in Section 5.1 in \cite{GnuSmi_ap06}. 
\begin{lem}
\label{lem: n minus weil}Let $\Gamma_{\lv}$ be a standard graph
of total length $L$ and let $k^{2}$ be a simple eigenvalue. Denote
$\kv=\left\{ k\lv\right\} $, then 
\begin{equation}
N_{osc}\left(\Gamma_{\lv},k\right)=n\left(\kv\right)-\frac{r_{2\pi}\left(\kv\right)}{\pi}.
\end{equation}
Moreover, the function $n\left(\kv\right)-\frac{r_{2\pi}\left(\kv\right)}{\pi}$
is continuous on $\Sigma^{reg}$ and is symmetric around $-\frac{\beta}{2}$
with respect to the inversion $\mathcal{I}\left(\kv\right)=\left\{ -\kv\right\} $.
Namely 
\[
n\left(\mathcal{I}\left(\kv\right)\right)-\frac{r_{2\pi}\left(\mathcal{I}\left(\kv\right)\right)}{\pi}=-\beta-\left(n\left(\kv\right)-\frac{r_{2\pi}\left(\kv\right)}{\pi}\right).
\]
\end{lem}

\begin{rem}
We will use the inversion notation $\I$ also for the entries of $\kv$,
namely $\I\left(\kappa_{e}\right):=\left\{ -\kappa_{e}\right\} $,
and for the diagonal exponent matrix $e^{i\hat{\kappa}}$ such that
$e^{i\I\hat{\kappa}}=e^{-i\hat{\kappa}}$.
\end{rem}

\begin{proof}
Let $U_{\kv}=e^{i\hat{\kappa}}S$ be the unitary evolution matrix
as defined in (\ref{eq: U=00003Dexp J S}), and consider the one parameter
family $t\mapsto U_{\left\{ t\lv\right\} }$ for $t>0$. Let $\left\{ e^{i\theta_{j}}\right\} _{j=1}^{2E}$
be a continuous choice of eigenvalues of $U_{\left\{ t\lv\right\} }$
with orthonormal eigenvectors $\left\{ a_{j}\right\} _{j=1}^{2E}$.
All $\left\{ e^{i\theta_{j}}\right\} _{j=1}^{2E}$ and $\left\{ a_{j}\right\} _{j=1}^{2E}$
depend smoothly on $t$ as long as the $\left\{ e^{i\theta_{j}}\right\} _{j=1}^{2E}$
eigenvalues are simple, and are continuous in $t$ when there are non-simple
eigenvalues. For readability we do not write the $t$ dependence explicitly.
Denote the tuple of angles by $\vec{\theta}\in\T^{2E}$ with derivatives
$\frac{d}{dt}\vec{\theta}\in\R^{2E}$. It is not hard to deduce from
(\ref{eq: U=00003Dexp J S}) that $\frac{d}{dt}U_{t\lv}=i\hat{L}U_{t\lv}$
where $\hat{L}$ is diagonal with entries $\left(\hat{L}\right)_{e,e}=\left(\hat{L}\right)_{\hat{e},\hat{e}}=l_{e}$.
We calculate $\frac{d}{dt}\theta_{j}$ by taking derivative of the
relation $\left(e^{i\theta_{j}}-U_{t\lv}\right)a_{j}=0$, which gives
\[
\left(i\frac{d}{dt}\theta_{j}e^{i\theta_{j}}-i\hat{L}U_{t\lv}\right)a_{j}+\left(e^{i\theta_{j}}-U_{t\lv}\right)\frac{d}{dt}a_{j}=0.
\]
Taking the inner product of the above equation with $a_{j}$ (where
$\left\langle *,*\right\rangle $ denote the standard $\C^{\vec{\E}}$
inner product) gives $\left\langle a_{j},\left(e^{i\theta_{j}}-U_{t\lv}\right)\frac{d}{dt}a_{j}\right\rangle =0$
for the right summand and therefore 
\[
\left\langle a_{j},i\frac{d}{dt}\theta_{j}e^{i\theta_{j}}a_{j}\right\rangle -\left\langle a_{j},i\hat{L}U_{t\lv}a_{j}\right\rangle =ie^{i\theta}\left(\frac{d}{dt}\theta_{j}-\left\langle a_{j},\hat{L}a_{j}\right\rangle \right)=0.
\]
We conclude that whenever it is defined (i.e $e^{i\theta_{j}}$ is
simple), $\frac{d}{dt}\theta_{j}$ is given by: 
\[
\frac{d}{dt}\theta_{j}=\left\langle a_{j},\hat{L}a_{j}\right\rangle =\sum_{e\in\E}l_{e}\left(\left|\left(a_{j}\right)_{e}\right|^{2}+\left|\left(a_{j}\right)_{\hat{e}}\right|^{2}\right)>0.
\]
In particular, if we denote $L_{min}=\min_{e\in\E}\left\{ l_{e}\right\} $
and $L_{max}=\max_{e\in\E}\left\{ l_{e}\right\} $, then the normalization
\[
\norm{a_{j}}^{2}=\sum_{e\in\E}\left|\left(a_{j}\right)_{e}\right|^{2}+\left|\left(a_{j}\right)_{\hat{e}}\right|^{2}=1,
\]
implies that 
\[
L_{min}\le\sum_{e\in\E}l_{e}\left(\left|\left(a_{j}\right)_{e}\right|^{2}+\left|\left(a_{j}\right)_{\hat{e}}\right|^{2}\right)\le L_{max}.
\]
Therefore, $\frac{d}{dt}\theta_{j}$ is strictly positive, and is
bounded by: 
\[
L_{min}\le\frac{d}{dt}\theta_{j}\le L_{max}.
\]
Moreover, as $\left\{ a_{j}\right\} _{j=1}^{2E}$ is an orthonormal
basis, then 
\begin{equation}
\sum_{j=1}^{2E}\frac{d}{dt}\theta_{j}=\sum_{j=1}^{2E}\left\langle a_{j},\hat{L}a_{j}\right\rangle =\mathrm{trace}\left(\hat{L}\right)=2L.\label{eq: sum of dtheth=00003D2L}
\end{equation}
If we now consider a simple eigenvalue $k^{2}$, then $N\left(\Gamma_{\lv},k\right)$
can be calculated by: 
\begin{align}
N\left(\Gamma_{\lv},k\right) & =\sum_{j=1}^{2E}\left|\set{0<t\le k}{e^{i\theta_{j}}|_{t}=1}\right|.
\end{align}
Since $\frac{d}{dt}\theta_{j}>0$, then each summand $\left|\set{0<t\le k}{e^{i\theta_{j}}|_{t}=1}\right|$
is given by integrating $\frac{d}{dt}\theta_{j}$ as follows. 
\[
\left|\set{0<t\le k}{e^{i\theta_{j}}|_{t}=1}\right|=\frac{r_{0}\left(\theta_{j}|_{0}\right)}{2\pi}+\frac{1}{2\pi}\int_{0}^{k}\frac{d}{dt}\theta_{j}dt-\frac{r_{0}\left(\theta_{j}|_{k}\right)}{2\pi},
\]
and summing over $j$, using $\sum_{j=1}^{2E}\frac{1}{2\pi}\int_{0}^{k}\frac{d}{dt}\theta_{j}dt=\frac{1}{2\pi}\int_{0}^{k}2Ldt=\frac{L}{\pi}k$,
gives:
\begin{align}
N\left(\Gamma_{\lv},k\right) & =\frac{L}{\pi}k+\sum_{j=1}^{2E}\frac{r_{0}\left(\theta_{j}|_{0}\right)}{2\pi}-\frac{r_{0}\left(\theta_{j}|_{k}\right)}{2\pi},\,\,\text{and so }\\
N_{osc}\left(\Gamma_{\lv},k\right) & =\frac{r_{0}\left(\tv|_{0}\right)}{2\pi}-\frac{r_{0}\left(\tv|_{k}\right)}{2\pi}.
\end{align}
Observe that $U_{t\lv}$ at $t=0$ is simply $S$, the real orthogonal
scattering matrix in Definition \ref{eq: S explicietly}, and therefore
its non-real eigenvalues come in conjugated pairs, so 
\[
\frac{r_{0}\left(\tv|_{0}\right)}{2\pi}=\frac{1}{2}\left|\set j{e^{i\theta_{j}}|_{0}\ne1}\right|=E-\frac{1}{2}\dim\ker\left(1-S\right).
\]
In \cite{KurNowCor_jpa06} (a correction to \cite{KurNow_jpa05})
it was shown that $\dim\ker\left(1-S\right)=E-V+2$, hence $\frac{r_{0}\left(\tv|_{0}\right)}{2\pi}=\frac{E+V}{2}-1$.
We may conclude that, 
\begin{equation}
N_{osc}\left(\Gamma_{\lv},k\right)=\frac{E+V}{2}-1-\frac{r_{0}\left(\tv|_{k}\right)}{2\pi}.\label{eq: N counting function}
\end{equation}
As $U_{t\lv}$ for $t=k$ is equal to $U_{\kv}$, then $\frac{r_{0}\left(\tv|_{k}\right)}{2\pi}$
is a function of $\kv$, and therefore $N_{osc}\left(\Gamma_{\lv},k\right)$
is a function of $\kv$, and so $N_{osc}\left(\Gamma_{\lv},k\right)=N_{osc}\left(\Gamma_{\kv},1\right)$
for any simple eigenvalue $k>0$. As the graph $\Gamma_{\kv}$ is
of total length $r_{2\pi}\left(\kv\right)$, we may conclude that
\begin{equation}
N_{osc}\left(\Gamma_{\lv},k\right)=n\left(\kv\right)-\frac{r_{2\pi}\left(\kv\right)}{\pi}.
\end{equation}
In order to prove that $n\left(\kv\right)-\frac{r_{2\pi}\left(\kv\right)}{\pi}$
is continuous over $\Sigma^{reg}$, we can use (\ref{eq: N counting function}),
denoting the angles of the eigenvalues of $U_{\kv}$ by $\tv|_{\kv}$:
\begin{equation}
\forall\kv\in\Sigma^{reg}\,\,\,\,n\left(\kv\right)-\frac{r_{2\pi}\left(\kv\right)}{\pi}=\frac{E+V}{2}-1-\frac{r_{0}\left(\tv|_{\kv}\right)}{2\pi}.\label{eq: n}
\end{equation}
So we need to prove that $r_{0}\left(\tv|_{\kv}\right)$ is continuous
on $\Sigma^{reg}$. Since $U_{\kv}$ is continuous in $\kv$, then
so does its eigenvalues. Therefore $\tv|_{\kv}:\Sigma^{reg}\rightarrow\T^{2E}$
is continuous. Let $\kv\in\Sigma^{reg}$ and number the eigenvalues
such that $e^{i\theta_{1}}|_{\kv}=1$ and hence $e^{i\theta_{j}}|_{\kv}\ne1$
for all $j>1$ (as $\dim\ker\left(1-U_{\kv}\right)$). By the continuity
of the eigenvalues, there is a neighborhood of $\kv$, $O\subset\Sigma^{reg}$,
such that $\forall\kv'\in O$,~$e^{i\theta_{1}}|_{\kv'}=1$ and $e^{i\theta_{j}}|_{\kv'}\ne1$
for all $j>1$. In particular, $r_{0}\left(\tv|_{\kv'}\right)=0+\sum_{j>1}r_{0}\left(\theta_{j}|_{\kv'}\right)$
in $O$, and each $r_{0}\left(\theta_{j}|_{\kv'}\right)$ is continuous
since $e^{i\theta_{j}}|_{\kv'}\ne1$. It follows that $r_{0}\left(\tv|_{\kv'}\right)$
is continuous in $O$ and therefore on $\Sigma^{reg}$ as $\kv$ was
arbitrary. This proves that $n\left(\kv\right)-\frac{r_{2\pi}\left(\tv|_{\kv}\right)}{2\pi}$
is continuous on $\Sigma^{reg}$ as needed.

To prove the symmetry recall that 
\[
U_{\mathcal{I}\left(\kv\right)}=e^{i\mathcal{I}\left(\hat{\kappa}\right)}S=e^{-i\hat{\kappa}}S=\overline{U_{\kv}},
\]
and so for any $\kv\in\Sigma^{reg}$, 
\begin{equation}
\frac{r_{0}\left(\tv|_{\I\left(\kv\right)}\right)}{2\pi}=\frac{1}{2\pi}\sum_{j=1}^{2E}\begin{cases}
0 & e^{i\theta_{j}}=1\\
2\pi-r_{0}\left(\theta_{j}|_{\kv'}\right) & e^{i\theta_{j}}\ne1
\end{cases}=2E-1-\frac{r_{0}\left(\tv|_{\kv}\right)}{2\pi}.
\end{equation}
Using (\ref{eq: n}) and $\beta=E-V+1$, we get the needed result:
\begin{align*}
n\left(\mathcal{I}\left(\kv\right)\right)-\frac{r_{2\pi}\left(\mathcal{I}\left(\kv\right)\right)}{\pi} & =\frac{E+V}{2}-1-\left(2E-1-\frac{r_{0}\left(\tv|_{\kv}\right)}{2\pi}\right)\\
= & -E+V-1-\left(\frac{E+V}{2}-1-\frac{r_{0}\left(\tv|_{\kv}\right)}{2\pi}\right)\\
= & -\beta-\left(n\left(\kv\right)-\frac{r_{2\pi}\left(\kv\right)}{\pi}\right).
\end{align*}
\end{proof}

\subsection{The differences $\phi\left(f\right)-\frac{L}{\pi}k$ and $\mu\left(f\right)-\frac{L}{\pi}k$}

Given a standard graph $\Gamma_{\lv}$ with a generic eigenfunction
$f$ of eigenvalue $k^{2}$, the nodal surplus and Neumann surplus
of $f$ are given by $\phi\left(f\right)-N\left(\Gamma_{\lv},k\right)$
and $\mu\left(f\right)-N\left(\Gamma_{\lv},k\right)$ as in Definition
\ref{def: Surpluses}. Since $N\left(\Gamma_{\lv},k\right)-\frac{L}{\pi}k$
was shown to be a continuous function on $\Sigma^{reg}$, we will
be able to show that the nodal and Neumann surplus are continuous
functions on $\Sigma_{\G}$ by proving that the differences $\phi\left(f\right)-\frac{L}{\pi}k$
and $\mu\left(f\right)-\frac{L}{\pi}k$ are continuous functions on
$\Sigma_{\G}$. 
\begin{lem}
\label{lem: counts minus weil}Let $\Gamma_{\lv}$ be a standard graph
of total length $L$ and let $f$ be a generic eigenfunction of eigenvalue
$k^{2}$. Denote $\kv=\left\{ k\lv\right\} \in\Sigma_{\G}$ and let
$f_{\kv}$ be the associated canonical eigenfunction. Then, 
\begin{align}
\phi\left(f\right)-\frac{L}{\pi}k & =\phi\left(f_{\kv}\right)-\frac{r_{2\pi}\left(\kv\right)}{\pi}\label{eq: nodal minus Weil}\\
\mu\left(f\right)-\frac{L}{\pi}k & =\mu\left(f_{\kv}\right)-\frac{r_{2\pi}\left(\kv\right)}{\pi}.\label{eq: Neumann minus weil}
\end{align}
Moreover, both $\phi\left(f_{\kv}\right)-\frac{r_{2\pi}\left(\kv\right)}{\pi}$
and $\mu\left(f_{\kv}\right)-\frac{r_{2\pi}\left(\kv\right)}{\pi}$
are continuous on $\Sigma^{\G}$, and satisfy the following symmetry:
\begin{align}
\phi\left(f_{\mathcal{I}\left(\kv\right)}\right)-\frac{r_{2\pi}\left(\mathcal{I}\left(\kv\right)\right)}{\pi} & =-\left(\phi\left(f_{\kv}\right)-\frac{r_{2\pi}\left(\kv\right)}{\pi}\right)\label{eq: symmetry of nodal minus weil}\\
\mu\left(f_{\mathcal{I}\left(\kv\right)}\right)-\frac{r_{2\pi}\left(\mathcal{I}\left(\kv\right)\right)}{\pi} & =-\left(\mu\left(f_{\kv}\right)-\frac{r_{2\pi}\left(\kv\right)}{\pi}\right)-\left|\partial\Gamma\right|.\label{eq: symmetry of Neumann minus weil}
\end{align}
\end{lem}

\begin{proof}
Denote the nodal and Neumann counts of the restriction $f|_{e}$ for
an edge $e$ by $\phi\left(f|_{e}\right)$ and $\mu\left(f|_{e}\right)$.
Since $f|_{e}\not\equiv0$ as it is generic, the by Definition \ref{def: edge restriction notations},
it is periodic on $e$ with period $\frac{2\pi}{k}$ and in every
interval $\opcl{a,b}$ of length $\frac{\pi}{k}$ it gets exactly
one nodal and one Neumann point. Therefore, both $\phi\left(f|_{e}\right)$
and $\mu\left(f|_{e}\right)$ are given by twice the number of periods,
$2\floor{\frac{kl_{e}}{2\pi}}$, plus a correction of either $0,1$
or $2$. In particular if the reminder (in half periods) $\frac{kl_{e}}{\pi}-2\floor{\frac{kl_{e}}{2\pi}}\le1$
then the corrections are either $0$ or $1$, and if $\frac{kl_{e}}{\pi}-2\floor{\frac{kl_{e}}{2\pi}}\ge1$
then they are either $1$ or $2$. 

Let $v,u$ be the vertices of $e$. To determine $\phi\left(f|_{e}\right)-2\floor{\frac{kl_{e}}{2\pi}}$,
notice that $f\left(v\right)f\left(u\right)<0$ implies $\phi\left(f|_{e}\right)$
is odd ,in which case $\phi\left(f|_{e}\right)=2\floor{\frac{kl_{e}}{2\pi}}+1$.
Otherwise $f\left(v\right)f\left(u\right)>0$ (by genericity) and
therefore $\phi\left(f|_{e}\right)$ is even, so $\phi\left(f|_{e}\right)-2\floor{\frac{kl_{e}}{2\pi}}$
is either $0$ or $2$. In such case 
\[
\phi\left(f|_{e}\right)=\begin{cases}
2\floor{\frac{kl_{e}}{2\pi}} & kl_{e}-2\pi\floor{\frac{kl_{e}}{2\pi}}\in\clop{0,\pi}\\
2\floor{\frac{kl_{e}}{2\pi}}+2 & kl_{e}-2\pi\floor{\frac{kl_{e}}{2\pi}}\in\left(\pi,2\pi\right)
\end{cases}.
\]
We excluded $kl_{e}-2\pi\floor{\frac{kl_{e}}{2\pi}}=\pi$ as it implies
$f\left(v\right)f\left(u\right)<0$. Using $\kv=\left\{ k\lv\right\} $
we can write $kl_{e}-2\pi\floor{\frac{kl_{e}}{2\pi}}=r_{0}\left(\kappa_{e}\right)$,
and given the canonical eigenfunction $f_{\kv}$ and Lemma \ref{lem: vertex values canonical ef and other ef }
we get $f\left(v\right)f\left(u\right)=f_{\kv}\left(v\right)f_{\kv}\left(u\right)$.
Therefore, the difference $\phi\left(f|_{e}\right)-\frac{kl_{e}}{\pi}$
is a $\kv$ dependent function: 
\begin{align}
\phi\left(f|_{e}\right)-\frac{kl_{e}}{\pi} & =\begin{cases}
-\frac{r_{0}\left(\kappa_{e}\right)}{\pi} & f_{\kv}\left(v\right)f_{\kv}\left(u\right)>0\,\,and\,\,r_{0}\left(\kappa_{e}\right)\in\clop{0,\pi}\\
1-\frac{r_{0}\left(\kappa_{e}\right)}{\pi} & f_{\kv}\left(v\right)f_{\kv}\left(u\right)<0\\
2-\frac{r_{0}\left(\kappa_{e}\right)}{\pi} & f_{\kv}\left(v\right)f_{\kv}\left(u\right)>0\,\,and\,\,r_{0}\left(\kappa_{e}\right)\in\left(\pi,2\pi\right)
\end{cases}.\label{eq: nodal minus weil explicit-1}
\end{align}
In particular, 
\begin{equation}
\phi\left(f|_{e}\right)-\frac{kl_{e}}{\pi}=\phi\left(f_{\kv}|_{e}\right)-\frac{r_{2\pi}\left(\kappa_{e}\right)}{\pi}.
\end{equation}
Let $M$ be a connected component of $\Sigma_{\G}$. According to
Lemma \ref{lem: vertex values of canonical function}, $f_{\kv}\left(v\right)f_{\kv}\left(u\right)$
is a real, continuous and non vanishing function on $M$ so it has
a fixed sign. The vertex values, as seen from Definition \ref{def: edge restriction notations},
satisfy $f_{\kv}\left(v\right)=f_{\kv}\left(u\right)$ in case that
$e^{i\kappa_{e}}=1$ and $f_{\kv}\left(v\right)=-f_{\kv}\left(u\right)$
if $e^{i\kappa_{e}}=-1$. It follows that if $f_{\kv}\left(v\right)f_{\kv}\left(u\right)<0$
on $M$, then $e^{i\kappa_{e}}\ne1$ on $M$ and therefore $\frac{r_{0}\left(\kappa_{e}\right)}{\pi}$
is continuous on $M$ and so does $\phi\left(f_{\kv}|_{e}\right)-\frac{r_{2\pi}\left(\kappa_{e}\right)}{\pi}=1-\frac{r_{0}\left(\kappa_{e}\right)}{\pi}$.
If $f_{\kv}\left(v\right)f_{\kv}\left(u\right)>0$ on $M$, then $e^{i\kappa_{e}}\ne-1$
and as $-\frac{r_{0}\left(\kappa_{e}\right)}{\pi}$ and $2-\frac{r_{0}\left(\kappa_{e}\right)}{\pi}$
agree when $e^{i\kappa_{e}}=1$ so the function 
\[
\phi\left(f_{\kv}|_{e}\right)-\frac{r_{2\pi}\left(\kappa_{e}\right)}{\pi}=\begin{cases}
-\frac{r_{0}\left(\kappa_{e}\right)}{\pi} & r_{0}\left(\kappa_{e}\right)\in\clop{0,\pi}\\
2-\frac{r_{0}\left(\kappa_{e}\right)}{\pi} & r_{0}\left(\kappa_{e}\right)\in\left(\pi,2\pi\right)
\end{cases}
\]
 is continuous on $M$. To prove the symmetry, consider the point
$\I\left(\kv\right)$ and notice that 
\[
r_{0}\left(\I\left(\kappa_{e}\right)\right)=\begin{cases}
2\pi-r_{0}\left(\kappa_{e}\right) & r_{0}\left(\kappa_{e}\right)\ne0\\
0 & r_{0}\left(\kappa_{e}\right)=0
\end{cases}
\]
 and $f_{\I\left(\kv\right)}\left(v\right)f_{\I\left(\kv\right)}\left(u\right)=f_{\kv}\left(v\right)f_{\kv}\left(u\right)$
according to Lemma \ref{lem: vertex values of canonical function}.
Substituting the latter into (\ref{eq: nodal minus weil explicit-1})
gives, 
\begin{align*}
\phi\left(f_{\I\left(\kv\right)}|_{e}\right)-\frac{r_{2\pi}\left(\I\left(\kappa_{e}\right)\right)}{\pi} & =\begin{cases}
0 & r_{0}\left(\kappa_{e}\right)=0\\
\frac{r_{0}\left(\kappa_{e}\right)}{\pi}-2 & f_{\kv}\left(v\right)f_{\kv}\left(u\right)>0\,\,and\,\,r_{0}\left(\I\left(\kappa_{e}\right)\right)\in\left(0,\pi\right)\\
\frac{r_{0}\left(\kappa_{e}\right)}{\pi}-1 & f_{\kv}\left(v\right)f_{\kv}\left(u\right)<0\\
\frac{r_{0}\left(\kappa_{e}\right)}{\pi} & f_{\kv}\left(v\right)f_{\kv}\left(u\right)>0\,\,and\,\,r_{0}\left(\I\left(\kappa_{e}\right)\right)\in\left(\pi,2\pi\right)
\end{cases}\\
= & -\left(\phi\left(f_{\kv}|_{e}\right)-\frac{r_{2\pi}\left(\kappa_{e}\right)}{\pi}\right).
\end{align*}
We may now sum over all edges to conclude that 
\[
\phi\left(f\right)-\frac{L}{\pi}k=\phi\left(f_{\kv}\right)-\frac{r_{2\pi}\left(\kv\right)}{\pi}.
\]
It is continuous on $\Sigma_{\G}$ and satisfy 
\[
\phi\left(f_{\I\left(\kv\right)}\right)-\frac{r_{2\pi}\left(\I\left(\kv\right)\right)}{\pi}=-\left(\phi\left(f_{\kv}\right)-\frac{r_{2\pi}\left(\kv\right)}{\pi}\right).
\]
The analysis of $\mu\left(f\right)-\frac{L}{\pi}k$ is in the same
spirit. To determine $\mu\left(f|_{e}\right)-2\floor{\frac{kl_{e}}{2\pi}}$
we consider first the case where both $v$ and $u$ are internal vertices.
In such case, $\partial_{e}f\left(v\right)\partial_{e}f\left(u\right)\ne0$
as $f$ is generic, and so if $\partial_{e}f\left(v\right)\partial_{e}f\left(u\right)>0$
then $\mu\left(f|_{e}\right)$ is odd ,in which case \[\mu\left(f|_{e}\right)=2\floor{\frac{kl_{e}}{2\pi}}+1.\]
Otherwise, $\partial_{e}f\left(v\right)\partial_{e}f\left(u\right)<0$
and 
\[
\mu\left(f|_{e}\right)=\begin{cases}
2\floor{\frac{kl_{e}}{2\pi}} & r_{0}\left(\kappa_{e}\right)\in\clop{0,\pi}\\
2\floor{\frac{kl_{e}}{2\pi}}+2 & r_{0}\left(\kappa_{e}\right)\in\left(\pi,2\pi\right)
\end{cases}.
\]
The same argument as before would give that $\mu\left(f_{\kv}|_{e}\right)-\frac{r_{2\pi}\left(\kappa_{e}\right)}{\pi}$
is a continuous function on $\Sigma_{\G}$ that satisfy,
\begin{align}
\mu\left(f|_{e}\right)-\frac{kl_{e}}{\pi} & =\mu\left(f_{\kv}|_{e}\right)-\frac{r_{2\pi}\left(\kappa_{e}\right)}{\pi},\,\,\,and\\
\mu\left(f_{\I\left(\kv\right)}|_{e}\right)-\frac{r_{2\pi}\left(\I\left(\kappa_{e}\right)\right)}{\pi} & =-\left(\mu\left(f_{\kv}|_{e}\right)-\frac{r_{2\pi}\left(\kappa_{e}\right)}{\pi}\right).
\end{align}
It is left to consider the case where one of the vertices, say $v$,
is a boundary vertex. In such case, $f|_{e}\left(x_{e}\right)=A_{e}\cos\left(kx_{e}\right)$
for $A_{e}\ne0$, so $\mu\left(f|_{e}\right)=\floor{\frac{kl_{e}}{\pi}}$
and the genericity gives $e^{ikl_{e}}\ne\pm1,\pm i$. It follows that
\begin{align*}
\mu\left(f|_{e}\right)-\frac{kl_{e}}{\pi} & =\begin{cases}
-\frac{r_{0}\left(\kappa_{e}\right)}{\pi} & r_{0}\left(\kappa_{e}\right)\in\left(0,\pi\right)\\
1-\frac{r_{0}\left(\kappa_{e}\right)}{\pi} & r_{0}\left(\kappa_{e}\right)\in\left(\pi,2\pi\right)
\end{cases},
\end{align*}
which a function of $\kv$ and therefore: 
\[
\mu\left(f|_{e}\right)-\frac{kl_{e}}{\pi}=\mu\left(f_{\kv}|_{e}\right)-\frac{r_{2\pi}\left(\kappa_{e}\right)}{\pi}.
\]
Since $e^{ikl_{e}}=e^{i\kappa_{e}}\ne\pm1$ for any $\kv\in\Sigma_{\G}$
then 
\[
\mu\left(f_{\kv}|_{e}\right)-\frac{r_{2\pi}\left(\kappa_{e}\right)}{\pi}=\begin{cases}
-\frac{r_{0}\left(\kappa_{e}\right)}{\pi} & r_{0}\left(\kappa_{e}\right)\in\left(0,\pi\right)\\
1-\frac{r_{0}\left(\kappa_{e}\right)}{\pi} & r_{0}\left(\kappa_{e}\right)\in\left(\pi,2\pi\right)
\end{cases}
\]
 is a continuous function on $\Sigma_{\G}$. Moreover, $e^{i\kappa_{e}}\ne\pm1$
implies that \[r_{0}\left(\I\left(\kappa_{e}\right)\right)=2\pi-r_{0}\left(\kappa_{e}\right),\ \text{and so,}\]
\begin{align*}
\mu\left(f_{\I\left(\kv\right)}|_{e}\right)-\frac{r_{2\pi}\left(\I\left(\kappa_{e}\right)\right)}{\pi} & =\begin{cases}
\frac{r_{0}\left(\kappa_{e}\right)}{\pi}-2 & r_{0}\left(\I\left(\kappa_{e}\right)\right)\in\left(0,\pi\right)\\
\frac{r_{0}\left(\kappa_{e}\right)}{\pi}-1 & r_{0}\left(\I\left(\kappa_{e}\right)\right)\in\left(\pi,2\pi\right)
\end{cases}\\
= & -\left(\mu\left(f_{\kv}|_{e}\right)-\frac{r_{2\pi}\left(\kappa_{e}\right)}{\pi}\right)-1.
\end{align*}
We may now sum over all edges to conclude that 
\[
\mu\left(f\right)-\frac{L}{\pi}k=\mu\left(f_{\kv}\right)-\frac{r_{2\pi}\left(\kv\right)}{\pi},
\]
where $\mu\left(f_{\kv}\right)-\frac{r_{2\pi}\left(\kv\right)}{\pi}$
is continuous on $\Sigma_{\G}$ and satisfies 
\[
\mu\left(f_{\I\left(\kv\right)}\right)-\frac{r_{2\pi}\left(\I\left(\kv\right)\right)}{\pi}=-\left(\mu\left(f_{\kv}\right)-\frac{r_{2\pi}\left(\kv\right)}{\pi}\right)-\left|\partial\Gamma\right|.
\]
\end{proof}

\subsection{Proof of Theorem \ref{thm:First}.}
\begin{proof}
Let $\Gamma_{\lv}$ be a standard graph with $\lv$ rationally independent
and let $\mathcal{F_{\G}}$ as in Definition \ref{def: sigma-algebra}.
According to Corollary \ref{cor: key lemma in equivalence of genericity}, $\Sigma_{\G}$ is open and Jordan and
therefore each of its connected components $\Sigma_{\G,i}$ is also
open and Jordan (as their boundary are included in its boundary).
As $\mu_{\lv}$ is strictly positive (Remark \ref{rem: BG measuer positivity})
then $\mu_{\lv}\left(\Sigma_{\G,i}\right)>0$ for every component.
We may use Theorem \ref{thm: BG equidistribution} to conclude that
every atom $\G_{i}$ has density, 
\begin{equation}
d\left(\G_{i}\right)=\mu_{\lv}\left(\Sigma_{\G,i}\right),\,\,\text{and }\,\,d_{\G}\left(\G_{i}\right)=\frac{\mu_{\lv}\left(\Sigma_{\G,i}\right)}{\mu_{\lv}\left(\Sigma_{\G}\right)}.
\end{equation}
The same reason gives that for any union of atoms $\sqcup_{i\in I}\G_{i}$,
\[
d_{\G}\left(\sqcup_{i\in I}\G_{i}\right)=\frac{\mu_{\lv}\left(\sqcup_{i\in I}\Sigma_{\G,i}\right)}{\mu_{\lv}\left(\Sigma_{\G}\right)}=\sum_{i\in I}\frac{\mu_{\lv}\left(\Sigma_{\G,i}\right)}{\mu_{\lv}\left(\Sigma_{\G}\right)}=\sum_{i\in I}d_{\G}\left(\G_{i}\right).
\]
Remark \ref{rem: elements in the sigma algebra} insure that every
element in $\mathcal{F_{\G}}$ is a union of atoms, and therefore
has density. The above proves that $d_{\G}$ is $\sigma$-additive
on $\mathcal{F_{\G}}$ and is therefor a probability measure, proving
(\ref{enu: triplet}). 

Let $f_{m}$ be the $m$'s eigenfunction of $\Gamma_{\lv}$ with eigenvalue
$k_{m}^{2}$ and assume that $f_{m}$ is generic. Let $\kv=\left\{ k_{m}\lv\right\} $
and denote the nodal surplus and Neumann surplus of the canonical
eigenfunction $f_{\kv}$ by 
\begin{align}
\boldsymbol{\sigma}\left(\kv\right) & :=\phi\left(f_{\kv}\right)-n\left(\kv\right)\\
\boldsymbol{\omega}\left(\kv\right) & :=\mu\left(f_{\kv}\right)-n\left(\kv\right).
\end{align}
Using both Lemma \ref{lem: n minus weil} and Lemma \ref{lem: counts minus weil}
we get that the nodal and Neumann surplus of $f_{m}$ are given by:
\begin{align}
\sigma\left(m\right) & =\phi\left(f_{m}\right)-\frac{L}{\pi}k_{m}+\frac{L}{\pi}k_{m}-m=\boldsymbol{\sigma}\left(\kv\right),\,\,and\\
\omega\left(m\right) & =\mu\left(f_{m}\right)-\frac{L}{\pi}k_{m}+\frac{L}{\pi}k_{m}-m=\boldsymbol{\omega}\left(\kv\right).
\end{align}
So, 
\begin{align}
\sigma^{-1}\left(j\right) & =\set{n\in\G}{\left\{ k_{n}\lv\right\} \in\boldsymbol{\sigma}^{-1}\left(j\right)},\label{eq: nodal surp level set}\\
\omega^{-1}\left(i\right) & =\set{n\in\G}{\left\{ k_{n}\lv\right\} \in\boldsymbol{\omega}^{-1}\left(i\right)},\,\,\,and\label{eq: Neumann surp level set}\\
\sigma^{-1}\left(j\right)\cap\omega^{-1}\left(i\right) & =\set{n\in\G}{\left\{ k_{n}\lv\right\} \in\boldsymbol{\sigma}^{-1}\left(j\right)\cap\boldsymbol{\omega}^{-1}\left(i\right)}.\label{eq: joint level set}
\end{align}
Moreover, using Lemma \ref{lem: n minus weil} and Lemma \ref{lem: counts minus weil}
we can deduce that $\boldsymbol{\sigma}\left(\kv\right)$ and $\boldsymbol{\omega}\left(\kv\right)$
are continuous (integer values) functions on $\Sigma_{\G}$ and so
constant on connected components. Therefore, their level sets $\boldsymbol{\sigma}^{-1}\left(j\right)$
and $\boldsymbol{\omega}^{-1}\left(i\right)$ are unions of connected
components, and by (\ref{eq: nodal surp level set}), (\ref{eq: Neumann surp level set})
and (\ref{eq: joint level set}), the level sets $\sigma^{-1}\left(j\right),\,\omega^{-1}\left(i\right)$
and $\sigma^{-1}\left(j\right)\cap\omega^{-1}\left(i\right)$ are
unions of atoms in $\mathcal{F_{\G}}$. This proves that $\sigma$
and $\omega$ are random variables on $\left(\G,\mathcal{F_{\G}},d_{\G}\right)$,
with probabilities given by 
\begin{align*}
P\left(\sigma=j\right):=d_{\G}\left(\sigma^{-1}\left(j\right)\right) & =\frac{\mu_{\lv}\left(\boldsymbol{\sigma}^{-1}\left(j\right)\right)}{\mu_{\lv}\left(\Sigma_{\G}\right)},\\
P\left(\omega=i\right):=d_{\G}\left(\omega^{-1}\left(i\right)\right) & =\frac{\mu_{\lv}\left(\boldsymbol{\omega}^{-1}\left(i\right)\right)}{\mu_{\lv}\left(\Sigma_{\G}\right)},\,\,\,\text{and}\\
P\left(\sigma=j\,\wedge\,\omega=i\right)=d_{\G}\left(\sigma^{-1}\left(j\right)\cap\omega^{-1}\left(i\right)\right) & =\frac{\mu_{\lv}\left(\boldsymbol{\sigma}^{-1}\left(j\right)\cap\boldsymbol{\omega}^{-1}\left(i\right)\right)}{\mu_{\lv}\left(\Sigma_{\G}\right)},
\end{align*}
which proves part (\ref{enu: sigma and omega random variables}) of
the theorem. As we showed that each joint level set of the form $\sigma^{-1}\left(j\right)\cap\omega^{-1}\left(i\right)$
is a union of atoms and each atom has positive density, then the level
set is either empty or with positive density, proving part (\ref{enu: levels sets are either ampty or of positive measure})
of the theorem. 

To prove the symmetry we use Lemma \ref{lem: n minus weil} and Lemma
\ref{lem: counts minus weil}: 
\begin{align*}
\boldsymbol{\sigma}\left(\I\left(\kv\right)\right) & =\phi\left(f_{\I\left(\kv\right)}\right)-\frac{r_{2\pi}\left(\I\left(\kv\right)\right)}{\pi}-\left(n\left(\I\left(\kv\right)\right)-\frac{r_{2\pi}\left(\I\left(\kv\right)\right)}{\pi}\right)\\
= & -\left(\phi\left(f_{\kv}\right)-\frac{r_{2\pi}\left(\kv\right)}{\pi}\right)+\beta+\left(n\left(\kv\right)-\frac{r_{2\pi}\left(\kv\right)}{\pi}\right)\\
= & \beta-\boldsymbol{\sigma}\left(\kv\right).
\end{align*}
And in the same way, 
\begin{align*}
\boldsymbol{\omega}\left(\I\left(\kv\right)\right) & =\mu\left(f_{\I\left(\kv\right)}\right)-\frac{r_{2\pi}\left(\I\left(\kv\right)\right)}{\pi}-\left(n\left(\I\left(\kv\right)\right)-\frac{r_{2\pi}\left(\I\left(\kv\right)\right)}{\pi}\right)\\
= & -\left(\mu\left(f_{\kv}\right)-\frac{r_{2\pi}\left(\kv\right)}{\pi}\right)-\left|\partial\Gamma\right|+\beta+\left(n\left(\kv\right)-\frac{r_{2\pi}\left(\kv\right)}{\pi}\right)\\
= & \beta-\left|\partial\Gamma\right|-\boldsymbol{\omega}\left(\kv\right).
\end{align*}
Using Theorem \ref{thm: inversion BG}, 
\[
d_{\G}\left(\sigma^{-1}\left(j\right)\cap\omega^{-1}\left(i\right)\right)=d_{\G}\left(\sigma^{-1}\left(\beta-j\right)\cap\omega^{-1}\left(\beta-\left|\partial\Gamma\right|-i\right)\right),
\]
which proves (\ref{enu: symmetric simultaniuosly}). 
\end{proof}

\newpage{}
\section{\label{sec: Properties-of-a Neumann domain} Properties of a Neumann
domain}

In this section we analyze the properties of a Neumann domain. We
refer to such properties as ``local'' versus the ``global'' properties
like the nodal and Neumann counts. The properties we investigate are
the spectral position of a Neumann domain, which was discussed in
the introduction, and a normalized total length of a Neumann domain,
which we call \emph{wavelength capacity}. Such a normalized total
length, but for nodal domains, was mentioned in \cite{GnuSmiWeb_wrm04}
where numerical examples of its statistics were provided to justify
the universality conjecture. The results of this section appear both
in \cite{AloBan19} and in \cite{AloBanBerEgg_Neumann} where they
are compared to their analogs on manifolds. 
\begin{defn}
\label{def: local observables} Let $\Gamma_{\lv}$ be a standard
graph, let $f$ be a generic eigenfunction of eigenvalue $k^{2}$
and let $\Omega$ be a Neumann domain of $f$. Let $L_{\Omega}$ denote
the total length of $\Omega$. We define the \emph{spectral position
}of $\Omega$ by $N\left(\Omega\right):=N\left(\Omega,k^{2}\right).$
We define the \emph{wavelength capacity} of $\Omega$ by
\begin{equation}
\rho\left(\Omega\right):=\frac{L_{\Omega}}{\pi}k.
\end{equation}
\end{defn}

\begin{rem}
The wavelength capacity is the Weyl term (at $k$) of the trace formula
of $N\left(\Omega,k^{2}\right)$ (see discussion in Subsection \ref{subsec: oscilatory part of the trace formula}). 
\end{rem}

Next we present some immediate properties of Neumann domains, some
of which are straightforward and some requires a short proof. 
\begin{lem}
\label{lem: local properties of Neumann domain}Let $\Gamma_{\lv}$
be a standard graph of total length $L$ and with minimal edge length
$L_{min}$. Let $f$ be a generic eigenfunction of $\Gamma_{\lv}$
with eigenvalue $k^{2}>0$ and let $\Omega$ be a Neumann domain of
$f$. Then,
\begin{enumerate}
\item \label{enu: local prop - restriction function}$k^{2}$ is an eigenvalue
of $\Omega$ with eigenfunction $f|_{\Omega}$. Moreover, $f|_{\Omega}$
satisfies properties I and II and is generic if and only if $k^{2}$
is a simple eigenvalue of $\Omega$.
\item \label{enu: local prop - segment}If $\Omega$ is a segment, then $\rho\left(\Omega\right)=N\left(\Omega\right)=1$. 
\item \label{enu: kl in () cup ()}If $\Omega$ is not a segment and $e$
is an edge of $\Omega$ of length $l_{e}$, then $kl_{e}<\pi$, and
if $e$ is a tail of $\Omega$, then $kl_{e}\in\left(0,\frac{\pi}{2}\right)\cap\left(\frac{\pi}{2},\pi\right)$.
In particular, $\rho\left(\Omega\right)<E_{\Omega}$ where $E_{\Omega}$
is the number of edges of $\Omega$. 
\item \label{enu: local prop -  bound on small eigenvalues}There are at
most $2\frac{L}{L_{min}}$ eigenvalues of $\Gamma_{\lv}$ for which
$k\le\frac{\pi}{L_{min}}$. 
\item \label{enu: local prop - star}If $k>\frac{\pi}{L_{min}}$, then $f$
has two kinds of Neumann domains, star graphs and segments. Every
interior vertex $v\in\V_{in}$ is the central vertex of a star Neumann
domain which we denote by $\Omega^{\left(v\right)}$. The rest of
the Neumann domains are segments. In both cases, $f|_{\Omega}$ is
generic with nodal count $\phi\left(f|_{\Omega}\right)=N\left(\Omega\right)$.
\end{enumerate}
\end{lem}

\begin{proof}
We first prove (\ref{enu: local prop - restriction function}). It
is clear that $f|_{\Omega}$ satisfies $f|_{\Omega}''=-k^{2}f|_{\Omega}$
on every edge of $\Omega$ and satisfies Neumann vertex condition
on internal vertices of $\Omega$ as these are internal vertices of
$\Gamma$. By the same reason $f|_{\Omega}$ satisfies property II.
The boundary vertices of $\Omega$ are either boundary vertices of
$\Gamma$ or Neumann points of $f$. In both cases the derivative
of $f|_{\Omega}$ vanish and therefore $f|_{\Omega}$ satisfies Neumann
vertex conditions and so $f|_{\Omega}\in Eig\left(\Omega,k^{2}\right)$
as needed. It is left to prove that $f|_{\Omega}$ satisfies property
I, namely does not vanish on vertices. It does not vanish on interior
vertices of $\Omega$ as these are also vertices of $\Gamma$. As
seen from Definition \ref{def: edge restriction notations}, an eigenfunction
that vanish on a boundary vertex must vanish on the entire tail connected
to it. This cannot happen as $f$ is generic, so $f|_{\Omega}$ satisfies
property I.

If $\Omega$ is a segment of length $L_{\Omega}$, then it has only
one (Neumann) eigenfunction with no Neumann points. It is that of
the first positive eigenvalue $k=\frac{\pi}{L_{\Omega}}$, so $N\left(\Omega\right)=1$
and $\rho\left(\Omega\right)=\frac{kL_{\Omega}}{\pi}=1$, proving
(\ref{enu: local prop - segment}).

Assume that $\Omega$ is not a segment, and let $e$ be an edge of
$\Omega$ of length $l_{e}$. As discussed in the proof of Lemma \ref{lem: counts minus weil},
in any interval $\opcl{a,b}$ of length $\frac{\pi}{k}$ inside $\Gamma_{\lv}$
there is one Neumann point, so $l_{e}\le\frac{\pi}{k}$. If $l_{e}=\frac{\pi}{k}$
then either there is a Neumann point inside $e$ or the derivative
vanish on both vertices of $e$. Since $e$ has no Neumann points
and has at least one vertex on which the derivative does not vanish
($f|_{\Omega}$ satisfies property $II$), then $l_{e}<\frac{\pi}{k}$.
The total length of $\Omega$ is therefore bounded by $L_{\Omega}<E_{\Omega}\frac{\pi}{k}$
so $\rho\left(\Omega\right)<E_{\Omega}$. If $e$ is a tail, then (using
Definition \ref{def: edge restriction notations}) $f|_{e}\left(x_{e}\right)=A_{e}\cos\left(kx_{e}\right)$,
so by properties II, $kl_{e}\ne\frac{\pi}{2}$. We thus proved (\ref{enu: kl in () cup ()}).

The number of eigenvalues for which $k\le\frac{\pi}{L_{min}}$ is
exactly $N\left(\Gamma_{\lv},\frac{\pi}{L_{min}}\right)$. Friedlander
proved in \cite[Theorem 1]{Fri_aif05} that, 
\begin{align}
k & \ge\frac{\pi\left(N\left(\Gamma_{\lv},k\right)+1\right)}{2L},\,\,\text{and so}\label{eq: Friedlander result}\\
\frac{2kL}{\pi}-1 & \ge N\left(\Gamma_{\lv},k^{2}\right).
\end{align}
This proves (\ref{enu: local prop -  bound on small eigenvalues})
by substituting $k=\frac{\pi}{L_{min}}$:
\[
N\left(\Gamma_{\lv},\frac{\pi}{L_{min}}\right)\le\frac{2L}{L_{min}}-1<\frac{2L}{L_{min}}.
\]
If $k>\frac{\pi}{L_{min}}$, then each edge $e$ of $\Gamma_{\lv}$
has length $l_{e}\ge L_{min}>\frac{\pi}{k}$ and according to the
arguments above, there is at least one Neumann point in $e$. It is
not hard to deduce that in such case every interior vertex is contained
in a star Neumann domain and all others are segments. We may notice
that both cases, star and segment, are trees. According to Corollary
3.1.9 in \cite{BerKuc_graphs}, an eigenfunction on a tree satisfies
property I if and only if its eigenvalue is simple. Therefore $k$
is a simple eigenvalue of $\Omega$ and so $f|_{\Omega}$ is generic.
Since $\Omega$ is a tree, then the nodal surplus bounds in (\ref{eq: nodal surplus bound})
implies $\phi\left(f|_{\Omega}\right)=N\left(\Omega\right)$. We thus
proved (\ref{enu: local prop - star}).
\end{proof}
Our discussion is oriented to statistical behaviour and is thus insensitive
to properties of low eigenvalues. As seen in the above lemma, for
high enough eigenvalues, all Neumann domains are either trivial (segments)
or star graphs which we label as $\Omega^{\left(v\right)}$ according
to their internal vertex $v$. We will now relate the properties of
these star Neumann domains to functions on the secular manifold. 
\begin{defn}
\label{def: Secular N and rho}Let $\Gamma$ be a graph and let $\Sigma_{\G}$
be the generic part of its secular manifold. For any interior vertex
$v\in\V_{in}$ we define the following functions on $\Sigma_{\G}$:
\begin{align*}
\boldsymbol{N}_{v}\left(\kv\right) & :=\frac{\deg v}{2}-\frac{1}{2}\sum_{e\in\E_{v}}\mathrm{sign}\left(f_{\kv}\left(v\right)\partial_{e}f_{\kv}\left(v\right)\right)\\
\boldsymbol{\rho}_{v}\left(\kv\right) & :=\frac{1}{\pi}\sum_{e\in\E_{v}}\tan^{-1}\left(\frac{f_{\kv}\left(v\right)\partial_{e}f_{\kv}\left(v\right)}{\left(f_{\kv}\left(v\right)\right)^{2}}\right).
\end{align*}
Where we consider $\tan^{-1}:\R\setminus\left\{ 0\right\} \rightarrow\left(0,\frac{\pi}{2}\right)\cup\left(\frac{\pi}{2},\pi\right)$.
\end{defn}

\begin{lem}
\label{lem: Nv rhov-1}Let $\Gamma_{\lv}$ be a standard graph, let
$f$ be a generic eigenfunction of eigenvalue $k^{2}$ and denote
the point $\kv=\left\{ k\lv\right\} \in\Sigma_{\G}$. Then for any
interior vertex $v\in\V_{in}$,
\begin{enumerate}
\item \label{enu: N=00003DN and rho=00003Drho}If $\Omega^{\left(v\right)}$,
the Neumann domain of $f$ that contains $v$, is a star graph, then
\begin{align*}
N\left(\Omega^{\left(v\right)}\right) & =\boldsymbol{N}_{v}\left(\kv\right),\\
\rho\left(\Omega^{\left(v\right)}\right) & =\boldsymbol{\rho}_{v}\left(\kv\right).
\end{align*}
\item \label{enu: N is constant}The function $\boldsymbol{N}_{v}$ is constant
on connected components of $\Sigma^{\G}$, and satisfies 
\[
\boldsymbol{N}_{v}\left(\mathcal{I}\left(\kv\right)\right)=\deg v-\boldsymbol{N}_{v}\left(\kv\right).
\]
\item \label{enu: rho is real analytic}The function $\boldsymbol{\rho}_{v}$
is real analytic on $\Sigma^{\G}$, and satisfies 
\begin{align*}
\boldsymbol{\rho}_{v}\left(\mathcal{I}\left(\kv\right)\right) & =\deg v-\boldsymbol{\rho}_{v}\left(\kv\right)
\end{align*}
\end{enumerate}
\end{lem}

\begin{proof}
Both $f_{\kv}\left(v\right)\partial_{e}f_{\kv}\left(v\right)$ and
$\left(f_{\kv}\left(v\right)\right)^{2}$ does not vanish on $\Sigma^{\G}$
by Definition \ref{def: Secular subsets}, and according to Lemma
\ref{lem: vertex values of canonical function}, are real analytic
that satisfy: 
\begin{align*}
f_{\mathcal{I}\left(\kv\right)}\left(v\right)f_{\mathcal{I}\left(\kv\right)}\left(v\right) & =f_{\kv}\left(v\right)f_{\kv}\left(v\right)\\
f_{\mathcal{I}\left(\kv\right)}\left(v\right)\partial_{e}f_{\mathcal{I}\left(\kv\right)}\left(v\right) & =-f_{\kv}\left(v\right)\partial_{e}f_{\kv}\left(v\right).
\end{align*}
Each of the terms $\mathrm{sign}\left(f_{\kv}\left(v\right)\partial_{e}f_{\kv}\left(v\right)\right)$
is therefore constant on connected components of $\Sigma^{\G}$ and
satisfies $\mathrm{sign}\left(f_{\mathcal{I}\left(\kv\right)}\left(v\right)\partial_{e}f_{\mathcal{I}\left(\kv\right)}\left(v\right)\right)=-\mathrm{sign}\left(f_{\kv}\left(v\right)\partial_{e}f_{\kv}\left(v\right)\right)$.
It follows that $\boldsymbol{N}_{v}\left(\kv\right)$ is constant
on connected components of $\Sigma^{\G}$ and satisfies 
\begin{align*}
\boldsymbol{N}_{v}\left(\mathcal{I}\left(\kv\right)\right)= & \frac{\deg v}{2}-\frac{1}{2}\sum_{e\in\E_{v}}\mathrm{sign}\left(f_{\mathcal{I}\left(\kv\right)}\left(v\right)\partial_{e}f_{\mathcal{I}\left(\kv\right)}\left(v\right)\right)\\
= & \frac{\deg v}{2}+\frac{1}{2}\sum_{e\in\E_{v}}\mathrm{sign}\left(f_{\kv}\left(v\right)\partial_{e}f_{\kv}\left(v\right)\right)\\
= & \deg v-\boldsymbol{N}_{v}\left(\kv\right).
\end{align*}
Similarly, each of the terms $\frac{f_{\kv}\left(v\right)\partial_{e}f_{\kv}\left(v\right)}{\left(f_{\kv}\left(v\right)\right)^{2}}$
is well defined, real analytic and non vanishing on $\Sigma_{\G}$.
As $\tan^{-1}:\R\setminus\left\{ 0\right\} \rightarrow\left(0,\frac{\pi}{2}\right)\cup\left(\frac{\pi}{2},\pi\right)$
is real analytic, then $\tan^{-1}\left(\frac{f_{\kv}\left(v\right)\partial_{e}f_{\kv}\left(v\right)}{\left(f_{\kv}\left(v\right)\right)^{2}}\right)$
is real analytic on $\Sigma^{\G}$. Notice that this choice of $\tan^{-1}$
satisfies $\tan^{-1}\left(-x\right)=\pi-\tan^{-1}\left(x\right)$
for any real $x\ne0$, and therefore: 
\begin{align*}
\tan^{-1}\left(\frac{f_{\mathcal{I}\left(\kv\right)}\left(v\right)\partial_{e}f_{\mathcal{I}\left(\kv\right)}\left(v\right)}{\left(f_{\mathcal{I}\left(\kv\right)}\left(v\right)\right)^{2}}\right) & =\pi-\tan^{-1}\left(\frac{f_{\kv}\left(v\right)\partial_{e}f_{\kv}\left(v\right)}{\left(f_{\kv}\left(v\right)\right)^{2}}\right).
\end{align*}
Hence, $\boldsymbol{\rho}_{v}\left(\kv\right)$ is real analytic on
$\Sigma^{\G}$ and satisfies 
\begin{align*}
\boldsymbol{\rho}_{v}\left(\mathcal{I}\left(\kv\right)\right)= & \frac{1}{\pi}\sum_{e\in\E_{v}}\tan^{-1}\left(\frac{f_{\mathcal{I}\left(\kv\right)}\left(v\right)\partial_{e}f_{\mathcal{I}\left(\kv\right)}\left(v\right)}{\left(f_{\mathcal{I}\left(\kv\right)}\left(v\right)\right)^{2}}\right)\\
= & \deg v-\frac{1}{\pi}\sum_{e\in\E_{v}}\tan^{-1}\left(\frac{f_{\kv}\left(v\right)\overline{\partial_{e}f_{\kv}\left(v\right)}}{f_{\kv}\left(v\right)\overline{f_{\kv}\left(v\right)}}\right)\\
= & \deg v-\boldsymbol{\rho}_{v}\left(\kv\right).
\end{align*}
To prove (\ref{enu: N=00003DN and rho=00003Drho}) let $f,k$ and
$\Omega^{\left(v\right)}$ be as stated in (\ref{enu: N=00003DN and rho=00003Drho}).
Let $e\in\E_{v}$ be directed from $v$ outwards and let $f_{e}|\left(x_{e}\right)=A_{e}\cos\left(kx_{e}-\varphi_{e}\right)$
according to Definition \ref{def: edge restriction notations}, where
$A_{e}\ne0$ and we may assume that $\varphi_{e}\in\opcl{0,\pi}$.
Let $l_{\tilde{e}}$ be the length of $\tilde{e}$, the corresponding
edge of $\Omega^{v}$. It is given by: 
\begin{align*}
l_{\tilde{e}}= & \min\set{x_{e}\in\left(0,l_{e}\right)}{f'|_{e}\left(x_{e}\right)=0}\\
= & \set{x_{e}\in\left(0,\frac{\pi}{k}\right)}{-kA_{e}\sin\left(kx_{e}-\varphi_{e}\right)=0}\\
= & \frac{\varphi_{e}}{k}.
\end{align*}
Therefore, 
\[
\tan\left(kl_{\tilde{e}}\right)=\tan\left(\varphi_{e}\right)=\frac{1}{k}\frac{f_{e}|'\left(0\right)}{f_{e}|\left(0\right)}=\frac{1}{k}\frac{f\left(v\right)\partial_{e}f\left(v\right)}{\left(f\left(v\right)\right)^{2}}.
\]
If we denote $\kv=\left\{ k\lv\right\} $ and use Lemma (\ref{lem: vertex values canonical ef and other ef }),
the equation above yields 
\begin{align}
\tan\left(kl_{\tilde{e}}\right) & =\frac{f_{\kv}\left(v\right)\partial_{e}f_{\kv}\left(v\right)}{\left(f_{\kv}\left(v\right)\right)^{2}},\,\,and
\end{align}
\begin{equation}
kl_{\tilde{e}}=\tan^{-1}\left(\frac{f_{\kv}\left(v\right)\partial_{e}f_{\kv}\left(v\right)}{f_{\kv}\left(v\right)f_{\kv}\left(v\right)}\right).\label{eq: kle =00003D atan-1}
\end{equation}
Where, as stated already, we consider $\tan^{-1}:\R\setminus\left\{ 0\right\} \rightarrow\left(0,\frac{\pi}{2}\right)\cup\left(\frac{\pi}{2},\pi\right)$,
and use the fact that $kl_{\tilde{e}}\in\left(0,\frac{\pi}{2}\right)\cup\left(\frac{\pi}{2},\pi\right)$
(see Lemma \ref{lem: local properties of Neumann domain} (\ref{enu: kl in () cup ()})).
Summing (\ref{eq: kle =00003D atan-1}) over all edges of $\Omega^{\left(v\right)}$
proves that $\rho\left(\Omega^{\left(v\right)}\right)=\boldsymbol{\rho}_{v}\left(\kv\right)$. 

Since $kl_{\tilde{e}}\in\left(0,\frac{\pi}{2}\right)\cup\left(\frac{\pi}{2},\pi\right)$
then there is at most one nodal point in $\tilde{e}$ and a simple
check gives: 
\begin{equation}
\phi\left(f|_{\tilde{e}}\right)=\begin{cases}
0 & f_{\kv}\left(v\right)\partial_{e}f_{\kv}\left(v\right)>0\\
1 & f_{\kv}\left(v\right)\partial_{e}f_{\kv}\left(v\right)<0
\end{cases}=\frac{1-\mathrm{sign}\left(f_{\kv}\left(v\right)\partial_{e}f_{\kv}\left(v\right)\right)}{2}.
\end{equation}
 Summing over all $\phi\left(f|_{\tilde{e}}\right)$ gives,
\[
\phi\left(f|_{\Omega}\right)=\sum_{e\in\E_{v}}\frac{1-\mathrm{sign}\left(f_{\kv}\left(v\right)\partial_{e}f_{\kv}\left(v\right)\right)}{2}=\boldsymbol{N}_{v}\left(\kv\right).
\]
This proves $N\left(\Omega^{\left(v\right)}\right)=\phi\left(f|_{\Omega}\right)=\boldsymbol{N}_{v}\left(\kv\right)$,
as needed.
\end{proof}
\begin{prop}
\label{prop:local_observables_bounds} Let $\Gamma_{\lv}$ be a standard
graph, let $f$ be a generic eigenfunction with eigenvalue $k^{2}$
and let $\Omega$ be a Neumann domain which is a star graph with central
vertex $v\in\V_{in}$. Then,
\begin{align}
1\le & N\left(\Omega\right)\le\deg v-1\label{eq:Spectral_Position_bounds}\\
1\leq\frac{N\left(\Omega\right)+1}{2}\le & \rho\left(\Omega\right)\le\frac{N\left(\Omega\right)+\deg v-1}{2}\leq\deg v-1\label{eq:Rho_bounds}
\end{align}
\end{prop}

\begin{proof}
Using Lemma \ref{lem: Nv rhov-1} we may prove Proposition \ref{prop:local_observables_bounds}
by proving that the bounds (\ref{eq:Spectral_Position_bounds}) and
(\ref{eq:Rho_bounds}) hold for $\boldsymbol{N}_{v}\left(\kv\right)$
and $\boldsymbol{\rho}_{v}\left(\kv\right)$ for every $\kv\in\Sigma_{\G}$.
Of course, the bound $N\left(\Omega\right)\ge1$ is trivial by the
definition of the spectral position, and hence $\boldsymbol{N}_{v}\left(\kv\right)\ge1$
for all $\kv\in\Sigma_{\G}$. Given $\kv\in\Sigma_{\G}$, then $\I\left(\kv\right)\in\Sigma_{\G}$
and so $\boldsymbol{N}_{v}\left(\I\left(\kv\right)\right)\ge1$. Using
Lemma \ref{lem: Nv rhov-1} (\ref{enu: N is constant}), we get $\boldsymbol{N}_{v}\left(\kv\right)\le\deg v-1$.
This proves (\ref{eq:Spectral_Position_bounds}). \\
Assume that the lower bound on $\boldsymbol{\rho}_{v}$ holds. Namely,
\begin{align*}
\forall\kv\in\Sigma^{\G}\,\,\,\frac{\boldsymbol{N}_{v}\left(\kv\right)+1}{2}\le & \boldsymbol{\rho}_{v}\left(\kv\right).
\end{align*}
Then replacing $\kv\mapsto\mathcal{I}\left(\kv\right)$ and using
the inversion symmetry (Lemma \ref{lem: Nv rhov-1}) we get, 
\[
\frac{\deg v-\boldsymbol{N}_{v}\left(\kv\right)+1}{2}\le\deg v-\boldsymbol{\rho}_{v}\left(\kv\right),
\]
which can be rearrange as $\boldsymbol{\rho}_{v}\left(\kv\right)\le\frac{\deg v+\boldsymbol{N}_{v}\left(\kv\right)-1}{2}$.
Thus it is only left to prove that $\forall\kv\in\Sigma^{\G}\,\,\,\frac{\boldsymbol{N}_{v}\left(\kv\right)+1}{2}\le\boldsymbol{\rho}_{v}\left(\kv\right),$
or equivalently to prove that $\frac{N\left(\Omega\right)+1}{2}\le\rho\left(\Omega\right)$.
To do so we will use \cite[Theorem 1]{Fri_aif05} as seen in (\ref{eq: Friedlander result}),
which can be rearranged such that for every eigenvalue $k^{2}$ of
$\Gamma_{\lv}$ with total length $L$, 
\[
\frac{kL}{\pi}\ge\frac{N\left(\Gamma_{\lv},k\right)+1}{2}.
\]
Applying the above to $\Gamma_{\lv}=\Omega$ we get $\frac{N\left(\Omega\right)+1}{2}\le\frac{kL}{\pi}=\rho\left(\Omega\right)$
as needed.
\end{proof}
\begin{rem}
As can be seen in the proof, the lower bounds in (\ref{eq:Spectral_Position_bounds})
and (\ref{eq:Rho_bounds}) hold for any Neumann domain and not only
star graphs.
\end{rem}

Next, we discuss the statistical properties of the spectral position
and wavelength capacity of Neumann domains. As can be seen in Proposition
\ref{prop:local_observables_bounds} and Lemma \ref{lem: Nv rhov-1},
in order to compare Neumann domains of different eigenfunctions it
would be convenient to consider such that contain the same interior
vertex. 
\begin{defn}
Let $\Gamma_{\lv}$ be a standard graph and let $v\in\V_{in}$ an
interior vertex. For every generic eigenfunction $f_{n}$ we denote
the Neumann domain that contain $v$ by $\Omega_{n}^{(v)}$. We define
the associated sequence of spectral positions $N_{v}:\G\rightarrow\N$
given by $N_{v}(n):=N\left(\Omega_{n}^{(v)}\right)$ and the associate
sequence of wave capacities $\rho_{v}:\G\rightarrow\R$ given by $\rho_{v}(n):=\rho\left(\Omega_{n}^{(v)}\right)$. 
\end{defn}

Lemma \ref{lem: Nv rhov-1} ensures that $N_{v}(n)=\boldsymbol{N}_{v}\left(\left\{ k_{n}\lv\right\} \right)$
whenever $\Omega_{n}^{\left(v\right)}$ is a star, and in particular
for all $n>\frac{2L}{L_{min}}$ (according to Lemma \ref{lem: local properties of Neumann domain}).
Therefore, the two sequences, $N_{v}(n)$ and $\boldsymbol{N}_{v}\left(\left\{ k_{n}\lv\right\} \right)$
have the same statistical behavior. For the sake of simplicity, avoiding
the technicality of refining the $\sigma$-algebra and defining a
new probability space as is done in \cite{AloBan19}, we may redefine
$N_{v}(n)$ and assume that $N_{v}(n)=\boldsymbol{N}_{v}\left(\left\{ k_{n}\lv\right\} \right)$
for all $n\in\G$.
\begin{thm}
\label{thm: statistics_of_local_observables-1}Let $\Gamma_{\lv}$
be a standard graph with $\lv$ rationally independent and let $\left(\G,\mathcal{F}_{\G},d_{\G}\right)$
be the probability space defined in Theorem \ref{thm:First}. Let
$v\in\V_{in}$ be an interior vertex and consider the two sequences
$N_{v}$ and $\rho_{v}$. Then,
\begin{enumerate}
\item \label{enu: Nv is an RV}$N_{v}$ is a random variable on $\left(\G,\mathcal{F}_{\G},d_{\G}\right)$,
with probability given by:
\begin{align}
P\left(N_{v}=j\right)=d_{\G}\left(N_{v}^{-1}\left(j\right)\right)= & \lim_{N\rightarrow\infty}\frac{\left|\set{n\in\G\left(N\right)}{N_{v}\left(n\right)=j}\right|}{\left|\G\left(N\right)\right|}.\label{eq:Thm1-1}
\end{align}
\item \label{enu: Nv is symmetric}Every $N_{v}$ is symmetric around $\frac{\deg v}{2}$,
namely $P\left(N_{v}=j\right)=P\left(N_{v}=\deg v-j\right)$. Moreover,
all $N_{v}$ for every $v\in\V_{in}$ together with $\sigma$ and
$\omega$ are symmetric simultaneously (see Theorem \ref{thm:First}).
\item \label{enu: rhov is almost RV} There exists a probability measure
$\xi_{v}$ on $\R$ such that for any open interval $\left(a,b\right)$,
the level set $\rho_{v}^{-1}\left(a,b\right)$ has density given by
\begin{equation}
\xi_{v}\left(a,b\right)=d_{\G}\left(\rho_{v}^{-1}\left(a,b\right)\right)=\lim_{N\rightarrow\infty}\frac{\left|\set{n\in\G\left(N\right)}{\rho_{v}\left(n\right)\in\left(a,b\right)}\right|}{\left|\G\left(N\right)\right|}.\label{eq:denstiy-as-integral-of-distribution-1}
\end{equation}
\item \label{enu: rhov is symmetric} The probability measure $\xi_{v}$
is supported inside $\left[1,\thinspace\deg v-1\right]$ symmetrically.
That is, if $I\subset\R$ is measurable and $\deg v-I=\set{x\in\R}{\deg v-x\in I}$
then $\xi_{v}\left(I\right)=\xi_{v}\left(\deg v-I\right)$. 
\end{enumerate}
\end{thm}

\begin{rem}
A further description of the measure $\xi_{v}$ is found in \cite{AloBan19},
where it is shown that $\xi_{v}$ has no singular continuous part.
\end{rem}

\begin{rem}
\label{rem: rho v is not actually an RV} Although $\rho_{v}^{-1}(a,b)$
has density for every $a<b$, $\rho_{v}$ may not be a random variable.
In fact, if $\boldsymbol{\rho}_{v}$ is not constant on connected
components (equivalently, $\xi_{v}$ has an absolutely continuous
part), then by Corollary \ref{cor: Not an RV} in Appendix \ref{sec: app Equidistribution},
there is no $\sigma$-algebra on $\G$ on which $d_{\G}$ is a measure
and $\rho_{v}$ is a random variable. 
\end{rem}

\begin{proof}
The proof of (\ref{enu: Nv is an RV}) and (\ref{enu: Nv is symmetric})
is similar to the proof of Theorem \ref{thm:First}. According to
our simplifying assumption, the level sets $N_{v}^{-1}\left(j\right)$
are given by $\set{n\in\G}{\left\{ k_{n}\lv\right\} =\boldsymbol{N}_{v}^{-1}\left(j\right)}$.
Since the level sets $\boldsymbol{N}_{v}^{-1}\left(j\right)$ are
unions of connected components of $\Sigma^{\G}$ (by Lemma \ref{lem: Nv rhov-1}),
then the level sets of $N_{v}$ are unions of atoms in $\mathcal{F}_{\G}$.
Therefore, $N_{v}$ is random variable on $\left(\G,\mathcal{F}_{\G},d_{\G}\right)$.
This proves (\ref{enu: Nv is an RV}). 

Lemma \ref{lem: Nv rhov-1} provides a symmetry of $\boldsymbol{N}_{v}$,
$\I\left(\boldsymbol{N}_{v}^{-1}\left(j\right)\right)=\boldsymbol{N}_{v}^{-1}\left(\deg v-j\right)$
for any $j$. Using Theorem \ref{thm: inversion BG} we get: 
\[
d_{\G}\left(\set{n\in\G}{\left\{ k_{n}\lv\right\} =\boldsymbol{N}_{v}^{-1}\left(j\right)}\right)=\set{n\in\G}{\left\{ k_{n}\lv\right\} =\I\left(\boldsymbol{N}_{v}^{-1}\left(j\right)\right)},
\]
which proves that $N_{v}$ is symmetric. In order to prove the simultaneous
symmetry, consider the set 
\[
\cap_{v\in\V_{in}}N_{v}^{-1}\left(i_{v}\right)\cap\sigma^{-1}\left(i_{\sigma}\right)\cap\omega^{-1}\left(i_{\omega}\right),
\]
for some choice of possible values $\left\{ i_{v}\right\} _{v\in\V_{in}},i_{\sigma},i_{\omega}$.
Then according to Lemma \ref{lem: Nv rhov-1} and the proof of Theorem
\ref{thm:First}, 
\[
\I\left(\cap_{v\in\V_{in}}\boldsymbol{N}_{v}^{-1}\left(i_{v}\right)\cap\boldsymbol{\sigma}^{-1}\left(i_{\sigma}\right)\cap\boldsymbol{\omega}^{-1}\left(i_{\omega}\right)\right)=...
\]
\[
...=\cap_{v\in\V_{in}}\boldsymbol{N}_{v}^{-1}\left(\deg v-i_{v}\right)\cap\boldsymbol{\sigma}^{-1}\left(\beta-i_{\sigma}\right)\cap\boldsymbol{\omega}^{-1}\left(\beta-\left|\partial\Gamma\right|-i_{\omega}\right),
\]
and by Theorem \ref{thm: inversion BG},
\[
d_{\G}\left(\cap_{v\in\V_{in}}N_{v}^{-1}\left(i_{v}\right)\cap\sigma^{-1}\left(i_{\sigma}\right)\cap\omega^{-1}\left(i_{\omega}\right)\right)=...
\]
\[
...=d_{\G}\left(\cap_{v\in\V_{in}}N_{v}^{-1}\left(\deg v-i_{v}\right)\cap\sigma^{-1}\left(\beta-i_{\sigma}\right)\cap\omega^{-1}\left(\beta-\left|\partial\Gamma\right|-i_{\omega}\right)\right).
\]
This proves (\ref{enu: Nv is symmetric}).

Let us now define $\xi_{v}$ as the push-forward of $\mu_{\lv}$ by
$\boldsymbol{\rho}_{v}$. That is, for any Borel set $A\subset\R$
we define 
\[
\xi_{v}\left(A\right):=\frac{\mu_{\lv}\left(\boldsymbol{\rho}_{v}^{-1}\left(A\right)\right)}{\mu_{\lv}\left(\Sigma_{\G}\right)}.
\]
The proof of (\ref{enu: rhov is symmetric}) follows from the definition
of $\xi_{v}$. By definition, $\xi_{v}$ is supported on the image
of $\boldsymbol{\rho}_{v}$ which is contained in $\left[1,\deg v-1\right]$
according to Proposition \ref{prop:local_observables_bounds}. It
is also symmetric by Lemma \ref{lem: Nv rhov-1} and the fact that
$\I$ is measure preserving. To see that, let $A\subset\R$ be a Borel
set, 
\begin{align*}
\xi_{v}\left(\deg v-A\right)= & \frac{\mu_{\lv}\left(\boldsymbol{\rho}_{v}^{-1}\left(\deg v-A\right)\right)}{\mu_{\lv}\left(\Sigma_{\G}\right)}\\
= & \frac{\mu_{\lv}\left(\I\left(\boldsymbol{\rho}_{v}^{-1}\left(A\right)\right)\right)}{\mu_{\lv}\left(\Sigma_{\G}\right)}\\
= & \frac{\mu_{\lv}\left(\boldsymbol{\rho}_{v}^{-1}\left(A\right)\right)}{\mu_{\lv}\left(\Sigma_{\G}\right)}=\xi_{v}\left(A\right).
\end{align*}
It is left to prove (\ref{enu: rhov is almost RV}). Let $a<b$ and
$I=\left(a,b\right)$. The fact that One can show that $\boldsymbol{\rho}_{v}$
is continuous implies that $\partial\boldsymbol{\rho}_{v}^{-1}\left(a,b\right)\subset\partial\boldsymbol{\rho}_{v}^{-1}\left(a\right)\sqcup\partial\boldsymbol{\rho}_{v}^{-1}\left(b\right)$.
To see that, notice that $\boldsymbol{\rho}_{v}^{-1}\left(a,b\right)$
is open and that $\boldsymbol{\rho}_{v}^{-1}\left[a,b\right],\boldsymbol{\rho}_{v}^{-1}\left(a\right)$
and $\boldsymbol{\rho}_{v}^{-1}\left(b\right)$ are closed. The closure
of $\boldsymbol{\rho}_{v}^{-1}\left(a,b\right)$ satisfies
\[
\overline{\boldsymbol{\rho}_{v}^{-1}\left(a,b\right)}\subset\boldsymbol{\rho}_{v}^{-1}\left[a,b\right]=\boldsymbol{\rho}_{v}^{-1}\left(a,b\right)\sqcup\boldsymbol{\rho}_{v}^{-1}\left(a\right)\sqcup\boldsymbol{\rho}_{v}^{-1}\left(b\right).
\]
It follows that 
\[
\partial\boldsymbol{\rho}_{v}^{-1}\left(a,b\right)=\overline{\boldsymbol{\rho}_{v}^{-1}\left(a,b\right)}\setminus\boldsymbol{\rho}_{v}^{-1}\left(a,b\right)\subset\boldsymbol{\rho}_{v}^{-1}\left(a\right)\sqcup\boldsymbol{\rho}_{v}^{-1}\left(b\right).
\]
But clearly if $\kv\in\mathrm{int}\boldsymbol{\rho}_{v}^{-1}\left(a\right)$
then $\kv\notin\partial\boldsymbol{\rho}_{v}^{-1}\left(a,b\right)$
and same for $\mathrm{int}\boldsymbol{\rho}_{v}^{-1}\left(b\right)$,
which means that 
\begin{align*}
\partial\boldsymbol{\rho}_{v}^{-1}\left(a,b\right) & \subset\boldsymbol{\rho}_{v}^{-1}\left(a\right)\sqcup\boldsymbol{\rho}_{v}^{-1}\left(b\right)\setminus\left(\mathrm{int}\boldsymbol{\rho}_{v}^{-1}\left(a\right)\sqcup\mathrm{int}\boldsymbol{\rho}_{v}^{-1}\left(b\right)\right)\\
= & \partial\boldsymbol{\rho}_{v}^{-1}\left(a\right)\sqcup\partial\boldsymbol{\rho}_{v}^{-1}\left(b\right).
\end{align*}
The set $\boldsymbol{\rho}_{v}^{-1}\left(a\right)$, is the zero set
of $h\left(\kv\right):=\boldsymbol{\rho}_{v}\left(\kv\right)-a$ which
is a real analytic function on $\Sigma_{\G}$, by Lemma \ref{lem: Nv rhov-1}.
Using Lemma \ref{lem: real analytic lemma}, and the fact that $\Sigma_{\G}$
is a real analytic manifold, then for any $M$ connected component
of $\Sigma_{\G}$, either $M\subset\boldsymbol{\rho}_{v}^{-1}\left(a\right)$
or $M\cap\boldsymbol{\rho}_{v}^{-1}\left(a\right)$ is a closed set
of positive co-dimension in $M$. It follows that $\partial\boldsymbol{\rho}_{v}^{-1}\left(a\right)\cap M$
is either empty or of positive co-dimension and in both cases of measure
zero. Summing over all connected components gives that $\mu_{\lv}\left(\partial\boldsymbol{\rho}_{v}^{-1}\left(a\right)\right)=0$,
and similarly $\mu_{\lv}\left(\partial\boldsymbol{\rho}_{v}^{-1}\left(b\right)\right)=0$.
Therefore, the set $\boldsymbol{\rho}_{v}^{-1}\left(I\right)\subset\Sigma_{\G}$
is Jordan in $\Sigma_{\G}$, and as the boundary of $\Sigma_{\G}$
has measure zero in $\Sigma^{reg}$ by Corollary  \ref{cor: key lemma in equivalence of genericity},
then $\boldsymbol{\rho}_{v}^{-1}\left(I\right)$ is Jordan in $\Sigma^{reg}$.
We may use Theorem \ref{thm: inversion BG} to conclude that, 
\begin{equation}
d_{\G}\left(\set{n\in\G}{\left\{ k_{n}\lv\right\} \in\boldsymbol{\rho}_{v}^{-1}\left(I\right)}\right)=\frac{\mu_{\lv}\left(\boldsymbol{\rho}_{v}^{-1}\left(I\right)\right)}{\mu_{\lv}\left(\Sigma_{\G}\right)}=\xi_{v}\left(I\right).\label{eq: dg of rhov}
\end{equation}
Since $\rho_{v}^{-1}\left(I\right)$ and $\set{n\in\G}{\left\{ k_{n}\lv\right\} \in\boldsymbol{\rho}_{v}^{-1}\left(I\right)}$
differ by a finite number of elements, then $d_{\G}\left(\rho_{v}^{-1}\left(I\right)\right)=\xi_{v}\left(I\right)$
as needed.
\end{proof}

\subsection{Local-global connections}

The following proposition connects between values of local properties,
the spectral position and wavelength capacity of a Neumann domain,
with global properties such as nodal count, Neumann count, and the
structure of the graph. 
\begin{prop}
\label{prop: local to global}Let $\Gamma_{\lv}$ be a standard graph
with minimal edge length $L_{min}$ and total length $L$. Let $f$
be a generic eigenfunction of eigenvalue $k>\frac{\pi}{L_{min}}$,
nodal count $\phi\left(f\right)$ and Neumann count $\mu\left(f\right)$.
For every $v\in\V_{in}$ let $\Omega^{\left(v\right)}$ be the Neumann
domain containing $v$, with spectral position $N\left(\Omega^{\left(v\right)}\right)$
and wavelength capacity $\rho\left(\Omega^{\left(v\right)}\right)$.
Then,
\begin{enumerate}
\item The sum of the spectral positions is given by,
\begin{equation}
\sum_{v\in\V_{in}}N\left(\Omega^{(v)}\right)=\phi(f)-\mu(f)+E-\left|\partial\Gamma\right|.\label{eq:sum_of_spec_pos}
\end{equation}
\item The sum of the wave capacities is given by,
\begin{equation}
\sum_{v\in\V_{in}}\rho\left(\Omega^{(v)}\right)=\frac{L}{\pi}k-\mu(f)+E-\left|\partial\Gamma\right|.\label{eq:sum_of_wavelength_capacities}
\end{equation}
\end{enumerate}
\end{prop}

\begin{rem}
Subtracting the two equation above, relates the oscillatory part of
the trace formula $N\left(\Gamma_{\lv},k\right)-\frac{L}{\pi}k$ (see
Subsection \ref{subsec: oscilatory part of the trace formula}) with those of each Neumann domain, which are $N\left(\Omega^{(v)}\right)-\rho\left(\Omega^{(v)}\right)$,
and the nodal surplus $\sigma\left(f\right)$: 
\begin{equation}
\sum_{v\in\V_{in}}\left(N\left(\Omega^{(v)}\right)-\rho\left(\Omega^{(v)}\right)\right)=\phi(f)-\frac{L}{\pi}k=N\left(\Gamma_{\lv},k\right)-\frac{L}{\pi}k+\sigma\left(f\right).
\end{equation}
And since every Neumann domain of $f$ which does not contain an interior
vertex is a segment with $N\left(\Omega\right)=\rho\left(\Omega\right)=1$
then we can sum over all Neumann domains of $f$. Using the notation
$N_{osc}\left(\Gamma_{\lv},k\right)=N\left(\Gamma_{\lv},k\right)-\frac{L}{\pi}k$,
we get
\begin{equation}
\sum N_{osc}\left(\Omega,k\right)=N_{osc}\left(\Gamma_{\lv},k\right)+\sigma\left(f\right).
\end{equation}
\end{rem}

\begin{proof}
It was shown in (\ref{eq:Diff-nodal-Neumann_by_vertices}) that $\phi(f)-\mu(f)=\frac{\left|\partial\Gamma\right|}{2}-\frac{1}{2}\sum_{v\in\V_{in}}\sum_{e\in\E_{v}}\mathrm{sign}\left(f\left(v\right)\partial_{e}f\left(v\right)\right)$.
It can be written, using Lemma \ref{lem: vertex values canonical ef and other ef },
as 
\[
\phi(f)-\mu(f)=\frac{\left|\partial\Gamma\right|}{2}+\sum_{v\in\V_{in}}\left(N\left(\Omega^{(v)}\right)-\frac{\deg v}{2}\right).
\]
A simple counting argument gives $E=\sum_{v\in\V}\frac{\deg v}{2}=\frac{\left|\partial\Gamma\right|}{2}+\sum_{v\in\V_{in}}\frac{\deg v}{2}$
so that 
\[
\phi(f)-\mu(f)=\sum_{v\in\V_{in}}N\left(\Omega^{(v)},k\right)-E+\left|\partial\Gamma\right|,
\]
proving (\ref{eq:sum_of_spec_pos}). To prove (\ref{eq:sum_of_wavelength_capacities}), let us denote by $\mathcal{W}$ the set of all Neumann domains of
$f$ that does not contain any interior vertex. By the definition
of $\rho\left(\Omega\right)=\frac{L_{\Omega}}{\pi}k$ and the fact
that the Neumann domains are a partition of the graph $\Gamma_{\lv}$
then summing $\rho\left(\Omega\right)$ over all Neumann domains gives,
\[
\sum_{\Omega\in\mathcal{W}}\rho\left(\Omega\right)+\sum_{v\in\V_{in}}\rho\left(\Omega^{\left(v\right)}\right)=\frac{L}{\pi}k.
\]
Since every $\Omega\in\mathcal{W}$ is a single segment (otherwise
it would have contained an interior vertex), then $\rho\left(\Omega\right)=1$
(Lemma \ref{lem: local properties of Neumann domain}) and therefore,
\[
\sum_{v\in\V_{in}}\rho\left(\Omega^{\left(v\right)}\right)=\frac{L}{\pi}k-\left|\mathcal{W}\right|.
\]
We are left to prove that $\left|\mathcal{W}\right|=\mu\left(f\right)+\left|\partial\Gamma\right|-E$
in order to prove (\ref{eq:sum_of_wavelength_capacities}). Each segment
$\Omega\in\mathcal{W}$ has $\left|\partial\Omega\right|=2$, one
point which is a Neumann point and one point which is either a Neumann
point or a boundary vertex. Let us now consider the set of Neumann
points $\left\{ x_{j}\right\} _{j=1}^{\mu\left(f\right)}$ and define
a counting function, 
\[
\delta\left(x_{j}\right)=\left|\set{\Omega\in\mathcal{W}}{x_{j}\in\partial\Omega}\right|.
\]
Then clearly 
\[
2\mathcal{W}=\sum_{\Omega\in\mathcal{\mathcal{W}}}\left|\partial\Omega\right|=\left|\partial\Gamma\right|+\sum_{j=1}^{\mu\left(f\right)}\delta\left(x_{j}\right).
\]
Consider an edge $e\in\E$ and let $J_{e}$ be such that $\left\{ x_{j}\right\} _{j\in J_{e}}$
are the Neumann points that lie in $e$. Notice that $J_{e}$ is not
empty by the assumption $k>\frac{\pi}{L_{min}}$. Let $v,u$ be the
vertices of $e$ (not necessarily distinct vertices), then it is a
simple observation that 
\[
\sum_{j\in J_{e}}\delta\left(x_{j}\right)=\begin{cases}
2\left|J_{e}\right|-1 & v\in\partial\Gamma\,\,or\,\,u\in\partial\Gamma\\
2\left|J_{e}\right|-2 & v,u\in\V_{in}
\end{cases},
\]
and therefore 
\begin{align*}
\sum_{j=1}^{\mu\left(f\right)}\delta\left(x_{j}\right) & =\sum_{e\in\E}\sum_{j\in J_{e}}\delta\left(x_{j}\right)=\sum_{e\in\E}\left(2\left|J_{e}\right|-2\right)+\left|\partial\Gamma\right|\\
= & 2\mu\left(f\right)-2E+\left|\partial\Gamma\right|.
\end{align*}
It follows that $2\mathcal{W}=2\mu\left(f\right)+2\left|\partial\Gamma\right|-2E$
as needed.
\end{proof}

\newpage{}
\section{\label{sec: Magnetic} The nodal magnetic relation and local magnetic
indices}

In this section we present the nodal magnetic theorem \cite{BerWey_ptrsa14},
using which the nodal surplus can be characterized in terms of magnetic
stability. The goal of this section is the decomposition of the nodal
surplus, a ``global'' quantity, into sum of ``local'' quantities.
As it is not in the scope of this manuscript, we will introduce the
magnetic potential and gauge invariance briefly, without proofs. An
elaborated explanations on magnetic potential and gauge invariance,
together with proofs and physical context can be found in \cite{GnuSmi_ap06,BerKuc_graphs}. 

\subsection{Magnetic potential and gauge invariance}

Given a metric graph $\Gamma_{\lv}$, a magnetic potential is a 1-form
on $\Gamma_{\lv}$. That is, a function $A:\Gamma_{\lv}\rightarrow\R$
whose sign depends on the orientation of an edge. By adding magnetic
potential $A$, the Laplace operator is changed from $-\left(\frac{d}{dx}\right)^{2}$
to $\left(i\frac{d}{dx}+A\right)^{2}$. The magnetic operator, $\left(i\frac{d}{dx}+A\right)^{2}$,
is a self-adjoint, non-negative operator on the domain of functions
satisfying Neumann vertex conditions. The \emph{magnetic flux} induced
by a magnetic potential $A$ along an oriented closed path $\gamma$
is given by $\varointclockwise_{\gamma}A\left(x\right)dx\,\,mod\,2\pi$.
Gauge invariance gives a characterization to unitary equivalence classes
of such magnetic operators in terms of the magnetic fluxes. That is,
two magnetic potentials $A$ and $\tilde{A}$ induce unitary equivalent
operators if and only if $e^{i\varointclockwise_{\gamma}A\left(x\right)dx}=e^{i\varointclockwise_{\gamma}\tilde{A}\left(x\right)dx}$
for any oriented closed path $\gamma$ (see Corollary 2.6.3 in \cite{BerKuc_graphs}).
Using these facts, the equivalence classes can be characterized by
the parameter space $\T^{\beta}$, where $\beta=E-V+1$ is the first
Betti number of the graph, as follows. Given some spanning tree $T$,
there are $\beta$ remaining edges in $\Gamma\setminus T$ which we
denote by $\left\{ e_{j}\right\} _{j=1}^{\beta}$. It can be showed
that for any magnetic potential $A$, there is a unique magnetic potential
$\tilde{A}$ which vanishes on $T$ and is constant on every edge $\left\{ e_{j}\right\} _{j=1}^{\beta}$.
We will therefore parameterize each unitary equivalence class of such
operators by the \emph{magnetic fluxes }(or magnetic parameters) $\av\in\T^{\beta}$,
defined as 
\[
\alpha_{j}:=\int_{e_{j}}\tilde{A}\left(x\right)dx\,\,mod\,2\pi\,\,\,\forall e_{j}\subset\Gamma\setminus T.
\]
It can be shown that a different choice of spanning tree corresponds
to a change of basis for $\T^{\beta}$. We may conclude that the eigenvalues
of the magnetic operators (which are invariant under unitary transformations)
can be considered as functions of $\av$ over $\T^{\beta}$. In fact,
it can be shown that if $k^{2}$ is a simple eigenvalue of a standard
graph $\Gamma_{\lv}$ (with no magnetic potential), then there is a
neighborhood in $\T^{\beta}$ around $\av=0$ on which $k\left(\av\right)$
is a smooth function such that $k\left(\av\right)$ is a simple eigenvalue
of $\left(i\frac{d}{dx}+A\right)^{2}$ for any $A$ corresponding
to $\av$.

\subsection{The nodal magnetic theorem}

As discussed in the introduction, the first classification of the
deviation of the nodal count from Courant's bound as a Morse index
of a certain functional appeared in \cite{BerKucSmi_gafa12} for planar
domains, in \cite{BanBerRazSmi_cmp12} for quantum graphs, and in
\cite{BerRazSmi_jpa12} for discrete graphs. This nodal count deviation,
in the quantum graphs setting, is the nodal surplus. The `nodal magnetic
theorem' is a similar relation between the nodal surplus (and its
discrete graphs' analog) to the Morse index of the eigenvalue with respect to changes in the magnetic field. That is, the `nodal magnetic theorem' characterizes
the nodal surplus as a stability index of the eigenvalue with respect
to magnetic perturbations. It was first proved by Berkolaiko for discrete
graphs \cite{Ber_apde13} after which Colin de Verdière provided a
different proof \cite{Col_apde13}. We will present the `nodal magnetic
theorem' for quantum graphs that was proved by Berkolaiko and Weyand
in \cite{BerWey_ptrsa14}.
\begin{defn}
Given an $N\times N$ self-adjoint matrix $B$, with eigenvalues $\left\{ \lambda_{j}\right\} _{j=1}^{N}$,
its \emph{Morse index} is defined by 
\begin{equation}
\M\left(B\right)=\left|\set{j\le N}{\lambda_{j}<0}\right|.
\end{equation}
Let $k\left(\av\right)$ be a smooth function with critical point
at $\av=0$ and denote its Hessian at $\av=0$ by $\hess_{\av}k$.
The \emph{Morse index }of $k\left(\av\right)$ at the critical point
$\av=0$ is $\M\left(\hess_{\av}k\right)$.
\end{defn}

The nodal magnetic theorem, for quantum graphs, can be stated as:
\begin{thm}
\cite{BerWey_ptrs13}\label{thm: nodal magnetic}Let $\Gamma_{\lv}$
be a standard graph and let $k^{2}$ be a simple eigenvalue with eigenfunction
$f$. Then the function $k\left(\av\right)$ has a critical point
at $\av=0$. If we further assume that $f$ satisfies property $I$,
then $\det\left(\hess_{\av}k\right)\ne0$ and the Morse index of $k\left(\av\right)$
at $\av=0$ is equal to the nodal surplus $\sigma\left(f\right)$.
Namely 
\begin{equation}
\sigma\left(f\right)=\M\left(\hess_{\av}k\right).
\end{equation}
\end{thm}

\begin{rem}
Although we stated the above with the Morse index of $k\left(\av\right)$,
it was stated in \cite{BerWey_ptrsa14} using the Morse index of $k^{2}\left(\av\right)$.
However, since $\nabla k\left(0\right)=0$, then \[\hess_{\av}(k^{2})=2k\hess_{\av}k,\ \text{with}\ k>0\] and so the two Morse indices are equal.
\end{rem}

A key ingredient in the proof of Theorem \ref{thm:First} was that
the nodal surplus is given as a function on the secular manifold.
This fact was first proved in \cite{Ban_ptrsa14}, using the nodal
magnetic relation and was further developed in \cite{AloBanBer_cmp18}.
This section will follow \cite{AloBanBer_cmp18}.
\begin{defn}
\label{def: U alpha and F alpha}Given a graph $\Gamma$, a spanning
tree $T$, and a choice of magnetic fluxes $\av\in\T^{\beta}$ such
that the $\Gamma\setminus T$ edge corresponding to $\alpha_{j}$
is denoted by $e_{j}$. Then the \emph{magnetic unitary evolution
matrix} is defined by 
\begin{equation}
U_{\kv;\av}:=e^{i\hat{\alpha}}U_{\kv},
\end{equation}
where $U_{\kv}=e^{i\hat{\kappa}}S$ is the unitary evolution matrix
defined in (\ref{eq: U=00003Dexp J S}) and $e^{i\hat{\alpha}}$ is
an $\av$ dependent unitary diagonal matrix, defined by 
\begin{equation}
\left(e^{i\hat{\alpha}}\right)_{e,e}=\overline{\left(e^{i\hat{\alpha}}\right)_{\hat{e},\hat{e}}}=\begin{cases}
e^{i\alpha_{j}} & e=e_{j}\\
1 & e\subset T
\end{cases}.
\end{equation}
We define the \emph{magnetic secular function, }similarly to Definition
\ref{def: secular function}, as 
\begin{align*}
\tilde{F}\left(\kv;\av\right) & :=\det\left(U_{\kv;\av}\right)^{\frac{1}{2}}\det\left(1-U_{\kv;\av}\right)=\det\left(U_{\kv}\right)^{\frac{1}{2}}\det\left(1-U_{\kv;\av}\right),\,\,\,\text{with}\\
\det\left(U_{\kv}\right)^{\frac{1}{2}} & =\left(i\right)^{\beta-1}e^{-i\sum_{e\in\E}\kappa_{e}}.
\end{align*}
Observe that $\det\left(e^{i\hat{\alpha}}\right)=1$, and that the
secular function $F$ is given by $F\left(\kv\right)=\tilde{F}\left(\kv;0\right)$.
\end{defn}

\begin{rem}
\label{rem: F alpha trig pol}Clearly $\tilde{F}$ is a trigonometric
polynomial in both $\kv$ and $\av$, and it is real using the same
argument as in Lemma \ref{lem:The-secular-functions properties} and
$\det\left(e^{i\hat{\alpha}}\right)=1$.
\end{rem}

The relation between the magnetic secular function $\tilde{F}$ and
the eigenvalues of the magnetic operator corresponding to $\av$ is
given in the following lemma, which can be found for example in both
\cite{GnuSmi_ap06,BerKuc_graphs}.
\begin{lem}
\label{lem: magnetic secular equation}Given a metric graph $\Gamma_{\lv}$
and a choice of magnetic fluxes $\av\in\T^{\beta}$, then $k^{2}>0$
is an eigenvalue of the corresponding magnetic operator if and only
if $\tilde{F}\left(k\lv;\av\right)=0$, and it is simple if the $k$
derivative $\frac{d\tilde{F}}{dk}\left(k\lv;\av\right)\ne0$. 
\end{lem}

\begin{defn}
Denote the Hessian of $\tilde{F}$ with respect to $\av$ at the point
$\left(\kv;0\right)$ by $\hess_{\av}F\left(\kv\right)$. That is,
$\hess_{\av}F\left(\kv\right)$ is a $\beta\times\beta$ real symmetric
matrix whose entries are given by the real trigonometric polynomials
\[
\left(\hess_{\av}F\left(\kv\right)\right)_{i,j}=\frac{\partial^{2}\tilde{F}}{\partial\alpha_{j}\partial\alpha_{i}}\left(\kv;0\right).
\]
\end{defn}

\begin{rem}
We write $\hess_{\av}F$ instead of $\hess_{\av}\tilde{F}$ in order
to emphasize that it is evaluated at $\av=0$ and therefore  it is only a
function of $\kv$.

We may now rewrite the nodal magnetic theorem in terms of the secular
manifold:
\end{rem}

\begin{prop}
\label{prop: nodal magnetic secular}Let $\Gamma$ be a graph, let
$\kv\in\Sigma_{\G}$ with canonical eigenfunction $f_{\kv}$, and
let $\boldsymbol{\sigma}\left(\kv\right)$ be the nodal surplus of
$f_{\kv}$. Then the magnetic secular function $\tilde{F}$ satisfies
$\frac{\partial\tilde{F}}{\partial\alpha_{j}}\left(\kv;0\right)=0$
for every $\alpha_{j}$, and $\boldsymbol{\sigma}\left(\kv\right)$
is given explicitly by 
\begin{equation}
\boldsymbol{\sigma}\left(\kv\right)=\M\left(-\frac{\hess_{\av}F}{p}\left(\kv\right)\right),\label{eq: nodal magnetic secular}
\end{equation}
recall Definition \ref{def: p} of $p\left(\kv\right)$. Moreover, $\frac{\hess_{\av}F}{p}$
is continuous on $\Sigma_{\G}$, and satisfies: 
\begin{align*}
\forall\kv\in\Sigma_{\G}\,\,\,\det\left(\frac{\hess_{\av}F}{p}\left(\kv\right)\right) & \ne0,\,\,\text{and}\\
\frac{\hess_{\av}F}{p}\left(\I\left(\kv\right)\right)= & -\frac{\hess_{\av}F}{p}\left(\kv\right).
\end{align*}
\end{prop}

\begin{proof}
Let $\kv\in\Sigma_{\G}$, and denote the edge lengths of $\Gamma_{\kv}$
by $\lv\in\opcl{0,2\pi}^{\E}$ such that $\left\{ \lv\right\} =\kv$.
Let $m_{\kv}$ be the weights vector at $\kv$, as in Definition \ref{def: mk weights},
namely if $\boldsymbol{a}$ is the amplitudes vector of $f_{\kv}$
then 
\[
\left(m_{\kv}\right)_{e}=\left|a_{e}\right|^{2}+\left|a_{\hat{e}}\right|^{2}.
\]
Since $\kv\in\Sigma_{\G}$, $k=1$ is a simple eigenvalue, and according
to Lemma \ref{lem: magnetic secular equation}, the function $k\left(\av\right)$
(with $k\left(0\right)=1$) is given by the implicit function $\tilde{F}\left(k\lv;\av\right)=0$
around the point $k=1$ and $\av=0$. According to Lemma \ref{lem:The-secular-functions properties},
the $k$ derivative of $\tilde{F}\left(k\lv;\av\right)$ at $k=1$
and $\av=0$ is given by:
\[
\frac{d\tilde{F}}{dk}\left(\lv;0\right)=\frac{d\tilde{F}}{dk}\left(\kv;0\right)=\sum_{e\in\E}l_{e}\frac{\partial F}{\partial\kappa_{e}}\left(\kv\right)=p\left(\kv\right)\cdot\left(m_{\kv}\cdot\lv\right).
\]
Lemma \ref{lem:The-secular-functions properties} also states that
$p$ is non vanishing on $\Sigma^{reg}$, and since $\kv\in\Sigma_{\G}$
then $m_{\kv}$ has positive entries. In particular 
\begin{align*}
\frac{d\tilde{F}}{dk}\left(\kv;0\right) & \ne0,\,\,\,\text{and}\\
\mathrm{sign}\frac{d\tilde{F}}{dk}\left(\kv;0\right) & =\mathrm{sign}\left(p\left(\kv\right)\right).
\end{align*}
By the implicit function theorem, locally around $k=1$ and $\av=0$
we get,
\begin{equation}
\frac{\partial k}{\partial\alpha_{j}}\left(\av\right)=-\frac{\frac{\partial\tilde{F}}{\partial\alpha_{j}}}{\frac{d\tilde{F}}{dk}}\left(k\lv;\av\right).\label{eq: implicit derivative}
\end{equation}
According to Theorem \ref{thm: nodal magnetic}, $\frac{\partial k}{\partial\alpha_{j}}\left(0\right)=0$
and therefore $\frac{\partial\tilde{F}}{\partial\alpha_{j}}\left(\lv;0\right)=0$
for all $j$. If we take an $\alpha_{i}$ derivative of (\ref{eq: implicit derivative})
and at $\av=0$, then there are two terms that vanish due to $\frac{\partial k}{\partial\alpha_{i}}\left(0\right)=\frac{\partial\tilde{F}}{\partial\alpha_{j}}\left(\lv;0\right)=0$
and we are left with
\begin{equation}
\frac{\partial^{2}k}{\partial\alpha_{j}\partial\alpha_{i}}\left(0\right)=-\frac{\frac{\partial^{2}\tilde{F}}{\partial\alpha_{j}\partial\alpha_{i}}}{\frac{\partial\tilde{F}}{\partial k}}\left(\lv;0\right).
\end{equation}
Notice that $\frac{\partial^{2}\tilde{F}}{\partial\alpha_{j}\partial\alpha_{i}}\left(\lv;0\right)=\left(\hess_{\av}F\left(\kv\right)\right)_{i,j}$
and that $\frac{\partial\tilde{F}}{\partial k}\left(\lv;0\right)=p\left(\kv\right)\cdot\left(\vec{m_{\kv}}\cdot\lv\right)$.
According to Theorem \ref{thm: nodal magnetic} we get that $\det\left(\frac{\partial^{2}k}{\partial\alpha_{j}\partial\alpha_{i}}\left(0\right)\right)\ne0$
so that $\det\left(\hess_{\av}F\left(\kv\right)\right)\ne0$ and the
nodal surplus of $f_{\kv}$ (the eigenfunction of $k\left(0\right)=1$),
is given by 
\[
\sigma\left(\kv\right)=\M\left(\frac{\partial^{2}k}{\partial\alpha_{j}\partial\alpha_{i}}\left(0\right)\right)=\M\left(-\frac{\hess_{\av}F\left(\kv\right)}{p\left(\kv\right)\cdot\left(m_{\kv}\cdot\lv\right)}\right).
\]
The factor $\frac{1}{m_{\kv}\cdot\lv}$ is strictly positive so $\M\left(-\frac{\hess_{\av}F\left(\kv\right)}{p\left(\kv\right)\cdot\left(m_{\kv}\cdot\lv\right)}\right)=\M\left(-\frac{\hess_{\av}F\left(\kv\right)}{p\left(\kv\right)}\right)$.
We conclude that, 
\[
\sigma\left(\kv\right)=\M\left(-\frac{\hess_{\av}F}{p}\left(\kv\right)\right).
\]
The entries of $\hess_{\av}F\left(\kv\right)$ and the function $p$
are real trigonometric polynomials. Using $p\left(\kv\right)\ne0$
on $\Sigma_{\G}$, we show that $\frac{\hess_{\av}F}{p}$ is continuous
on $\Sigma_{\G}$. 

As for the inversion symmetry, Lemma \ref{lem: inversion on F} gives
$p\left(\mathcal{I}\left(\kv\right)\right)=\left(-1\right)^{\beta}p\left(\kv\right)$.
By Definition \ref{def: U alpha and F alpha}, $U_{\I\kv;\I\av}=\overline{U_{\kv;\av}}$,
and so 
\[
\tilde{F}\left(\I\kv;\I\av\right)=\left(i\right)^{\beta-1}\overline{e^{-i\sum_{e\in\E}\kappa_{e}}\det\left(1-U_{\kv;\av}\right)}=\left(-1\right)^{\beta-1}\overline{\tilde{F}\left(\kv;\av\right)}.
\]
As $\tilde{F}$ is real, then $\tilde{F}\left(\I\kv;\I\av\right)=\left(-1\right)^{\beta-1}\tilde{F}\left(\kv;\av\right)$
and therefore, for any $\alpha_{i}$ and $\alpha_{j}$, 
\begin{align*}
\frac{\partial^{2}\tilde{F}}{\partial\alpha_{j}\partial\alpha_{i}}\left(\I\kv;\I\av\right) & =\left(-1\right)^{\beta+1}\frac{\partial^{2}\tilde{F}}{\partial\alpha_{j}\partial\alpha_{i}}\left(\kv;\av\right).
\end{align*}
It follows that, 
\[
\frac{\hess_{\av}F}{p}\left(\I\kv\right)=-\frac{\hess_{\av}F}{p}\left(\kv\right).
\]
\end{proof}

\subsection{Local magnetic index}
\begin{defn}
Let $\E_{bridges}$ be the set of bridges of a graph $\Gamma$, and
consider the decomposition of $\Gamma\setminus\E_{bridges}$ into
connected components\footnote{Using graph theoretic terminology, each component is a 2-edge-connected
sub-graph of $\Gamma$.}, as shown in Figure \ref{fig: bridge partition}. Such a connected
component may be a single vertex with no edges, in which case we call
it \emph{trivial.} We denote the non-trivial connected components
by $\left\{ \Gamma_{j}\right\} _{j=1}^{m}$. 

The \emph{edge separation decomposition }of $\Gamma$, denoted by
$\left[\Gamma_{1},\Gamma_{2},...\Gamma_{m}\right]$, is the set of
non-trivial connected components of $\Gamma\setminus\E_{bridges}$.
\end{defn}

\begin{defn}
Let $\Gamma$ be a graph with first Betti number $\beta$ and magnetic
fluxes $\av\in\T^{\beta}$. Let $\left[\Gamma_{1},\Gamma_{2},...\Gamma_{m}\right]$
be the edge-separation decomposition of $\Gamma$. Since $\E_{bridges}$
is contained in every spanning tree, then every magnetic flux is associated
to an edge in $\Gamma\setminus\E_{bridges}$, and therefore $\av$
is decomposed accordingly, $\av=\left(\av_{1},\av_{2}...\av_{m}\right)$. 

Given an eigenvalue $k>0$ of $\Gamma_{\lv}$, we consider a block
decomposition of its hessian, $\hess_{\av}k$, according to $\av=\left(\av_{1},\av_{2}...\av_{m}\right)$,
and we denote the diagonal blocks by $\hess_{\av_{j}}k$. We call
$\hess_{\av_{j}}k$ a \emph{local magnetic hessian} and define its
\emph{local magnetic index }by 
\[
\iota_{j}:=\M\left(\hess_{\av_{j}}k\right)
\]
.

\begin{figure}
\includegraphics[width=0.6\paperwidth]{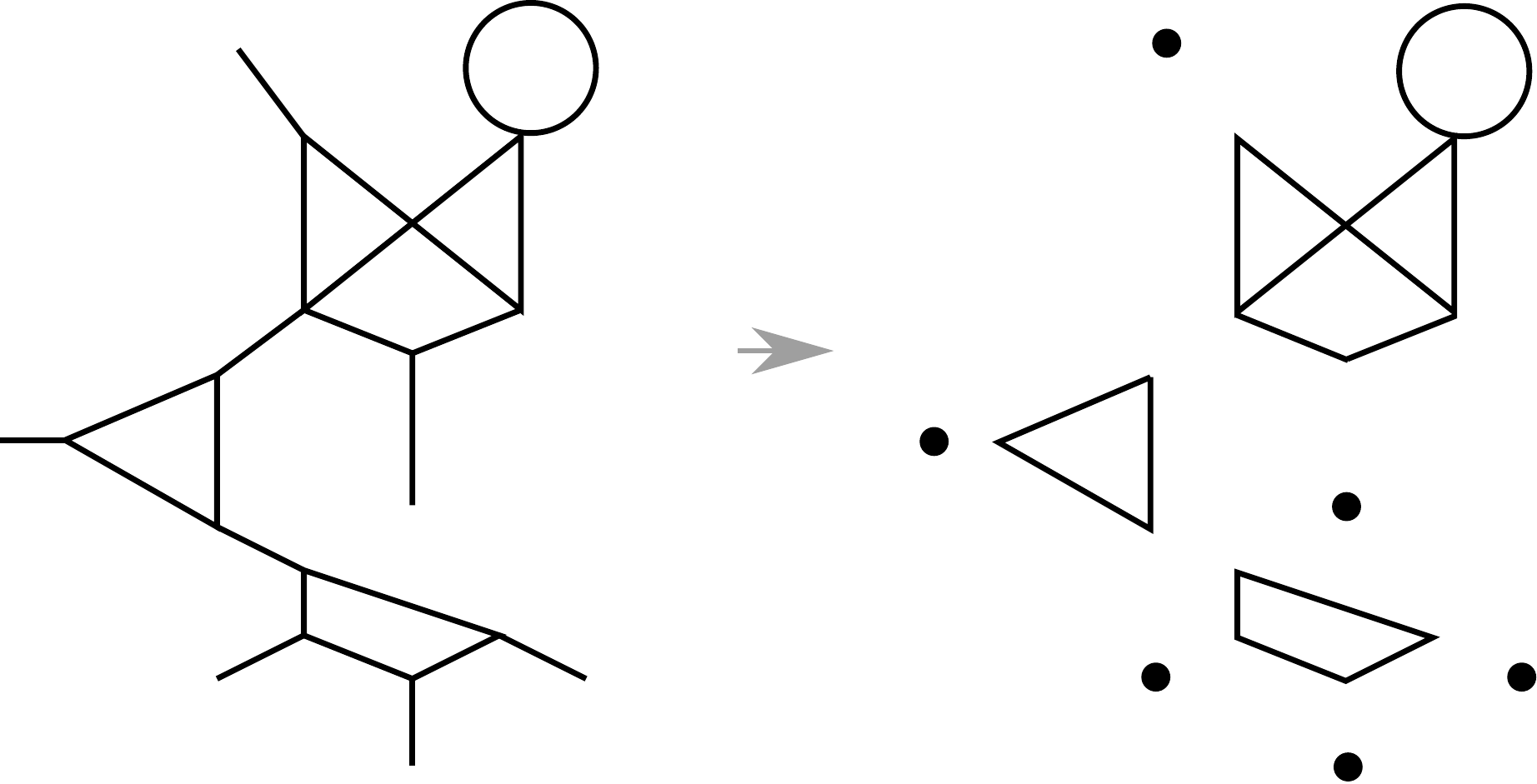}

\caption[Edge separation partition]{\label{fig: bridge partition} On the left, a graph $\Gamma$ with
$\left|\protect\E_{bridges}\right|=8$, $\left|\partial\Gamma\right|=6$
and edge separation decomposition $\left[\Gamma_{1},\Gamma_{2},\Gamma_{3}\right]$.
On the right, $\Gamma\setminus\protect\E_{bridges}$, with 6 trivial
connected components and 3 non trivial connected components $\Gamma_{1},\Gamma_{2}$
and $\Gamma_{3}$. }
\end{figure}
\end{defn}

\begin{defn}
\label{def: local magnetic indicies}Given a standard graph $\Gamma_{\lv}$
with edge separation decomposition\emph{ }$\left[\Gamma_{1},\Gamma_{2},...\Gamma_{m}\right]$.
For any $\Gamma_{j}$ (of first Betti number $\beta_{j}$) we define
its local (magnetic) index sequence $\iota_{j}:\G\rightarrow\left\{ 0,1,...\beta_{j}\right\} $
by 
\[
\iota_{j}\left(n\right):=\M\left(\hess_{\av_{j}}k_{n}\right).
\]
 
\end{defn}

\begin{thm}
\label{thm: local magnetic indices}Let $\Gamma_{\lv}$ be a standard
graph and let $f_{n}$ be a generic eigenfunction with nodal surplus
$\sigma\left(n\right)$. If $\Gamma$ has edge separation decomposition
$\left[\Gamma_{1},\Gamma_{2},...\Gamma_{m}\right]$, then $\hess_{\av}k_{n}$
is block diagonal, $\hess_{\av}k_{n}=\oplus\hess_{\av_{j}}k_{n}$,
according to the decomposition. In particular, 
\[
\sum_{j=1}^{m}\iota_{j}\left(n\right)=\sigma\left(n\right).
\]
\end{thm}

\begin{thm}
\label{thm: local mag indx 2} If $\lv$ is rationally independent
then each local index $\iota_{j}$ is a random variable on $\left(\G,\mathcal{F}_{\G},d_{\G}\right)$
(the probability space defined in Theorem \ref{thm:First}) and it
is symmetric around $\frac{\beta_{j}}{2}$. That is, for any possible
value $i$, 
\begin{align*}
P\left(\iota_{j}=i\right) & =d_{\G}\left(\iota_{j}^{-1}\left(i\right)\right)=\lim_{N\rightarrow\infty}\frac{\left|\set{n\in\G\left(N\right)}{\iota_{j}\left(n\right)=i}\right|}{\left|\G\left(N\right)\right|},\,\,\,and\\
P\left(\iota_{j}=i\right) & =P\left(\iota_{j}=\beta_{j}-i\right).
\end{align*}
Where $\left(\G,\mathcal{F}_{\G},d_{\G}\right)$ is the probability
space defined in Theorem \ref{thm:First}.
\end{thm}

To prove the above two theorems, as in the cases of Theorems \ref{thm:First}
and \ref{thm: statistics_of_local_observables-1}, we first define
the relevant functions on $\Sigma_{\G}$:
\begin{defn}
\label{def: secular local magnetic index}Let $\Gamma$ be a graph
with edge separation decomposition $\left[\Gamma_{1},\Gamma_{2},...\Gamma_{m}\right]$,
then for any $\Gamma_{j}$ we define $\boldsymbol{\iota}_{j}:\Sigma_{\G}\rightarrow\left\{ 0,1,...\beta_{j}\right\} $
by 
\begin{equation}
\boldsymbol{\iota}_{j}\left(\kv\right)=\M\left(-\frac{\hess_{\av_{j}}F}{p}\left(\kv\right)\right).\label{eq: secular local mangetic index}
\end{equation}
Where, $\hess_{\av_{j}}F$ is the block of $\hess_{\av}F$ corresponding
to $\av_{j}$. 
\end{defn}

\begin{rem}
According to the proof of Proposition \ref{prop: nodal magnetic secular},
if $\Gamma_{\lv}$ is a standard graph and $\iota_{j}$ is as defined
in Theorem \ref{thm: local magnetic indices}, then 
\begin{equation}
\forall n\in\G\,\,\iota_{j}\left(n\right)=\boldsymbol{\iota}_{j}\left(\left\{ k_{n}\lv\right\} \right).\label{eq: local magnetic index n to secular}
\end{equation}
\end{rem}

In order to prove that $\hess_{\av}k$ is block diagonal we will prove
that $\hess_{\av}F$ is block diagonal as they are proportional by
a scalar function. In fact it is enough to show that $\hess_{\av}F$
is block diagonal with respect to every single bridge decomposition.
Let $\Gamma$ be a graph with a bridge $e$ and bridge decomposition
$\Gamma\setminus\left\{ e\right\} =\Gamma_{1}\sqcup\Gamma_{2}$, where
$e$ is oriented from $\Gamma_{1}$ to $\Gamma_{2}$. Denote the corresponding
decomposition of the torus coordinates $\kv=\left(\kv_{1},\kappa_{e},\kv_{2}\right)$
and the magnetic parameters $\av=\left(\av_{1},\av_{2}\right)$.
\begin{prop}
\label{prop: block decomposition}If $e$ is a bridge of a graph $\Gamma$
with bridge decomposition $\Gamma\setminus\left\{ e\right\} =\Gamma_{1}\sqcup\Gamma_{2}$,
then $\hess_{\av}F\left(\kv\right)$ is block diagonal with respect
to $\av=\left(\av_{1},\av_{2}\right)$ for any $\kv\in\Sigma_{\G}$.
Moreover, if $\hess_{\av_{i}}F$ is the block corresponding to $\av_{i}$,
then $\frac{\hess_{\av_{i}}F}{\frac{\partial F}{\partial\kappa_{e}}}\left(\kv_{1},\kappa_{e},\kv_{2}\right)$
depends only on $\kv_{i}$ (under the restriction of $\left(\kv_{1},\kappa_{e},\kv_{2}\right)\in\Sigma_{\G}$).
\end{prop}

\begin{proof}
As in the proof of Proposition \ref{prop: simple bridge det decomposition},
we use Lemma \ref{lem: Technical  bridge decomp} to decompose $\tilde{F}\left(\kv,\av\right)$.
Denote $\left(\kv,\av\right)=\left(\kv_{1},\kappa_{e},\kv_{2};\av_{1},\av_{2}\right)$
according to the decomposition, and decompose $e^{i\hat{\alpha}}$
and $e^{i\hat{\kappa}}$ correspondingly. We substitute $\zi=e^{i\hat{\alpha_{i}}}e^{i\hat{\kappa}_{i}}$
for $i\in\left\{ 1,2\right\} $ and $z_{e}=e^{i\kappa_{e}}$ into
Proposition \ref{lem: Technical  bridge decomp} and Lemma \ref{lem: scattering phase},
and define $g_{i}\left(\kv_{i},\av_{i}\right):=\det D_{i}\left(\zi\right),\,e^{i\Theta_{i}\left(\kv_{i},\av_{i}\right)}=\mathcal{S}\left(\zi\right)$
for $i\in\left\{ 1,2\right\} $. As discussed in Proposition \ref{prop: simple bridge det decomposition},
the $g_{i}$'s are trigonometric polynomials and each $\Theta_{i}$
is smooth whenever $g_{i}\ne0$. Let $\kv'\in\Sigma_{\G}$, so the
amplitudes of $f_{\kv'}$ does not vanish, and according to Proposition
\ref{lem: Technical  bridge decomp}, $g_{1}\left(\kv'_{1},0\right)g_{2}\left(\kv'_{2},0\right)\ne0$.
By continuity, there is a neighborhood of $\left(\kv',0\right)$,
$O\subset\T^{\E}\times\T^{\beta}$, such that $g_{1}\left(\kv_{1},\av\right)g_{2}\left(\kv_{2},\av\right)\ne0$
for any $\left(\kv,\av\right)\in O$. According to Proposition \ref{lem: Technical  bridge decomp},
\[
\det\left(1-U_{\kv;\av}\right)=g_{1}\left(\kv_{1},\av_{1}\right)g_{2}\left(\kv_{2},\av_{2}\right)\left(1-e^{i2\kappa_{e}}e^{i\Theta_{1}\left(\kv_{1},\av_{1}\right)}e^{i\Theta_{2}\left(\kv_{2},\av_{2}\right)}\right),
\]
and denote $h\left(\kv,\av\right):=\det\left(U_{\kv}\right)^{-\frac{1}{2}}g_{1}\left(\kv_{1},\av_{1}\right)g_{2}\left(\kv_{2},\av_{2}\right)$
so that, 
\begin{equation}
\tilde{F}\left(\kv;\av\right)=h\left(\kv,\av\right)\left(1-e^{i2\kappa_{e}}e^{i\Theta_{1}\left(\kv_{1},\av_{1}\right)}e^{i\Theta_{2}\left(\kv_{2},\av_{2}\right)}\right).\label{eq: F tilde for bridge}
\end{equation}
Since $h$ is smooth and non vanishing on $O$, we can conclude that
for any $\left(\kv,\av\right)\in O$,
\begin{align}
\tilde{F} & =0\iff\frac{\tilde{F}}{h}=0\iff e^{i2\kappa_{e}}e^{i\Theta_{1}}e^{i\Theta_{2}}=1,\label{eq: secular equation for bridge}
\end{align}
\begin{align}
\frac{\partial\tilde{F}}{\partial\kappa_{e}} & =\frac{\tilde{F}}{h}\frac{\partial h}{\partial\kappa_{e}}+h\frac{\partial}{\partial\kappa_{e}}\left(\frac{\tilde{F}}{h}\right),\label{eq: dFdk magnetic with bridge}\\
\frac{\partial\tilde{F}}{\partial\alpha_{j}} & =\frac{\tilde{F}}{h}\frac{\partial h}{\partial\alpha_{j}}+h\frac{\partial}{\partial\alpha_{j}}\left(\frac{\tilde{F}}{h}\right),\,\,\,\text{and}\label{eq: dfda mag with bridg}\\
\frac{\partial^{2}\tilde{F}}{\partial\alpha_{j}\partial\alpha_{i}}=\frac{\tilde{F}}{h}\frac{\partial^{2}h}{\partial\alpha_{j}\partial\alpha_{i}}+ & \frac{\partial h}{\partial\alpha_{i}}\frac{\partial}{\partial\alpha_{j}}\left(\frac{\tilde{F}}{h}\right)+\frac{\partial h}{\partial\alpha_{j}}\frac{\partial}{\partial\alpha_{i}}\left(\frac{\tilde{F}}{h}\right)+h\frac{\partial^{2}}{\partial\alpha_{j}\partial\alpha_{i}}\left(\frac{\tilde{F}}{h}\right).\label{eq: Hessian mag with bridge}
\end{align}
If $\kv\in\Sigma_{\G}$ such that $\left(\kv;0\right)\in O$, then
$h\left(\kv;0\right)\ne0,\,\frac{\tilde{F}}{h}\left(\kv;0\right)=0$
and so
\begin{align}
\frac{\partial F}{\partial\kappa_{e}}\left(\kv\right)= & \frac{\partial\tilde{F}}{\partial\kappa_{e}}\left(\kv;0\right)=h\left(\kv;0\right)\frac{\partial}{\partial\kappa_{e}}\left(\frac{\tilde{F}}{h}\right)\left(\kv;0\right),\,\,\,\text{and}\label{eq: dfdk and dfda on Sigma}\\
\frac{\partial\tilde{F}}{\partial\alpha_{j}}\left(\kv;0\right)= & h\left(\kv;0\right)\frac{\partial}{\partial\alpha_{j}}\left(\frac{\tilde{F}}{h}\right)\left(\kv;0\right).
\end{align}
According to Proposition \ref{prop: nodal magnetic secular}, $\frac{\partial\tilde{F}}{\partial\alpha_{j}}\left(\kv;0\right)=0$
for any $\alpha_{j}$, so that the above equation gives $\frac{\partial}{\partial\alpha_{j}}\left(\frac{\tilde{F}}{h}\right)\left(\kv;0\right)=0$
and therefore (\ref{eq: Hessian mag with bridge}) is simplified:
\begin{equation}
\left(\hess_{\av}F\left(\kv\right)\right)_{j,i}=\frac{\partial^{2}\tilde{F}}{\partial\alpha_{j}\partial\alpha_{i}}\left(\kv;0\right)=h\left(\kv;0\right)\frac{\partial^{2}}{\partial\alpha_{j}\partial\alpha_{i}}\left(\frac{\tilde{F}}{h}\right)\left(\kv;0\right).
\end{equation}
We may now substitute $\frac{\tilde{F}}{h}\left(\kv;\av\right)=\left(1-e^{i2\kappa_{e}}e^{i\Theta_{1}\left(\kv_{1},\av_{1}\right)}e^{i\Theta_{2}\left(\kv_{2},\av_{2}\right)}\right)$
into the above expressions and get 
\begin{equation}
\frac{\left(\hess_{\av}F\left(\kv\right)\right)_{j,i}}{\frac{\partial F}{\partial\kappa_{e}}\left(\kv\right)}=\frac{\frac{\partial^{2}}{\partial\alpha_{j}\partial\alpha_{i}}\left(\frac{\tilde{F}}{h}\right)}{\frac{\partial}{\partial\kappa_{e}}\left(\frac{\tilde{F}}{h}\right)}\left(\kv;0\right)=\frac{\frac{\partial^{2}}{\partial\alpha_{j}\partial\alpha_{i}}\left(e^{i\Theta_{1}}e^{i\Theta_{2}}\right)}{2ie^{i\Theta_{1}}e^{i\Theta_{2}}}\left(\kv;0\right).
\end{equation}
Denote the entries of $\av_{1}$ and $\av_{2}$ by $\left\{ \alpha_{1,i}\right\} _{i=1}^{\beta_{1}}$
and $\left\{ \alpha_{2,j}\right\} _{j=1}^{\beta_{2}}$, and the cosponsoring
gradients by $\nabla_{\av_{1}}$ and $\nabla_{\av_{2}}$. Since both
$\nabla_{\av_{1}}\left(\frac{\tilde{F}}{h}\right)\left(\kv;0\right)$
and $\nabla_{\av_{2}}\left(\frac{\tilde{F}}{h}\right)\left(\kv;0\right)$
vanish, then so does $\nabla_{\av_{1}}e^{i\Theta_{1}}\left(\kv_{1};0\right)$
and $\nabla_{\av_{2}}e^{i\Theta_{2}}\left(\kv_{2};0\right)$. It follows
that the off diagonal blocks of $\hess_{\av}F$ vanish: 
\begin{equation}
\frac{\partial^{2}\tilde{F}}{\partial\alpha_{1,i}\partial\alpha_{2,j}}\left(\kv;0\right)\propto\frac{\partial^{2}}{\partial\alpha_{1,i}\partial\alpha_{2,j}}\left(e^{i\Theta_{1}}e^{i\Theta_{2}}\right)\left(\kv;0\right)=\frac{\partial}{\partial\alpha_{1,i}}e^{i\Theta_{1}}\left(\kv_{1};0\right)\frac{\partial}{\partial\alpha_{2,j}}e^{i\Theta_{2}}\left(\kv_{2};0\right)=0.
\end{equation}
On the first block, $\frac{\hess_{\av_{1}}F}{\frac{\partial F}{\partial\kappa_{e}}}\left(\kv\right)=\frac{1}{2}\hess_{\av_{1}}\Theta_{1}\left(\kv_{1}\right)$,
since 
\begin{align}
\frac{1}{\frac{\partial F}{\partial\kappa_{e}}}\frac{\partial^{2}\tilde{F}}{\partial\alpha_{1,i}\partial\alpha_{1,j}}\left(\kv;0\right)= & \frac{1}{2ie^{i\Theta_{1}}e^{i\Theta_{2}}}\frac{\partial^{2}\left(e^{i\Theta_{1}}e^{i\Theta_{2}}\right)}{\partial\alpha_{1,i}\partial\alpha_{1,j}}\left(\kv;0\right)\\
= & \frac{1}{2}\frac{\partial^{2}\Theta_{1}}{\partial\alpha_{1,i}\partial\alpha_{1,j}}\left(\kv;0\right).
\end{align}
And same for the second block, $\frac{\hess_{\av_{2}}F}{\frac{\partial F}{\partial\kappa_{e}}}\left(\kv\right)=\frac{1}{2}\hess_{\av_{2}}\Theta_{2}\left(\kv_{2}\right)$.
In particular, each block $\frac{\hess_{\av_{i}}F}{\frac{\partial F}{\partial\kappa_{e}}}$
is a function of its local coordinates $\kv_{i}$. 
\end{proof}
\begin{cor}
\label{cor: block diagonal}If $\Gamma$ has edge separation decomposition
$\left[\Gamma_{1},\Gamma_{2},...\Gamma_{m}\right]$, then for any
$\kv\in\Sigma_{\G}$, $\hess_{\av}F$ is block diagonal with respect
to $\av=\left(\av_{1},...\av_{m}\right)$. In particular, 
\[
\sum_{j=1}^{m}\boldsymbol{\iota}_{j}\left(\kv\right)=\boldsymbol{\sigma}\left(\kv\right).
\]
\end{cor}

\begin{proof}
Let $\kv\in\Sigma_{\G}$, by considering the block decomposition of
$\hess_{\av}F\left(\kv\right)$ for every bridge decomposition we
get that $\hess_{\av}F\left(\kv\right)$ is block diagonal with respect
to $\av=\left(\av_{1},...\av_{m}\right)$. It now follows from Proposition
\ref{prop: nodal magnetic secular} and Definition \ref{def: secular local magnetic index}
that $\sum_{j=1}^{m}\boldsymbol{\iota}_{j}\left(\kv\right)=\boldsymbol{\sigma}\left(\kv\right)$. 
\end{proof}
The proof of Theorem \ref{thm: local magnetic indices} follows:
\begin{proof}
If $\Gamma_{\lv}$ is a standard graph, $n\in\G$ and we denote $\kv=\left\{ k_{n}\lv\right\} \in\G$,
then we have showed in the proof of Proposition \ref{prop: nodal magnetic secular}
that $\hess_{\av}k_{n}\propto\hess_{\av}F\left(\kv\right)$. Therefore
$\hess_{\av}k_{n}$ is block diagonal with respect to $\av=\left(\av_{1},...\av_{m}\right)$,
and so $\M\left(\hess_{\av}k_{n}\right)=\sum_{j=1}^{m}\iota_{j}\left(n\right)$.
As $\M\left(\hess_{\av}k_{n}\right)=\sigma\left(n\right)$ by Theorem
\ref{thm: nodal magnetic}, we are done. 
\end{proof}
In order to prove Theorem \ref{thm: local mag indx 2}, we first need
to show the symmetry for $\boldsymbol{\iota}_{j}$. 
\begin{lem}
\label{lem: local index constant on connected components}If $\Gamma$
has edge separation decomposition $\left[\Gamma_{1},\Gamma_{2},...\Gamma_{m}\right]$,
then for any $\Gamma_{j}$, $\boldsymbol{\iota}_{j}$ is constant
on connected components of $\Sigma_{\G}$ and satisfies $\boldsymbol{\iota}_{j}\circ\I=\beta_{j}-\boldsymbol{\iota}_{j}$.
\end{lem}

\begin{proof}
According to Proposition \ref{prop: nodal magnetic secular}, $\frac{\hess_{\av}F}{p}$
is continuous on $\Sigma_{\G}$ with $\det\left(\frac{\hess_{\av}F}{p}\right)\ne0$
and $\frac{\hess_{\av}F}{p}\circ\I=-\frac{\hess_{\av}F}{p}$ on $\Sigma_{\G}$.
As $\frac{\hess_{\av}F}{p}$ is block diagonal , by Corollary \ref{cor: block diagonal},
then each of its blocks, $\frac{\hess_{\av_{j}}F}{p}$, is also continuous
with $\det\left(\frac{\hess_{\av_{j}}F}{p}\right)\ne0$ and\\ $\frac{\hess_{\av_{j}}F}{p}\circ\I=-\frac{\hess_{\av_{j}}F}{p}$
on $\Sigma_{\G}$. It follows that the eigenvalues of $\frac{\hess_{\av_{j}}F}{p}$
are real continuous and non-vanishing on $\Sigma_{\G}$ and therefore
$\boldsymbol{\iota}_{j}$, the number of positive eigenvalues, is
constant on connected components. As $\frac{\hess_{\av_{j}}F}{p}$
is of size $\beta_{j}$, then $\frac{\hess_{\av_{j}}F}{p}\circ\I=-\frac{\hess_{\av_{j}}F}{p}$
implies $\boldsymbol{\iota}_{j}\circ\I=\beta_{j}-\boldsymbol{\iota}_{j}$
on $\Sigma_{\G}$.
\end{proof}
We may now prove Theorem \ref{thm: local mag indx 2}:
\begin{proof}
Let $\Gamma_{\lv}$ be a standard graph with rationally independent
$\lv$ and edge separation decomposition $\left[\Gamma_{1},\Gamma_{2},...\Gamma_{m}\right]$.
According to (\ref{eq: local magnetic index n to secular}), for any
$\iota_{j}$ and any possible value $i$, the level set $\iota_{j}^{-1}\left(i\right)$
is given by 
\[
\iota_{j}^{-1}\left(i\right)=\set{n\in\G}{\left\{ k_{n}\lv\right\} \in\boldsymbol{\iota}_{j}^{-1}\left(i\right)}.
\]
Just as in the proof of Theorem \ref{thm:First}, we have showed in
Lemma \ref{lem: local index constant on connected components} that
$\boldsymbol{\iota}_{j}$ is constant on connected components of $\Sigma_{\G}$,
and so $\iota_{j}^{-1}\left(i\right)$ is a union of atoms of $\mathcal{F_{\G}}$,
which proves that $\iota_{j}$ is a random variable on $\left(\G,\mathcal{F}_{\G},d_{\G}\right)$.
Using Theorem \ref{thm: inversion BG} and Lemma \ref{lem: local index constant on connected components},
we get 
\begin{align*}
d_{\G}\left(\iota_{j}^{-1}\left(i\right)\right) & =\mu_{\lv}\left(\boldsymbol{\iota}_{j}^{-1}\left(i\right)\right)\\
= & \mu_{\lv}\left(\mathcal{I}\left(\boldsymbol{\iota}_{j}^{-1}\left(i\right)\right)\right)\\
= & \mu_{\lv}\left(\boldsymbol{\iota}_{j}^{-1}\left(\beta_{j}-i\right)\right)=d_{\G}\left(\iota_{j}^{-1}\left(\beta_{j}-i\right)\right).
\end{align*}
\end{proof}

\newpage{}
\section{\label{sec: Binomial} Binomial distributions and universality}

This section is the highlight of the thesis, in which we present a
result published in \cite{AloBanBer_cmp18} that proves the universality
conjecture of Gnutzmann Smilansky and Webber in \cite{GnuSmiWeb_wrm04}
for a certain family of graphs that we call \emph{trees of cycles}.
The method of our proof uses the probabilistic machinery developed
in Section \ref{sec: proof-of-existence} and the local indices of
Section \ref{sec: Magnetic} to conclude that the nodal surplus random
variable is a sum of local random variables. We then apply the symmetries
developed in Section \ref{sec: The-secular-manifold} to show that
the local random variables are independent, by which we prove that
the nodal surplus is binomial, and the central limit theorem ensures
that it will converge to a Gaussian in the limit of big graphs. Using
the same method, a similar result for Neumann count statistics was
achieved in \cite{AloBan19}, where the local indices are replaced
by different local random variables, the spectral position of Neumann
domains from Section \ref{sec: Properties-of-a Neumann domain}. This
motivated us to conjecture a universal behaviour for the Neumann statistics.
Let restate the (suitably modified) nodal statistics conjecture of
Gnutzmann Smilansky and Webber using our terminology (as appears in
the introduction): 
\begin{conjecture}
\label{conj: universality}Let $\left\{ \Gamma_{\lv}^{\left(\beta\right)}\right\} _{\beta\nearrow\infty}$
be any sequence of standard graph parameterized by their first Betti
numbers, and assume that each graph has rationally independent edge
lengths. Then the corresponding sequence of nodal surplus random variables
$\left\{ \sigma^{\left(\beta\right)}\right\} _{\beta\nearrow\infty}$
converge to a Gaussian distribution as follows:
\[
\frac{\sigma^{\left(\beta\right)}-\frac{\beta}{2}}{\sqrt{\mathrm{Var}\left(\sigma^{\left(\beta\right)}\right)}}\xrightarrow[\beta\rightarrow\infty]{\mathcal{D}}N\left(0,1\right).
\]
Where the convergence above is in distribution and the variances are
of order $\mathrm{Var}\left(\sigma^{\left(\beta\right)}\right)=\mathcal{O}\left(\beta\right)$.
\end{conjecture}

In a work in progress \cite{AloBanBer_conj}, we provide a vast numerical
evidence affirming the conjecture. We present some of these results in Figure \ref{fig: universality evidense}. In \cite{AloBanBer_conj} we also prove the convergence to Gaussian for some
other families of graphs which are not trees of cycles, which is the proof we present here. In \cite{AloBan19} we discuss an analogous
conjecture for Neumann count statistics, only we require the limit
where $\left|\V_{in}\right|\rightarrow\infty$. 

\begin{figure}

\includegraphics[width=0.6\paperwidth]{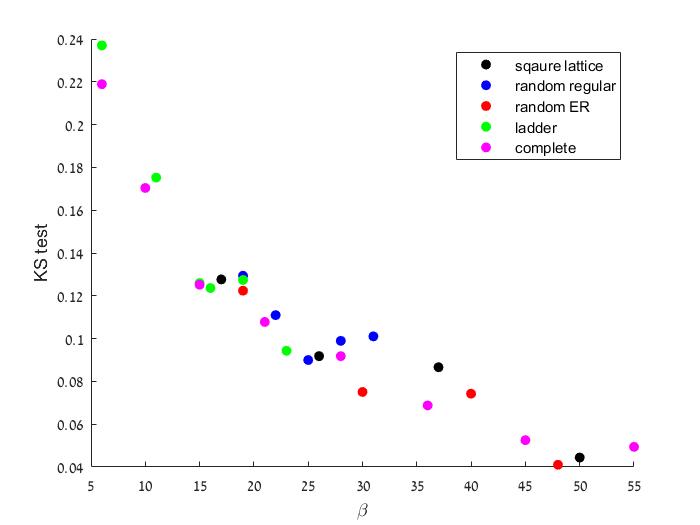}\\\includegraphics[width=0.6\paperwidth]{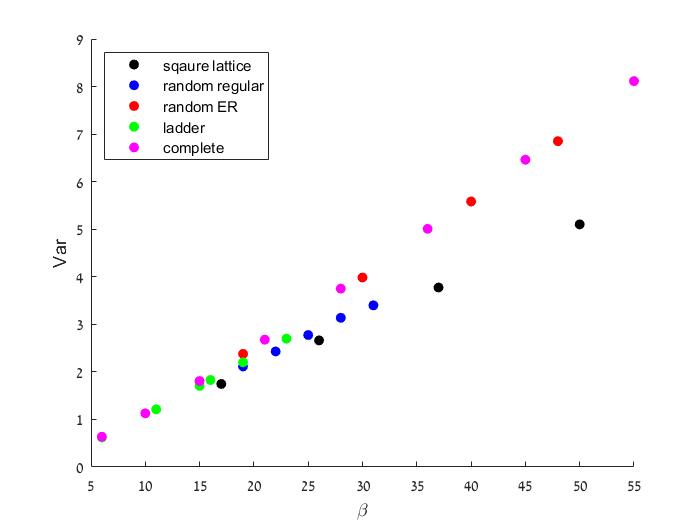}\caption[Universality evidence]{\label{fig: universality evidense} The results of several numerical experiments of growing families of graphs. For each graph we choose edge lengths at random and compute its nodal count statistics using $10^6$ eigenfunctions. The families of graphs we examine are complete graphs, ladder graphs, square lattices of the form $\Z^2/n\Z^2$, random Erdos-Reny graphs and random regular graphs. In the upper picture, the convergence to Gaussian is presented by means of the values of Kolmogorov–Smirnov tests versus $\beta$. In the lower picture, the variance growth in $\beta$ is presented.}

\end{figure}

\begin{defn}
We say that a graph $\Gamma$ is a \emph{finite (3,1)-regular tree
}if it is a (finite) tree with $\deg v=3$ for every $v\in\V_{in}$.
We say that a graph $\Gamma$ is a \emph{tree of cycles }if it has
an edge separation decomposition $\left[\Gamma_{1},\Gamma_{2}...\Gamma_{m}\right]$
where every $\Gamma_{j}$ has first Betti number $\beta_{j}=1$. See
Figure \ref{fig: tree of cycles} for example.
\end{defn}
\begin{figure}

\includegraphics[width=0.3\paperwidth]{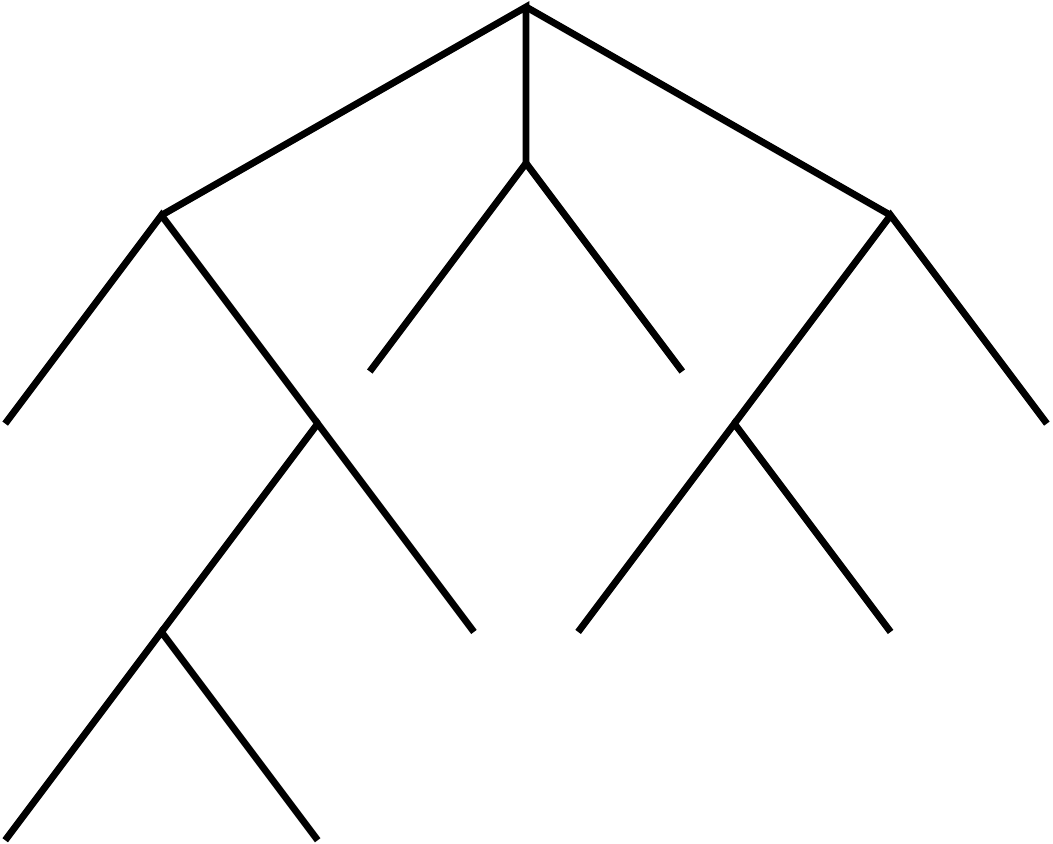}~~~~~~~~~~\includegraphics[width=0.3\paperwidth]{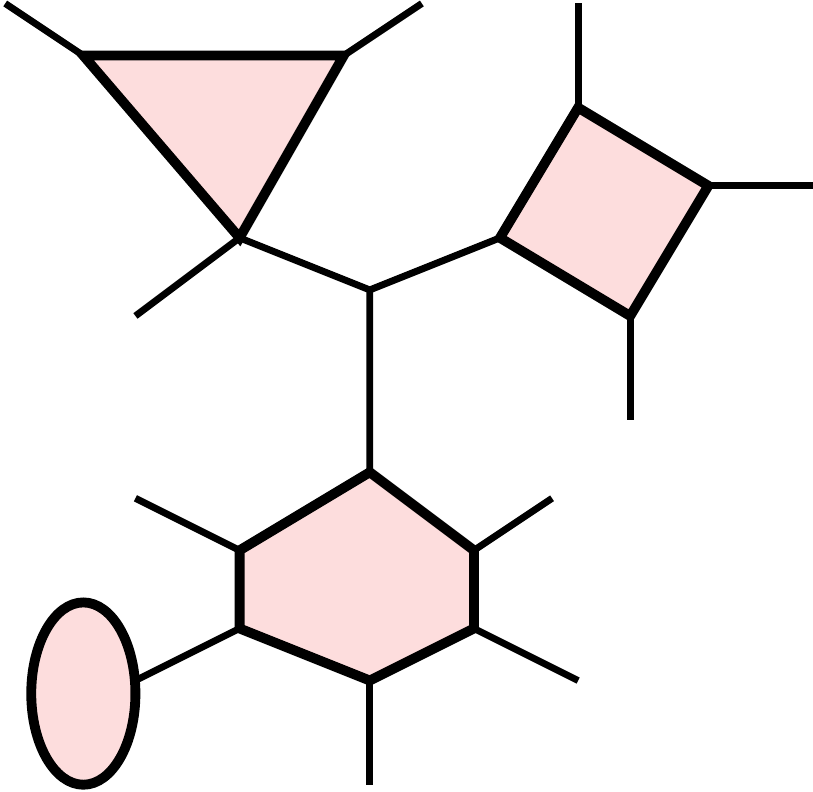}\caption[Tree of cycles]{\label{fig: tree of cycles} On the left, a (3,1)-regular~tree.~On
the right, a tree of cycles, with cycles filled for emphasis. ~}

\end{figure}

The following theorem combines Theorem 2.3 of \cite{AloBanBer_cmp18}
and Theorem 3.7 \cite{AloBan19}.
\begin{thm}
\label{thm:Second}Let $\Gamma_{\lv}$ be a standard graph and assume
that $\lv$ is rationally independent. Then,
\begin{enumerate}
\item \label{enu: binomial nodal surplus}If $\Gamma_{\lv}$ is a tree of
cycles, then the nodal surplus distribution is binomial, 
\[
\sigma\sim Bin\left(\beta,\frac{1}{2}\right).
\]
That is:
\begin{equation}
\forall j\in\left\{ 0,1,...,\beta\right\} \,\,P\left(\sigma-j\right)={\beta \choose j}2^{-\beta}.\label{eq:thm2}
\end{equation}
\item \label{enu:binomial Neumann surplus}If $\Gamma_{\lv}$ is a (3,1)-regular
finite tree, then the Neumann surplus distribution is given by $\omega+\left|\V_{in}\right|+1\sim Bin\left(\left|\V_{in}\right|,\frac{1}{2}\right)$.
That is:
\[
\forall j\in\left\{ -\left|\V_{in}\right|-1,...,-1\right\} \,\,\,\P(\omega=j)=\binom{\left|\V_{in}\right|}{j+\left|\V_{in}\right|+1}2^{-\left|\V_{in}\right|}.
\]
\end{enumerate}
\end{thm}

It is now straightforward, using central limit theorem, that the
nodal surplus distributions of trees of cycles converge to a Gaussian
limit. We also know that their variance in such case is $\frac{\beta}{4}$. 
\begin{cor}
Conjecture \ref{conj: universality} holds in the case where the graphs
are trees of cycles.
\end{cor}

As we said, the proof of Theorem \ref{thm:Second} relies on breaking
of $\sigma$ into sum of local indices and the breaking of $\omega$
into a sum of spectral positions of Neumann domains. The common feature
of both is that in the cases above each of these local random variables
gets only two values. A random variable $X$ that takes only two values,
0 and 1, is called a \emph{Bernoulli random variable. }We will now
define \emph{cut-flips} on tree graphs and provide a lemma regarding
the action of cut-flips on Bernoulli random variables.
\begin{defn}
\label{def: cut-flip} Let $\Gamma$ be a tree graph, consider a subset
of vertices $\V_{0}\subseteq\V$, and the binary set $\left\{ 0,1\right\} ^{\V_{0}}$.
Let $u\in\V$ (may or may not be in $\V_{0}$) with edge $e\in\E_{v}$
and decomposition $\Gamma\setminus e=\Gamma_{1}\sqcup\Gamma_{2}$
such that $u\in\Gamma_{1}$. \\We define the \emph{cut-flip }$g_{u,e}:\left\{ 0,1\right\} ^{\V_{0}}\rightarrow\left\{ 0,1\right\} ^{\V_{0}}$,
as seen in Figure \ref{fig: cutflips}, such that 
\[
\forall s\in\left\{ 0,1\right\} ^{\V_{0}},\,\,\forall v\in\V_{0}\,\,\,\,\,\,\left(g_{u,e}.s\right)_{v}=\begin{cases}
s_{v} & v\in\Gamma_{1}\\
1-s_{v} & v\in\Gamma_{2}
\end{cases}.
\]
Namely, $g_{u,e}$ fixes the values of $s$ on $\V_{0}\cap\Gamma_{1}$
and ``flips'' the values on $\V_{0}\cap\Gamma_{2}$.
\end{defn}

\begin{lem}
\label{lem: tree lemma}Let $\Gamma$ be a tree graph and consider
a subset of vertices $\V_{0}\subseteq\V$ with Bernoulli random variables
$\left\{ X_{v}\right\} _{v\in\V_{0}}$ assigned to $\V_{0}$. That
is, each $X_{v}$ takes the values 0 and 1, and their joint random
vector $\vec{X}$ takes its values in $\left\{ 0,1\right\} ^{\V_{0}}$.
We do not assume that the $X_{v}$'s are identical nor that they are
independent. 

If the joint probability distribution is invariant under all cut-flips,
namely
\begin{align}
\forall s\in\left\{ 0,1\right\} ^{\V_{0}},\,\forall u\in\V,\,\forall e\in\E_{u}\,\,\,\,\,\,P\left(\vec{X}=s\right) & =P\left(\vec{X}=g_{u,e}.s\right),\label{eq: Tree lemma assumption-1}
\end{align}
then $\left|\vec{X}\right|:=\sum_{v\in\V_{0}}X_{v}$, has binomial
distribution $X\sim Bin\left(\left|\V_{0}\right|,\frac{1}{2}\right)$.
\end{lem}

\begin{proof}
Let $e$ be an edge connecting the vertices $u_{1}$ and $u_{2}$.
Observe that applying both cut flips of that edge, $TF:=g_{u_{1},e}\circ g_{u_{1},e}$
gives a \emph{total flip:} 
\[
\forall s\in\left\{ 0,1\right\} ^{\V_{0}},\,\,\forall v\in\V_{0}\,\,\,\left(TF.s\right)_{v}=\left(g_{u_{1},e}.\left(g_{u_{2},e}.s\right)\right)=1-s_{v}.
\]
As demonstrated in Figure \ref{fig: cutflips}, given some $u\in\V_{0}$,
applying a total flip and then each $g_{u,e}$ for all $e\in\E_{v}$,
will result in a \emph{single flip}, $F_{u}$, flipping only $u$:
\[
\forall s\in\left\{ 0,1\right\} ^{\V_{0}},\,\,\forall v\in\V_{0}\,\,\,\left(F_{u}.s\right)_{v}=\begin{cases}
1-s_{v} & v=u\\
s_{v} & v\ne u
\end{cases}.
\]
Therefore, the group generated by all cut-flips $\left\langle \left\{ g_{u,e}\right\} _{u\in\V,\,e\in\E_{u}}\right\rangle $
contains every single flip and therefore acts transitively on $\left\{ 0,1\right\} ^{\V_{0}}$.
To see that, consider any two distinct elements $s,s'\in\left\{ 0,1\right\} ^{\V_{0}}$
and let $\left\{ v_{j}\right\} _{j=1}^{n}$ be the vertices in $\V_{0}$
on which $s_{v_{j}}\ne s'_{v_{j}}$. Let $g$ be the decomposition
of all single flips $F_{v_{j}}$, then clearly $g\in\left\langle \left\{ g_{u,e}\right\} _{u\in\V,\,e\in\E_{u}}\right\rangle $
and $g.s=s'$. Therefore, $\left\langle \left\{ g_{u,e}\right\} _{u\in\V,\,e\in\E_{u}}\right\rangle $
acts transitively on $\left\{ 0,1\right\} ^{\V_{0}}$, and since $\vec{X}$
takes values in $\left\{ 0,1\right\} ^{\V_{0}}$ with probability
which is invariant under the action of $\left\langle \left\{ g_{u,e}\right\} _{u\in\V,\,e\in\E_{u}}\right\rangle $
then $\vec{X}$ must be uniform. To see that, consider $s,s'\in\left\{ 0,1\right\} ^{\V_{0}}$
and $g\in\left\langle \left\{ g_{u,e}\right\} _{u\in\V,\,e\in\E_{u}}\right\rangle $
as before, such that $g.s=s'$. Since the probability is invariant
under $g$, then 
\[
P\left(\vec{X}=s\right)=P\left(\vec{X}=g.s\right)=P\left(\vec{X}=s'\right),
\]
and since this is true for any $s$ and $s'$, then 
\[
\forall s\in\left\{ 0,1\right\} ^{\V_{0}}\,\,\,P\left(\vec{X}=s\right)=\frac{1}{\left|\left\{ 0,1\right\} ^{\V_{0}}\right|}=2^{-\left|\V_{0}\right|}.
\]
As the probability of the sum, $P\left(\left|\vec{X}\right|=j\right)$,
is given by summing the probabilities of all elements $s\in\left\{ 0,1\right\} ^{\V_{0}}$
for which $\sum_{v\in\V_{0}}s_{v}=j$, then clearly, 
\[
P\left(\left|\vec{X}\right|=j\right)={\left|\V_{0}\right| \choose j}2^{-\left|\V_{0}\right|}.
\]
\end{proof}
\begin{figure}

\includegraphics[width=0.6\paperwidth]{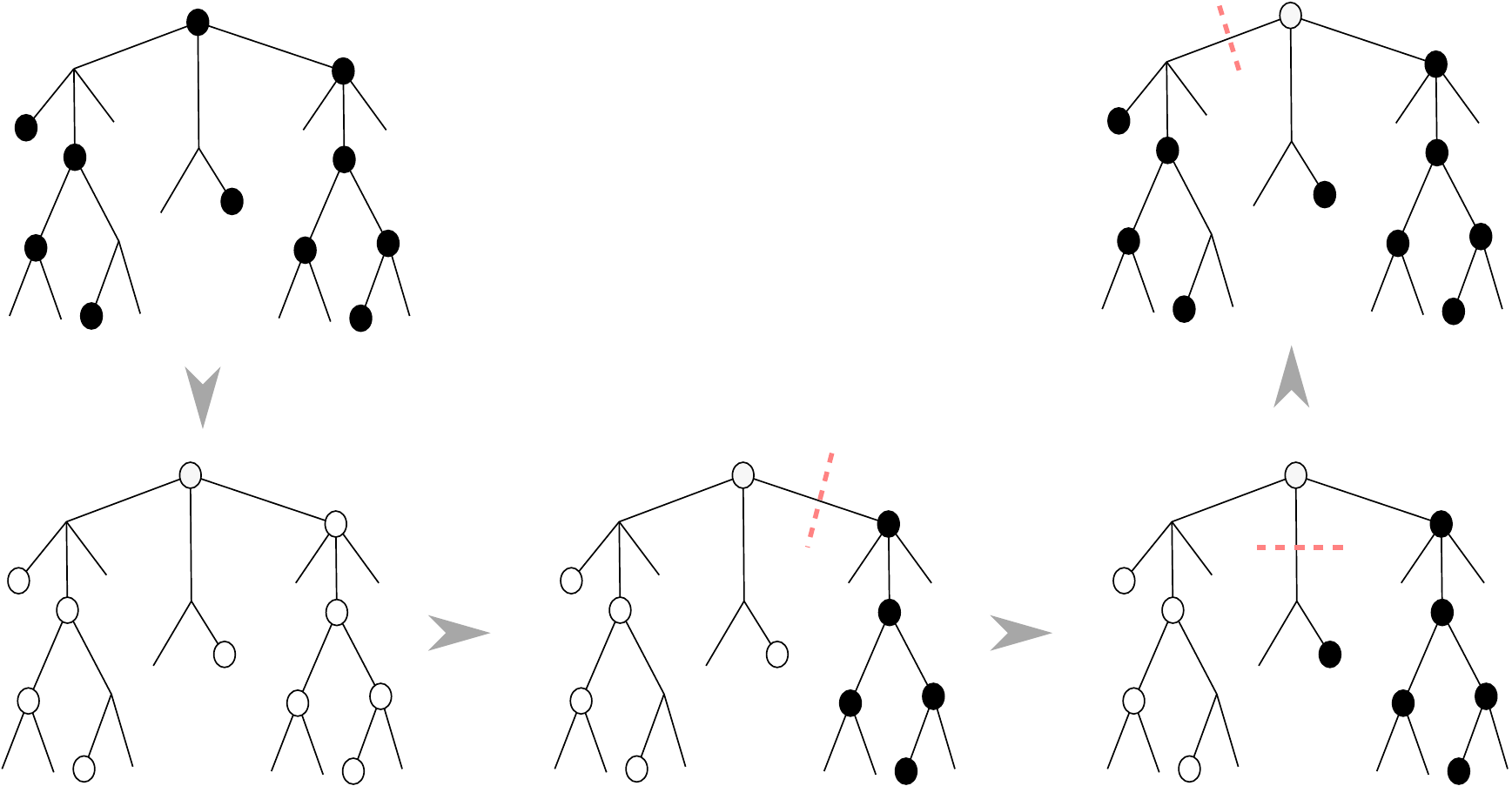}\caption[Cut-flips]{\label{fig: cutflips} A tree with a subset of vertices of two colors,
white or black. Starting with a configuration of all blacks, the first
step is a total flip. The second, third and forth steps are cut flips,
according to the dashed lines, that keep the root vertex unchanged.
The result is a flip only of the root vertex. }

\end{figure}

We will now show that $\Rev:\Sigma_{\G}\rightarrow\Sigma_{\G}$, that
was defined in Definition \ref{def: Tau_=00007Be=00007D} by 
\[
\Rev\left(\kv_{1},\kappa_{e},\kv_{2}\right):=\left\{ \left(\kv_{1},\kappa_{e}+\Theta_{2}\left(\kv_{2}\right),\I\left(\kv_{2}\right)\right)\right\} ,
\]
is acting similarly to a cut-flip on the level sets of the local random
variables $\boldsymbol{N}_{u}$ and $\boldsymbol{\iota}_{j}$ (see
Definitions \ref{def: Secular N and rho} and \ref{def: secular local magnetic index}). 
\begin{lem}
\label{lem: R_e}Let $\Gamma$ be a graph, let $e$ be a bridge with
a bridge decomposition\\ $\Gamma\setminus e=\Gamma_{1}\sqcup\Gamma_{2}$
and let $v$ be the vertex connecting $e$ to $\Gamma_{1}$.Then,
\begin{enumerate}
\item \label{enu: Re on Nv} For any interior vertex $u\in\V_{in}$, 
\[
\boldsymbol{N}_{u}\circ\Rev=\begin{cases}
\boldsymbol{N}_{u} & u\in\V_{1}\\
\deg u-\boldsymbol{N}_{u} & u\notin\V_{1}.
\end{cases}
\]
\item \label{enu: Re on magnetic index} If $\Gamma$ has further edge separation
$\left[\Gamma_{1,1},\Gamma_{1,2},...\Gamma_{1,m_{1}},\Gamma_{2,1},\Gamma_{2,2}...\Gamma_{2,m_{2}}\right]$,
such that $\Gamma_{1,j}\subset\Gamma_{1}$ and $\Gamma_{2,i}\subset\Gamma_{2}$.
Then the corresponding local magnetic indices $\boldsymbol{\iota}_{1,j}$
and $\boldsymbol{\iota}_{2,i}$ satisfy 
\begin{align*}
\boldsymbol{\iota}_{1,j}\circ\Rev & =\boldsymbol{\iota}_{1,j}\\
\boldsymbol{\iota}_{2,i}\circ\Rev & =\beta_{2,i}-\boldsymbol{\iota}_{2,i}.
\end{align*}
Where $\beta_{2,i}$ is the first Betti number of $\Gamma_{2,i}$.
\end{enumerate}
\end{lem}

\begin{proof}
In order to prove (\ref{enu: Re on Nv}), let $\kv\in\Sigma_{\G}$
and notice that according to Lemma \ref{lem: Re and tau e are measure preseving},
for any $e'\in\E_{u}$, 
\[
f_{\Rev\left(\kv\right)}\left(u\right)\partial_{e'}f_{\Rev\left(\kv\right)}\left(u\right)=\begin{cases}
f_{\kv}\left(u\right)\partial_{e'}f_{\kv}\left(u\right) & u\in\V_{1}\\
-f_{\kv}\left(u\right)\partial_{e'}f_{\kv}\left(u\right) & u\notin\V_{1}.
\end{cases}
\]
Therefore, 
\begin{align*}
\boldsymbol{N}_{u}\left(\Rev\left(\kv\right)\right) & =\frac{\deg u}{2}-\frac{1}{2}\sum_{e'\in\E_{u}}\mathrm{sign}\left(f_{\Rev\left(\kv\right)}\left(u\right)\partial_{e'}f_{\Rev\left(\kv\right)}\left(u\right)\right)\\
= & \begin{cases}
\frac{\deg u}{2}-\frac{1}{2}\sum_{e'\in\E_{u}}\mathrm{sign}\left(f_{\kv}\left(u\right)\partial_{e'}f_{\kv}\left(u\right)\right) & u\in\V_{1}\\
\frac{\deg u}{2}+\frac{1}{2}\sum_{e'\in\E_{u}}\mathrm{sign}\left(f_{\kv}\left(u\right)\partial_{e'}f_{\kv}\left(u\right)\right) & u\notin\V_{1}.
\end{cases}\\
= & \begin{cases}
\boldsymbol{N}_{u}\left(\kv\right) & u\in\V_{1}\\
\deg u-\boldsymbol{N}_{u}\left(\kv\right) & u\notin\V_{1}.
\end{cases}.
\end{align*}
In order to prove (\ref{enu: Re on magnetic index}), let $\kv=\left(\kv_{1},\kappa_{e},\kv_{2}\right)\in\Sigma_{\G}$
and consider the block decomposition of $\frac{\hess_{\av}F}{p}$
into $\frac{\hess_{\av_{1}}F}{p}$ and $\frac{\hess_{\av_{2}}F}{p}$.
In Lemma \ref{lem:The-secular-functions properties} we showed that
\begin{align*}
\frac{\partial F}{\partial\kappa_{e}}\left(\kv\right)= & p\left(\kv\right)\left(m_{\kv}\right)_{e},\,\,and\,so\\
\frac{\partial F}{\partial\kappa_{e}}\left(\Rev\left(\kv\right)\right)= & p\left(\Rev\left(\kv\right)\right)\left(m_{\Rev\left(\kv\right)}\right)_{e},
\end{align*}
where both $\left(m_{\kv}\right)_{e}$ and $\left(m_{\Rev\left(\kv\right)}\right)_{e}$
are strictly positive (since $\kv\in\Sigma_{\G}$). Denote the positive
scalar $c:=\frac{\left(m_{\Rev\left(\kv\right)}\right)_{e}}{\left(m_{\kv}\right)_{e}}$.
As $\Rev\left(\kv\right)=\left(\kv_{1},\tilde{\kappa_{e}},\I\left(\kv_{2}\right)\right)$
for some $\tilde{\kappa_{e}}$, then $\kv$ and $\Rev\left(\kv\right)$
agree on their $\kv_{1}$ coordinates. It then follows from Proposition
\ref{prop: block decomposition} that, 
\[
\frac{\hess_{\av_{1}}F}{\frac{\partial F}{\partial\kappa_{e}}}\left(\Rev\left(\kv\right)\right)=\frac{\hess_{\av_{1}}F}{\frac{\partial F}{\partial\kappa_{e}}}\left(\kv\right),
\]
which means that 
\[
\frac{\hess_{\av_{1}}F}{p}\left(\Rev\left(\kv\right)\right)=c\frac{\hess_{\av_{1}}F}{p}\left(\kv\right).
\]
Therefore, the Morse indices of their sub-block agree: 
\[
\boldsymbol{\iota}_{1,j}\left(\Rev\left(\kv\right)\right)=\M\left(-\frac{\hess_{\av_{1,j}}F}{p}\left(\Rev\left(\kv\right)\right)\right)=\M\left(-c\frac{\hess_{\av_{1,j}}F}{p}\left(\kv\right)\right)=\boldsymbol{\iota}_{1,j}\left(\kv\right),
\]
In the same way, $\I\left(\kv\right)$ and $\Rev\left(\kv\right)$
agree on their $\kv_{2}$ coordinates, so 
\[
\frac{\hess_{\av_{2}}F}{\frac{\partial F}{\partial\kappa_{e}}}\left(\Rev\left(\kv\right)\right)=\frac{\hess_{\av_{2}}F}{\frac{\partial F}{\partial\kappa_{e}}}\left(\I\kv\right)=-\frac{\hess_{\av_{2}}F}{\frac{\partial F}{\partial\kappa_{e}}}\left(\kv\right),
\]
using Proposition \ref{prop: block decomposition} once more, for
the inversion. Therefore, 
\[
\boldsymbol{\iota}_{2,i}\left(\Rev\left(\kv\right)\right)=\M\left(-\frac{\hess_{\av_{2,i}}F}{p}\left(\Rev\left(\kv\right)\right)\right)=\M\left(c\frac{\hess_{\av_{2,i}}F}{p}\left(\kv\right)\right)=\beta_{2,i}-\boldsymbol{\iota}_{2,i}\left(\kv\right).
\]
\end{proof}
Using both Lemmas \ref{lem: tree lemma} and \ref{lem: R_e}, we can
prove Theorem \ref{thm:Second}. Let us first consider the Neumann
case.

\subsection{Proof of Theorem \ref{thm:Second} (\ref{enu:binomial Neumann surplus}).}

Let $\Gamma_{\lv}$ be a (3,1)-regular finite tree with rationally
independent $\lv$. Consider the subset of interior vertices $\V_{0}=\V_{in}$
and define $X_{v}:=\boldsymbol{N}_{v}-1$ for any $v\in\V_{in}$ such
that $\left\{ X_{v}\right\} _{v\in\V_{in}}$ are random variables
on $\Sigma_{\G}$ with Borel $\sigma$-algebra and BG measure $\frac{1}{\mu_{\lv}\left(\Sigma_{\G}\right)}\mu_{\lv}$.
According to Proposition \ref{prop:local_observables_bounds} and
since every $v\in\V_{in}$ is of $\deg v=3$, then $\boldsymbol{N}_{v}$
takes the values $1$ and $2$ and therefore $X_{v}$ takes the values
0 and 1, and is therefore Bernoulli random variable. Consider $v\in\V,\,e\in\E_{v}$
and the partition $\Gamma\setminus\left\{ e\right\} =\Gamma_{1}\sqcup\Gamma_{2}$
such that $v\in\Gamma_{1}$, with corresponding cut-flip $g_{v,e}$
(as defined in Lemma \ref{lem: tree lemma}). Consider the level sets
\[
\forall s\in\left\{ 0,1\right\} ^{\V_{0}}\,\,\,\vec{X}^{-1}\left(s\right):=\set{\kv\in\Sigma_{\G}}{\forall v\in\V_{0}\,\,\,X_{v}\left(\kv\right)=s_{v}}.
\]
According to Lemma \ref{lem: R_e}, 
\begin{align*}
X_{u}\left(\Rev\left(\kv\right)\right) & =\begin{cases}
X_{u}\left(\kv\right) & u\in\Gamma_{1}\\
1-X_{u}\left(\kv\right) & u\in\Gamma_{2}
\end{cases},
\end{align*}
and a simple observation leads to
\[
\forall s\in\left\{ 0,1\right\} ^{\V_{0}}\,\,\,\,\Rev\left(\vec{X}^{-1}\left(s\right)\right)=\vec{X}^{-1}\left(g_{v,e}.s\right).
\]
By Lemma \ref{lem: Re and tau e are measure preseving}, $\Rev$ is
$\mu_{\lv}$ preserving and therefore 
\begin{align*}
\forall s\in\left\{ 0,1\right\} ^{\V_{0}}\,\,\,\,P\left(\vec{X}=s\right)= & \frac{\mu_{\lv}\left(\vec{X}^{-1}\left(s\right)\right)}{\mu_{\lv}\left(\Sigma_{\G}\right)}\\
= & \frac{\mu_{\lv}\left(\Rev\left(\vec{X}^{-1}\left(s\right)\right)\right)}{\mu_{\lv}\left(\Sigma_{\G}\right)}\\
= & \frac{\mu_{\lv}\left(\vec{X}^{-1}\left(g_{v,e}.s\right)\right)}{\mu_{\lv}\left(\Sigma_{\G}\right)}\\
= & P\left(\vec{X}=g_{v,e}.s\right).
\end{align*}
Therefore, by Lemma \ref{lem: tree lemma}, $\left|\vec{X}\right|:=\sum_{v\in\V_{in}}X_{u}$
is binomial $\left|\vec{X}\right|\sim Bin\left(\left|\V_{in}\right|,\frac{1}{2}\right)$.
According to Proposition \ref{prop: local to global} and Lemma \ref{lem: Nv rhov-1},
for every $n\in\G$ such that $n>2\frac{L}{L_{min}}$, 
\begin{equation}
\sum_{u\in\V_{in}}\boldsymbol{N}_{u}\left(\left\{ k_{n}\lv\right\} \right)=\phi\left(f_{n}\right)-\mu\left(f_{n}\right)+E-\left|\partial\Gamma\right|.
\end{equation}
Since $\Gamma$ is a tree, then $\phi\left(f_{n}\right)=n$, and therefore
$\phi\left(f_{n}\right)-\mu\left(f_{n}\right)=-\omega\left(f_{n}\right)$
which is equal to $\boldsymbol{\omega}\left(\left\{ k_{n}\lv\right\} \right)$
(by Lemma \ref{lem: counts minus weil}). As all $\boldsymbol{N}_{u}$'s
and $\boldsymbol{\omega}$ are constant on connected components, and
$\left\{ k_{n}\lv\right\} $ is dense in $\Sigma$ (Theorem \ref{thm: BG equidistribution}),
then the following relation holds on $\Sigma_{\G}$: 
\[
\boldsymbol{\omega}=-\sum_{u\in\V_{in}}\boldsymbol{N}_{u}+E-\left|\partial\Gamma\right|.
\]
In particular, since $\Gamma$ is a tree so $E-V=-1$, we get 
\[
\boldsymbol{\omega}=-\sum_{u\in\V_{in}}X_{u}+E-V=-\left|\vec{X}\right|-1.
\]
We may now deduce that $-\boldsymbol{\omega}-1\sim Bin\left(\left|\V_{in}\right|,\frac{1}{2}\right)$
and therefore so does $\left|\V_{in}\right|-\left(-\boldsymbol{\omega}-1\right)=\boldsymbol{\omega}+\left|\V_{in}\right|+1$.
This proves that 
\[
\forall j\in\left\{ -\left|\V_{in}\right|-1,...,-1\right\}\ \ \  \,\,\,\frac{\mu_{\lv}\left(\boldsymbol{\omega}^{-1}\left(j\right)\right)}{\mu_{\lv}\left(\Sigma_{\G}\right)}=\binom{\left|\V_{in}\right|}{j+\left|\V_{in}\right|+1}2^{-\left|\V_{in}\right|}.
\]
This is the needed result as the Neumann surplus probability is given
by \[P\left(\omega=j\right)=\frac{\mu_{\lv}\left(\boldsymbol{\omega}^{-1}\left(j\right)\right)}{\mu_{\lv}\left(\Sigma_{\G}\right)}.\]

\subsection{Proof of Theorem \ref{thm:Second} (\ref{enu: binomial nodal surplus}).}

Let $\Gamma_{\lv}$ be a standard graph with rationally independent
$\lv$, and assume it is a tree of cycles. Denote the set of bridges
by $\E_{bridges}$, and consider the edge separation $\left[\Gamma_{1},\Gamma_{2}...\Gamma_{m}\right]$.
As $\Gamma_{\lv}$ is a tree of cycles, then $m=\beta$ and each $\Gamma_{j}$
has first Betti number $\beta_{j}=1$. It now follows from Theorem
\ref{thm: local magnetic indices} that every $\boldsymbol{\iota}_{j}$
satisfies 
\[
\frac{1}{\mu_{\lv}\left(\Sigma_{\G}\right)}\mu_{\lv}\left(\boldsymbol{\iota}_{j}^{-1}\left(0\right)\right)=\frac{1}{\mu_{\lv}\left(\Sigma_{\G}\right)}\mu_{\lv}\left(\boldsymbol{\iota}_{j}^{-1}\left(1\right)\right)=\frac{1}{2}.
\]
Consider $\left\{ \boldsymbol{\iota}_{j}\right\} _{j=1}^{\beta}$
as random variables over $\Sigma_{\G}$ with Borel $\sigma$-algebra
and probability measure $\frac{1}{\mu_{\lv}\left(\Sigma_{\G}\right)}\mu_{\lv}$.
Let $\tilde{\Gamma}$ be an auxiliary graph whose edges are $\E_{edges}$
and whose vertices are the connected components of $\Gamma\setminus\E_{edges}$
and let $\V_{0}=\left\{ v_{j}\right\} _{j=1}^{\beta}$ be the vertices
corresponding to the non-trivial connected components $\left\{ \Gamma_{j}\right\} _{j=1}^{\beta}$,
as demonstrated in Figure \ref{fig: TOC proof}. Clearly by the construction,
$\tilde{\Gamma}$ is a tree. Let $\left\{ \boldsymbol{\iota}_{j}\right\} _{v_{j}\in\V_{0}}$
be the random variables associated to the vertices in $\V_{0}$ as
in Lemma \ref{lem: tree lemma}, with probability $P\left(\boldsymbol{\iota}_{j}=i\right)=\frac{1}{\mu_{\lv}\left(\Sigma_{\G}\right)}\mu_{\lv}\left(\boldsymbol{\iota}_{j}^{-1}\left(i\right)\right)$.
Let $\vec{\boldsymbol{\iota}}$ be their joint vector, taking values
in $\left\{ 0,1\right\} ^{\beta}$, with probability given by 
\[
\forall s\in\left\{ 0,1\right\} ^{\beta}\,\,P\left(\vec{\boldsymbol{\iota}}=s\right)=\frac{1}{\mu_{\lv}\left(\Sigma_{\G}\right)}\mu_{\lv}\left(\cap_{v_{j}\in\V_{0}}\boldsymbol{\iota}_{j}^{-1}\left(s_{v_{j}}\right)\right).
\]
Let $e\in\E_{bridges}$ and let $v\in\V$ be a vertex of $\Gamma$
connected to $e$ with $\tilde{v}\in\tilde{\V}$ being the corresponding
vertex of $\tilde{\Gamma}$ (the connected component that contain
$v$). Consider the bridge decomposition $\tilde{\Gamma}\setminus\left\{ e\right\} =\tilde{\Gamma}_{1}\sqcup\tilde{\Gamma}_{2}$
such that $\tilde{v}\in\tilde{\Gamma}_{1}$, and let $g_{\tilde{v},e}$
be the corresponding cut-flip (on $\tilde{\Gamma}$). It follows from
Lemma \ref{lem: R_e}, that every $\boldsymbol{\iota}_{j}$ satisfies:
\begin{align*}
\boldsymbol{\iota}_{j}\left(\Rev\left(\kv\right)\right)= & \begin{cases}
\boldsymbol{\iota}_{j}\left(\kv\right) & v_{j}\in\tilde{\Gamma}_{1}\\
1-\boldsymbol{\iota}_{j}\left(\kv\right) & v_{j}\in\tilde{\Gamma}_{2}
\end{cases},
\end{align*}
and exactly as in the proof of Theorem \ref{thm:Second} (\ref{enu:binomial Neumann surplus}),
it follows that
\[
\forall s\in\left\{ 0,1\right\} ^{\beta}\,\,\,\,\Rev\left(\vec{\boldsymbol{\iota}}^{-1}\left(s\right)\right)=\vec{\boldsymbol{\iota}}^{-1}\left(g_{\tilde{v},e}.s\right),
\]
which leads to 
\[
\forall s\in\left\{ 0,1\right\} ^{\beta}\,\,\,\,P\left(\vec{\boldsymbol{\iota}}=s\right)=P\left(\vec{\boldsymbol{\iota}}=g_{\tilde{v},e}.s\right).
\]
And then Lemma \ref{lem: tree lemma} can be applied to show that
the sum $\left|\vec{\boldsymbol{\iota}}\right|=\sum_{j=1}^{\beta}\boldsymbol{\iota}_{j}$
has binomial distribution $\left|\vec{\boldsymbol{\iota}}\right|\sim Bin\left(\beta,\frac{1}{2}\right)$.
By Corollary \ref{cor: block diagonal}, $\boldsymbol{\sigma}=\sum_{j=1}^{\beta}\boldsymbol{\iota}_{j}=\left|\vec{\boldsymbol{\iota}}\right|$
which means that $\boldsymbol{\sigma}\sim Bin\left(\beta,\frac{1}{2}\right)$.
So, 
\[
\forall j\in\left\{ 0,1,...\beta\right\} \,\,\frac{1}{\mu_{\lv}\left(\Sigma_{\G}\right)}\mu_{\lv}\left(\boldsymbol{\sigma}^{-1}\left(j\right)\right)={\beta \choose j}2^{-\beta},
\]
as needed since $P\left(\sigma=j\right)=\frac{1}{\mu_{\lv}\left(\Sigma_{\G}\right)}\mu_{\lv}\left(\boldsymbol{\sigma}^{-1}\left(j\right)\right)$.

\begin{figure}
\includegraphics[width=0.6\paperwidth]{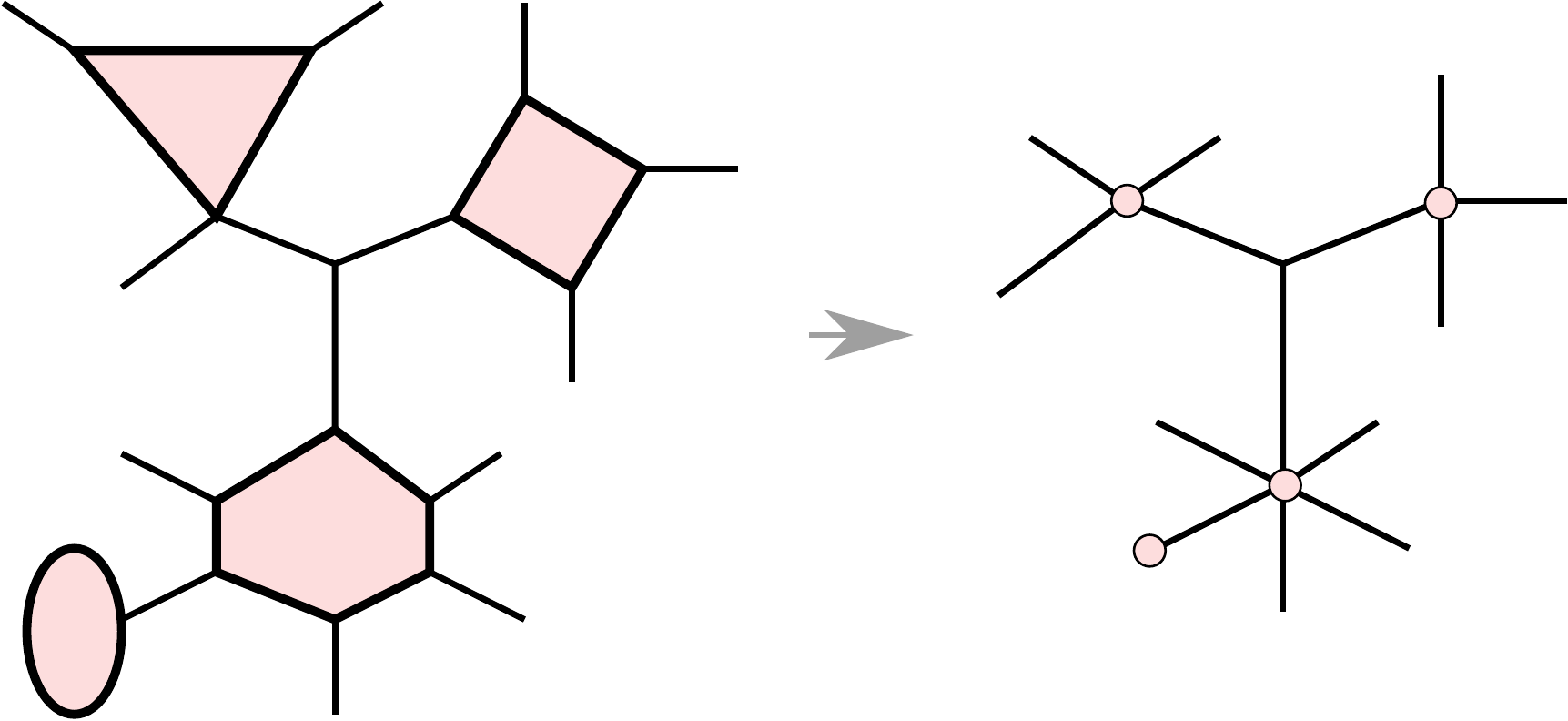}

\caption[The auxiliary tree of a tree of cycles ]{\label{fig: TOC proof}On the left, $\Gamma$ a tree of cycles, with
cycles filled for emphasis. On the right, $\tilde{\Gamma}$ the constructed
tree, with the corresponding $\protect\V_{0}$ chosen vertices marked
by circles. }
\end{figure}
\newpage{}
\section{Summary}{}
This thesis deals with generic eigenfunctions of standard graphs and the statistics of their nodal and Neumann counts. The nodal count $\phi(n)$ is the number of points where the nth eigenfunction vanishes. The Neumann count $\mu(n)$ is the number of local extrema of the nth eigenfunction. Both nodal and Neumann counts cannot be defined on every eigenfunction and we restrict the discussion to generic eigenfunctions on which the two counts are well defined. We call an eigenfunction generic if it corresponds to a simple eigenvalue, does not vanish on vertices, and its outgoing derivatives at every interior vertex do not vanish as well.\\ 
\par{}The first main result of this thesis, Theorem \ref{thm: density-of-generic-and-loop-eigenfunctions}, is a generalization of the genericity results in \cite{BerLiu_jmaa17,Fri_ijm05}, in which we justify the generality of our discussion. We prove, using two different notions of genericity, that generically every eigenfunction which is not supported on a loop (if such exist) is generic and therefore has well defined nodal and Neumann counts. Given a standard quantum graph $\Gamma_{\lv{}}$, we denote the index set of $n$'s for which the nth eigenfunction is generic by $\G$, and we denote its intersection with  $\{1,2,...N\}$ by $\G(N)$. As was shown in \cite{Ber_cmp08,GnuSmiWeb_wrm04}, the nodal surplus $\sigma(n):=\phi(n)-n$ is uniformly bounded between $0$ and $\beta$, the first Betti number of the graph. Therefore, the statistics of the nodal count is encapsulated in the nodal surplus distribution. The nodal surplus distribution is well defined if the following limits exist for any $j\in\{0,1,...\beta\}$:
\begin{align}
P\left(\sigma=j\right) & :=\lim_{N\rightarrow\infty}\frac{\left|\set{n\in\G\left(N\right)}{\sigma(n)=j}\right|}{\left|\G\left(N\right)\right|}
\end{align}
Similarly, we define the Neumann surplus $\omega(n):=\mu(n)-n$ and in Theorem \ref{thm:Neumann_surplus_main-1} we provide a uniform bound:
\begin{equation}
\forall n\in\G\,\,\,\,\,\,\,\,1-\beta-\left|\d\Gamma\right|\leq\omega(n)\leq2\beta-1.
\end{equation}
Hence, the Neumann count statistics is described by the Neumann surplus distribution, which is well defined if the following limits exist for any $j$ in the allowed range of $\omega$:
\begin{align}
P\left(\omega=j\right) & :=\lim_{N\rightarrow\infty}\frac{\left|\set{n\in\G\left(N\right)}{\omega(n)=j}\right|}{\left|\G\left(N\right)\right|}
\end{align}
\par{}The second main result of this thesis is Theorem \ref{thm:First}, in which we prove that if the edge lengths of the graph are rationally independent then the above limits exist and so the nodal and Neumann surplus distributions are well defined. We do so by providing a probabilistic setting for the statistical discussion. We define a $\sigma$-algebra $\mathcal{F}_{\G}$ on $\G$ such that $d_{\G}$ the (restricted) natural density is a measure on $(\G,\mathcal{F}_{\G})$. We then prove that the functions $\sigma, \omega:\G\rightarrow\Z$ are (finite) random variables on $(\G,\mathcal{F}_{\G},d_{\G})$. This proves that the distributions are well defined, and moreover, it allows to consider their joint distribution. In Theorem \ref{thm:First} we also show that both $\sigma$
and $\omega$ are symmetric around $\frac{\beta}{2}$ and $\frac{\beta-\left|\partial\Gamma\right|}{2}$
simultaneously. That is, 
\begin{align}
P\left((\sigma,\omega)=(j,i)\right) & =P\left((\sigma,\omega)=(\beta-j,\beta-\left|\partial\Gamma\right|-i)\right).
\end{align}
As a corollary, if the edge lengths of a graph are rationally independent then both $\beta$ and $\left|\partial\Gamma\right|$ can be obtained by the averages of the nodal and Neumann surplus distributions:
\begin{align*}
\beta= 2 \mathbb{E}\left(\sigma\right) & =\lim_{N\rightarrow\infty}\frac{2}{\left|\G(N)\right|}\sum_{n\in\G(N)}\sigma(n) ,\,\,\text{and}\\
\left|\partial\Gamma\right|=2 \mathbb{E}\left(\sigma-\omega\right) & =\lim_{N\rightarrow\infty}\frac{2}{\left|\G(N)\right|}\sum_{n\in\G(N)}(\sigma(n)-\omega(n)).
\end{align*}\\
There are two immediate consequences of the above, regarding the "geometric information" of the graph that is stored in the nodal and Neumann counts. The first, is that although counting zero points and counting extrema seems equivalent, it is not. For example consider tree graphs, so that $\sigma\equiv0$ and cannot distinguish between different trees, but the average Neumann surplus distinguishes between trees of different $\left|\partial\Gamma\right|$. The second consequence is that given both the average nodal surplus and the average Neumann surplus, we obtain  $\beta$ and $\left|\partial\Gamma\right|$. It can be shown that there are finitely many discrete graphs with a given $\beta$ and $\left|\partial\Gamma\right|$, and so the inverse question of retrieving the underlying discrete is narrowed down to a finite set of graphs, given only the averages of the nodal and Neumann surpluses.\\
\par{} The third main result is Theorem \ref{thm:Second} in which we explicitly calculate the nodal and Neumann surplus distributions for two specific families of graphs. If a graph is a tree of cycles (see Definition \ref{def: tree of cycles}) with rationally independent edge lengths, then its nodal surplus distribution is binomial, 
\[
\sigma\sim Bin\left(\beta,\frac{1}{2}\right).
\]
If a graph is a (3,1)-regular
finite tree (see Definition \ref{def: tree of cycles}) with rationally independent edge lengths then its Neumann surplus distribution is a shifted binomial distribution given by \[\omega+\left|\V_{in}\right|+1\sim Bin\left(\left|\V_{in}\right|,\frac{1}{2}\right),\]
where $\left|\V_{in}\right|=V-\left|\partial\Gamma\right|$ is the number of interior vertices.\\ 
The binomial distribution converge to Gaussian by the Central Limit Theorem. Therefore, the conjectured universal behavior of the nodal surplus statistics, as reformulated in Conjecture \ref{conj: universality}, holds for trees of cycles.

\newpage{}

\appendix

\section{\label{sec: appendix bridge}Decomposition for a bridge}

In this subsection we provide the technical details of the bridge
decomposition. As already discussed, more general results can be found
in Section 4 of \cite{AloBanBer_cmp18}. We will follow the method
of Section 4 of \cite{AloBanBer_cmp18}, which relies on the scattering
results obtained in \cite{BanBerSmi_ahp12}.

Let $e$ be a bridge of $\Gamma$ and consider the bridge decomposition
$\Gamma\setminus\left\{ e\right\} =\Gamma_{1}\sqcup\Gamma_{2}$. Denote
the edges of $\Gamma_{1},\Gamma_{2}$ by $\E_{1},\E_{2}$ correspondingly.
Consider the direction of $e$ from $\Gamma_{1}$ to $\Gamma_{2}$
and the basis of directed edges in the following order: directed edges
of $\Gamma_{1}$, $e$, $\hat{e}$, directed edges of $\Gamma_{2}$.
With this order and edge grouping the real orthogonal scattering matrix
$S$ has the following block structure:
\[
S=\begin{pmatrix}S_{1} & 0 & t_{1} & 0\\
t'_{1} & 0 & r_{1} & 0\\
0 & r_{2} & 0 & t'_{2}\\
0 & t_{2} & 0 & S_{2}
\end{pmatrix}.
\]
Where for $j\in\left\{ 1,2\right\} $, $S_{j}$ is square matrix of
dimension $2\left|\E_{j}\right|$, $t_{j}$ and $t_{j}'$ are column
and row vectors of dimension $2\left|\E_{j}\right|$ and $r_{j}$
is a scalar. Let $\zone,\ztwo$ be unitary diagonal matrices of dimensions
$2\left|\E_{1}\right|$ and $2\left|\E_{2}\right|$and let $\ze$
be uni-modular. Consider the unitary matrix 
\begin{equation}
U_{\ztrip}:=\left(\zone\oplus\ze\oplus\ze\oplus\ztwo\right)S=\begin{pmatrix}\zone S_{1} & 0 & \zone t_{1} & 0\\
\ze t'_{1} & 0 & \ze r_{1} & 0\\
0 & \ze r_{2} & 0 & \ze t'_{2}\\
0 & \ztwo t_{2} & 0 & \ztwo S_{2}
\end{pmatrix}.
\end{equation}
It is apparent that $U_{\ztrip}$ has a special block structure, but
it is not block diagonal. The following definitions are motivated
by the infinite leads scattering matrix which is used in \cite{BanBerSmi_ahp12}.
\begin{defn}
\label{def: D,S} Define $D_{i}\left(\zi\right):=\zi S_{i}-1$ and
$\mathcal{S}_{i}\left(\zi\right):=r_{i}-t'_{i}\left(D_{i}\left(\zi\right)\right)^{+}\zi t_{i}$
for $i\in\left\{ 1,2\right\} $. Where $\left(D_{i}\left(\zi\right)\right)^{+}$
is the Moore-Penrose inverse and is equal to $D_{i}\left(\zi\right)^{-1}$
whenever $D_{i}\left(\zi\right)$ is invertible.
\end{defn}

The definition of the Moore-Penrose inverse can be found in \cite{genSchurComp},
and we will only use its following property:
\begin{lem}
\cite{genSchurComp}\label{lem: Schur decom}~Given a matrix $A$
we denote its Moore-Penrose inverse by $A^{+}$. It satisfies the
property that $1-A^{+}A$ is the orthogonal projection on $\ker\left(A\right)$
and $A^{+}A$ is the orthogonal projection on $\left(\ker\left(A\right)\right)^{\perp}$. 
\end{lem}

\begin{rem}
Although it was not mentioned by name, the Moore-Penrose inverse is
being used in \cite{BanBerSmi_ahp12}. 
\end{rem}

\begin{lem}
\label{lem: D-1}\cite{AloBanBer_cmp18,BanBerSmi_ahp12} For both
$i\in\left\{ 1,2\right\} $, $t_{i}'\in\left(\ker\left(D_{i}\left(\zi\right)\right)\right)^{\perp}$
and if $a_{1}\in\ker\left(D_{i}\left(\zi\right)\right)$, then (in
block structure) $\boldsymbol{a}=\begin{pmatrix}a_{1}\\
0\\
0\\
0
\end{pmatrix}\in\ker\left(1-U_{\ztrip}\right)$.
\end{lem}

\begin{proof}
Let $a_{1}\in\ker D_{1}\left(\zone\right)$ and let $\boldsymbol{a}=\begin{pmatrix}a_{1}\\
0\\
0\\
0
\end{pmatrix}$. Since $D_{1}\left(\zone\right)a_{1}=0$, then $\zone S_{1}a_{1}=a_{1}$
and therefore, using the block structure of $U_{\ztrip}$, 
\[
U_{\ztrip}\boldsymbol{a}=\begin{pmatrix}\zone S_{1}a_{1}\\
\ze t'_{1}a_{1}\\
0\\
0
\end{pmatrix}=\begin{pmatrix}a_{1}\\
\ze t'_{1}a_{1}\\
0\\
0
\end{pmatrix}.
\]
which means that $\norm{U_{\ztrip}\boldsymbol{a}}^{2}=\norm{a_{1}}^{2}+\norm{\ze t'_{1}a_{1}}^{2}$.
Since $U_{\ztrip}$ is unitary, then $\norm{U_{\ztrip}\boldsymbol{a}}^{2}=\norm{\boldsymbol{a}}^{2}=\norm{a_{1}}^{2}$
and therefore $\ze t'_{1}a_{1}=0$. Since $z_{e}\ne0$ (uni-modular)
then it follows that $a_{1}$, and hence $\ker D_{1}\left(\zone\right)$,
is orthogonal to $t_{1}'$. In particular it follows that $U_{\ztrip}\boldsymbol{a}=\boldsymbol{a}$
and therefore $\boldsymbol{a}\in\ker\left(1-U_{\ztrip}\right)$.
\end{proof}
\begin{lem}
\label{lem: Technical  bridge decomp} \cite{AloBanBer_cmp18}The
determinant of $1-U_{\ztrip}$ can be decomposed as follows:
\begin{equation}
\det\left(1-U_{\ztrip}\right)=\det D_{1}\left(\zone\right)\cdot\det D_{2}\left(\ztwo\right)\cdot\left(1-z_{e}^{2}\mathcal{S}\left(\zone\right)\mathcal{S}\left(\ztwo\right)\right).\label{eq: general decompoisition}
\end{equation}
Moreover, let \[M:=\begin{pmatrix}D_{1}\left(\zone\right) & 0 & \zone t_{1} & 0\\
0 & -1 & \ze\mathcal{S}_{1}\left(\zone\right) & 0\\
0 & \ze\mathcal{S}_{2}\left(\ztwo\right) & -1 & 0\\
0 & \ztwo t_{2} & 0 & D_{2}\left(\ztwo\right)
\end{pmatrix},\] then 
\begin{equation}
\ker\left(1-U_{\ztrip}\right)=\ker M.\label{eq: ker identity schur}
\end{equation}
\end{lem}

\begin{proof}
Let $L_{1}=\begin{pmatrix}1 & 0 & 0 & 0\\
-\ze t'_{1}D_{1}\left(\zone\right)^{+} & 1 & 0 & 0\\
0 & 0 & 1 & 0\\
0 & 0 & 0 & 1
\end{pmatrix}$ and let $L_{2}=\begin{pmatrix}1 & 0 & 0 & 0\\
0 & 1 & 0 & 0\\
0 & 0 & 1 & -\ze t'_{2}D_{2}\left(\ztwo\right)^{+}\\
0 & 0 & 0 & 1
\end{pmatrix}$ so that $L_{1}L_{2}\left(U_{\ztrip}-1\right)$ is equal to 
\[
\begin{pmatrix}D_{1}\left(\zone\right) & 0 & \zone t_{1} & 0\\
\ze t'_{1}\left(1-D_{1}\left(\zone\right)^{+}D_{1}\left(\zone\right)\right) & -1 & \ze\mathcal{S}_{1}\left(\zone\right) & 0\\
0 & \ze\mathcal{S}_{2}\left(\ztwo\right) & -1 & \ze t'_{2}\left(1-D_{2}\left(\ztwo\right)^{+}D_{2}\left(\ztwo\right)\right)\\
0 & \ztwo t_{2} & 0 & D_{2}\left(\ztwo\right)
\end{pmatrix}.
\]
According to Lemma \ref{lem: D-1}, $\ze t'_{i}\in\left(\ker D_{i}\left(\zi\right)\right)^{\perp}$
and by Lemma \ref{lem: Schur decom}, $\left(1-D_{i}\left(\zi\right)^{+}D_{i}\left(\zi\right)\right)$
is the orthogonal projection on $\ker D_{i}\left(\zi\right)$. Therefore
$\ze t'_{i}\left(1-D_{i}\left(\zi\right)^{+}D_{i}\left(\zi\right)\right)=0$
for both $i\in\left\{ 1,2\right\} $, and so $L_{1}L_{2}\left(U_{\ztrip}-1\right)=M$.
Notice that both $L_{1}$ and $L_{2}$ are invertible with $\det\left(L_{1}\right)=\det\left(L_{2}\right)=1$
so left multiplication by $L_{1}L_{2}$ does not change the (right)
kernel: 
\[
\ker\left(1-U_{\ztrip}\right)=\ker\left(U_{\ztrip}-1\right)=\ker L_{1}L_{2}\left(U_{\ztrip}-1\right)=\ker M,
\]
and $\det M=\det\left(U_{\ztrip}-1\right)$. The matrix $M=L_{1}L_{2}\left(U_{\ztrip}-1\right)$
has triangular block structure $M=\begin{pmatrix}A & 0\\
C & D
\end{pmatrix}$ with blocks: $A=\begin{pmatrix}D_{1}\left(\zone\right) & 0 & \zone t_{1}\\
0 & -1 & \ze\mathcal{S}_{1}\left(\zone\right)\\
0 & \ze\mathcal{S}_{2}\left(\ztwo\right) & -1
\end{pmatrix}$, $B=\begin{pmatrix}0\\
0\\
0
\end{pmatrix},\,C=\begin{pmatrix}0 & \ztwo t_{2} & 0\end{pmatrix}$ and $D=D_{2}\left(\ztwo\right)$, so $\det\left(M\right)=\det\left(A\right)\det\left(D\right).$
Namely,
\[
det\left(U_{\ztrip}-1\right)=\det\left(D_{2}\left(\ztwo\right)\right)\det\begin{pmatrix}D_{1}\left(\zone\right) & 0 & \zone t_{1}\\
0 & -1 & \ze\mathcal{S}_{1}\left(\zone\right)\\
0 & \ze\mathcal{S}_{2}\left(\ztwo\right) & -1
\end{pmatrix}.
\]
Using the same argument, $A=\begin{pmatrix}D_{1}\left(\zone\right) & 0 & \zone t_{1}\\
0 & -1 & \ze\mathcal{S}_{1}\left(\zone\right)\\
0 & \ze\mathcal{S}_{2}\left(\ztwo\right) & -1
\end{pmatrix}$ has upper triangular blocks structure, so 
\[
\det\left(A\right)=\det\left(D_{1}\left(\zone\right)\right)\det\begin{pmatrix}-1 & \ze\mathcal{S}_{1}\left(\zone\right)\\
\ze\mathcal{S}_{2}\left(\ztwo\right) & -1
\end{pmatrix}.
\]
We may conclude that, 
\[
\det\left(U_{\ztrip}-1\right)=\det\left(D_{2}\left(\ztwo\right)\right)\det\left(D_{1}\left(\zone\right)\right)\left(1-z_{e}^{2}\mathcal{S}\left(\zone\right)\mathcal{S}\left(\ztwo\right)\right),
\]
and since $U_{\ztrip}-1$ is of even dimension, then $\det\left(U_{\ztrip}-1\right)=\det\left(1-U_{\ztrip}\right)$
and we are done.
\end{proof}
\begin{lem}
\label{lem: scattering phase}\cite{KosSch_jpa99,BanBerSmi_ahp12}
Both $\mathcal{S}\left(\zone\right)$ and $\mathcal{S}\left(\ztwo\right)$
are uni-modular and satisfy $\mathcal{S}\left(\overline{\zi}\right)=\overline{\mathcal{S}\left(\zi\right)}$
and each $\mathcal{S}\left(\zi\right)$ is analytic in the entries
of $\zi$ in the region where $D_{i}\left(\zi\right)\ne0$.
\end{lem}

\begin{proof}
If we consider the scattering system of $\Gamma_{1}$ with an infinite
lead attached instead of $e$, then both in \cite{BanBerSmi_ahp12}
Theorem 2.1(2) and in \cite{KosSch_jpa99} Theorem 3.3 it was shown
that $\mathcal{S}_{1}\left(\zone\right)=r_{1}-t'_{1}D_{1}\left(\zone\right)^{+}\zone t_{1}$\footnote{In \cite{BanBerSmi_ahp12} Theorem 2.1(2), they allow more freedom
in the choice of pseudo inverse to $D_{1}\left(\zone\right)$ but
it can be shown, using Lemma \ref{lem: D-1} that $D_{1}\left(\zone\right)^{+}$
is included in their possible choices.} is unitary (and one dimensional in our case) so it has magnitude
one. Same for $\mathcal{S}_{2}\left(\ztwo\right)$. Since $D_{i}\left(\zi\right)$
is linear in $\zi$ with real coefficients, then $D_{i}\left(\overline{\zi}\right)=\overline{D_{i}\left(\zi\right)}$
and the Moore-Penrose inverse commute with conjugation therefore $D_{i}\left(\overline{\zi}\right)^{+}=\overline{D_{i}\left(\zi\right)^{+}}$.
Therefore, as $r_{1},t'_{1}$ and $t_{1}$ are real, then $\mathcal{S}_{i}\left(\overline{\zi}\right)=r_{i}-t'_{i}D_{i}\left(\overline{\zi}\right)^{-1}\overline{\zi}t_{i}=\overline{\mathcal{S}_{i}\left(\zi\right)}$.
As for analyticity, clearly $r_{i}-t'_{i}D_{i}\left(\zi\right)^{-1}\zi t_{i}$
is a rational function in the entries of $\zi$ and its poles are
exactly when $\det D_{i}\left(\zi\right)=0$, so it is analytic in
the region of $\det D_{i}\left(\zi\right)\ne0$. 
\end{proof}

\newpage{}
\section{\label{sec: Equidistribution} \label{sec: app Equidistribution}Equidistribution
and the natural density}

In Lemma \ref{lem: equidistribution Jordan}, we state that if $\left\{ x_{n}\right\} _{n\in N}$
is equidistributed on $X$ (a compact metric space) with respect to
$m$ (a Borel regular measure) and $A\subset X$ is Jordan with respect
to $m$, then
\[
d\left(\set{n\in\N}{x_{n}\in A}\right)=m\left(A\right).
\]
The proof is the following standard approximation: 
\begin{proof}
This is a standard approximation argument. For every $\epsilon>0$
we can define an open set $O_{\epsilon}$ that contains the closure
$\overline{A}$ and a compact set $K_{\epsilon}$ contained in in
the interior $\mathrm{int}\left(A\right)$ such that both $m\left(O_{\epsilon}\setminus\overline{A}\right),m\left(\mathrm{int}\left(A\right)\setminus K_{\epsilon}\right)<\epsilon$.
We can define (using Urysohn's lemma) two continuous functions, $h_{\epsilon,+}$
and $h_{\epsilon,-}$ from $X$ to $\left[0,1\right]$ that bound
the indicator function of $A$
\[
h_{\epsilon,-}\le\chi_{A}\le h_{\epsilon,+},
\]
such that $h_{\epsilon,-}$ is supported inside $int\left(A\right)$
and $h_{\epsilon,-}|_{K_{\epsilon}}\equiv1$ and $h_{\epsilon,+}$
is supported inside $O_{\epsilon}$ and $h_{\epsilon,-}|_{\overline{A}}\equiv1$.
Since $A$ is Jordan, then \[m(\overline{A})=m(A)=m(\mathrm{int}\left(A\right)),\]
and so \[\int_{X}\left(h_{\epsilon,+}-\chi_{A}\right)dm\le m\left(O_{\epsilon}\setminus\overline{A}\right)<\epsilon, \ \text{and}\] \[\int_{X}\left(\chi_{A}-h_{\epsilon,-}\right)dm\le m\left(\mathrm{int}\left(A\right)\setminus K_{\epsilon}\right)<\epsilon.\]
Therefore, 
\begin{equation}
\int_{X}h_{\epsilon,+}dm-\epsilon\le m\left(A\right)\le\int_{X}h_{\epsilon,+}dm+\epsilon.\label{eq: measure difference}
\end{equation}
For a given $N$, the bound $h_{\epsilon,-}\le\chi_{A}\le h_{\epsilon,+}$
gives 
\[
\frac{\sum_{n=1}^{N}h_{\epsilon,-}\left(x_{n}\right)}{N}\le\frac{\left|\set{n\le N}{x_{n}\in A}\right|}{N}\le\frac{\sum_{n=1}^{N}h_{\epsilon,+}\left(x_{n}\right)}{N}
\]
As both $h_{\epsilon,+}$ and $h_{\epsilon,-}$ are continuous, then
by the equidistribution $\lim_{N\rightarrow\infty}\frac{\sum_{n=1}^{N}h_{\epsilon,\pm}\left(x_{n}\right)}{N}=\int_{X}h_{\epsilon,\pm}dm$
and so taking $N\rightarrow\infty$, the above gives
\begin{equation}
\int_{X}h_{\epsilon,-}dm\le\liminf_{N\rightarrow\infty}\frac{\left|\set{n\le N}{x_{n}\in A}\right|}{N}\le\limsup_{N\rightarrow\infty}\frac{\left|\set{n\le N}{x_{n}\in A}\right|}{N}\le\int_{X}h_{\epsilon,+}dm.\label{eq: density bounds}
\end{equation}
As $\epsilon\rightarrow0$, (\ref{eq: measure difference}) together
with (\ref{eq: density bounds}) gives 
\[
\liminf_{N\rightarrow\infty}\frac{\left|\set{n\le N}{x_{n}\in A}\right|}{N}=\limsup_{N\rightarrow\infty}\frac{\left|\set{n\le N}{x_{n}\in A}\right|}{N}=m\left(A\right).
\]
\end{proof}
A more general statement is that under the above conditions, the following
equality, 
\begin{equation}
\lim_{N\rightarrow\infty}\frac{\sum_{n=1}^{N}f\left(x_{n}\right)}{N}=\int_{X}fdm,
\end{equation}
that holds for continuous functions can be generalized to \emph{Riemann
integrable} functions. Where by\emph{ Riemann integrable, }we mean
functions whose set of discontinuity points is of measure zero (with
respect to $m$).

The method of proof in Theorems \ref{thm:First} and \ref{thm: statistics_of_local_observables-1},
in terms of equidistribution, was as follows. Given a finite Riemann
integrable step function on $X$, $f=\sum_{j=1}^{n}a_{j}\chi_{A_{j}}$
we push it forward to a function $\boldsymbol{f}:\N\rightarrow\mathrm{Image}\left(f\right)$
defined by $\boldsymbol{f}\left(n\right):=f\left(x_{n}\right)$, using
the equidistributed sequence. In such case, as we showed, the level
sets of $\boldsymbol{f}$ have density according to the measures of
the level sets of $f$. It follows that \textbf{$\boldsymbol{f}$}
is a random variable on $\N$ with the $\sigma$-algebra generated
by $\boldsymbol{f}$, say $\mathcal{F}_{\boldsymbol{f}}$, and the
natural density $d$.

It is only natural to ask whether this procedure can be generalized
to continuous functions, or even to any Riemann integrable function,
and the answer appears to be negative:
\begin{prop}
\label{prop: random variables no continuous}Let $X$ be a compact
metric space with Borel regular probability measure $m$ and let $\left\{ x_{n}\right\} _{n\in N}$
be a sequence equidistributed with respect to $m$.\\ Let $f:X\rightarrow\R$
be Riemann integrable function, and let $\boldsymbol{f}:\N\rightarrow\mathrm{Image}\left(f\right)$
defined by $\boldsymbol{f}\left(n\right):=f\left(x_{n}\right)$. Let
$\mathcal{F}_{\boldsymbol{f}}$ denote the $\sigma$-algebra generated
by $\boldsymbol{f}$. If the natural density, $d$, is a probability
measure on $\N$ with $\mathcal{F}_{\boldsymbol{f}}$, then $f$ is
a countable step function up to measure zero. Namely, there is countable
disjoint collection of Borel sets $\left\{ A_{n}\right\} _{n\in\N}$
with $\tilde{X}=\sqcup_{n=1}^{\infty}A_{n}$ of measure $m\left(\tilde{X}\right)=1$,
such that the restriction of $f$ to $\tilde{X}$ is 
\[
f|_{\tilde{X}}=\sum_{n=1}^{\infty}a_{n}\chi_{A_{n}},
\]
 for some real $a_{n}$'s.
\end{prop}

As an immediate corollary:
\begin{cor}
\label{cor: Not an RV}If $f$ is continuous and non-constant on some
open set, then $d$ is not a measure on $\N$ with $\mathcal{F}_{\boldsymbol{f}}$.
\end{cor}

Let us now prove Proposition \ref{prop: random variables no continuous}:
\begin{proof}
Denote the values of $\boldsymbol{f}$ by $t_{n}:=\boldsymbol{f}\left(n\right)$.
As there might be repetitions, let $J\subset\N$ be the set without
repetitions: 
\[
J:=\set{n\in\N}{\forall j<n\,\,t_{n}\ne t_{j}}.
\]
Define the index sets 
\[
\forall j\in J\,\,\,\boldsymbol{A}_{j}:=\set{n\in\N}{t_{n}=t_{j}}.
\]
Notice that $\mathcal{F}_{\boldsymbol{f}}$ is generated by these
$\boldsymbol{A}_{j}'s$ and as they are disjoint by definition, then
every set in $\mathcal{F}_{\boldsymbol{f}}$ is a countable union
of such $\boldsymbol{A}_{j}'s$. As $d$ was assumed to be a probability
measure on $\mathcal{F}_{\boldsymbol{f}}$, then every $\boldsymbol{A}_{j}$
has density and $\sum_{j\in J}d\left(\boldsymbol{A}_{j}\right)=1.$
We might restrict this sum to a smaller set $\tilde{J}\subset J$
defined by 
\[
\tilde{J}:=\set{j\in J}{d\left(\boldsymbol{A}_{j}\right)>0},
\]
such that 
\begin{equation}
\sum_{j\in\tilde{J}}d\left(\boldsymbol{A}_{j}\right)=1.
\end{equation}
In particular $\tilde{J}\ne\emptyset$. Let us define the corresponding
level sets of $f$ by 
\[
A_{j}:=f^{-1}\left(t_{j}\right)=\set{x\in X}{f\left(x\right)=t_{j}},
\]
and let 
\[
\tilde{X}:=\cup_{j\in\tilde{J}}A_{j}.
\]
Therefore, $A_{j}$ and $\boldsymbol{A}_{j}$ are related through
$\left\{ x_{n}\right\} _{n=1}^{\infty}$ as 
\begin{align*}
\boldsymbol{A}_{j} & =\set{n\in\N}{x_{n}\in A_{j}},\,\,\,\mathrm{and}\\
\left\{ x_{n}\right\} _{n=1}^{\infty} & \subset\tilde{X}.
\end{align*}
As the $A_{j}$'s are disjoint by definition, so $m\left(\tilde{X}\right)=\sum_{j\in\tilde{J}}m\left(A_{j}\right)$,
it is left to show that $m\left(\tilde{X}\right)=1$ to conclude the
proof. Since $m$ is a probability measure, it will be enough to prove
that $m\left(\tilde{X}\right)\ge1$. 

Let $j_{0}\in\tilde{J}$. Since $A_{j_{0}}$ is a level set, then the
points in $\overline{A_{j_{0}}}\setminus A_{j_{0}}$ are discontinuity
points of $f$. As $f$ is Riemann integrable, then $m\left(\overline{A_{j_{0}}}\setminus A_{j_{0}}\right)=0$,
namely $m\left(A_{j_{0}}\right)=m\left(\overline{A_{j_{0}}}\right).$
As $m$ is regular, then for any $\epsilon>0$ there exists an open
set $O_{\epsilon}$ such that $\overline{A_{j_{0}}}\subset O_{\epsilon}$
and $m\left(O_{\epsilon}\setminus\overline{A_{j_{0}}}\right)<\epsilon$.
Let $h_{\epsilon}:X\rightarrow\left[0,1\right]$ be a continuous function
supported inside $O_{\epsilon}$ and such that the restriction of
$h_{\epsilon}$ to $\overline{A_{j_{0}}}$ is constant $h_{\epsilon}|_{\overline{A_{j_{0}}}}\equiv1$.
Such a function exists by Uryson's Lemma. Then 
\[
0\le\chi_{\boldsymbol{A}_{j_{0}}}\le h_{\epsilon},
\]
and so 
\[
d\left(\boldsymbol{A}_{j_{0}}\right)=\lim_{N\rightarrow\infty}\frac{1}{N}\sum_{n\le N}\chi_{\boldsymbol{A}_{j_{0}}}\left(x_{n}\right)\le\lim_{N\rightarrow\infty}\frac{1}{N}\sum_{n\le N}h_{\epsilon}\left(x_{n}\right).
\]
As $h_{\epsilon}$ is continuous and $\left\{ x_{n}\right\} _{n\in\N}$
is equidistributed, then the RHS is equal to $\int_{X}h_{\epsilon}dm$.
Since $0\le h_{\epsilon}\le1$, is supported on $O_{\epsilon}$ and
$\int_{\overline{A_{j_{0}}}}h_{\epsilon}dm=m\left(\overline{A_{j_{0}}}\right)=m\left(A_{j_{0}}\right)$
then, 
\[
d\left(\boldsymbol{A}_{j_{0}}\right)\le\int_{X}h_{\epsilon}dm=\int_{\overline{A_{j_{0}}}}h_{\epsilon}dm+\int_{O_{\epsilon}\setminus\overline{A_{j_{0}}}}h_{\epsilon}dm<m\left(A_{j_{0}}\right)+\epsilon.
\]
Taking $\epsilon\rightarrow0$ we get that 
\[
d\left(\boldsymbol{A}_{j_{0}}\right)\le m\left(A_{j_{0}}\right).
\]
As this is true for any $j\in\tilde{J}$, and $\sum_{j\in\tilde{J}}d\left(\boldsymbol{A}_{j}\right)=1$,
then 
\[
m\left(\tilde{X}\right)=\sum_{j\in\tilde{J}}m\left(A_{j}\right)\ge\sum_{j\in\tilde{J}}d\left(\boldsymbol{A}_{j}\right)=1.
\]
Therefore, $m\left(\tilde{X}\right)=1$ and we are done, as $\tilde{X}=\sqcup_{j\in\tilde{J}}A_{j}$
and $f|_{\tilde{X}}=\sum_{j\in\tilde{J}}t_{j}\chi_{A_{j}}$.
\end{proof}

\newpage{}
\section{\label{sec: App Gluing and contracting}Gluing vertices and contracting
edges}

There are many results on operations on quantum graphs such as gluing
vertices and contracting edges, but we will not present them here.
For example, gluing vertices among other operations called ``surgery''
operations, can be found in \cite{berkolaiko2019surgery}, a recent
work which combines all of these results into a ``surgeons toolkit''
as the authors call it. Another important work is \cite{BerLatSuk19}
where the limit of operators on quantum graphs with shrinking edges
(the continuous contraction of an edge) is analyzed using tools from
symplectic geometry.

We will only provide the very specific results needed for this section
and for completeness we will prove them.

We describe the gluing process as follows. Let $\Gamma$ be a graph
with boundary, and consider a boundary vertex $v\in\partial\Gamma$
and an interior vertex $u\in\V_{in}$. We define $\tilde{\Gamma}$,
the graph obtained by gluing $v$ and $u$, using a new vertex $\tilde{v}$
such that $\E_{\tilde{v}}=\E_{u}\cup\E_{v}$. Notice that $\deg{\tilde{v}}=\deg u+1>3$
since $u\in\V_{in}$. If $\Gamma_{\lv}$ is a standard graph, then
$\tilde{\Gamma}_{\lv}$ has the same edge lengths, and we consider
the two as having the same function spaces, but with different vertex
conditions. Let $e$ be the edge connected to $v$. The vertex conditions
on $v$ and $u$ are 
\begin{align}
\partial_{e}f\left(v\right) & =0,\\
\forall e',e''\in\E_{u}\,\,\,f|_{e'}\left(u\right) & =f|_{e''}\left(u\right),\,\,\mathrm{and}\\
\sum_{e'\in\E_{u}}\partial_{e'}f\left(u\right) & =0.
\end{align}
While the vertex conditions on $\tilde{v}$ can be written as 
\begin{align}
\forall e'\in\E_{u}\,\,\,f|_{e'}\left(u\right)= & f|_{e}\left(v\right),\,\,\mathrm{and}\\
\sum_{e'\in\E_{u}}\partial_{e'}f\left(u\right)= & -\partial_{e}f\left(v\right).
\end{align}
So it is clear that 
\begin{align}
Eig\left(\Gamma_{\lv},k^{2}\right)\cap Eig\left(\tilde{\Gamma}_{\lv},k^{2}\right)= & \set{f\in Eig\left(\Gamma_{\lv},k^{2}\right)}{f\left(u\right)=f\left(v\right)}\\
= & \set{f\in Eig\left(\tilde{\Gamma}_{\lv},k^{2}\right)}{\partial_{e}f\left(\tilde{v}\right)=0}.
\end{align}
The following is immediate:
\begin{lem}
\label{lem: gluing vertices}Let $\Gamma$ be a graph with $v\in\partial\Gamma$,
$e$ connected to $v$ and $u\in\V_{in}$. Let $\tilde{\Gamma}$ be
the graph obtained by gluing $v$ and $u$, denoting the new vertex
$\tilde{v}$. Then, 
\begin{align*}
\Sigma^{reg}\cap\tilde{\Sigma} & =\set{\kv\in\Sigma^{reg}}{f_{\kv}\left(v\right)=f_{\kv}\left(u\right)},\,\,a\mathrm{nd}\\
\Sigma\cap\tilde{\Sigma}^{reg} & =\set{\kv\in\tilde{\Sigma}^{reg}}{\partial_{e}f_{\kv}\left(\tilde{v}\right)=0}.
\end{align*}
Where the decorated sets relate to $\tilde{\Gamma}$ and the non decorated
to $\Gamma$. 
\end{lem}

\begin{rem}
Although the canonical eigenfunctions should be decorated to resolve
ambiguity, we will not do that unless it is not clear from the context
to which regular part they belong. 
\end{rem}

Another situation we need to consider is the gluing of a boundary
vertex to an interior point. This can only be done on a metric graph
(unless we enforce a degree two vertex). Let $\Gamma_{\lv}$ be a
standard graph with $v\in\partial\Gamma$, $e$ connected to $v$
and let $e'$ be an edge of length $l_{e'}$, with arc-length parameterization
$x_{e'}\in\left[0,l_{e'}\right]$. Let $u$ be an interior point located
at $x_{e'}=l_{1}$ so it partition $e'$ into two edge $e_{1}$ and
$e_{2}$ of lengths $l_{1}$ and $l_{2}=l_{e'}-l_{1}$. We define
$\tilde{\Gamma}_{\tilde{\lv}}$ as the graph obtained by gluing $v$
to $u$, denoting the new vertex by $\tilde{v}$ so that $\deg{\tilde{v}}=3$
with $e,e_{1},e_{2}$ attached to it. In this case, $\tilde{\Gamma}$
has $E+1$ edges, and therefore different edge lengths. We denote
$\tilde{\lv}=\left(l_{e},l_{1},l_{2},...\right)$ and $\lv=\left(l_{e},l_{1}+l_{2},...\right)$.
In such case, a similar analysis of vertex conditions gives:
\begin{lem}
\label{lem: deg3 construction}Let $\Gamma$ and $\tilde{\Gamma}$
as above, and define $T\left(\kappa_{e},\kappa_{1},\kappa_{2},...\right)=\left(\kappa_{e},\kappa_{1}+\kappa_{2},...\right)$.
Then,
\end{lem}

\begin{enumerate}
\item If $\tilde{\kv}=\left(\kappa_{e},\kappa_{1},\kappa_{2},...\right)\in\tilde{\Sigma}^{reg}$
and $T\left(\tilde{\kv}\right)\in\Sigma^{reg}$, then $f_{\kv}$ and
$f_{\tilde{\kv}}$ agree on their joint vertices, and if $t\in\left[0,2\pi\right]$
is such that $\left\{ t\right\} =\kappa_{1}$ (namely, $x_{e'}=t$
is the gluing point), then 
\[
f_{\kv}|_{e'}\left(t\right)=f_{\kv}\left(v\right)=f_{\tilde{\kv}}\left(\tilde{v}\right).
\]
\item If $\tilde{\kv}=\left(\kappa_{e},\kappa_{1},\kappa_{2},...\right)\in\tilde{\Sigma}^{reg}$
then $T\left(\tilde{\kv}\right)\in\Sigma$ if and only if $\partial_{e}f_{\tilde{\kv}}\left(\tilde{v}\right)=0$. 
\end{enumerate}
We may now discuss edge contraction. Let $\Gamma$ be a graph with
an edge $e$, which is not a loop, connecting two distinct vertices
$v_{1}$ and $v_{2}$. We define $\Gamma'$, the graph obtained by
contracting the edge $e$, as follows. We remove $e$ from $\Gamma$
and identify $v_{1}$ and $v_{2}$, denoting the new vertex by $v$.
Thus $\Gamma'$ has $V'=V-1$ vertices and $E'=E-1$ edges and the
same first Betti number. If $\Gamma_{\lv}$ is a standard graph, we
denote the new edge lengths (removing $e$) by $\lv_{\overline{e}}$
such that $\Gamma'_{\lv_{\overline{e}}}$ is a standard graph, and
we consider the restriction of functions $f\mapsto f|_{\Gamma'}$
as a linear map from $L^{2}\left(\Gamma_{\lv}\right)$ to $L^{2}\left(\Gamma'_{\lv_{\overline{e}}}\right)$.
\begin{lem}
\label{lem: contracting lemma-1}Let $\Gamma,e$ and $\Gamma'$ as
the above such that $e$ is not a loop. Consider the decomposition
$\kv=\left(\kappa_{e},\kv_{\overline{e}}\right)\in\T^{\E}=\T\times\T^{\E'}$
such that $\T^{\E'}$ is the characteristic torus of $\Gamma'$. Let
$\kv_{\overline{e}}\in\T^{\E'}$ and let $\kv=\left(2\pi,\kv_{\overline{e}}\right)\in\T^{\E}$,
then the restriction $f\mapsto f|_{\Gamma'}$ is a linear bijection
between $Eig\left(\Gamma_{\kv},1\right)$ and $Eig\left(\Gamma'_{\kv_{\overline{e}}},1\right)$
that preserve the values of functions on vertices. 
\end{lem}

\begin{rem}
In particular, if $\kv=\left(2\pi,\kv_{\overline{e}}\right)\in\Sigma^{reg}\iff\kv_{\overline{e}}\in\Sigma'^{reg}$,
in which case $\tr\left(f_{\kv}\right)$ restricted to $\Gamma'$
is equal to $\mathrm{\tr}\left(f_{\kv_{\overline{e}}}\right)$.
\end{rem}

\begin{proof}
Let $\Gamma,e$ and $\Gamma'$ be as above and denote the vertices
of $e$ by $v_{1},v_{2}$ and the identified vertex in $\Gamma'$
by $v$. Let $\kv_{\overline{e}}\in\T^{\E\setminus e}$ and let $\kv=\left(2\pi,\kv_{\overline{e}}\right)\in\T^{\E}$.
Let $f\in Eig\left(\Gamma_{\kv},1\right)$. Since $f|_{e}$ is $2\pi$
periodic, and $e$ has length $2\pi$, then 
\begin{align}
f\left(v_{1}\right) & =f\left(v_{2}\right)\label{eq: equal on vertices-1}\\
\partial_{e}f\left(v_{1}\right) & =-\partial_{e}f\left(v_{2}\right).\label{eq: derivatives on vertices-1}
\end{align}
To show that $f|_{\Gamma'}\in Eig\left(\Gamma'_{\kv_{\overline{e}}},1\right)$,
it is enough to show that if it satisfies Neumann condition on $v$.
The continuity at $v$ follows from the continuity of $f$ at $v_{1},v_{2}$
and (\ref{eq: equal on vertices-1}). Using the notation of $\E'_{v}$
as the edges in $\Gamma'$ connected to $v$ and since the construction
gives that $\E'_{v}=\E_{v_{1}}\cup\E_{v_{2}}\setminus\left\{ e\right\} $
it follows that 
\[
\sum_{e'\in\E_{v}}\partial_{e'}f|_{\Gamma'}\left(v\right)=\sum_{e'\in\E_{v_{1}}\setminus\left\{ e\right\} }\partial_{e'}f\left(v_{1}\right)+\sum_{e'\in\E_{v_{2}}\setminus\left\{ e\right\} }\partial_{e'}f\left(v_{2}\right).
\]
The Neumann conditions on $v_{1}$ and $v_{2}$ implies that $\sum_{e'\in\E_{v_{1}}\setminus\left\{ e\right\} }\partial_{e'}f\left(v_{1}\right)=-\partial_{e}f\left(v_{1}\right)$
and $\sum_{e'\in\E_{v_{2}}\setminus\left\{ e\right\} }\partial_{e'}f\left(v_{2}\right)=-\partial_{e}f\left(v_{2}\right)$.
Together with (\ref{eq: derivatives on vertices-1}), it follows that
\[
\sum_{e'\in\E_{v}}\partial_{e'}f|_{\Gamma'}\left(v\right)=-\partial_{e}f\left(v_{1}\right)-\partial_{e}f\left(v_{2}\right)=0.
\]
This proves that $f|_{\Gamma'}\in Eig\left(\Gamma'_{\kv_{\overline{e}}},1\right)$.
As the restriction map is linear, the map is injective if its kernel
is trivial. Let $f\in Eig\left(\Gamma_{\kv},1\right)$ such that $f|_{\Gamma'}\equiv0$.
Therefore $f$ is supported on $e$, but since $e$ is not a loop
then and eigenfunction cannot be supported on $e$ and therefore $f\equiv0$.
Hence the restriction map is injective. To show that it is onto let
$g\in Eig\left(\Gamma'_{\kv_{\overline{e}}},1\right)$ and denote
the two constants $A:=g\left(v\right)$ and $B=-\sum_{e'\in\E_{v_{1}}\setminus\left\{ e\right\} }\partial_{e'}g\left(v\right)$.
The Neumann condition of $g$ at $v$ implies that $B=\sum_{e'\in\E_{v_{2}}\setminus\left\{ e\right\} }\partial_{e'}g\left(v\right)$.
We may now define $f$ on $\Gamma$ by its restrictions 
\[
f|_{e'}\left(x_{e'}\right):=\begin{cases}
g|_{e'}\left(x_{e'}\right) & e'\ne e\\
A\cos\left(x_{e'}\right)+B\sin\left(x_{e'}\right) & e'=e
\end{cases},
\]
where for $e'=e$ we choose the coordinate such that $x_{e'}=0$ at
$v_{1}$ and $x_{e'}=2\pi$ at $v_{2}$. It follows that $f|_{\Gamma'}=g$
and it is left to prove that $f$ satisfies Neumann condition on $v_{1}$
and $v_{2}$. The continuity at $v_{1},v_{2}$ follows from $f|_{e}\left(0\right)=f|_{e}\left(2\pi\right)=A=g\left(v\right)$
and the continuity of $g$ at $v$. By the definition of $B$, the
derivatives of $f$ at $v_{1}$ satisfy 
\[
\sum_{e'\in\E_{v_{1}}\setminus\left\{ e\right\} }\partial_{e'}f\left(v_{1}\right)+\partial_{e}f\left(v_{1}\right)=\sum_{e'\in\E_{v_{1}}\setminus\left\{ e\right\} }\partial_{e'}g\left(v_{1}\right)+B=0,
\]
 and the derivatives at $v_{2}$ satisfy 
\[
\sum_{e'\in\E_{v_{2}}\setminus\left\{ e\right\} }\partial_{e'}f\left(v_{2}\right)+\partial_{e}f\left(v_{2}\right)=\sum_{e'\in\E_{v_{2}}\setminus\left\{ e\right\} }\partial_{e'}g\left(v_{2}\right)-B=0.
\]
This proves that $f\in Eig\left(\Gamma_{\kv},1\right)$ and therefore
the restriction map is onto. Clearly, restriction preserve vertex
values and derivatives, and we are done.
\end{proof}

\newpage{}
\section{\label{sec: Appendix examples}Secular manifolds for 3-edges graphs}

In this appendix we provide all examples of (allowed) graphs with
3 edges and their secular manifolds. The purpose of this appendix
is to provide some visual motivation for the definitions of different
parts of the secular manifold. For example, if a graph has loops, then
we emphasize $\Sigma_{\L}$ by a different color (blue). We also present
$Z_{0}$ alone for each graph with loops, so that the geometric meaning
of excluding loop eigenfunctions becomes apparent. 

There are 6 (allowed) graphs of 3 edges, when vertices of degree two
are prohibited:

\begin{figure}[H]
\includegraphics[width=0.35\paperwidth]{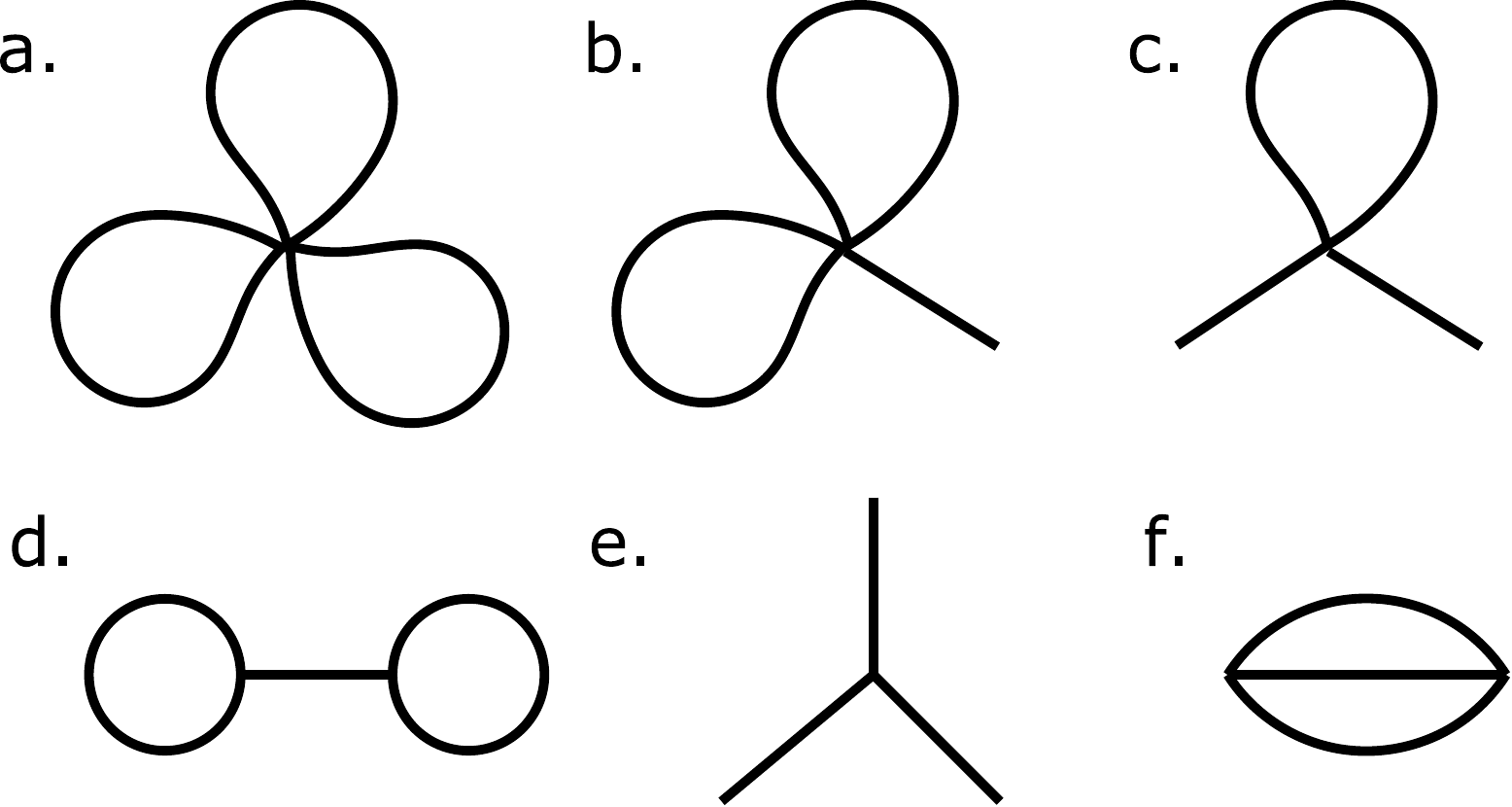}

\caption[All 3 edges graphs]{\label{fig: 3edges}Graphs of 3 edges. Their common names are:\protect \\
 a. 3-flower. b. $\left(2,1\right)$-stower. c. $\left(1,2\right)$-stower.
d. Dumbbell graph. e. 3-star. f. 3-mandarin. }
\end{figure}

For graphs a-d, which have loops, we will present both the secular
manifold $\Sigma$ and the main factor $Z_{0}$, with $\kv\in\left(-\pi,\pi\right)^{3}$
in order for the planes of $\Sigma_{\L}$ to be visible.

\begin{figure}[H]
\includegraphics[width=0.4\paperwidth]{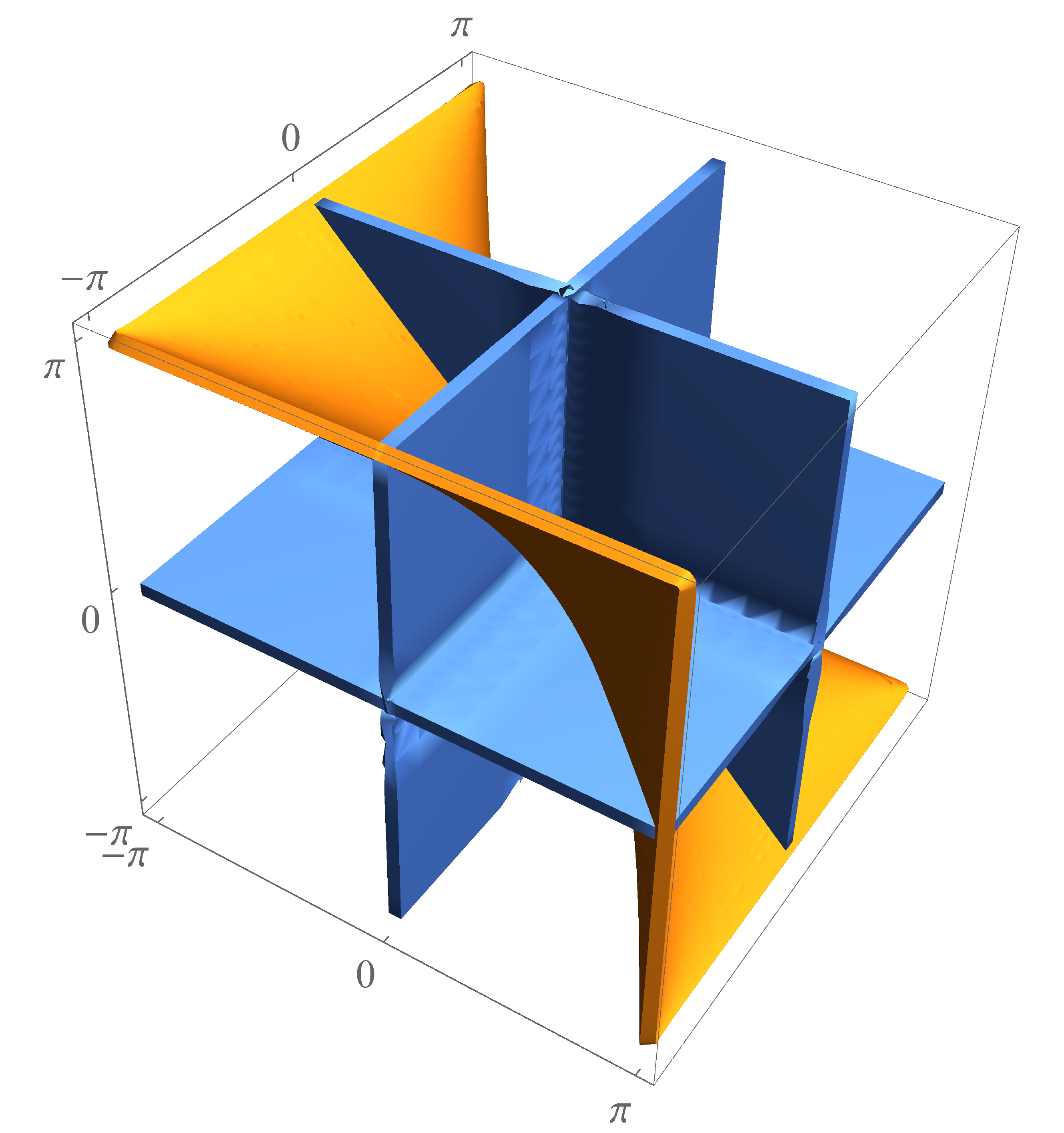}~~\includegraphics[width=0.4\paperwidth]{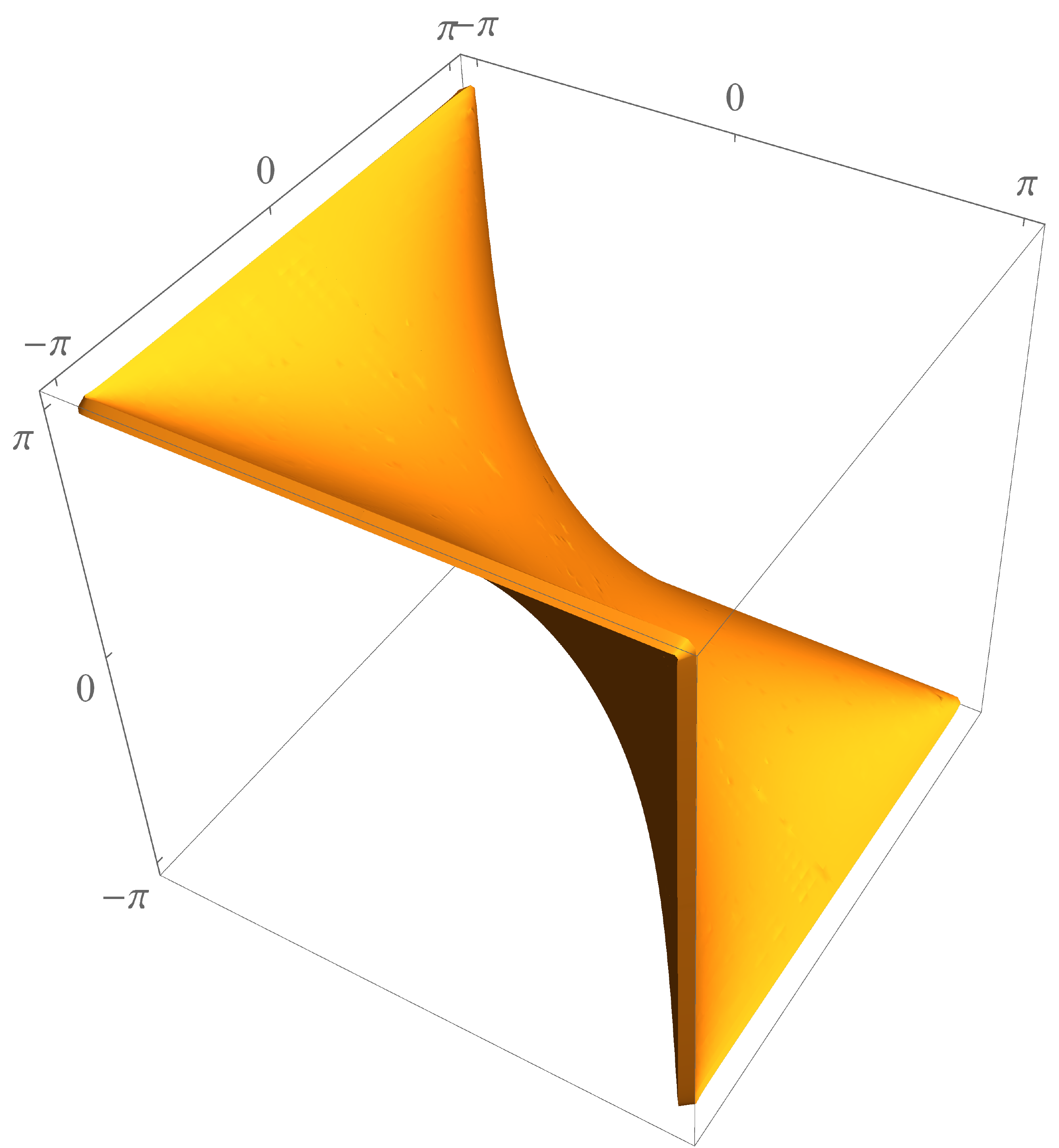}

\caption[ Secular manifold of 3-flower]{\label{fig: secman with loops 3flower} The secular manifold of the
3-flower. On the left, the secular manifold where $\Sigma_{\protect\L}$
is in blue and $Z_{0}$ is in orange. On the right, only $Z_{0}$
in orange.}
\end{figure}

\begin{figure}[H]
\includegraphics[width=0.4\paperwidth]{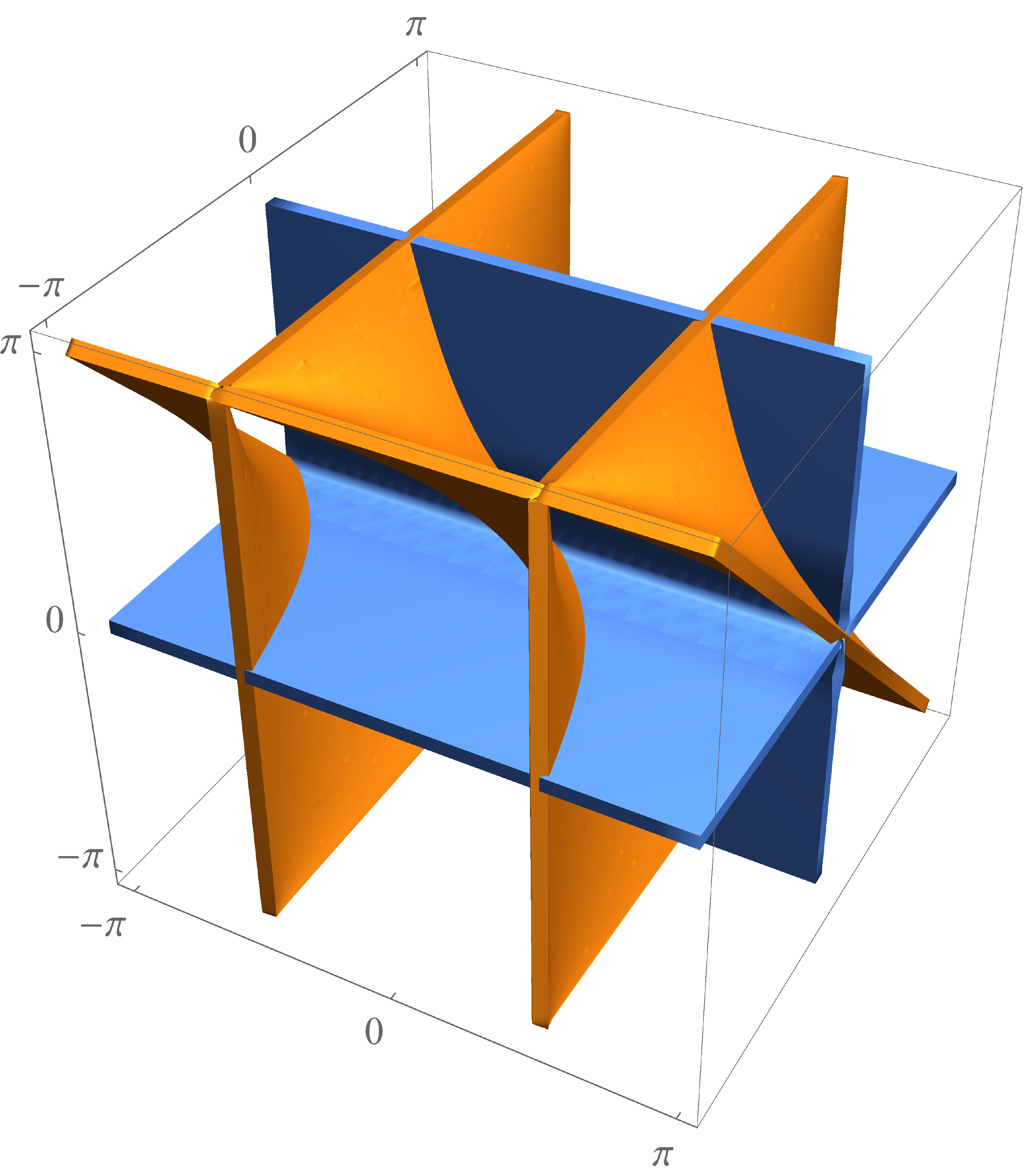}~~\includegraphics[width=0.4\paperwidth]{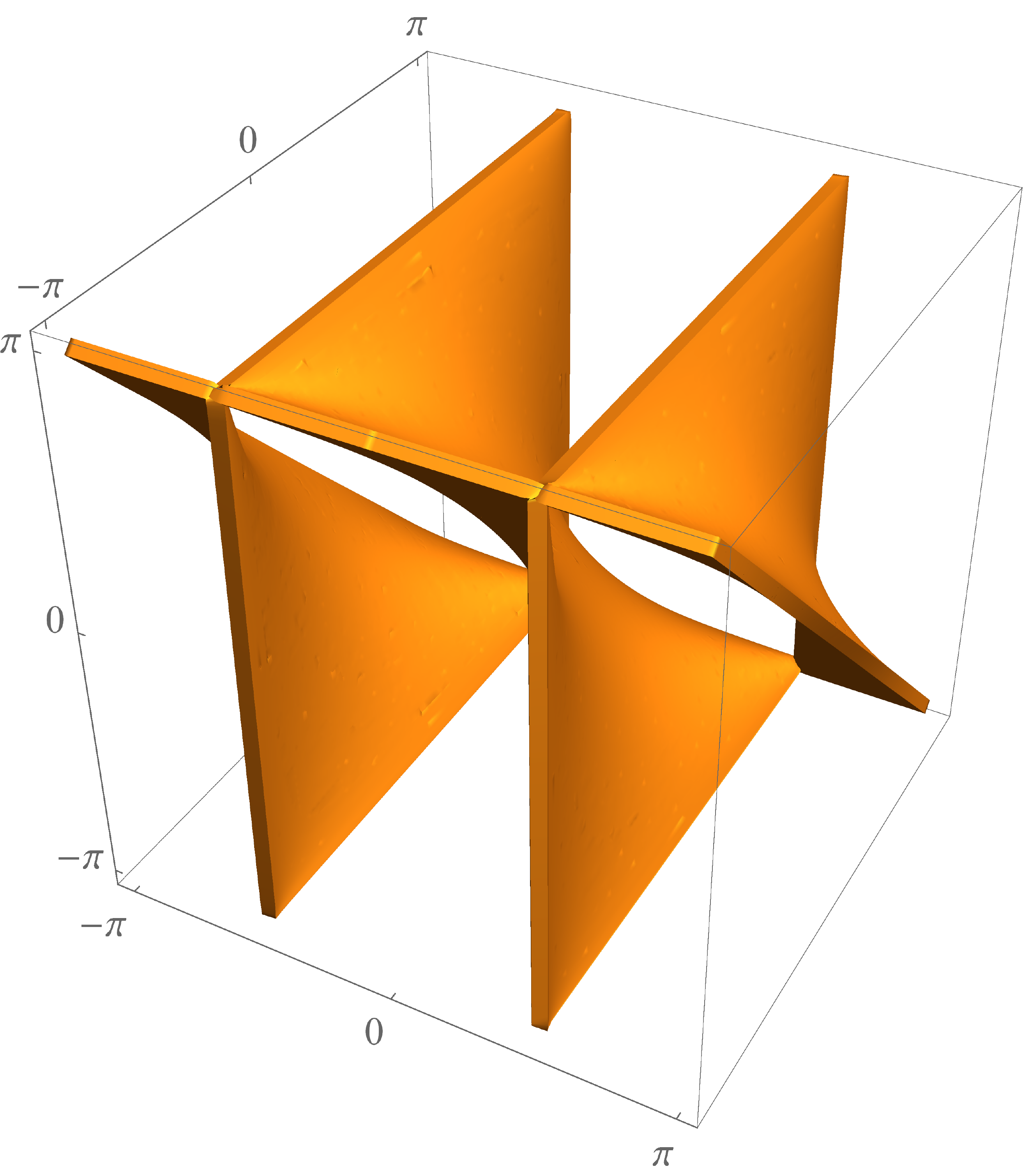}

\caption[ Secular manifold of (2,1)-stower]{\label{fig: secman with loops 2-1 stower} The secular manifold of
the $\left(2,1\right)$-stower. On the left, the secular manifold
where $\Sigma_{\protect\L}$ is in blue and $Z_{0}$ is in orange.
On the right, only $Z_{0}$ in orange.}
\end{figure}
\begin{figure}[H]
\includegraphics[width=0.4\paperwidth]{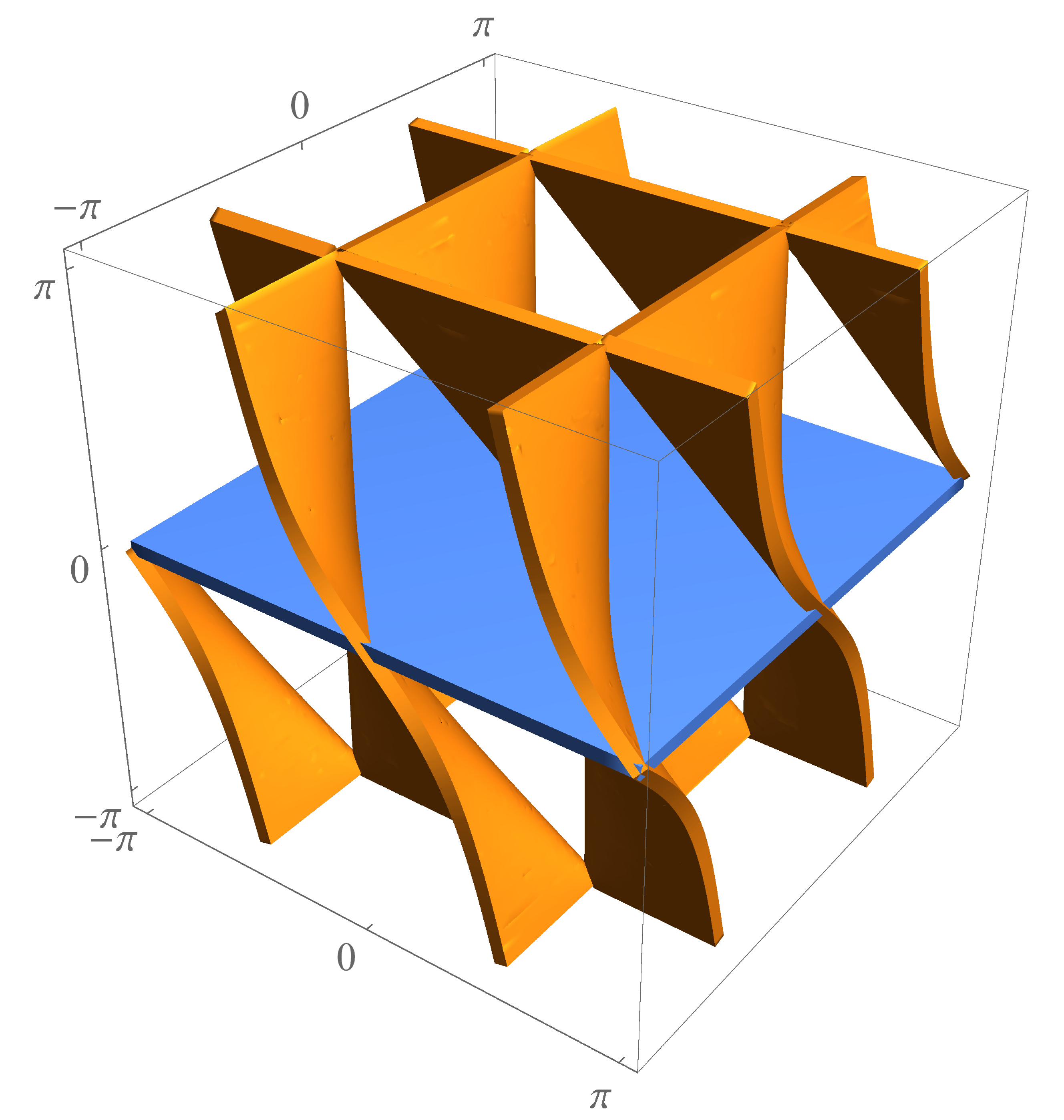}~~\includegraphics[width=0.4\paperwidth]{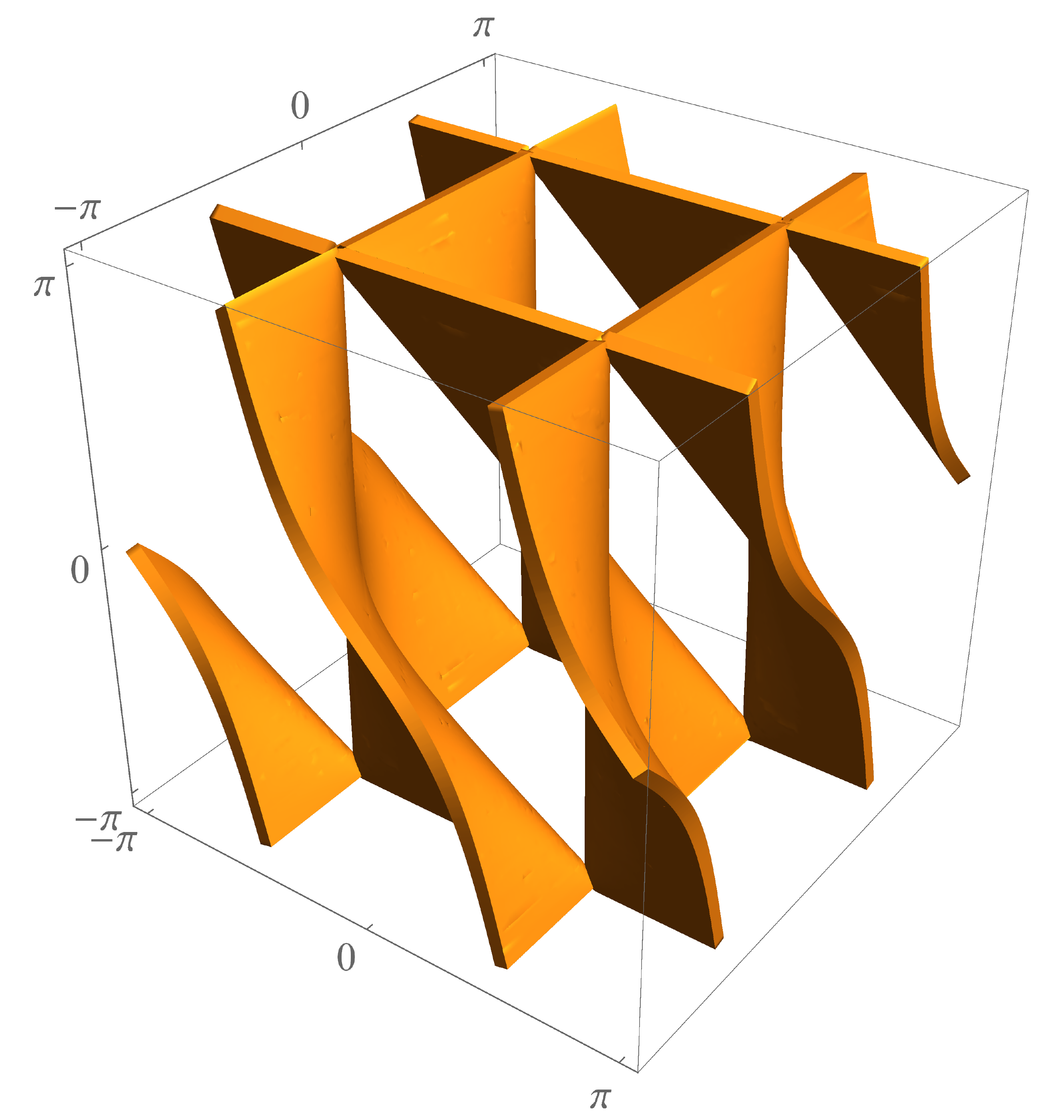}

\caption[ Secular manifold of (1,2)-stower]{\label{fig: secman with loops 12 stower} The secular manifold of
the $\left(1,2\right)$-stower. On the left, the secular manifold
where $\Sigma_{\protect\L}$ is in blue and $Z_{0}$ is in orange.
On the right, only $Z_{0}$ in orange.}
\end{figure}
\begin{figure}[H]
\includegraphics[width=0.4\paperwidth]{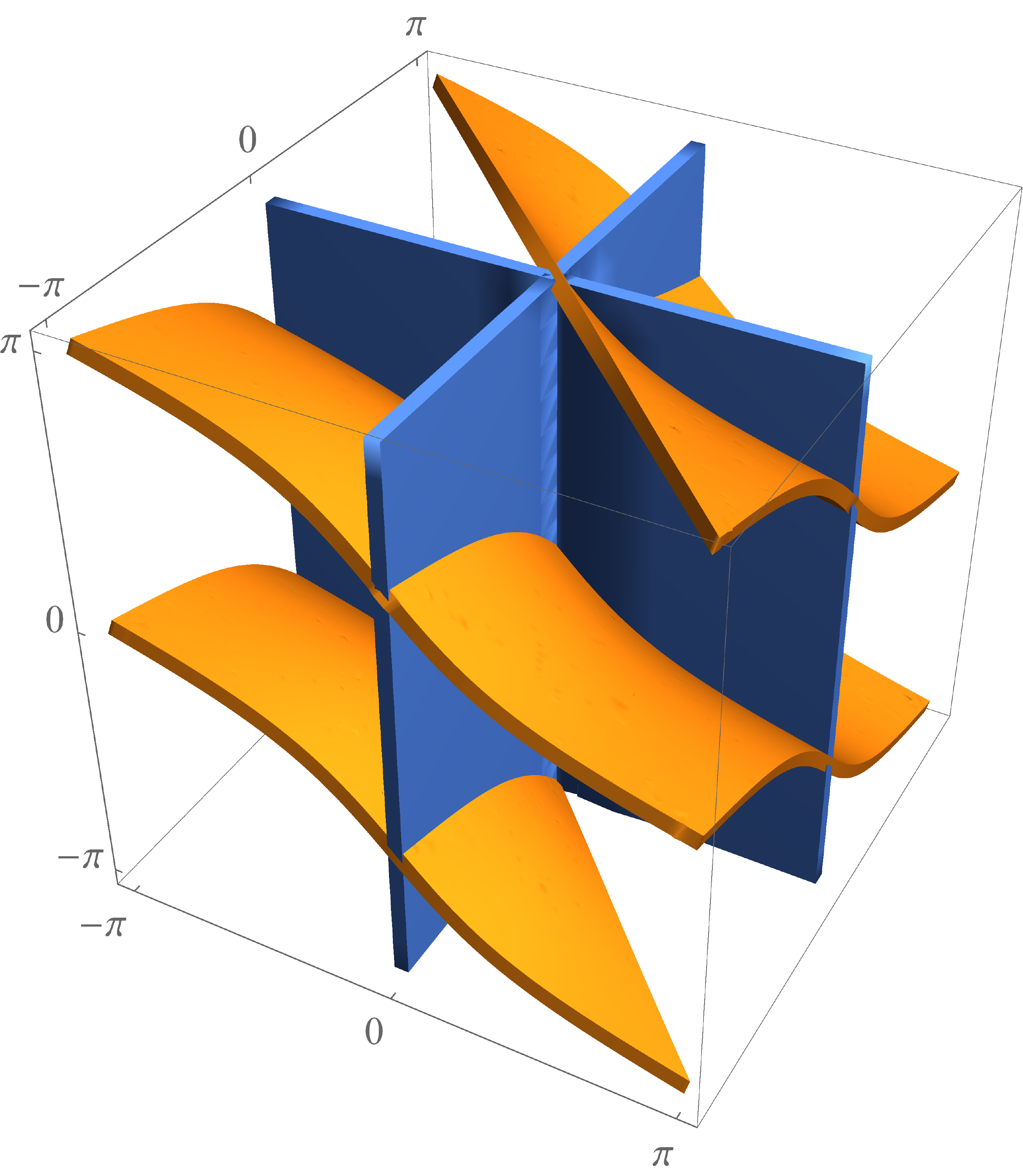}~~\includegraphics[width=0.4\paperwidth]{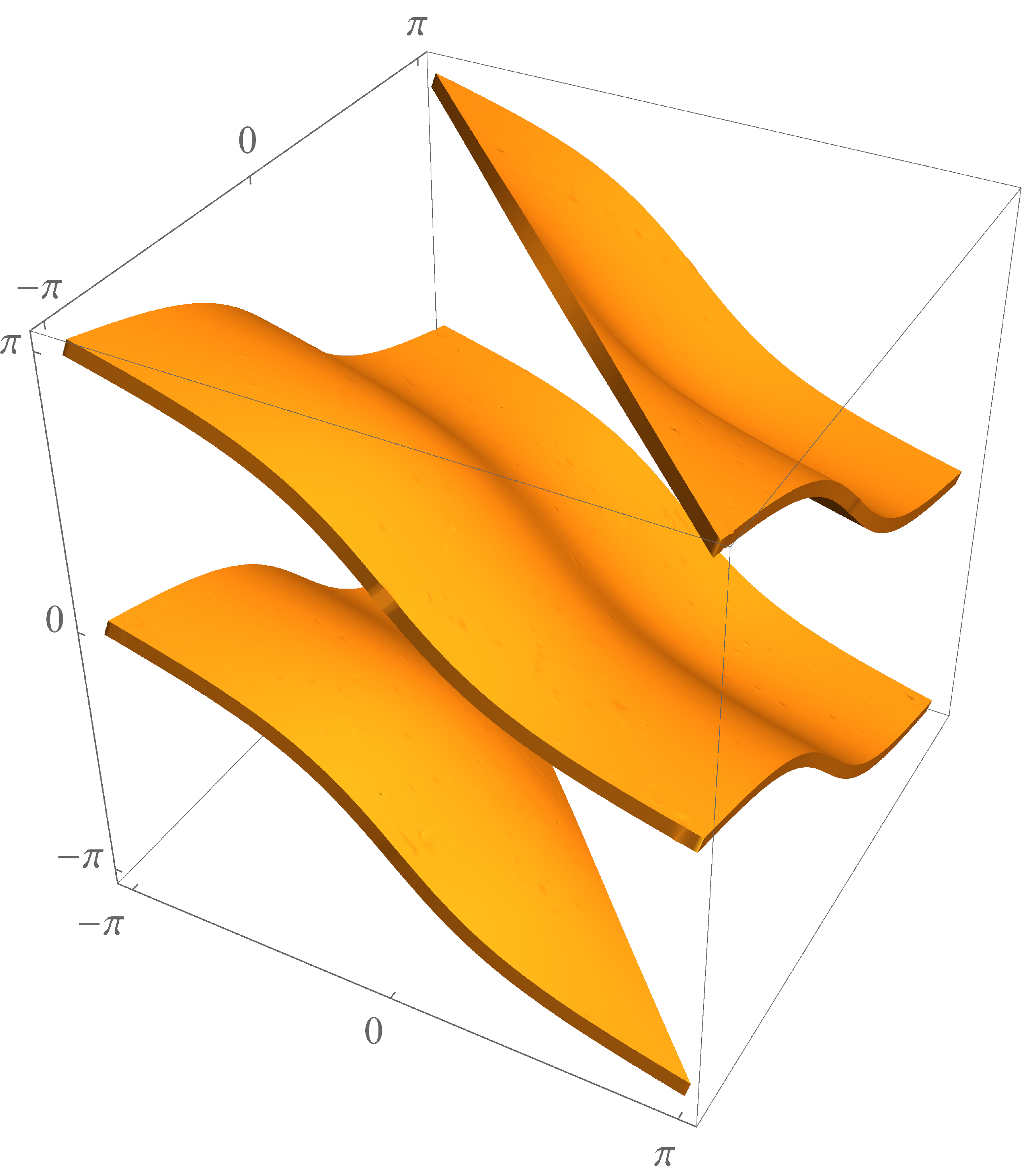}

\caption[ Secular manifold of a dumbbell graph]{\label{fig: secman with loops dumbbell} The secular manifold of the
dumbbell graph. On the left, the secular manifold where $\Sigma_{\protect\L}$
is in blue and $Z_{0}$ is in orange. On the right, only $Z_{0}$
in orange.}
\end{figure}

For graphs e and f which have no loops, we will present the secular
manifold $\Sigma$ with $\kv\in\left(0,2\pi\right)^{3}$:

\begin{figure}[H]
\includegraphics[width=0.4\paperwidth]{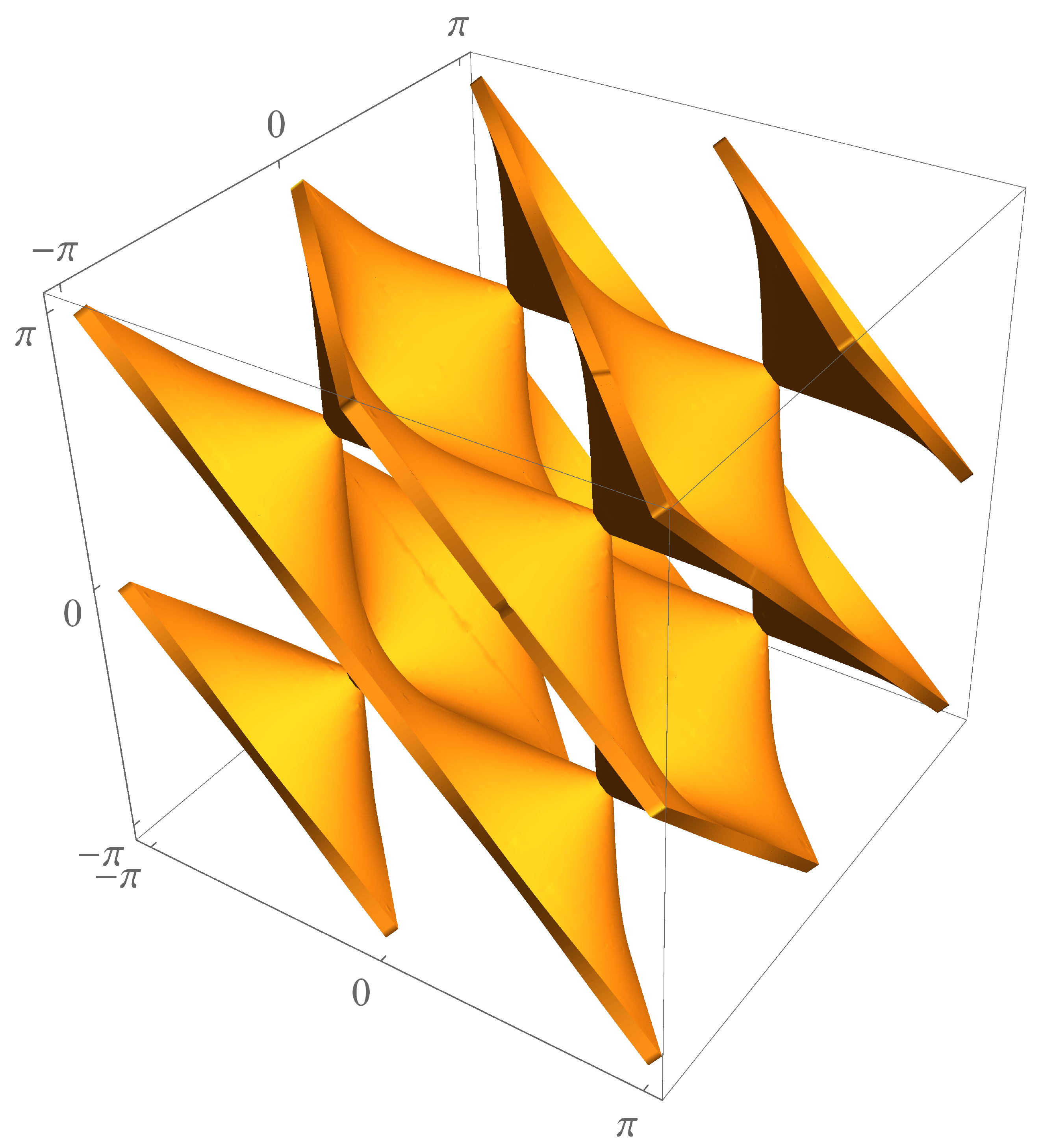}~~\includegraphics[width=0.4\paperwidth]{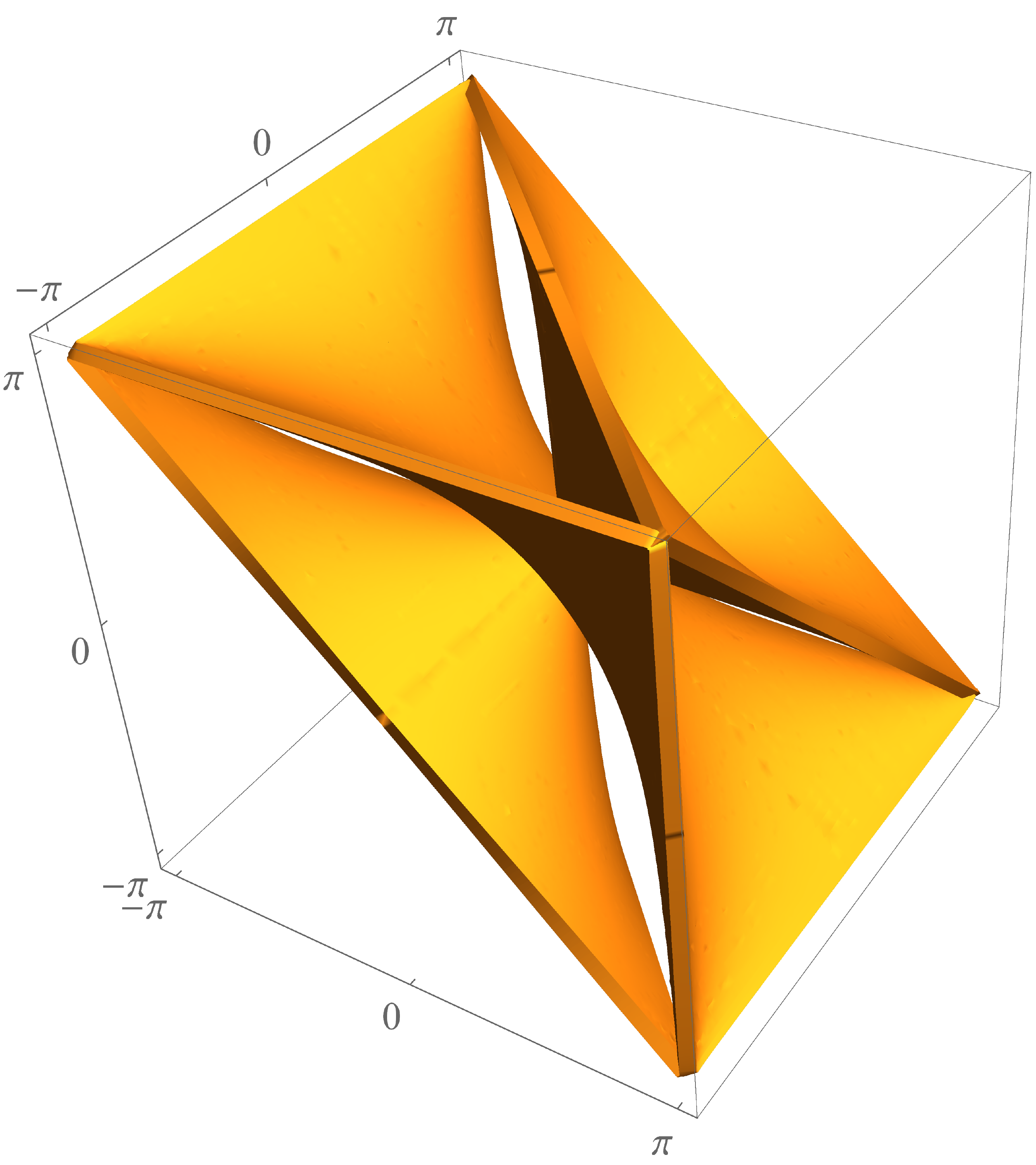}\caption[ Secular manifolds of 3-star and 3-mandarin]{\label{fig: secman without loops} On the left, the secular manifold
of the 3-star graph. On the right, the secular manifold of the 3-mandarin
graph.}
\end{figure}

One may notice that $\Sigma$ of graph f. contains $Z_{0}$ of graph
a. In fact, its second sheet is a translation by $\pi$ in each coordinate
of $Z_{0}$ of a. It can be explained in terms of the symmetric and
antisymmetric eigenfunctions on graph f. \cite{Alon}

\bibliographystyle{siam}
\bibliography{phdthesis}

\pagestyle{empty}
\end{document}